\documentclass[runningheads]{llncs}
\usepackage{amsmath}%
\usepackage{pifont}%
\usepackage{amssymb}%
\usepackage{stmaryrd}%
\usepackage{array}%
\usepackage{epic}%
\usepackage[fixamsmath]{mathtools}
\usepackage{cancel}

\usepackage{url}
\usepackage{bigstrut}
\usepackage{wrapfig}

\DeclareMathAlphabet\mathbb{U}{msb}{m}{n}
\DeclareMathAlphabet\mathbfcal{OMS}{cmsy}{b}{n}
\newtheorem{fact}{Fact}
\DeclareMathAlphabet{\mathrm}    {OT1}{cmr}{m}{n}
\DeclareMathAlphabet{\mathrmbf}  {OT1}{cmr}{bx}{n}
\DeclareMathAlphabet{\mathrmit}  {OT1}{cmr}{m}{it}
\DeclareMathAlphabet{\mathrmbfit}{OT1}{cmr}{bx}{it}
\DeclareMathAlphabet{\mathsf}    {OT1}{cmss}{m}{n}
\DeclareMathAlphabet{\mathsfbf}  {OT1}{cmss}{bx}{n}
\DeclareMathAlphabet{\mathsfit}  {OT1}{cmss}{m}{sl}
\DeclareMathAlphabet{\mathtt}    {OT1}{cmtt}{m}{n}
\DeclareMathAlphabet{\mathttbf}  {OT1}{cmtt}{bx}{n}
\DeclareMathAlphabet{\mathttit}  {OT1}{cmtt}{m}{it}
\DeclareMathAlphabet{\mathpzc}   {OT1}{pzc}{m}{it}
\newcommand{\keywords}[1]{\par\addvspace\baselineskip\noindent\enspace\ignorespaces{\bfseries Keywords:\,}#1}
\newcommand{\comment}[1]{}
\makeatletter
\def\dual#1{\expandafter\dual@aux#1\@nil}
\def\dual@aux#1/#2\@nil{\begin{tabular}{@{}c@{}}#1\\#2\end{tabular}}
\makeatother
\begin{document}

\pagestyle{headings}
\title{The {\ttfamily FOLE} Database} 
\titlerunning{The {\ttfamily FOLE} Database}  
\author{Robert E. Kent}
\institute{Ontologos}
\maketitle

\begin{abstract}

This paper continues the discussion 
of the representation 
and interpretation
of ontologies
in the first-order logical environment {\ttfamily FOLE}.
%(Kent~\cite{kent:iccs2013}).
%\begin{itemize}
%\item 
%Ontologies are represented and interpreted in
%%\item 
%(many-sorted) first-order logic.
%\end{itemize}
%
%%%%%%%%%%%%%%%%%%%%%%%%%%%%%%%%%%%%%%%%%%%%%%%%%%%%%%%%%%%%%%%%%%%%%%%%%%%%%%%%%%%%%%%%%%
%\newline
%\rule{290pt}{1pt}
%\newline
%%%%%%%%%%%%%%%%%%%%%%%%%%%%%%%%%%%%%%%%%%%%%%%%%%%%%%%%%%%%%%%%%%%%%%%%%%%%%%%%%%%%%%%%%%
%
%Five papers 
%We
%provide a rigorous mathematical representation for 
%the {\ttfamily ERA} (entity-relationship-attribute) data model 
%%(Chen \cite{chen:76}) 
%in particular,
%and ontologies in general,
%within the first-order logical environment {\ttfamily FOLE}.
%
The formalism and semantics of (many-sorted) first-order logic 
can be developed in both
a \emph{classification form}
and
an 
\emph{interpretation form}.
%
%
%\begin{description}
%\comment{
%\item[classification form:] 
Two papers,
``The {\ttfamily ERA} of {\ttfamily FOLE}: Foundation''
%Kent~\cite{kent:fole:era:found} 
and 
``The {\ttfamily ERA} of {\ttfamily FOLE}: Superstructure'',
%\cite{kent:fole:era:supstruc})
%, 
%which provide a \emph{foundation} and \emph{superstructure} for {\ttfamily FOLE},
represent 
the 
%formalism and semantics of (many-sorted) first-order logic 
%in a 
\emph{classification form},
corresponding to ideas discussed in the Information Flow Framework. 
%(IFF~\cite{iff}).
%}
%\item[interpretation form:] 
Two papers, 
``The {\ttfamily FOLE} Table''
%Kent \cite{kent:fole:era:tbl} 
and the current paper,
represent 
%the interpretation of 
%(many-sorted) first-order logic in an 
the \emph{interpretation form},
expanding on material found in the paper 
``Database Semantics''.
%(Kent~\cite{kent:db:sem}).
%
%\end{description}
%\end{itemize}
%
Although 
the classification form 
%(and \texttt{FOLE} itself)
follows the entity-relationship-attribute data model of Chen,
the interpretation form incorporates the relational data model of Codd.
%
%%%%%%%%%%%%%%%%%%%%%%%%%%%%%%%%%%%%%%%%%%%%%%%%%%%%%%%%%%%%%%%%%%%%%%%%%%%%%%%%%%%%%%%%%%
%\newline
%\rule{290pt}{1pt}
%\newline
%%%%%%%%%%%%%%%%%%%%%%%%%%%%%%%%%%%%%%%%%%%%%%%%%%%%%%%%%%%%%%%%%%%%%%%%%%%%%%%%%%%%%%%%%%
%
%Two further papers discuss the ``relational algebra'' 
%(Kent \cite{kent:fole:rel:ops})
%and the ``relational calculus''. 
%
%%%%%%%%%%%%%%%%%%%%%%%%%%%%%%%%%%%%%%%%%%%%%%%%%%%%%%%%%%%%%%%%%%%%%%%%%%%%%%%%%%%%%%%%%%
%\newline
%\rule{290pt}{1pt}
%\newline
%%%%%%%%%%%%%%%%%%%%%%%%%%%%%%%%%%%%%%%%%%%%%%%%%%%%%%%%%%%%%%%%%%%%%%%%%%%%%%%%%%%%%%%%%%
%
In general,
the {\ttfamily FOLE} representation uses a conceptual structures approach,
that is completely compatible with 
the theory of institutions (Goguen and Burstall),
%~\cite{goguen:burstall:92}), 
formal concept analysis (Ganter and Wille),
%~\cite{ganter:wille:99}), 
and information flow (Barwise and Seligman).
%~\cite{barwise:seligman:97}).
%
%\keywords{formula, constraint, interpretation, satisfaction, consequence.}
%\keywords{formula,sequent,constraint,interpretation,satisfaction,specification,entailment,consequence,logic,soundness,residuation.}
%\keywords{entity, attribute, relationship, schema, universe, structure.}
%\keywords{interpretation, specification, structure, logic, information system, interoperability.}
%\keywords{structure, specification, satisfaction, logic, interpretation.}
%\keywords{specification, structure, logic, information system, interoperability.}
\end{abstract}

\tableofcontents

%%%%%%%%%%%%%%%%%%%%%%%%%%%%%%%%%%%%%%%%%%%%%%%%%%%%%%%%%%%%%%%%%%%%%%%%%%%%%%%%%%%%%%%%%%
%%%%%%%%%%%%%%%%%%%%%%%%%%%%%%%%%%%%%%%%%%%%%%%%%%%%%%%%%%%%%%%%%%%%%%%%%%%%%%%%%%%%%%%%%%

%%\input{db-shape}
%%\input{db-sem}
%%%%%%%%%%%%%%%%%%%%%
%%\input{background}

%%%%%%%%%%%%%%%%%%%%%%%%%%%%%%%%%%%%%%%%%%%%%%%%%%%%%%%%%%%%%%%%%%%%%%%%%%%%%%%%%%%%%%%%%%
%%%%%%%%%%%%%%%%%%%%%%%%%%%%%%%%%%%%%%%%%%%%%%%%%%%%%%%%%%%%%%%%%%%%%%%%%%%%%%%%%%%%%%%%%%
%%%%%%%%%%%%%%%%%%%%%%%%%%%%%%%%%%%%%%%%%%%%%%%%%%%%%%%%%%%%%%%%%%%%%%%%%%%%%%%%%%%%%%%%%%
\newpage
\section{Introduction}\label{sec:intro}
%%%%%%%%%%%%%%%%%%%%%%%%%%%%%%%%%%%%%%%%%%%%%%%%%%%%%%%%%%%%%%%%%%%%%%%%%%%%%%%%%%%%%%%%%%
%%%%%%%%%%%%%%%%%%%%%%%%%%%%%%%%%%%%%%%%%%%%%%%%%%%%%%%%%%%%%%%%%%%%%%%%%%%%%%%%%%%%%%%%%%
%%%%%%%%%%%%%%%%%%%%%%%%%%%%%%%%%%%%%%%%%%%%%%%%%%%%%%%%%%%%%%%%%%%%%%%%%%%%%%%%%%%%%%%%%%

%%%%%%%%%%%%%%%%%%%%%%%%%%%%%%%%%%%%%%%%%%%%%%%%%%%%%%%%%%%%%%%%%%
%%%%%%%%%%%%%%%%%%%%%%%%%%%%%%%%%%%%%%%%%%%%%%%%%%%%%%%%%%%%%%%%%%
%\newpage
\subsection{Knowledge Representation}
%\label{subsec:rel:mod}
%\paragraph{The Relational Model.}
%%%%%%%%%%%%%%%%%%%%%%%%%%%%%%%%%%%%%%%%%%%%%%%%%%%%%%%%%%%%%%%%%
%%%%%%%%%%%%%%%%%%%%%%%%%%%%%%%%%%%%%%%%%%%%%%%%%%%%%%%%%%%%%%%%%

Many-sorted (multi-sorted) first-order predicate logic 
represents a community's ``universe of discourse'' as 
a heterogeneous collection of objects
by conceptually scaling
%partitioning 
the universe according to types.
The \emph{relational model} (Codd~\cite{codd:90}) 
%for database management 
%\cite{codd:90}
%uses a structure and language 
%consistent with this logic.
%The relational model  
is an approach for the information management of a ``community of discourse''
%%%%%%%%%%%%%%%%%%%%%%%%%%%%%%%%%%%%%%%%%%%%%%%%%%%%%%%%%%%%%%%%%%%%%%%%%%%%%%%%
%%%%%%%%%%%%%%%%%%%%%%%%%%%%%%%%%%%%%%%%%%%%%%%%%%%%%%%%%%%%%%%%%%%%%%%%%%%%%%%%
\footnote{Examples include:
an academic discipline;
a commercial enterprise;
library science;
the legal profession;
%fans of a particular sport, such as soccer;
%a religious group;
%one branch of the armed forces;
%the choreographed ballroom dance community;
etc.}
%A query in the latter community might be:
%find all dances,
%whose \emph{rhythm} is ``foxtrot'',
%whose \emph{choreographer} is ``smith'',
%which contains a ``telemark'' \emph{figure}.}
%%%%%%%%%%%%%%%%%%%%%%%%%%%%%%%%%%%%%%%%%%%%%%%%%%%%%%%%%%%%%%%%%%%%%%%%%%%%%%%%
%%%%%%%%%%%%%%%%%%%%%%%%%%%%%%%%%%%%%%%%%%%%%%%%%%%%%%%%%%%%%%%%%%%%%%%%%%%%%%%%
using the semantics and formalism of (many-sorted) first-order predicate logic. 
%\textsf{MFOL}.
%
The relational model was initially discussed in two papers:
``A Relational Model of Data for Large Shared Data Banks''
by Codd \cite{codd:70} 
and
``The Entity-Relationship Model -- Toward a Unified View of Data'' 
by Chen \cite{chen:76}. 
The relational model follows many-sorted logic
by representing data in terms of many-sorted relations, 
subsets of the Cartesian product of multiple domains. 
All data is represented horizontally in terms of tuples, 
which are grouped vertically into relations. 
A database organized in terms of the relational model 
is a called relational database.
The relational model provides a method 
for modeling the data stored in a relational database 
and for defining queries upon it. 
%

%%%%%%%%%%%%%%%%%%%%%%%%%%%%%%%%%%%%%%%%%%%%%%%%%%%%%%%%%%%%%%%%%%
%%%%%%%%%%%%%%%%%%%%%%%%%%%%%%%%%%%%%%%%%%%%%%%%%%%%%%%%%%%%%%%%%%
%\newpage
\subsection{First Order Logical Environment}\label{subsec:FOLE}
%\subsection{\texttt{FOLE}}\label{subsec:FOLE}
%\paragraph{\texttt{FOLE}.}
%%%%%%%%%%%%%%%%%%%%%%%%%%%%%%%%%%%%%%%%%%%%%%%%%%%%%%%%%%%%%%%%%
%%%%%%%%%%%%%%%%%%%%%%%%%%%%%%%%%%%%%%%%%%%%%%%%%%%%%%%%%%%%%%%%%

%%%%%%%%%%%%%%%%%%%%%%%%%%%%%%%%%%%%%%%%%%%%%%%%%%%%%%%%%%%%%%%%%%%%%%%%%%%%%%%%%%%%%%%%%%
%\newpage
\paragraph{Basics.}
%%%%%%%%%%%%%%%%%%%%%%%%%%%%%%%%%%%%%%%%%%%%%%%%%%%%%%%%%%%%%%%%%%%%%%%%%%%%%%%%%%%%%%%%%%

%The first order logical environment {\texttt{FOLE}} is a category-theoretic approach to many-sorted first order predicate logic.
The \emph{first-order logical environment} \texttt{FOLE}
%(Kent~\cite{kent:iccs2013})
is a category-theoretic representation for 
many-sorted (multi-sorted) first-order predicate logic. 
%
%%%%%%%%%%%%%%%%%%%%%%%%%%%%%%%%%%%%%%%%%%%%%%%%%%%%%%%%%%%%%%%%%%%%%%%%%%%%%%%%
%%%%%%%%%%%%%%%%%%%%%%%%%%%%%%%%%%%%%%%%%%%%%%%%%%%%%%%%%%%%%%%%%%%%%%%%%%%%%%%%
\footnote{Following the original discussion of {\ttfamily FOLE} (Kent~\cite{kent:iccs2013}), 
we use 
the term \emph{mathematical context} for the concept of a category,
the term \emph{passage} for the concept of a functor, and
the term \emph{bridge} for the concept of a natural transformation.
A context represents some ``species of mathematical structure''. 
A passage is a ``natural construction on structures of one species, 
yielding structures of another species'' 
(Goguen \cite{goguen:cm91}).}
%%%%%%%%%%%%%%%%%%%%%%%%%%%%%%%%%%%%%%%%%%%%%%%%%%%%%%%%%%%%%%%%%%%%%%%%%%%%%%%%
%%%%%%%%%%%%%%%%%%%%%%%%%%%%%%%%%%%%%%%%%%%%%%%%%%%%%%%%%%%%%%%%%%%%%%%%%%%%%%%%
% 
%Hence, 
The relational model can naturally be represented in \texttt{FOLE}.
The {\texttt{FOLE}} approach to logic, 
and hence to databases, 
%uses lists and classifications.
%The {\texttt{FOLE}} representation for database tables 
relies upon two mathematical concepts:
(1) lists and (2) classifications.
Lists represent database signatures and tuples;
classifications represent data-types and logical predicates.
{\texttt{FOLE}} 
represents the header of a database table as a list of sorts, and
represents the body of a database table as a set of tuples 
classified by the header.
The notion of a list is common in category theory.
The notion of a classification is described in two books:
``Information Flow: The Logic of Distributed Systems''
by Barwise and Seligman \cite{barwise:seligman:97} and
''Formal Concept Analysis: Mathematical Foundations''
by Ganter and Wille \cite{ganter:wille:99}.

%%%%%%%%%%%%%%%%%%%%%%%%%%%%%%%%%%%%%%%%%%%%%%%%%%%%%%%%%%%%%%%%%%%%%%%%%%%%%%%%%%%%%%%%%%
%\newpage
%\subsection{Philosophy.}\label{sub:sec:overview:old}
\paragraph{Philosophy.}
%%%%%%%%%%%%%%%%%%%%%%%%%%%%%%%%%%%%%%%%%%%%%%%%%%%%%%%%%%%%%%%%%%%%%%%%%%%%%%%%%%%%%%%%%%

A relational database system
should satisfy the following properties.
%\begin{enumerate}
%\item 
It should present the data to the user as relations.
%(a presentation in tabular form, i.e. as a collection of tables with each table consisting of a set of rows and columns).
In a {\ttfamily FOLE} database this is accomplished by using {\ttfamily FOLE} tables,
which are discussed in the paper \cite{kent:fole:era:tbl}.
%\item
It should provide relational operations to manipulate the data in tabular form.
Relational operations are discussed in the paper \cite{kent:fole:rel:ops}.
%``Relational Operations in \texttt{FOLE}''
%\end{enumerate}
%
%\item[Intro: ``Database Semantics''] 
%
%All information in the relational model is represented within relations
%(relational tables, or just tables).
A relational database is a diagram (linked collection) of tables.
%relations 
%(relational tables, or just tables). 
%
The information in a database is accessed by specifying queries, 
which use operations such as 
%\begin{itemize}
%\item 
\emph{select} to identify tuples, 
%\item 
\emph{project} to identify attributes, and 
%\item 
\emph{join} to combine tables. 
%\end{itemize}
%
In this paper,
projection refers to a primitive generalization-specialization operation between pairs of relational tables
(they are specified by the database schema, project from joined table to components, or other),
whereas join is a composite operation on a linked collection of tables.
Selection is a special case of join,
which uses reference relations (tables).

%{\fbox{\bf{Add discussion about how ontologies fit with databases.}}}

%\item[Overview of the Paper:] 
%We first review various aspects of the \texttt{FOLE} notion of a table.
%

%\newpage

%
%%%%%%%%%%%%%%%%%%%%%%%%%%%%%%%%%%%%%%%%%%%%%%%%%%%%%%%%%%%%%%%%%%%%%%%%%%%%%%%%%%%%%%%%%%
%\mbox{}\newline
%\rule{142pt}{1pt}{\fbox{\textbf{Work Zone}}}\rule{142pt}{1pt}
%\newline\newline
%%%%%%%%%%%%%%%%%%%%%%%%%%%%%%%%%%%%%%%%%%%%%%%%%%%%%%%%%%%%%%%%%%%%%%%%%%%%%%%%%%%%%%%%%%
%

%%%%%%%%%%%%%%%%%%%%%%%%%%%%%%%%%%%%%%%%%%%%%%%%%%%%%%%%%%%%%%%%%%%%%%%%%%%%%%%%%%%%%%%%%%
%%%%%%%%%%%%%%%%%%%%%%%%%%%%%%%%%%%%%%%%%%%%%%%%%%%%%%%%%%%%%%%%%%%%%%%%%%%%%%%%%%%%%%%%%%
%\comment{% interpretion-form-fragment
%%%%%%%%%%%%%%%%%%%%%%%%%%%%%%%%%%%%%%%%%%%%%%%%%%%%%%%%%%%%%%%%%%%%%%%%%%%%%%%%%%%%%%%%%%
%%%%%%%%%%%%%%%%%%%%%%%%%%%%%%%%%%%%%%%%%%%%%%%%%%%%%%%%%%%%%%%%%%%%%%%%%%%%%%%%%%%%%%%%%%
%\newpage
%\subsection{Preface.}\label{sub:sec:preface}
%%%%%%%%%%%%%%%%%%%%%%%%%%%%%%%%%%%%%%%%%%%%%%%%%%%%%%%%%%%%%%%%%%%%%%%%%%%%%%%%%%%%%%%%%%
%%%%%%%%%%%%%%%%%%%%%%%%%%%%%%%%%%%%%%%%%%%%%%%%%%%%%%%%%%%%%%%%%%%%%%%%%%%%%%%%%%%%%%%%%%

%%%%%%%%%%%%%%%%%%%%%%%%%%%%%%%%%%%%%%%%%%%%%%%%%%%%%%%%%%%%%%%%%%%%%%%%%%%%%%%%%%%%%%%%%%
%\newpage
\paragraph{Architecture.}
%%%%%%%%%%%%%%%%%%%%%%%%%%%%%%%%%%%%%%%%%%%%%%%%%%%%%%%%%%%%%%%%%%%%%%%%%%%%%%%%%%%%%%%%%%

%\begin{description}
%
%\item[{\fbox{\bf Codd} interpretation form:}]
%.\newline

%\newpage

A series of papers provides a rigorous mathematical basis for {\ttfamily FOLE} 
%interpretation 
%and to define 
by defining 
an architectural semantics for the relational data model,
thus providing the foundation for
the formalism and semantics of first-order logical/relational database systems.
%\item 
%``The First-order Logical Environment''
%\cite{kent:iccs2013},
This architecture 
%of
%are concerned with the presentation of
%{\ttfamily FOLE}
consists of two hierarchies of two nodes each:
the classification hierarchy
and
the interpretation hierarchy.

%These papers expand
%upon material found within the paper ``Database Semantics'' \cite{kent:db:sem}.
%\end{description}
%
%This paper falls into two parts: fixed-shape versus variable-shape structures 
%(schemas, schemed domains, or databases).
%Fixed-shape structures are develop by lifting various components of 
%(signatures, signed domains, or tables) to diagrams. 
%Variable-shape structures are develop by using Kan extensions on fixed-shape structures.

%The first order logical environment {\texttt{FOLE}} is described in multiple papers by the author:
%\begin{itemize}
%\item 
%\item 
%``The {\ttfamily ERA} of {\ttfamily FOLE}: Foundation''
%\cite{kent:fole:era:found},
%\item 
%``The {\ttfamily ERA} of {\ttfamily FOLE}: Superstructure''
%\cite{kent:fole:era:supstruc},
%\item 
%``The {\ttfamily FOLE} Table''
%\cite{kent:fole:era:tbl},
%\end{itemize}
%

\begin{itemize}
%\begin{description}
%
\item
%[classification hierarchy:] 
Two papers provide a precise mathematical basis for \texttt{FOLE} classification.
%\begin{itemize}
%\item 
The paper 
%\item 
``The {\ttfamily ERA} of {\ttfamily FOLE}: Foundation''
\cite{kent:fole:era:found}
%(2015),
develops the notion of a \texttt{FOLE} \underline{\emph{structure}},
following the entity-relationship model of Chen~\cite{chen:76}.
%providing a rigorous mathematical representation for the {\ttfamily ERA}
%data model (Chen~\cite{chen:76}).
% in particular, and ontologies in general, within the first-order
%logical environment \texttt{FOLE}.
%\item 
This provided a basis for
the paper 
%\item 
``The {\ttfamily ERA} of {\ttfamily FOLE}: Superstructure''
\cite{kent:fole:era:supstruc},
%(2016),
which develops the notion of a \texttt{FOLE} \underline{\emph{sound logic}}.
%\end{itemize}
\newline
\item
%[interpretation hierarchy:] 
Two papers provide a precise mathematical basis for \texttt{FOLE} interpretation.
Both of these papers expand on material found in the paper 
%(Kent~\cite{kent:db:sem}).
``Database Semantics''
\cite{kent:db:sem}.
%\begin{itemize}
%\item 
The paper 
%(Kent \cite{kent:fole:era:tbl})
``The {\ttfamily FOLE} Table''
\cite{kent:fole:era:tbl},
develops the notion of a \texttt{FOLE} \underline{\emph{table}}
following the relational model of Codd~\cite{codd:90}.
%\item 
This provides a basis for the current paper 
%(following Kent~\cite{kent:db:sem})
``The {\ttfamily FOLE} Database'',
%\cite{kent:fole:era:tbl},
which develops the notion of a \texttt{FOLE} relational \underline{\emph{database}}.
\end{itemize}
%\end{description}
%
%In addition,

%%%%%%%%%%%%%%%%%%%%%%%%%%%%%%%%%%%%%%%%%%%%%%%%%%%%%%%%%%%%%%%%%%
%%%%%%%%%%%%%%%%%%%%%%%%%%%%%%%%%%%%%%%%%%%%%%%%%%%%%%%%%%%%%%%%%%
%\newpage
%\paragraph{Architecture.}
%\subsection{{\ttfamily FOLE} Architecture}\label{sub:sec:architecture}
%\paragraph{Overview.}
%%%%%%%%%%%%%%%%%%%%%%%%%%%%%%%%%%%%%%%%%%%%%%%%%%%%%%%%%%%%%%%%%
%%%%%%%%%%%%%%%%%%%%%%%%%%%%%%%%%%%%%%%%%%%%%%%%%%%%%%%%%%%%%%%%%

%
\begin{flushleft}
{{\setlength{\extrarowheight}{1.6pt}
{{
{\begin{tabular}[t]{l@{\hspace{20pt}}l}
%\rule[-5pt]{0pt}{1pt}
%{\footnotesize{\textsf{Figures}}} &
%{\footnotesize{\textsf{Tables}}}
%\\
{{{\begin{minipage}{230pt}
%{\fbox{\begin{tabular}[t]{p{2.5in}l}
The architecture of
%are concerned with the presentation of
{\ttfamily FOLE}
is pictured briefly
on the right
%in Fig.\,\ref{hierarchy},
and more completely in
Fig.\,1
of the preface of
the paper
%``{\ttfamily FOLE} Equivalence''
\cite{kent:fole:equiv}.
This consists of two hierarchies of two nodes each.
\comment{
\\
\vspace{-16pt}
\begin{description}
\item[{The classification hierarchy}]
%(LHS of Fig.\,\ref{hierarchy})
on the left
defines
%\begin{itemize}
%\item
{\ttfamily FOLE} Structures
\cite{kent:fole:era:found}
at the bottom
%\item
and {\ttfamily FOLE} Sound Logics
\cite{kent:fole:era:supstruc}
at the top.
%\end{itemize}
%
\item[\emph{The interpretation hierarchy}]
on the right
%with
%RHS of Fig.\,\ref{hierarchy}
%\begin{itemize}
%\item
defines
{\ttfamily FOLE} Tables \cite{kent:fole:era:tbl}
at the bottom
%\item
and {\ttfamily FOLE} Databases 
%\cite{kent:fole:era:db}
(this paper)
at the top.
%\end{itemize}
%
\end{description}
\vspace{-5pt}
}
The paper
``{\ttfamily FOLE} Equivalence''
\cite{kent:fole:equiv}
proves that
{\ttfamily FOLE} sound logics
are equivalent to
{\ttfamily FOLE} databases.
%proves the equivalence of these two hierarchies.
%\end{tabular}}}
%
\end{minipage}}}}
%%%%%%%%%%%%%%%%%%%%%%%%%%%%%%%%%%%%%%%%%%%%%%%%%%%%%%%%%%%%%%%%%%%%%%%%%%%%%%%%
%%%%%%%%%%%%%%%%%%%%%%%%%%%%%%%%%%%%%%%%%%%%%%%%%%%%%%%%%%%%%%%%%%%%%%%%%%%%%%%%
&
%%%%%%%%%%%%%%%%%%%%%%%%%%%%%%%%%%%%%%%%%%%%%%%%%%%%%%%%%%%%%%%%%%%%%%%%%%%%%%%%
%%%%%%%%%%%%%%%%%%%%%%%%%%%%%%%%%%%%%%%%%%%%%%%%%%%%%%%%%%%%%%%%%%%%%%%%%%%%%%%%
{{\begin{tabular}{c@{\hspace{5pt}}}
\setlength{\unitlength}{0.36pt}
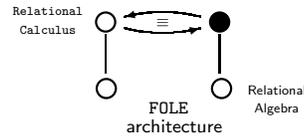
\begin{picture}(180,120)(-50,-20)
%\put(60.3,80){\makebox(0,0){\huge{$\bullet$}}}
\put(60.5,80.5){\makebox(0,0){\tiny{$\equiv$}}}
%\put(60.5,79.8){\makebox(0,0){\huge{$\circ$}}}
%\put(60.5,79.8){\makebox(0,0){\huge{$\bullet$}}}
%\put(100,70){\vector(-1,0){80}}
%\qbezier(100,75)(60,65)(20,75)
%
\put(-60,93){\makebox(0,0){\tiny{\tt{Relational}}}}
\put(-60,73){\makebox(0,0){\tiny{\tt{Calculus}}}}
%\put(180,30){\makebox(0,0){\huge{$\bullet$}}}
\put(180,10){\makebox(0,0){\tiny{\sf{Relational}}}}
\put(180,-10){\makebox(0,0){\tiny{\sf{Algebra}}}}
\qbezier(100,87)(60,97)(20,87)
\put(20,87){\vector(-4,-1){0}}
\qbezier(20,75)(60,65)(100,75)
\put(100,75){\vector(4,1){0}}
\put(0.3,80){\makebox(0,0){\huge{$\circ$}}}
\put(120.3,80){\makebox(0,0){\huge{$\bullet$}}}
\put(1,10){\makebox(0,0){\huge{$\circ$}}}
\put(122,10){\makebox(0,0){\huge{$\circ$}}}
%
%\put(45,128){\line(-1,-1){33}}
%\put(75,128){\line(1,-1){33}}
\put(0,68){\line(0,-1){40}}
\put(120,68){\line(0,-1){40}}
\put(46,-15){\scriptsize{\ttfamily FOLE}}
\put(22,-35){\scriptsize{\textsf{architecture}}}
\end{picture}
\end{tabular}}}
\end{tabular}}
}}}}
%\end{tabular}}}
\end{flushleft}
In the relational model there are two approaches for database management:
%The relational algebra and the relational calculus;
the relational algebra,
%(Kent \cite{kent:fole:rel:ops}), 
which defines an imperative language,
and
the relational calculus, which defines a declarative language.
The paper 
``Relational Operations in \texttt{FOLE}''
\cite{kent:fole:rel:ops}
%demonstrate how {\ttfamily FOLE} 
represents relational algebra
by
expressing the relational operations of database theory in a clear and implementable representation.
The relational calculus
will be represented in \texttt{FOLE} in a future paper.

%
%%%%%%%%%%%%%%%%%%%%%%%%%%%%%%%%%%%%%%%%%%%%%%%%%%%%%%%%%%%%%%%%%%%%%%%%%%%%%%%%%%%%%%%%%%
%\mbox{}\newline
%\rule{142pt}{1pt}{\fbox{\textbf{Work Zone}}}\rule{142pt}{1pt}
%\newline\newline
%%%%%%%%%%%%%%%%%%%%%%%%%%%%%%%%%%%%%%%%%%%%%%%%%%%%%%%%%%%%%%%%%%%%%%%%%%%%%%%%%%%%%%%%%%
%

%%%%%%%%%%%%%%%%%%%%%%%%%%%%%%%%%%%%%%%%%%%%%%%%%%%%%%%%%%%%%%%%%%
%%%%%%%%%%%%%%%%%%%%%%%%%%%%%%%%%%%%%%%%%%%%%%%%%%%%%%%%%%%%%%%%%%
\newpage
\subsection{Overview}\label{sub:sec:overview}
%\paragraph{Overview.}
%%%%%%%%%%%%%%%%%%%%%%%%%%%%%%%%%%%%%%%%%%%%%%%%%%%%%%%%%%%%%%%%%
%%%%%%%%%%%%%%%%%%%%%%%%%%%%%%%%%%%%%%%%%%%%%%%%%%%%%%%%%%%%%%%%%
%
This paper defines an architectural semantics for the relational data model.
%
%%%%%%%%%%%%%%%%%%%%%%%%%%%%%%%%%%%%%%%%%%%%%%%%%%%%%%%%%%%%%%%%%%%%%%%%%%%%%%%%
%%%%%%%%%%%%%%%%%%%%%%%%%%%%%%%%%%%%%%%%%%%%%%%%%%%%%%%%%%%%%%%%%%%%%%%%%%%%%%%%
\footnote{Older architectures of data include the hierarchical model and network model.
Of these, nothing will be said.
A newer architecture of data, 
called the object-relation-object model, 
is a presentation form for the relational data model described here.}
%%%%%%%%%%%%%%%%%%%%%%%%%%%%%%%%%%%%%%%%%%%%%%%%%%%%%%%%%%%%%%%%%%%%%%%%%%%%%%%%
%%%%%%%%%%%%%%%%%%%%%%%%%%%%%%%%%%%%%%%%%%%%%%%%%%%%%%%%%%%%%%%%%%%%%%%%%%%%%%%%
%
It defines the notion of a relational database in terms of the
the first-order logical environment {\ttfamily FOLE}. 
%\begin{itemize}
%
%\item 
In \S\ref{sub:sec:schema}
we define \texttt{FOLE} Schemas in the context $\mathrmbf{LIST}$.
A schema is a diagram of signatures.
%
%\item 
In \S\ref{sec:sch:dom}
we define \texttt{FOLE} Schemed Domains.
A schemed domain is a diagram of signed domains.
%\begin{itemize}
%\item 
In \S\ref{sub:sec:sch:dom:gen}
schemed domains are defined in general
in the context $\mathrmbf{DOM}$.
%\item 
%
In \S\ref{sub:sec:rel:sch:typ:dom}
schemed domains are defined 
with fixed type domain (datatypes).
%\begin{itemize}
%\item 
%
In 
\S\ref{sub:sub:sec:dom:typ:dom:lower}
schemed domains are constrained 
to a particular fixed type domain 
(the lower aspect)
in the context 
$\mathring{\mathrmbf{Dom}}(\mathcal{A})$.
%and then 
In 
\S\ref{sub:sub:sec:dom:typ:dom:upper}
schemed domains are constrained 
to a fixed type domain morphism 
(the upper aspect)
in the context 
$\mathring{\mathrmbf{Dom}}$.
%\end{itemize}
%\item 
%
%Following tables,
The projective components of schemed domains are defined:
%%%%%%%%%%%%%%%%%%%%%%%%%%%%%%%%%%%%%%%%%%%%%%%%%%%%%%%%%%%%%%%%%%%%%%%%%%%%%%%%
%%%%%%%%%%%%%%%%%%%%%%%%%%%%%%%%%%%%%%%%%%%%%%%%%%%%%%%%%%%%%%%%%%%%%%%%%%%%%%%%
%\footnote{Tables are based upon lists (signatures) and classifications (type domains),
%which are linked through sort sets.}
%%%%%%%%%%%%%%%%%%%%%%%%%%%%%%%%%%%%%%%%%%%%%%%%%%%%%%%%%%%%%%%%%%%%%%%%%%%%%%%%
%%%%%%%%%%%%%%%%%%%%%%%%%%%%%%%%%%%%%%%%%%%%%%%%%%%%%%%%%%%%%%%%%%%%%%%%%%%%%%%%
signature diagrams (schemas) and type domain diagrams.
Hence, 
a schemed domain is a diagram of signed domains
consisting of the following components:
a signature diagram
and a type domain diagram
with the same shape
and a common sort diagram.
%\end{itemize}
%
%\item 
In \S\ref{sec:rel:db}
we define \texttt{FOLE} Databases.
A relational database is a diagram of tables.
%
%%%%%%%%%%%%%%%%%%%%%%%%%%%%%%%%%%%%%%%%%%%%%%%%%%%%%%%%%%%%%%%%%%%%%%
%\paragraph{In brief.}
%%%%%%%%%%%%%%%%%%%%%%%%%%%%%%%%%%%%%%%%%%%%%%%%%%%%%%%%%%%%%%%%%%%%%%
%\begin{itemize}
%\item 
%A database is a diagram of tables with a shape context.
%\item 
%The organization 
%%(construction, framework) 
%of databases follows the 
%organization 
%%(construction, framework) 
%of tables.
%\item 
%In each case, 
%We first 
%Following tables,
%\begin{itemize}
%\item 
%
In \S\ref{sub:sec:rel:db:gen}
databases are defined in general
in the context $\mathrmbf{DB}$.
%\item 
%
In \S\ref{sub:sec:rel:db:typ:dom}
databases are defined 
with fixed type domain (datatypes).
%\item 
%
In \S\ref{sub:sub:sec:rel:db:typ:dom:lower}
databases are constrained 
to a particular fixed type domain 
(the lower aspect)
in the context $\mathrmbf{Db}(\mathcal{A})$.
%and then 
In \S\ref{sub:sub:sec:rel:db:typ:dom:upper}
databases are constrained 
to a fixed type domain morphism 
(the upper aspect)
in the context $\mathrmbf{Db}$.
%\end{itemize}
%\item 
%
Following tables,
the projective components of databases are defined:
%%%%%%%%%%%%%%%%%%%%%%%%%%%%%%%%%%%%%%%%%%%%%%%%%%%%%%%%%%%%%%%%%%%%%%%%%%%%%%%%
%%%%%%%%%%%%%%%%%%%%%%%%%%%%%%%%%%%%%%%%%%%%%%%%%%%%%%%%%%%%%%%%%%%%%%%%%%%%%%%%
\footnote{Tables are based upon lists (signatures) and classifications (type domains),
which are linked through sort sets.}
%%%%%%%%%%%%%%%%%%%%%%%%%%%%%%%%%%%%%%%%%%%%%%%%%%%%%%%%%%%%%%%%%%%%%%%%%%%%%%%%
%%%%%%%%%%%%%%%%%%%%%%%%%%%%%%%%%%%%%%%%%%%%%%%%%%%%%%%%%%%%%%%%%%%%%%%%%%%%%%%%
signed domain diagrams (schemed domains) and key diagrams, linked by tuple bridges.
Hence, 
a relational database is a diagram of tables
consisting of the following components:
a signed domain diagram
and a key set diagram
with the same shape, 
which are 
connected by
a tuple bridge.
%\item 
%Schemed domains project to schemas and type domains.
%\end{itemize}
%
%
%\item 
%
The appendix \S\ref{sec:append} 
reviews general theory and database components. 
%\begin{itemize}
%\item 
%
In \S\ref{sub:sec:append:gen:th}
we review general theory: 
Grothendieck construction in \S\ref{append:grothen:construct},
diagram contexts in \S\ref{sub:sub:sec:lax:comma:cxt}, and 
Kan extensions in \S\ref{append:kan:ext}.
%\item 
In \S\ref{sub:sec:append:db:app}
we review the various database components of {\ttfamily FOLE}: 
%consisting of
contexts in \S\ref{sub:sub:sec:math:context},
passages in \S\ref{sub:sub:sec:pass}, and
bridges in \S\ref{sub:sec:bridges}.
%\end{itemize}
%

%contexts, and signed domains; and tables, tabular flow and relations in § A.2.

%
\begin{table}
\begin{center}
{{\scriptsize{\setlength{\extrarowheight}{1.6pt}
{\begin{tabular}{l@{\hspace{10pt}}l}
\\
{\fbox{
\begin{tabular}[t]{|l@{\hspace{8pt}}l@{\hspace{2pt}:\hspace{5pt}}l|}
%\\\hline\hline
\hline

%\S\ref{sec:intro}
%&
%Fig.\,\ref{hierarchy}
%&
%{{{\ttfamily FOLE} Architecture}}
%%\\\hline\hline
%%\S\ref{sec:intro}
%%&
%%Fig.\,\ref{hierarchy}
%%&
%%{\ttfamily FOLE} Hierarchy
%\\\hline\hline
\S\ref{sub:sec:schema}
&
Fig.\,\ref{fig:schema:mor}
&
Schema Morphism: $\mathrmbf{LIST}$
\\
%\cline{1-1}
&
Fig.\,\ref{fig:cxt:schema}
&
Schema Context: $\mathrmbf{LIST}$
\\\hline\hline
\S\ref{sub:sec:sch:dom:gen}
&
Fig.\,\ref{fig:sch:dom:mor}
&
Schemed Domain Morphism: $\mathrmbf{DOM}$
\\
%\cline{1-1}
&
Fig.\,\ref{fig:cxt:sch:dom}
&
Schemed Domain Context: $\mathrmbf{DOM}$
\\\cline{1-2}
\S\ref{sub:sub:sec:dom:typ:dom:lower}
&
Fig.\,\ref{fig:sch:dom:mor:A}
&
Schemed Domain Morphism: $\mathring{\mathrmbf{Dom}}(\mathcal{A})$
\\\cline{1-1}
\S\ref{sub:sub:sec:dom:typ:dom:upper}
&
Fig.\,\ref{fig:sch:dom:mor:typ:dom}
&
Schemed Domain Morphism: $\mathring{\mathrmbf{Dom}}$
%\\\cline{1-2}
%\S\ref{sub:sub:sec:tup:pass:db}
%&
%Fig.\,\ref{fig:tup:brid:adj}
%&
%Tuple Bridge Adjointness
\\\hline\hline
\S\ref{sub:sec:rel:db:gen}
& 
Fig.\,\ref{fig:db:mor:gen}
& 
Database Morphism: $\mathrmbf{DB}$
\\
%\cline{1-1}
%\S\ref{sub:sub:sec:rel:db:var:shape:gen}
& 
Fig.\,\ref{fig:fole:db:cxt:var}
& 
Database Context: $\mathrmbf{DB}$
\\
%\cline{1-1}
%\S\ref{sub:sub:sec:rel:db:var:shape:gen}
& 
Fig.\,\ref{fig:db:morph:proj}
& 
Database Morphism (proj): $\mathrmbf{DB}$
\\\cline{1-2}
\S\ref{sub:sub:sec:rel:db:typ:dom:lower}
& 
Fig.\,\ref{fig:db:A:mor:adj}
& 
Database Morphism: $\mathrmbf{Db}(\mathcal{A})$
\\
%\cline{1-1}
& 
Fig.\,\ref{fig:fole:db:cxt:fbr:var:A}
& 
Database Context: $\mathrmbf{Db}(\mathcal{A})$
\\
%\cline{1-1}
& 
Fig.\,\ref{fig:rel:db:mor:var:A}
& 
Database Morphism (proj): $\mathrmbf{Db}(\mathcal{A})$
\\\cline{1-1}
\S\ref{sub:sub:sec:rel:db:typ:dom:upper}
&
Fig.\,\ref{fig:db:mor:Db}
&
Database Morphism: $\mathrmbf{Db}$
\\
%\cline{1-2}
&
Fig.\,\ref{fig:fole:db:cxt:typ:dom:var}
&
Database Context: $\mathrmbf{Db}$
\\
%\cline{1-2}
&
Fig.\,\ref{fig:rel:db:mor:var:typ:dom}
&
Database Morphism (proj): $\mathrmbf{Db}$
%\\\hline\hline
%%%%%%%%%%%%%%%%%%%%%%%%%%%%%%%%%%%%%%%%%%%%%%%%%%%%%%%%%%%%%%%%%%%%%%
%%%%%%%%%%%%%%%%%%%%%%%%%%%%%%%%%%%%%%%%%%%%%%%%%%%%%%%%%%%%%%%%%%%%%%
\comment{
\\\hline
\S\ref{sub:sec:comma:cxt:DB}
& 
Fig.\,\ref{fig:comma:cxt:morph}
& 
$\widehat{\mathrmbf{DB}}$\,-\,morphism
\\\cline{1-2}
\S\ref{sub:sec:comma:cxt:DB:A}
& 
Fig.\,\ref{fig:comma:cxt:morph:A}
& 
%$\bigl(\mathrmbf{SET}{\;\Downarrow\,}\mathring{\mathrmbfit{tup}_{\mathcal{A}}}\bigr)$ morphism
{$\widehat{\mathrmbf{Db}}(\mathcal{A})$\,-\,morphism}
\\\cline{1-2}
\S\ref{sub:sub:sec:comma:cxt:Db}
& 
Fig.\,\ref{fig:comma:cxt:morph:typ:dom}
& 
%$\bigl(\mathrmbf{SET}{\;\Downarrow\,}\mathring{\mathrmbfit{tup}}\bigr)$ morphism
{$\widehat{\mathrmbf{Db}}$\,-\,morphism}
}
%%%%%%%%%%%%%%%%%%%%%%%%%%%%%%%%%%%%%%%%%%%%%%%%%%%%%%%%%%%%%%%%%%%%%%
%%%%%%%%%%%%%%%%%%%%%%%%%%%%%%%%%%%%%%%%%%%%%%%%%%%%%%%%%%%%%%%%%%%%%%
\\\hline\hline
\comment{
%\S\ref{sub:sec:inc:bridge}
\S\ref{sub:sub:sec:comma:cxt}
& 
Fig.\,\ref{fig:comma:cxt}
& 
Comma Context: 
{{$\bigl(\mathrmbfit{F}{\,\downarrow\,}\mathrmbfit{G}\bigr)$}}
\\\cline{1-2}
}
\S\ref{append:grothen:construct}
& 
Fig.\,\ref{fig:incl:bridge:fbr:cxt}
& 
Inclusion Bridge: Fibered Context
\\\cline{1-2}
\S\ref{sub:sub:sec:math:context}
& 
Fig.\,\ref{fig:diag:comma:cxts}
& 
Diagram/Comma/Fibered Contexts
\\\hline
\end{tabular}
}}
%%%%%%%%%%%%%%%%%%%%%%%%%%%%%%%%%%%%%%%%%%%%%%%%%%%%%%%%%%%%%%%%%%%%%%%%%%%%%%%%
%%%%%%%%%%%%%%%%%%%%%%%%%%%%%%%%%%%%%%%%%%%%%%%%%%%%%%%%%%%%%%%%%%%%%%%%%%%%%%%%
&
%%%%%%%%%%%%%%%%%%%%%%%%%%%%%%%%%%%%%%%%%%%%%%%%%%%%%%%%%%%%%%%%%%%%%%%%%%%%%%%%
%%%%%%%%%%%%%%%%%%%%%%%%%%%%%%%%%%%%%%%%%%%%%%%%%%%%%%%%%%%%%%%%%%%%%%%%%%%%%%%%
{\fbox{
\begin{tabular}[t]{|l@{\hspace{8pt}}l@{\hspace{2pt}:\hspace{5pt}}l|}
\hline
\S\ref{sub:sec:overview}
&
Tbl.\,\ref{tbl:figs:tbls}
&
Figures and Tables
\\\hline\hline
\S\ref{sec:sch:dom}
& 
Tbl.\,\ref{tbl:sch:dom:cxts}
& 
{Schemed Domain Contexts} 
%\\\hline
%%\S\ref{sec:sch:dom}
%& 
%Tbl.\,\ref{tbl:tup:pass:sign:dom}
%& 
%{Signed Domain Tuple Passages} 
%\\\cline{1-2}
%%\S\ref{sub:sec:sch:dom:gen}
%& 
%Tbl.\,\ref{tbl:tup:pass:sch:dom}
%& 
%{Schemed Domain Tuple Passages}
\\\hline\hline
\S\ref{sec:rel:db}
& 
Tbl.\,\ref{tbl:db:cxts}
& 
{Relational Database Contexts}
%\\\cline{1-2}
%%\S\ref{sec:rel:db}
%&
%Tbl.\,\ref{tbl:adb:char}
%& 
%Characterization
\\\hline\hline
\S\ref{sub:sub:sec:math:context}
& 
Tbl.\,\ref{tbl:sch:dom:morphs}
& 
Schemed Domain Morphisms
\\
%\cline{1-1}
& 
Tbl.\,\ref{tbl:rel:db:morphs}
& 
Relational Database Morphisms
\\\cline{1-1}
\S\ref{sub:sub:sec:pass}
& 
Tbl.\,\ref{defs:proj:pass}
& 
Projection Passages
\\
& 
Tbl.\,\ref{tbl:colim:lim:pass}
& 
Lim (Colim) Passages
\\\cline{1-1}
%\\\hline\hline
\S\ref{sub:sec:bridges}
& 
Tbl.\,\ref{bridge:descr}
& 
Bridges
\\
%\cline{1-2}
& 
Tbl.\,\ref{adjoints:composites}
& 
Bridge Adjoints and Composites
%\\\cline{1-2}
%\S\ref{sub:sub:sec:lim:colim:all}
%& 
%Tbl.\,\ref{tbl:colim:lim:pass}
%& 
%{Colim (Lim) Passages} 
\\\hline
\end{tabular} 
}}
\end{tabular}}}}}
\end{center}
\caption{Figures and Tables}
\label{tbl:figs:tbls}
\end{table}
%

%%%%%%%%%%%%%%%%%%%%%%%%%%%%%%%%%%%%%%%%%%%%%%%%%%%%%%%%%%%%%%%%%%%%%%%%%%%%%%%%
%%%%%%%%%%%%%%%%%%%%%%%%%%%%%%%%%%%%%%%%%%%%%%%%%%%%%%%%%%%%%%%%%%%%%%%%%%%%%%%%
%\comment{% temporary beautification 

%%%%%%%%%%%%%%%%%%%%%%%%%%%%%%%%%%%%%%%%%%%%%%%%%%%%%%%%%%%%%%%%%%%%%%%%%%%%%%%%%%%%%%%%%%
%%%%%%%%%%%%%%%%%%%%%%%%%%%%%%%%%%%%%%%%%%%%%%%%%%%%%%%%%%%%%%%%%%%%%%%%%%%%%%%%%%%%%%%%%%
\newpage
\section{\texttt{FOLE} Schemas: $\mathrmbf{LIST}$}\label{sub:sec:schema}
%%%%%%%%%%%%%%%%%%%%%%%%%%%%%%%%%%%%%%%%%%%%%%%%%%%%%%%%%%%%%%%%%%%%%%%%%%%%%%%%%%%%%%%%%%
%%%%%%%%%%%%%%%%%%%%%%%%%%%%%%%%%%%%%%%%%%%%%%%%%%%%%%%%%%%%%%%%%%%%%%%%%%%%%%%%%%%%%%%%%%

\comment{
In this section
we discuss the context of schemas 
$\mathrmbf{LIST}=\mathrmbf{List}^{\scriptscriptstyle{\Uparrow}}$ 
(Fig.~\ref{fig:cxt:schema}),
which mirrors the context of signatures (lists) $\mathrmbf{List}$ 
at a higher dimension. 
%
%%%%%%%%%%%%%%%%%%%%%%%%%%%%%%%%%%%%%%%%%%%%%%%%%%%%%%%%%%%%%%%%%%%%%%%%%%%%%%%%
%%%%%%%%%%%%%%%%%%%%%%%%%%%%%%%%%%%%%%%%%%%%%%%%%%%%%%%%%%%%%%%%%%%%%%%%%%%%%%%%
\footnote{The contexts
$\mathrmbf{LIST}=\mathrmbf{List}^{\!\scriptscriptstyle{\Uparrow}}$ (schemas), 
and
$\mathrmbf{SET}=\mathrmbf{Set}^{\!\scriptscriptstyle{\Downarrow}}$
are defined in this section.}
%%%%%%%%%%%%%%%%%%%%%%%%%%%%%%%%%%%%%%%%%%%%%%%%%%%%%%%%%%%%%%%%%%%%%%%%%%%%%%%%
%%%%%%%%%%%%%%%%%%%%%%%%%%%%%%%%%%%%%%%%%%%%%%%%%%%%%%%%%%%%%%%%%%%%%%%%%%%%%%%%
%
%$\mathrmbf{List}\subseteq\mathrmbf{LIST}$ is the context of schemas with fixed sort set.
}

%%%%%%%%%%%%%%%%%%%%%%%%%%%%%%%%%%%%%%%%%%%%%%%%%%%%%%%%%%%%%%%%%%%%%%%%%%%%%%%%%%%%%%%%%%
%\newpage
%\paragraph{Diagram Context.}
%%%%%%%%%%%%%%%%%%%%%%%%%%%%%%%%%%%%%%%%%%%%%%%%%%%%%%%%%%%%%%%%%%%%%%%%%%%%%%%%%%%%%%%%%%

A schema ${\langle{\mathrmbf{R},\mathrmbfit{S}}\rangle}$
consists of 
%a shape context $\mathrmbf{R}$ 
%and an object $\mathrmbfit{S}$ in the fiber context $\mathrmbf{LIST}(\mathrmbf{R})$
%(\S~\ref{sub:sec:sch:shape});
%that is,
a shape context $\mathrmbf{R}$ 
and a diagram of signatures (lists) $\mathrmbf{R}\xrightarrow{\;\mathrmbfit{S}\;}\mathrmbf{List}$.
A schema morphism 
${\langle{\mathrmbf{R}_{2},\mathrmbfit{S}_{2}}\rangle} 
\xrightarrow{{\langle{\mathrmbfit{R},\,\sigma}\rangle}}
{\langle{\mathrmbf{R}_{1},\mathrmbfit{S}_{1}}\rangle}$
consists of 
a shape-changing passage $\mathrmbf{R}_{1}\xrightarrow{\;\mathrmbfit{R}\;\,}\mathrmbf{R}_{1}$
and a bridge $\mathrmbfit{S}_{2}\xRightarrow{\;\sigma\;\,}\mathrmbfit{R}{\,\circ\,}\mathrmbfit{S}_{1}$.
Composition is component-wise.
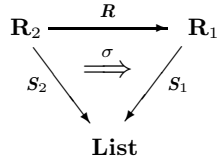
\begin{figure}
\begin{center}
{{\begin{tabular}{c}
\setlength{\unitlength}{0.56pt}
\begin{picture}(120,80)(8,0)
\put(5,80){\makebox(0,0){\footnotesize{$\mathrmbf{R}_{2}$}}}
\put(125,80){\makebox(0,0){\footnotesize{$\mathrmbf{R}_{1}$}}}
\put(65,0){\makebox(0,0){\footnotesize{$\mathrmbf{List}$}}}
\put(60,92){\makebox(0,0){\scriptsize{$\mathrmbfit{R}$}}}
\put(20,42){\makebox(0,0)[r]{\scriptsize{$\mathrmbfit{S}_{2}$}}}
\put(100,42){\makebox(0,0)[l]{\scriptsize{$\mathrmbfit{S}_{1}$}}}
\put(60,57){\makebox(0,0){\shortstack{\scriptsize{$\sigma$}\\\large{$\Longrightarrow$}}}}
\put(20,80){\vector(1,0){80}}
\put(10,67){\vector(3,-4){38}}
\put(110,68){\vector(-3,-4){38}}
\end{picture}
\end{tabular}}}
\end{center}
\caption{Schema Morphism: $\mathrmbf{LIST}$}
\label{fig:schema:mor}
\end{figure}
\begin{definition}\label{def:lax:com:cxt:schema:var}
The context of schemas is 
%the diagram
%lax comma 
%context
$\mathrmbf{LIST}
= \mathrmbf{List}^{\!\scriptscriptstyle{\Uparrow}}
= \bigl(\mathrmbf{Cxt}{\,\Uparrow\,}\mathrmbf{List}\bigr)$,
a diagram context over signatures (lists).
\end{definition}
\begin{proposition}\label{prop:lim:colim:schema}
%
%The context of 
%%schemas 
%%$\mathrmbf{LIST}
%%= \mathrmbf{List}^{\!\scriptscriptstyle{\Uparrow}}$
%%%= \bigl(\mathrmbf{Cxt}{\,\Uparrow\,}\mathrmbf{List}\bigr)$,
%%(Def.~\ref{def:lax:com:cxt:schema:var})
%%and
%%the context of 
%schemed domains 
%$\mathrmbf{DOM}
%= \mathrmbf{Dom}^{\!\scriptscriptstyle{\Uparrow}}$
%%= \bigl(\mathrmbf{Cxt}{\,\Uparrow\,}\mathrmbf{Dom}\bigr)$,
%%a diagram context over signed domains.
%(Def.~\ref{def:lax:com:cxt:sch:dom:var})
%is an example of a lax comma context
%(Def.~\ref{def:lax:oplax}).
%\begin{itemize}
%\item 
%The fibered context (Grothendieck construction) $\mathrmbf{LIST}$ is complete and cocomplete
%and the projection 
%$\mathrmbf{LIST} = \int\widetilde{\mathrmbf{List}}\rightarrow\mathrmbf{Cxt}$
%is continuous and cocontinuous.
%\item 
The fibered context (Grothendieck construction) 
$\mathrmbf{LIST} =
\mathrmbf{List}^{\scriptscriptstyle{\Uparrow}} = \int\hat{\mathrmbf{List}}$ is complete and cocomplete 
and the projection 
$\mathrmbf{List}^{\scriptscriptstyle{\Uparrow}}\rightarrow\mathrmbf{Cxt}
:{\langle{\mathrmbf{R},\mathrmbfit{S}}\rangle}\mapsto\mathrmbf{R}$ 
is continuous and cocontinuous.
%The fibered context (Grothendieck construction) $\mathrmbf{DOM}$ is complete and cocomplete
%and the projection 
%$\mathrmbf{DOM} = \int
%\widetilde{\mathrmbf{Dom}}
%%\widehat{\mathrmbf{Dom}}
%\rightarrow\mathrmbf{Cxt}$
%is continuous and cocontinuous.
%\end{itemize}
%
\end{proposition}
\begin{proof}
Use the lax parts of 
Prop.~\ref{prop:lax:fibered:A}
and
Prop.~\ref{prop:lim:colim:A} 
in \S\,\ref{append:kan:ext},
since
the context of schemas 
$\mathrmbf{LIST}$
is the diagram context 
%(Grothendieck construction) 
$\mathrmbf{LIST} = 
%\mathrmbf{Tbl}^{\scriptscriptstyle{\Downarrow}}
% = \int\hat{\mathrmbf{Tbl}}
\mathrmbf{List}^{\scriptscriptstyle{\Uparrow}} = \int\hat{\mathrmbf{List}}$ 
%$\mathrmbf{DB} = \mathrmbf{Tbl}^{\scriptscriptstyle{\Downarrow}}$
and
$\mathrmbf{List}$ is complete and cocomplete.
\mbox{}\hfill\rule{5pt}{5pt}
\end{proof}
%

%By Prop.~\ref{prop:lim:colim:A} of \S\,\ref{append:kan:ext},
%since
%%$\mathrmbf{DOM} = \mathrmbf{Dom}^{\scriptscriptstyle{\Uparrow}}$ and 
%$\mathrmbf{List}$ is complete and cocomplete.

%%%%%%%%%%%%%%%%%%%%%%%%%%%%%%%%%%%%%%%%%%%%%%%%%%%%%%%%%%%%%%%%%%%%%%%%%%%%%%%%%%%%%%%%%%
%\newpage
\paragraph{Projections.}
%%%%%%%%%%%%%%%%%%%%%%%%%%%%%%%%%%%%%%%%%%%%%%%%%%%%%%%%%%%%%%%%%%%%%%%%%%%%%%%%%%%%%%%%%%

%
%%%%%%%%%%%%%%%%%%%%%%%%%%%%%%%%%%%%%%%%%%%%%%%%%%%%%%%%%%%%%%%%%%%%%%%%%%%%%%%%%%%%%%%%%%
%\mbox{}\newline
%\rule{142pt}{1pt}{\fbox{\textbf{Work Zone}}}\rule{142pt}{1pt}
%\newline\newline
%%%%%%%%%%%%%%%%%%%%%%%%%%%%%%%%%%%%%%%%%%%%%%%%%%%%%%%%%%%%%%%%%%%%%%%%%%%%%%%%%%%%%%%%%%
%

%Composition with table projection passages 
%(Expo.~\ref{expo:intro:tbl} in \S~\ref{sec:intro}) 
%define database projection passages.
%Projections offer an alternate representation,
%defining the three primary components of 
%databases and database morphisms: 
%diagram shapes, schemed domains and key diagrams.
Projections offer an alternate representation,
defining the two primary components of 
schemas and schema morphisms: 
arity diagrams and sort diagrams.
%Diagram shapes are direct projections,
arity diagrams and sort diagrams
%are indirect,
come from composition with list projection passages 
%It has index and set projection passages
$\mathrmbf{Set}\xleftarrow{\mathrmbfit{arity}}\mathrmbf{List}\xrightarrow{\mathrmbfit{sort}} \mathrmbf{Set}$.
%(Expo.~\ref{expo:intro:tbl} in \S~\ref{sub:sub:sec:tables}). 
%
%The schema projections are described in Fig.\,\ref{fig:fole:db:cxt:var}
%and are defined as follows.
%
The context $\mathrmbf{LIST}$ has projection passages
$\mathrmbf{Set}^{\!\scriptscriptstyle{\Uparrow}}
\xleftarrow{\mathring{\mathrmbfit{arity}}}
\mathrmbf{LIST}
\xrightarrow{\mathring{\mathrmbfit{sort}}}
\mathrmbf{Set}^{\!\scriptscriptstyle{\Uparrow}}$.
\begin{itemize}
\item 
The arity domain projection 
$\mathring{\mathrmbfit{arity}} 
= {(\mbox{-})} \circ \mathrmbfit{arity} 
: \mathrmbf{LIST} 
= \mathrmbf{List}^{\!\scriptscriptstyle{\Uparrow}}
\rightarrow \mathrmbf{Set}^{\!\scriptscriptstyle{\Uparrow}}$
\begin{itemize}
\item 
maps
a schema ${\langle{\mathrmbf{R},\mathrmbfit{S}}\rangle}$ 
to the arity diagram
$\mathring{\mathrmbfit{arity}}(\mathrmbf{R},\mathrmbfit{S})
=
{\langle{\mathrmbf{R},\mathrmbfit{S} \circ \mathrmbfit{arity}}\rangle}
%=
%{\langle{\mathrmbf{R},\mathrmbfit{S}}\rangle}
$
with the arity passage
$\mathrmbf{R}\xrightarrow[\mathrmbfit{S}{\,\circ\,}\mathrmbfit{arity}]{\,A\;}\mathrmbf{Set}$,
and 
\item 
maps 
a schema morphism 
${\langle{\mathrmbf{R}_{2},\mathrmbfit{S}_{2}}\rangle} 
\xrightarrow{{\langle{\mathrmbfit{R},\,\sigma}\rangle}}
{\langle{\mathrmbf{R}_{1},\mathrmbfit{S}_{1}}\rangle}$
to the 
%signature domain morphism 
%$\mathring{\mathrmbfit{sign}}(\mathrmbfit{R},\varsigma)
%=
%{\langle{\mathrmbf{R}_{2},\mathrmbfit{C}_{2}}\rangle} 
%\xleftarrow[{\langle{\mathrmbfit{R},\,\varsigma \circ \mathrmbfit{sign}}\rangle}]
%{{\langle{\mathrmbfit{R},\,\gamma}\rangle}}
%{\langle{\mathrmbf{R}_{1},\mathrmbfit{C}_{1}}\rangle}
%=
%{\langle{\mathrmbf{R}_{2},\mathrmbfit{C}_{2}}\rangle} 
%\xleftarrow{{\langle{\mathrmbfit{R},\,\gamma}\rangle}}
%{\langle{\mathrmbf{R}_{1},\mathrmbfit{C}_{1}}\rangle}
%$
arity diagram morphism
{\footnotesize{$
%\mathring{\mathrmbfit{dom}}(\mathcal{R}_{2})=
\mathring{\mathrmbfit{arity}}(\mathrmbfit{R},\sigma) 
= {\langle{\mathrmbfit{R},\alpha}\rangle}
:
{\langle{\mathrmbf{R}_{2},\mathrmbfit{S}_{2}}\rangle}
\rightarrow
%\xrightarrow[\mathring{\mathrmbfit{dom}}(\mathrmbfit{R},\xi)]
%{{\langle{\mathrmbfit{R},\varsigma}\rangle}} 
{\langle{\mathrmbf{R}_{1},\mathrmbfit{S}_{1}}\rangle}
%=\mathring{\mathrmbfit{dom}}(\mathcal{R}_{1})
%in $\mathrmbf{DOM}$
%with equivalent bridge pair
%$\hat{\zeta} = {\langle{\acute{\zeta},\grave{\zeta}}\rangle}
%= \xi^{\mathrm{op}}\!{\circ\,}\mathrmbfit{dom}$,
%:
%\mathrmbfit{K}_{2} \Leftarrow
%%[\;\,\psi{\;\circ\;}\mathrmbfit{key}_{\mathcal{A}_{1}}]
%%{\;\,\kappa\;}
%\mathrmbfit{R}^{\mathrm{op}}{\circ\;}\mathrmbfit{K}_{1}
$}}
%\newline
with bridge
$
%\varsigma = 
%\xi^{\mathrm{op}}{\,\circ\,}\mathrmbfit{dom} :
\mathrmbfit{A}_{2} \xRightarrow
[\,\sigma{\,\circ\,}\mathrmbfit{arity}]
{\alpha}
\mathrmbfit{R}{\;\circ\;}\mathrmbfit{A}_{1}$.
\end{itemize}
\item 
The sort domain projection 
$\mathring{\mathrmbfit{sort}} 
= {(\mbox{-})} \circ \mathrmbfit{sort} 
: \mathrmbf{LIST} 
= \mathrmbf{List}^{\!\scriptscriptstyle{\Uparrow}}
\rightarrow \mathrmbf{Set}^{\!\scriptscriptstyle{\Uparrow}}$
\begin{itemize}
\item 
maps
a schema ${\langle{\mathrmbf{R},\mathrmbfit{S}}\rangle}$ 
to the sort diagram
$\mathring{\mathrmbfit{sort}}(\mathrmbf{R},\mathrmbfit{S})
={\langle{\mathrmbf{R},\mathrmbfit{S}{\,\circ\,}\mathrmbfit{sort}}\rangle}$ 
with the sort passage
$\mathrmbf{R}\xrightarrow[\,\mathrmbfit{S}{\,\circ\,}\mathrmbfit{sort}]{\,\mathrmbfit{B}\;}\mathrmbf{Set}$,
\item 
maps 
a schema morphism 
${\langle{\mathrmbf{R}_{2},\mathrmbfit{S}_{2}}\rangle} 
\xrightarrow{{\langle{\mathrmbfit{R},\,\sigma}\rangle}}
{\langle{\mathrmbf{R}_{1},\mathrmbfit{S}_{1}}\rangle}$
to the sort diagram morphism
%to the $\mathrmbf{CLS}$-morphism 
%\newline\mbox{}\hfill
{\footnotesize{$
\mathring{\mathrmbfit{sort}}(\mathrmbfit{R},\,\sigma)
=
{\langle{\mathrmbfit{R},\beta}\rangle} :
{\langle{\mathrmbf{R}_{2},\mathrmbfit{B}_{2}}\rangle}
\rightarrow
{\langle{\mathrmbf{R}_{1},\mathrmbfit{B}_{1}}\rangle}$}}
%\hfill\mbox{}\newline
with a bridge
$\mathrmbfit{B}_{2}
\xRightarrow[\,\sigma{\,\circ\,}\mathrmbfit{sort}]{\beta}
\mathrmbfit{R} \circ \mathrmbfit{B}_{1}$.
%between (key) set diagrams.
%
\end{itemize}
\end{itemize}
\begin{figure}
\begin{center}
{{\begin{tabular}{c}
\setlength{\unitlength}{0.5pt}
\begin{picture}(240,20)(0,0)
%\put(0,80){\makebox(0,0){\footnotesize{$\mathrmbf{Set}^{\!\scriptscriptstyle{\Uparrow}}$}}}
%\put(120,80){\makebox(0,0){\footnotesize{$\mathrmbf{DOM}$}}}
%\put(240,80){\makebox(0,0){\footnotesize{$\mathrmbf{CLS}$}}}
\put(0,0){\makebox(0,0){\footnotesize{$\mathrmbf{Set}^{\!\scriptscriptstyle{\Uparrow}}$}}}
\put(120,0){\makebox(0,0){\footnotesize{$\mathrmbf{LIST}$}}}
\put(240,0){\makebox(0,0){\footnotesize{$\mathrmbf{Set}^{\!\scriptscriptstyle{\Uparrow}}$}}}
%\put(55,92){\makebox(0,0){\scriptsize{$\mathring{\mathrmbfit{arity}}$}}}
\put(55,12){\makebox(0,0){\scriptsize{$\mathring{\mathrmbfit{arity}}$}}}
%\put(185,92){\makebox(0,0){\scriptsize{$\mathring{\mathrmbfit{data}}$}}}
\put(185,12){\makebox(0,0){\scriptsize{$\mathring{\mathrmbfit{sort}}$}}}
%\put(-8,40){\makebox(0,0)[r]{\scriptsize{$\mathrmbfit{1}$}}}
%\put(112,40){\makebox(0,0)[r]{\scriptsize{$\mathring{\mathrmbfit{sign}}$}}}
%\put(248,40){\makebox(0,0)[l]{\scriptsize{$\mathring{\mathrmbfit{sort}}$}}}
%\put(85,80){\vector(-1,0){60}}
\put(85,0){\vector(-1,0){60}}
%\put(150,80){\vector(1,0){60}}
\put(150,0){\vector(1,0){60}}
%\put(0,65){\vector(0,-1){50}}
%\put(120,65){\vector(0,-1){50}}
%\put(240,65){\vector(0,-1){50}}
%\qbezier(200,30)(210,30)(220,30)\qbezier(200,30)(200,20)(200,10)
\end{picture}
\end{tabular}}}
\end{center}
\caption{\texttt{FOLE} Schema Context}
\label{fig:cxt:schema}
\end{figure}

% not needed?
%%%%%%%%%%%%%%%%%%%%%%%%%%%%%%%%%%%%%%%%%%%%%%%%%%%%%%%%%%%%%%%%%%%%%%%%%%%%%%%%
%%%%%%%%%%%%%%%%%%%%%%%%%%%%%%%%%%%%%%%%%%%%%%%%%%%%%%%%%%%%%%%%%%%%%%%%%%%%%%%%

%}% temporary beautification 
%%%%%%%%%%%%%%%%%%%%%%%%%%%%%%%%%%%%%%%%%%%%%%%%%%%%%%%%%%%%%%%%%%%%%%%%%%%%%%%%

%%%%%%%%%%%%%%%%%%%%%%%%%%%%%%%%%%%%%%%%%%%%%%%%%%%%%%%%%%%%%%%%%%%%%%%%%%%%%%%%%%%%%%%%%%
%%%%%%%%%%%%%%%%%%%%%%%%%%%%%%%%%%%%%%%%%%%%%%%%%%%%%%%%%%%%%%%%%%%%%%%%%%%%%%%%%%%%%%%%%%
%%%%%%%%%%%%%%%%%%%%%%%%%%%%%%%%%%%%%%%%%%%%%%%%%%%%%%%%%%%%%%%%%%%%%%%%%%%%%%%%%%%%%%%%%%
\newpage
\section{\texttt{FOLE} Schemed Domains}\label{sec:sch:dom}
%%%%%%%%%%%%%%%%%%%%%%%%%%%%%%%%%%%%%%%%%%%%%%%%%%%%%%%%%%%%%%%%%%%%%%%%%%%%%%%%%%%%%%%%%%
%%%%%%%%%%%%%%%%%%%%%%%%%%%%%%%%%%%%%%%%%%%%%%%%%%%%%%%%%%%%%%%%%%%%%%%%%%%%%%%%%%%%%%%%%%
%%%%%%%%%%%%%%%%%%%%%%%%%%%%%%%%%%%%%%%%%%%%%%%%%%%%%%%%%%%%%%%%%%%%%%%%%%%%%%%%%%%%%%%%%%

%{\fbox{{Need to briefly define $\mathrmbf{CLS}$ and  $\mathrmbf{LIST}$!}}}
%\newline

The schemed domain contexts defined in this paper
are listed in Tbl.~\ref{tbl:sch:dom:cxts}.
\begin{table}
\begin{center}
{\fbox{\scriptsize{\setlength{\extrarowheight}{2pt}{\begin{tabular}
{|@{\hspace{3pt}}r@{\hspace{5pt}=\hspace{5pt}}l@{\hspace{5pt}:\hspace{5pt}}l@{\hspace{20pt}}
l@{\hspace{5pt}\text{in}\hspace{5pt}}l|}
\hline
\rule{0pt}{9pt}
$\mathrmbf{DOM}$
&
$\mathrmbf{Dom}^{\scriptscriptstyle{\Uparrow}}$
&
%oplax comma 
\textit{diagram context} 
& 
Def.~\ref{def:lax:com:cxt:sch:dom:var}
&
%\S~\ref{sub:sub:sec:sch:dom:var:shape:gen}
\S~\ref{sub:sec:sch:dom:gen}
\\
%&
%$\int\mathrmbf{Cxt}\xrightarrow{\,\mathring{\mathrmbfit{dom}}\;}\mathrmbf{Adj}$
%&
%fibered context 
%& 
%Prop.~\ref{prop:fib:cxt:sch:dom:var}
%&
%\S~\ref{sub:sub:sec:sch:dom:var:shape:gen}
%\\
&
$\mathrmbf{LIST}{\times_{\mathrmbf{Set}^{\!\scriptscriptstyle{\Uparrow}}}}\mathrmbf{CLS}$
&
\textit{pullback context}
&
Prop.~\ref{prop:pullback:sch:dom:var}
&
%\S~\ref{sub:sub:sec:sch:dom:var:shape:gen}
\S~\ref{sub:sec:sch:dom:gen}
\\\hline
\rule{0pt}{10pt}
$
%\underset{\cong\mathring{\mathrmbf{List}}(X)} 
{\mathring{\mathrmbf{Dom}}(\mathcal{A})}$
&
$
%\underset{\cong\mathrmbf{List}(X)^{\scriptscriptstyle{\Uparrow}}}
{\mathrmbf{Dom}(\mathcal{A})^{\scriptscriptstyle{\Uparrow}}}$ 
&
%oplax comma 
\textit{diagram context}
& 
%Def.~\ref{def:sch:dom:A:lax:cxt}
Def.~\ref{def:dom:A:oplax:cxt}
&
%\S~\ref{sub:sub:sec:sch:dom:var:shape:typ:dom}
\S~\ref{sub:sub:sec:dom:typ:dom:lower}
\\
%&
%$\int\mathrmbf{Cxt}\xrightarrow{\,\mathring{\mathrmbfit{dom}}_{\mathcal{A}}\;}\mathrmbf{Adj}$ 
%&
%\textit{fibered context} 
%& 
%Prop.~\ref{prop:sch:dom:A:fib:cxt}
%&
%\S~\ref{sub:sub:sec:sch:dom:var:shape:typ:dom}
%\\
$\mathring{\mathrmbf{Dom}}$
&
$\int\mathrmbf{Cls}\xrightarrow{\,\hat{\mathrmbfit{dom}}\;}\mathrmbf{Adj}$
&
\textit{Grothendieck construction}
&
%Def.~\ref{def:fib:cxt:sch:dom:var:typ:dom}
Prop.~\ref{prop:fib:cxt:sch:dom:var:typ:dom:sh}
&
%\S~\ref{sub:sub:sec:sch:dom:var:shape:typ:dom}
\S~\ref{sub:sub:sec:dom:typ:dom:upper}
\\\hline
\end{tabular}}}}}
\end{center}
\caption{Schemed Domain Contexts}
\label{tbl:sch:dom:cxts}
\end{table}
%
%\newline
%Each schemed domain context has an associated tuple passage
%(Tbl.~\ref{tbl:tup:pass:sch:dom}).

%
%%%%%%%%%%%%%%%%%%%%%%%%%%%%%%%%%%%%%%%%%%%%%%%%%%%%%%%%%%%%%%%%%%%%%%%%%%%%%%%%
%%%%%%%%%%%%%%%%%%%%%%%%%%%%%%%%%%%%%%%%%%%%%%%%%%%%%%%%%%%%%%%%%%%%%%%%%%%%%%%%
\comment{% redundant
\begin{table}
\begin{center}
{\fbox{\footnotesize{\begin{tabular}{r@{\hspace{20pt}}l}
$\mathrmbf{DOM}
= \bigl(\mathrmbf{Cxt}{\,\Uparrow\,}\mathrmbf{Dom}\bigr)$
&
\textit{diagram context}
%\textit{lax comma context}
\\
$\mathrmbf{DOM}
= \mathrmbf{LIST}{\times_{\mathrmbf{Set}^{\!\scriptscriptstyle{\Uparrow}}}}\mathrmbf{CLS}$
&
\textit{pullback context}
%\\
%$\mathrmbf{DOM}=\int\mathrmbf{dom}
%\xrightarrow{\mathrmbfit{shape}}\mathrmbf{Cxt}$
%& 
%\textit{fibered context}
\end{tabular}}}}
\end{center}
\caption{Alternate Definitions of $\mathrmbf{DOM}$}
\label{tbl:alt:defs:sch:dom}
\end{table}
}% redundant
%%%%%%%%%%%%%%%%%%%%%%%%%%%%%%%%%%%%%%%%%%%%%%%%%%%%%%%%%%%%%%%%%%%%%%%%%%%%%%%%
%%%%%%%%%%%%%%%%%%%%%%%%%%%%%%%%%%%%%%%%%%%%%%%%%%%%%%%%%%%%%%%%%%%%%%%%%%%%%%%%
%

%%%%%%%%%%%%%%%%%%%%%%%%%%%%%%%%%%%%%%%%%%%%%%%%%%%%%%%%%%%%%%%%%%%%%%%%%%%%%%%%%%%%%%%%%%
%%%%%%%%%%%%%%%%%%%%%%%%%%%%%%%%%%%%%%%%%%%%%%%%%%%%%%%%%%%%%%%%%%%%%%%%%%%%%%%%%%%%%%%%%%
%\newpage
%\subsection{\texttt{FOLE} Schemed Domains}\label{sub:sec:sch:dom:var:shape}
%%%%%%%%%%%%%%%%%%%%%%%%%%%%%%%%%%%%%%%%%%%%%%%%%%%%%%%%%%%%%%%%%%%%%%%%%%%%%%%%%%%%%%%%%%
%%%%%%%%%%%%%%%%%%%%%%%%%%%%%%%%%%%%%%%%%%%%%%%%%%%%%%%%%%%%%%%%%%%%%%%%%%%%%%%%%%%%%%%%%%

%%%%%%%%%%%%%%%%%%%%%%%%%%%%%%%%%%%%%%%%%%%%%%%%%%%%%%%%%%%%%%%%%%%%%%%%%%%%%%%%%%%%%%%%%%
%\newpage
\subsection{General Case:
$\mathrmbf{DOM}$}\label{sub:sec:sch:dom:gen}
%%%%%%%%%%%%%%%%%%%%%%%%%%%%%%%%%%%%%%%%%%%%%%%%%%%%%%%%%%%%%%%%%%%%%%%%%%%%%%%%%%%%%%%%%%

In this section
we discuss the context of schemed domains 
$\mathrmbf{DOM}$,
%=\mathrmbf{Dom}^{\scriptscriptstyle{\Uparrow}}
%= \bigl(\mathrmbf{Cxt}{\,\Uparrow\,}\mathrmbf{Dom}\bigr)$, 
%(Tbl.~\ref{tbl:alt:defs:sch:dom}),
%(Fig.~\ref{fig:cxt:sch:dom}),
which mirrors 
%the context of signed domains $\mathrmbf{Dom}$ 
%(Fig.~4 in \cite{kent:fole:era:tbl}) 
the context of signed domains 
%$\mathrmbf{Dom}$,
%which is the comma context
%\[\mbox
{\footnotesize{$
%\mathrmbf{Set}\xleftarrow{\mathrmbfit{arity}}
\mathrmbf{Dom}
% = {\bigl(\mathrmbf{Set}{\,\downarrow\,}\mathrmbfit{sort}\bigr)}
%\xrightarrow{\mathrmbfit{data}}\mathrmbf{Cls}
$}\normalsize}
%\]
%
at a higher dimension. 
%
%%%%%%%%%%%%%%%%%%%%%%%%%%%%%%%%%%%%%%%%%%%%%%%%%%%%%%%%%%%%
%%%%%%%%%%%%%%%%%%%%%%%%%%%%%%%%%%%%%%%%%%%%%%%%%%%%%%%%%%%%
\footnote{
The schemed domain 
$\mathrmbf{R} \xrightarrow{\mathrmbfit{Q}} \mathrmbf{Dom}$
as developed in this paper is a derived concept built up from the basic concept of 
the signed domain 
$\mathcal{D} = {\langle{I,s,\mathrmbfit{A}}\rangle} \in \mathrmbf{Dom}$
as defined in the paper ``The {\ttfamily FOLE} Table'' \cite{kent:fole:era:tbl}.
We replace signed domains $\mathcal{D} \in \mathrmbf{Dom}$ 
with diagrams (passages) $\mathrmbf{R}\xrightarrow{\mathrmbfit{Q}}\mathrmbf{Dom}$, and 
replace signed domain morphisms $\mathcal{D}_{2}\xrightarrow{{\langle{h,f,g}\rangle}}\mathcal{D}_{1}$ in $\mathrmbf{Dom}$ 
with bridges 
%$\mathrmbfit{Q}_{2}\xRightarrow{\,\zeta\;\,}\mathrmbfit{Q}_{1}$ over $\mathrmbf{Dom}$
$\varsigma : \mathrmbfit{Q}_{2}\Rightarrow\mathrmbfit{R}{\,\circ\,}\mathrmbfit{Q}_{1}$.}
%%%%%%%%%%%%%%%%%%%%%%%%%%%%%%%%%%%%%%%%%%%%%%%%%%%%%%%%%%%%
%%%%%%%%%%%%%%%%%%%%%%%%%%%%%%%%%%%%%%%%%%%%%%%%%%%%%%%%%%%%
%
A schemed domain 
(signed domain diagram)
${\langle{\mathrmbf{R},\mathrmbfit{Q}}\rangle}$
consists of 
%a shape context $\mathrmbf{R}$ 
%and an object $\mathrmbfit{Q}$ in the fiber context $\mathrmbf{DOM}(\mathrmbf{R})$
%(\S~\ref{sub:sub:sec:sch:dom:fix:shape:gen});
%that is,
a shape context $\mathrmbf{R}$ 
and a diagram of signed domains $\mathrmbf{R}\xrightarrow{\;\mathrmbfit{Q}\;}\mathrmbf{Dom}$.
A schemed domain morphism 
${\langle{\mathrmbf{R}_{2},\mathrmbfit{Q}_{2}}\rangle} 
\xrightarrow{{\langle{\mathrmbfit{R},\,\varsigma}\rangle}}
{\langle{\mathrmbf{R}_{1},\mathrmbfit{Q}_{1}}\rangle}$
(Fig.\,\ref{fig:sch:dom:mor})
consists of 
a shape-changing passage $\mathrmbf{R}_{1}\xrightarrow{\;\mathrmbfit{R}\;\,}\mathrmbf{R}_{1}$
and a bridge 
$\varsigma : \mathrmbfit{Q}_{2}\Rightarrow\mathrmbfit{R}{\,\circ\,}\mathrmbfit{Q}_{1}$.
Composition is component-wise.
\begin{figure}
\begin{center}
{{\begin{tabular}{c}
\setlength{\unitlength}{0.56pt}
\begin{picture}(120,80)(8,0)
\put(5,80){\makebox(0,0){\footnotesize{$\mathrmbf{R}_{2}$}}}
\put(125,80){\makebox(0,0){\footnotesize{$\mathrmbf{R}_{1}$}}}
\put(65,0){\makebox(0,0){\footnotesize{$\mathrmbf{Dom}$}}}
\put(60,92){\makebox(0,0){\scriptsize{$\mathrmbfit{R}$}}}
\put(20,42){\makebox(0,0)[r]{\scriptsize{$\mathrmbfit{Q}_{2}$}}}
\put(100,42){\makebox(0,0)[l]{\scriptsize{$\mathrmbfit{Q}_{1}$}}}
\put(60,57){\makebox(0,0){\shortstack{\scriptsize{$\varsigma$}\\\large{$\Longrightarrow$}}}}
\put(20,80){\vector(1,0){80}}
\put(10,67){\vector(3,-4){38}}
\put(110,68){\vector(-3,-4){38}}
\end{picture}
\end{tabular}}}
\end{center}
\caption{Schemed Domain Morphism: $\mathrmbf{DOM}$}
\label{fig:sch:dom:mor}
\end{figure}
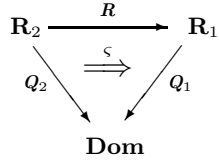
\begin{definition}\label{def:lax:com:cxt:sch:dom:var}
The context of schemed domains is 
the lax comma context
$\mathrmbf{DOM}
= \mathrmbf{Dom}^{\!\scriptscriptstyle{\Uparrow}}$,
%= \bigl({(\mbox{-})}{\,\Uparrow\,}\mathrmbf{Dom}\bigr)
%= \bigl(\mathrmbf{Cxt}{\,\Uparrow\,}\mathrmbf{Dom}\bigr)
a diagram context over signed domains.
(Def.\,\ref{def:lax:oplax} in \S\,\ref{sub:sub:sec:lax:comma:cxt})
\end{definition}
\begin{proposition}\label{prop:lim:colim:dom}
%
%The context of 
%%schemas 
%%$\mathrmbf{LIST}
%%= \mathrmbf{List}^{\!\scriptscriptstyle{\Uparrow}}$
%%%= \bigl(\mathrmbf{Cxt}{\,\Uparrow\,}\mathrmbf{List}\bigr)$,
%%(Def.~\ref{def:lax:com:cxt:schema:var})
%%and
%%the context of 
%schemed domains 
%$\mathrmbf{DOM}
%= \mathrmbf{Dom}^{\!\scriptscriptstyle{\Uparrow}}$
%%= \bigl(\mathrmbf{Cxt}{\,\Uparrow\,}\mathrmbf{Dom}\bigr)$,
%%a diagram context over signed domains.
%(Def.~\ref{def:lax:com:cxt:sch:dom:var})
%is an example of a lax comma context
%(Def.~\ref{def:lax:oplax}).
%\begin{itemize}
%\item 
%The fibered context (Grothendieck construction) $\mathrmbf{LIST}$ is complete and cocomplete
%and the projection 
%$\mathrmbf{LIST} = \int\widetilde{\mathrmbf{List}}\rightarrow\mathrmbf{Cxt}$
%is continuous and cocontinuous.
%\item 
%Since $\mathrmbf{Dom}$ is complete and cocomplete,
The fibered context (Grothendieck construction) 
$\mathrmbf{DOM} =
\mathrmbf{Dom}^{\scriptscriptstyle{\Uparrow}} = \int\hat{\mathrmbf{Dom}}$ is complete and cocomplete 
and the projection 
$\mathrmbf{Dom}^{\scriptscriptstyle{\Uparrow}}\rightarrow\mathrmbf{Cxt}
:{\langle{\mathrmbf{R},\mathrmbfit{Q}}\rangle}\mapsto\mathrmbf{R}$ 
is continuous and cocontinuous.
%The fibered context (Grothendieck construction) $\mathrmbf{DOM}$ is complete and cocomplete
%and the projection 
%$\mathrmbf{DOM} = \int
%\widetilde{\mathrmbf{Dom}}
%%\widehat{\mathrmbf{Dom}}
%\rightarrow\mathrmbf{Cxt}$
%is continuous and cocontinuous.
%%\end{itemize}
%
\end{proposition}
\begin{proof}
Use the lax parts of 
Prop.~\ref{prop:lax:fibered:A}
and
Prop.~\ref{prop:lim:colim:A} 
in \S\,\ref{append:kan:ext},
since
the context of schemed domains 
$\mathrmbf{DOM}$
is the diagram context 
%(Grothendieck construction) 
$\mathrmbf{DOM} = 
%\mathrmbf{Tbl}^{\scriptscriptstyle{\Downarrow}}
% = \int\hat{\mathrmbf{Tbl}}
\mathrmbf{Dom}^{\scriptscriptstyle{\Uparrow}} = \int\hat{\mathrmbf{Dom}}$ 
%$\mathrmbf{DB} = \mathrmbf{Tbl}^{\scriptscriptstyle{\Downarrow}}$
and
$\mathrmbf{Dom}$ is complete and cocomplete.
\mbox{}\hfill\rule{5pt}{5pt}
\end{proof}
%

%By Prop.~\ref{prop:lim:colim:A} of \S\,\ref{append:kan:ext},
%since
%%$\mathrmbf{DOM} = \mathrmbf{Dom}^{\scriptscriptstyle{\Uparrow}}$ and 
%$\mathrmbf{Dom}$ is complete and cocomplete.

%Let $r_{2} \in \mathrmbf{R}_{2}$
%be a source predicate
%with image target predicate
%$r_{1} = \mathrmbfit{R}(r_{2}) \in \mathrmbf{R}_{1}$.

%
%A signed domain morphism
%$\mathcal{D}_{2}={\langle{I_{2},s_{2},\mathcal{A}_{2}}\rangle}
%\xrightarrow{{\langle{h,f,g}\rangle}}
%{\langle{I_{1},s_{1},\mathcal{A}_{1}}\rangle}=\mathcal{D}_{1}$
%consists of 
%a signature morphism ${\langle{I_{2},s_{2},X_{2}}\rangle}\xrightarrow{{\langle{h,f}\rangle}}{\langle{I_{1},s_{1},X_{1}}\rangle}$ and 
%a type domain morphism $\mathcal{A}_{2}\xrightleftharpoons{{\langle{f,g}\rangle}}\mathcal{A}_{1}$
%with a common sort function $X_{2}\xrightarrow{f}X_{1}$.
%

%%%%%%%%%%%%%%%%%%%%%%%%%%%%%%%%%%%%%%%%%%%%%%%%%%%%%%%%%%%%%%%%%%%%%%%%%%%%%%%%%%%%%%%%%%
%\newpage
%\subsubsection{Projections in $\mathrmbf{DOM}$.}
%\label{sub:sub:sec:dom:proj}
%%%%%%%%%%%%%%%%%%%%%%%%%%%%%%%%%%%%%%%%%%%%%%%%%%%%%%%%%%%%%%%%%%%%%%%%%%%%%%%%%%%%%%%%%%

%%%%%%%%%%%%%%%%%%%%%%%%%%%%%%%%%%%%%%%%%%%%%%%%%%%%%%%%%%%%%%%%%%%%%%%%%%%%%%%%%%%%%%%%%%
%\newpage
\paragraph{Projections.}
%%%%%%%%%%%%%%%%%%%%%%%%%%%%%%%%%%%%%%%%%%%%%%%%%%%%%%%%%%%%%%%%%%%%%%%%%%%%%%%%%%%%%%%%%%

%Composition with table projection passages 
%(Expo.~\ref{expo:intro:tbl} in \S~\ref{sec:intro}) 
%define database projection passages.
%Projections offer an alternate representation,
%defining the three primary components of 
%databases and database morphisms: 
%diagram shapes, schemed domains and key diagrams.
Projections offer an alternate representation,
defining the three primary components of 
schemed domains and schemed domain morphisms: 
diagram shapes, signature domains and type domain diagrams.
Diagram shapes are direct projections,
whereas signature domains and type domain diagrams
are indirect,
coming from composition with signed domain projection passages 
%
%\[\mbox{\footnotesize{$
%\mathrmbf{List}\xleftarrow{\mathrmbfit{sign}}
%\mathrmbf{Dom}=\mathrmbf{List}{\times_{\mathrmbf{Set}}}\mathrmbf{Cls}
%\xrightarrow{\mathrmbfit{data}}\mathrmbf{Cls}
%$,}\normalsize}\]
%
(defined in detail in 
\S2.3 of the paper
``The {\ttfamily FOLE} Table''
\cite{kent:fole:era:tbl}).
The schemed domain projections are described in 
Fig.\,\ref{fig:cxt:sch:dom}
%{fig:fole:sch:dom:cxt:var}
and defined as follows.
\begin{itemize}
\item 
The signature domain projection 
$\mathring{\mathrmbfit{sign}} 
= {(\mbox{-})} \circ \mathrmbfit{sign} 
: \mathrmbf{DOM} \rightarrow \mathrmbf{LIST}$
\begin{itemize}
\item 
maps
a schemed domain ${\langle{\mathrmbf{R},\mathrmbfit{Q}}\rangle}$ 
to the schema
$\mathring{\mathrmbfit{sign}}(\mathrmbf{R},\mathrmbfit{Q})
=
{\langle{\mathrmbf{R},\mathrmbfit{Q} \circ \mathrmbfit{sign}}\rangle}
=
{\langle{\mathrmbf{R},\mathrmbfit{S}}\rangle}$
with the signature passage
$\mathrmbf{R}\xrightarrow[\,\mathrmbfit{Q}{\,\circ\,}\mathrmbfit{sign}]{\,\mathrmbfit{S}\;}\mathrmbf{List}$,
and 
\item 
maps 
a schemed domain morphism 
${\langle{\mathrmbf{R}_{2},\mathrmbfit{Q}_{2}}\rangle} 
\xrightarrow{{\langle{\mathrmbfit{R},\,\varsigma}\rangle}}
{\langle{\mathrmbf{R}_{1},\mathrmbfit{Q}_{1}}\rangle}$
to the 
%signature domain morphism 
%$\mathring{\mathrmbfit{sign}}(\mathrmbfit{R},\varsigma)
%=
%{\langle{\mathrmbf{R}_{2},\mathrmbfit{C}_{2}}\rangle} 
%\xleftarrow[{\langle{\mathrmbfit{R},\,\varsigma \circ \mathrmbfit{sign}}\rangle}]
%{{\langle{\mathrmbfit{R},\,\gamma}\rangle}}
%{\langle{\mathrmbf{R}_{1},\mathrmbfit{C}_{1}}\rangle}
%=
%{\langle{\mathrmbf{R}_{2},\mathrmbfit{C}_{2}}\rangle} 
%\xleftarrow{{\langle{\mathrmbfit{R},\,\gamma}\rangle}}
%{\langle{\mathrmbf{R}_{1},\mathrmbfit{C}_{1}}\rangle}
%$
schema morphism
{\footnotesize{$
%\mathring{\mathrmbfit{dom}}(\mathcal{R}_{2})=
\mathring{\mathrmbfit{sign}}(\mathrmbfit{R},\varsigma) 
= {\langle{\mathrmbfit{R},\gamma}\rangle}
:
{\langle{\mathrmbf{R}_{2},\mathrmbfit{S}_{2}}\rangle}
\rightarrow
%\xrightarrow[\mathring{\mathrmbfit{dom}}(\mathrmbfit{R},\xi)]
%{{\langle{\mathrmbfit{R},\varsigma}\rangle}} 
{\langle{\mathrmbf{R}_{1},\mathrmbfit{S}_{1}}\rangle}
%=\mathring{\mathrmbfit{dom}}(\mathcal{R}_{1})
%in $\mathrmbf{DOM}$
%with equivalent bridge pair
%$\hat{\zeta} = {\langle{\acute{\zeta},\grave{\zeta}}\rangle}
%= \xi^{\mathrm{op}}\!{\circ\,}\mathrmbfit{dom}$,
%:
%\mathrmbfit{K}_{2} \Leftarrow
%%[\;\,\psi{\;\circ\;}\mathrmbfit{key}_{\mathcal{A}_{1}}]
%%{\;\,\kappa\;}
%\mathrmbfit{R}^{\mathrm{op}}{\circ\;}\mathrmbfit{K}_{1}
$}}
%\newline
with bridge
$
%\varsigma = 
%\xi^{\mathrm{op}}{\,\circ\,}\mathrmbfit{dom} :
\mathrmbfit{S}_{2} \xRightarrow
[\,\varsigma{\,\circ\,}\mathrmbfit{sign}]
{\gamma}
\mathrmbfit{R}{\;\circ\;}\mathrmbfit{S}_{1}$.
\end{itemize}
\item 
The type domain projection 
$\mathring{\mathrmbfit{data}} 
= {(\mbox{-})} \circ \mathrmbfit{data} 
: \mathrmbf{DOM} \rightarrow \mathrmbf{CLS}$ 
\begin{itemize}
\item 
maps
a schemed domain ${\langle{\mathrmbf{R},\mathrmbfit{Q}}\rangle}$ 
to the type domain diagram
$\mathring{\mathrmbfit{data}}(\mathrmbf{R},\mathrmbfit{Q})
={\langle{\mathrmbf{R},\mathrmbfit{C}}\rangle}$ 
with the type domain passage
$\mathrmbf{R}\xrightarrow[\,\mathrmbfit{Q}{\,\circ\,}\mathrmbfit{data}]{\,\mathrmbfit{C}\;}\mathrmbf{Set}$,
\item 
and maps 
a schemed domain morphism 
${\langle{\mathrmbf{R}_{2},\mathrmbfit{Q}_{2}}\rangle} 
\xrightarrow{{\langle{\mathrmbfit{R},\,\varsigma}\rangle}}
{\langle{\mathrmbf{R}_{1},\mathrmbfit{Q}_{1}}\rangle}$
to the $\mathrmbf{CLS}$-morphism 
%\newline\mbox{}\hfill
{\footnotesize{$
\mathring{\mathrmbfit{data}}(\mathrmbfit{R},\,\varsigma)
=
{\langle{\mathrmbfit{R},\gamma}\rangle} :
{\langle{\mathrmbf{R}_{2},\mathrmbfit{C}_{2}}\rangle}
\rightarrow
{\langle{\mathrmbf{R}_{1},\mathrmbfit{C}_{1}}\rangle}$}}
%\hfill\mbox{}\newline
with a bridge
$\mathrmbfit{C}_{2}
\xRightarrow[\,\varsigma{\,\circ\,}\mathrmbfit{data}]{\gamma}
\mathrmbfit{R} \circ \mathrmbfit{C}_{1}$.
%between (key) set diagrams.
%
\end{itemize}
\end{itemize}
%

%%%%%%%%%%%%%%%%%%%%%%%%%%%%%%%%%%%%%%%%%%%%%%%%%%%%%%%%%%%%%%%%%%%%%%%%%%%%%%%%%%%%%%%%%%
%\newpage
\paragraph{Pullback Context.}
%%%%%%%%%%%%%%%%%%%%%%%%%%%%%%%%%%%%%%%%%%%%%%%%%%%%%%%%%%%%%%%%%%%%%%%%%%%%%%%%%%%%%%%%%%

The context $\mathrmbf{DOM}$ has projection passages
(RHS Fig.~\ref{fig:cxt:sch:dom})
$\mathrmbf{LIST}
\xleftarrow{\mathring{\mathrmbfit{sign}}}\mathrmbf{DOM}\xrightarrow{\mathring{\mathrmbfit{data}}}
\mathrmbf{CLS}$.
%
%%%%%%%%%%%%%%%%%%%%%%%%%%%%%%%%%%%%%%%%%%%%%%%%%%%%%%%%%%%%%%%%%%%%%%%%%%%%%%%%
%%%%%%%%%%%%%%%%%%%%%%%%%%%%%%%%%%%%%%%%%%%%%%%%%%%%%%%%%%%%%%%%%%%%%%%%%%%%%%%%
\footnote{
$\mathrmbf{LIST}=\mathrmbf{List}^{\!\scriptscriptstyle{\Uparrow}}$ 
is the context of schemas.
A schema ${\langle{\mathrmbf{R},\mathrmbfit{S}}\rangle}$
consists of 
a shape context $\mathrmbf{R}$ 
and a diagram of signatures (lists) 
$\mathrmbf{R}\xrightarrow{\;\mathrmbfit{S}\;}\mathrmbf{List}$.
A schema morphism 
${\langle{\mathrmbf{R}_{2},\mathrmbfit{S}_{2}}\rangle} 
\xrightarrow{{\langle{\mathrmbfit{R},\,\sigma}\rangle}}
{\langle{\mathrmbf{R}_{1},\mathrmbfit{S}_{1}}\rangle}$
consists of 
a shape-changing passage 
$\mathrmbf{R}_{1}\xrightarrow{\;\mathrmbfit{R}\;\,}\mathrmbf{R}_{1}$
and a bridge 
$\mathrmbfit{S}_{2}\xRightarrow{\;\sigma\;\,}\mathrmbfit{R}{\,\circ\,}\mathrmbfit{S}_{1}$.
}
%%%%%%%%%%%%%%%%%%%%%%%%%%%%%%%%%%%%%%%%%%%%%%%%%%%%%%%%%%%%%%%%%%%%%%%%%%%%%%%%
%%%%%%%%%%%%%%%%%%%%%%%%%%%%%%%%%%%%%%%%%%%%%%%%%%%%%%%%%%%%%%%%%%%%%%%%%%%%%%%%
%
%
%%%%%%%%%%%%%%%%%%%%%%%%%%%%%%%%%%%%%%%%%%%%%%%%%%%%%%%%%%%%%%%%%%%%%%%%%%%%%%%%
%%%%%%%%%%%%%%%%%%%%%%%%%%%%%%%%%%%%%%%%%%%%%%%%%%%%%%%%%%%%%%%%%%%%%%%%%%%%%%%%
\footnote{
$\mathrmbf{CLS}=\mathrmbf{Cls}^{\!\scriptscriptstyle{\Uparrow}}$.
is the context of type domain diagrams.
A type domain diagram 
${\langle{\mathrmbf{R},\mathrmbfit{A}}\rangle}$
consists of 
a shape context $\mathrmbf{R}$ 
and a diagram of type domains 
$\mathrmbf{R}\xrightarrow{\;\mathrmbfit{A}\;}\mathrmbf{Cls}$.
A morphism of type domain diagrams 
${\langle{\mathrmbf{R}_{2},\mathrmbfit{A}_{2}}\rangle}
\xrightarrow{{\langle{\mathrmbfit{R},\alpha}\rangle}}
{\langle{\mathrmbf{R}_{1},\mathrmbfit{A}_{1}}\rangle}$
consists of 
a shape-changing passage 
$\mathrmbf{R}_{1}\xrightarrow{\;\mathrmbfit{R}\;\,}\mathrmbf{R}_{1}$
and a bridge 
$\mathrmbfit{A}_{2}\xRightarrow{\;\alpha\;\,}\mathrmbfit{R}{\,\circ\,}\mathrmbfit{A}_{1}$.
}
%%%%%%%%%%%%%%%%%%%%%%%%%%%%%%%%%%%%%%%%%%%%%%%%%%%%%%%%%%%%%%%%%%%%%%%%%%%%%%%%
%%%%%%%%%%%%%%%%%%%%%%%%%%%%%%%%%%%%%%%%%%%%%%%%%%%%%%%%%%%%%%%%%%%%%%%%%%%%%%%%
%

%
\begin{itemize}
\item 
Using projections,
%a $\mathrmbf{DOM}$-object 
%${\langle{\mathrmbf{R},\mathrmbfit{S},\mathrmbfit{A}}\rangle}$,
%called 
a schemed domain
${\langle{\mathrmbf{R},\mathrmbfit{S},\mathrmbfit{A}}\rangle}$
consists of 
a schema
(signature diagram)
%signature diagram 
${\langle{\mathrmbf{R},\mathrmbfit{S}}\rangle}\in\mathrmbf{LIST}$
and a type domain diagram ${\langle{\mathrmbf{R},\mathrmbfit{A}}\rangle}\in\mathrmbf{CLS}$
with common sort diagram 
$\mathring{\mathrmbfit{sort}}(\mathrmbf{R},\mathrmbfit{S})
={\langle{\mathrmbf{R},\mathrmbfit{X}}\rangle}
=\mathring{\mathrmbfit{sort}}(\mathrmbf{R},\mathrmbfit{A})
\in\mathrmbf{Set}^{\!\scriptscriptstyle{\Uparrow}}$
(and hence common shape $\mathrmbf{R}$).
The passages
$\mathrmbf{R}\xrightarrow{\mathrmbfit{S}}\mathrmbf{List}$
and
$\mathrmbf{R}\xrightarrow{\mathrmbfit{A}}\mathrmbf{Cls}$
satisfy
$\mathrmbfit{S}{\,\circ\,}\mathrmbfit{sort}=\mathrmbfit{X}=\mathrmbfit{A}{\,\circ\,}\mathrmbfit{sort}$.
%
%%%%%%%%%%%%%%%%%%%%%%%%%%%%%%%%%%%%%%%%%%%%%%%%%%%%%%%%%%%%%%%%%%%%%%%%%%%%%%%%
%%%%%%%%%%%%%%%%%%%%%%%%%%%%%%%%%%%%%%%%%%%%%%%%%%%%%%%%%%%%%%%%%%%%%%%%%%%%%%%%
\footnote{
$\mathrmbf{LIST}\xrightarrow{\mathring{\mathrmbfit{sort}}}\mathrmbf{Set}^{\!\scriptscriptstyle{\Uparrow}}$
is the sort projection for the fibered context of schemas (signature diagrams),
%(Thm.~\ref{thm:fib:cxt:sign:set}),
%the Grothendieck construction of 
%the sort diagram indexed adjunction of signature diagrams
%$\mathrmbf{Set}^{\!\scriptscriptstyle{\Uparrow}}\xrightarrow{\;\mathring{\mathrmbfit{list}}\;}\mathrmbf{ADJ}$
and 
$\mathrmbf{CLS}\xrightarrow{\mathring{\mathrmbfit{sort}}}\mathrmbf{Set}^{\!\scriptscriptstyle{\Uparrow}}$
is the sort projection for type domain diagrams.}
%%%%%%%%%%%%%%%%%%%%%%%%%%%%%%%%%%%%%%%%%%%%%%%%%%%%%%%%%%%%%%%%%%%%%%%%%%%%%%%%
%%%%%%%%%%%%%%%%%%%%%%%%%%%%%%%%%%%%%%%%%%%%%%%%%%%%%%%%%%%%%%%%%%%%%%%%%%%%%%%%
%
%
\item 
A 
%$\mathrmbf{DOM}$-morphism
%called a 
schemed domain morphism
${\langle{\mathrmbf{R}_{2},\mathrmbfit{S}_{2},\mathrmbfit{A}_{2}}\rangle}
\xrightarrow{{\langle{\mathrmbfit{R},\sigma,\alpha}\rangle}}
{\langle{\mathrmbf{R}_{1},\mathrmbfit{S}_{1},\mathrmbfit{A}_{1}}\rangle}$
consists of 
a morphism of signature diagrams 
${\langle{\mathrmbf{R}_{2},\mathrmbfit{S}_{2}}\rangle}
\xrightarrow{{\langle{\mathrmbfit{R},\sigma}\rangle}}
{\langle{\mathrmbf{R}_{1},\mathrmbfit{S}_{1}}\rangle}$
in $\mathrmbf{LIST}$
and
a morphism of type domain diagrams 
${\langle{\mathrmbf{R}_{2},\mathrmbfit{A}_{2}}\rangle}
\xrightarrow{{\langle{\mathrmbfit{R},\alpha}\rangle}}
{\langle{\mathrmbf{R}_{1},\mathrmbfit{A}_{1}}\rangle}$
in $\mathrmbf{CLS}$
with common sort diagram morphism 
$\mathring{\mathrmbfit{sort}}(\mathrmbfit{R},\sigma)
={\langle{\mathrmbfit{R},\omega}\rangle}
=\mathring{\mathrmbfit{sort}}(\mathrmbfit{R},\alpha)
\in\mathrmbf{Set}^{\!\scriptscriptstyle{\Uparrow}}$
(and hence common shape-changing passage $\mathrmbf{R}_{2}\xrightarrow{\,\mathrmbfit{R}\;}\mathrmbf{R}_{1}$).
The bridges
$\mathrmbfit{S}_{2}\xRightarrow{\;\sigma\;\,}\mathrmbfit{R}{\,\circ\,}\mathrmbfit{S}_{1}$
and
$\mathrmbfit{A}_{2}\xRightarrow{\;\alpha\;\,}\mathrmbfit{R}{\,\circ\,}\mathrmbfit{A}_{1}$
satisfy
$\sigma{\,\circ\,}\mathrmbfit{sort}=\alpha{\,\circ\,}\mathrmbfit{sort}$.
\end{itemize}
\begin{proposition}\label{prop:pullback:sch:dom:var}
The context of schemed domains $\mathrmbf{DOM}$ is the \underline{pullback context} 
(fibered product)
\[\mbox{\footnotesize{$
\mathrmbf{LIST}\xleftarrow{\mathring{\mathrmbfit{sign}}}
\mathrmbf{DOM}=\mathrmbf{LIST}{\times_{\mathrmbf{Set}^{\!\scriptscriptstyle{\Uparrow}}}}\mathrmbf{CLS}
\xrightarrow{\mathring{\mathrmbfit{data}}}\mathrmbf{CLS}
$}\normalsize}\]
for the opspan of passages
$\mathrmbf{LIST}\xrightarrow{\mathring{\mathrmbfit{sort}}}\mathrmbf{Set}^{\!\scriptscriptstyle{\Uparrow}}\xleftarrow{\mathring{\mathrmbfit{sort}}}\mathrmbf{CLS}$
(Fig.~\ref{fig:cxt:sch:dom} RHS).
\end{proposition}
\begin{proof}
%
%From a different point-of-view,
From 
\S2.3 of the paper
``The {\ttfamily FOLE} Table''
\cite{kent:fole:era:tbl},
%
%%%%%%%%%%%%%%%%%%%%%%%%%%%%%%%%%%%%%%%%%%%%%%%%%%%%%%%%%%%%%%%%%%%%%%%%%%%%%%%%
%%%%%%%%%%%%%%%%%%%%%%%%%%%%%%%%%%%%%%%%%%%%%%%%%%%%%%%%%%%%%%%%%%%%%%%%%%%%%%%%
\comment{% details
a signed domain
$\mathcal{D} = {\langle{\mathcal{S},\mathcal{A}}\rangle}$
consists of 
a signature $\mathcal{S}={\langle{I,s,X}\rangle}$ and 
a type domain $\mathcal{A}={\langle{X,Y,\models_{\mathcal{A}}}\rangle}$ 
with common sort set $X$,
and
a signed domain morphism
${\langle{\mathcal{S}_{2},\mathcal{A}_{2}}\rangle}\xrightarrow{{\langle{h,f,g}\rangle}}{\langle{\mathcal{S}_{1},\mathcal{A}_{1}}\rangle}$
consists of a signature morphism $\mathcal{S}_{2}\xrightarrow{{\langle{h,f}\rangle}}\mathcal{S}_{1}$
and a type domain morphism $\mathcal{A}_{2}\xrightleftharpoons{{\langle{f,g}\rangle}}\mathcal{A}_{1}$
with common sort function $X_{2}\xrightarrow{f}X_{1}$.
%Hence,
%$\mathrmbf{Dom}$ can also be defined as the fibered product
%
Hence,
}% details
%%%%%%%%%%%%%%%%%%%%%%%%%%%%%%%%%%%%%%%%%%%%%%%%%%%%%%%%%%%%%%%%%%%%%%%%%%%%%%%%
%%%%%%%%%%%%%%%%%%%%%%%%%%%%%%%%%%%%%%%%%%%%%%%%%%%%%%%%%%%%%%%%%%%%%%%%%%%%%%%%
%
the context of signed domains $\mathrmbf{Dom}$ can be defined as the fibered product
\[\mbox{\footnotesize{$
\mathrmbf{List}\xleftarrow{\mathrmbfit{sign}}
\mathrmbf{Dom}=\mathrmbf{List}{\times_{\mathrmbf{Set}}}\mathrmbf{Cls}
\xrightarrow{\mathrmbfit{data}}\mathrmbf{Cls}
$,}\normalsize}\]
for the opspan of passages
$\mathrmbf{List}\xrightarrow{\mathrmbfit{sort}}\mathrmbf{Set}\xleftarrow{\mathrmbfit{sort}}\mathrmbf{Cls}$.
%\end{itemize}
%
\begin{center}
\begin{tabular}{c}
%%%%%%%%%%%%%%%%%%%%%%%%%%%%%%%%%%%%%%%%%%%%%%%%%%%%%%%%%%%%%%%%%%%%%%%%%%%%%%%%%%%%%%%%%%
{{\begin{tabular}{c}
\setlength{\unitlength}{0.55pt}
%\begin{picture}(220,120)(110,-25)
%
%{{\begin{tabular}{c}
%\setlength{\unitlength}{0.5pt}
\begin{picture}(320,100)(0,2)
%\put(0,0){\framebox(320,80){}}
\put(0,80){\makebox(0,0){\footnotesize{$\mathrmbf{Set}$}}}
\put(163,80){\makebox(0,0){\footnotesize{$\mathrmbf{Dom}=\bigl(\mathrmbf{Set}{\,\downarrow\,}\mathrmbfit{sort}\bigr)$}}}
\put(320,80){\makebox(0,0){\footnotesize{$\mathrmbf{Cls}$}}}
\put(0,0){\makebox(0,0){\footnotesize{$\mathrmbf{Set}$}}}
\put(163,0){\makebox(0,0){\footnotesize{$\mathrmbf{List}=\bigl(\mathrmbf{Set}{\,\downarrow\,}\mathrmbf{Set}\bigr)$}}}
\put(320,0){\makebox(0,0){\footnotesize{$\mathrmbf{Set}$}}}
\put(45,92){\makebox(0,0){\scriptsize{$\mathrmbfit{arity}$}}}
\put(275,92){\makebox(0,0){\scriptsize{$\mathrmbfit{data}$}}}
\put(55,-12){\makebox(0,0){\scriptsize{$\mathrmbfit{arity}$}}}
\put(265,-12){\makebox(0,0){\scriptsize{$\mathrmbfit{sort}$}}}
\put(-8,40){\makebox(0,0)[r]{\scriptsize{$\mathrmbfit{1}$}}}
\put(152,40){\makebox(0,0)[r]{\scriptsize{$\mathrmbfit{sign}$}}}
\put(328,40){\makebox(0,0)[l]{\scriptsize{$\mathrmbfit{sort}$}}}
\put(68,80){\vector(-1,0){43}}
\put(252,80){\vector(1,0){43}}
\put(78,0){\vector(-1,0){55}}
\put(242,0){\vector(1,0){55}}
\put(0,65){\vector(0,-1){50}}
\put(160,65){\vector(0,-1){50}}
\put(320,65){\vector(0,-1){50}}
\qbezier(280,30)(290,30)(300,30)\qbezier(280,30)(280,20)(280,10)
\end{picture}
\end{tabular}}}
%%%%%%%%%%%%%%%%%%%%%%%%%%%%%%%%%%%%%%%%%%%%%%%%%%%%%%%%%%%%%%%%%%%%%%%%%%%%%%%%%%%%%%%%%%
\\
\\
%%%%%%%%%%%%%%%%%%%%%%%%%%%%%%%%%%%%%%%%%%%%%%%%%%%%%%%%%%%%%%%%%%%%%%%%%%%%%%%%%%%%%%%%%%
\comment{\footnotesize{$\begin{array}{r@{\hspace{10pt}}r@{\hspace{5pt}}l@{\hspace{5pt}:\hspace{5pt}}l}
\text{schema}
&
& \mathrmbfit{sign} 
& \mathrmbf{Dom} \rightarrow \mathrmbf{List}
\\
\text{data}
&
& \mathrmbfit{data} 
& \mathrmbf{Dom} \rightarrow \mathrmbf{Cls}
\end{array}$}}
%%%%%%%%%%%%%%%%%%%%%%%%%%%%%%%%%%%%%%%%%%%%%%%%%%%%%%%%%%%%%%%%%%%%%%%%%%%%%%%%%%%%%%%%%%
\end{tabular}
\end{center}
Extend this to diagram contexts
$\mathrmbf{LIST} 
= \mathrmbf{List}^{\!\scriptscriptstyle{\Uparrow}}$
and
$\mathrmbf{CLS} 
= \mathrmbf{Cls}^{\!\scriptscriptstyle{\Uparrow}}$.
\mbox{}\hfill\rule{5pt}{5pt}
\end{proof}
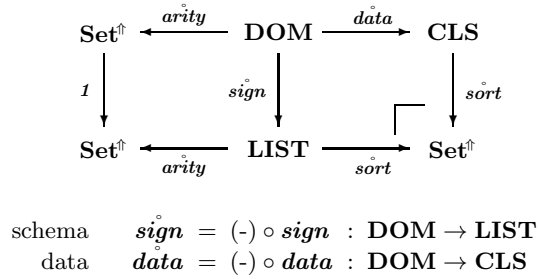
\begin{figure}
\begin{center}
\begin{tabular}{c}
%%%%%%%%%%%%%%%%%%%%%%%%%%%%%%%%%%%%%%%%%%%%%%%%%%%%%%%%%%%%%%%%%%%%%%%%%%%%%%%%%%%%%%%%%%
{{\begin{tabular}{c}
\setlength{\unitlength}{0.55pt}
%\begin{picture}(220,120)(110,-25)
%
%{{\begin{tabular}{c}
%\setlength{\unitlength}{0.5pt}
\begin{picture}(240,120)(0,-20)
\put(0,80){\makebox(0,0){\footnotesize{$\mathrmbf{Set}^{\!\scriptscriptstyle{\Uparrow}}$}}}
\put(120,80){\makebox(0,0){\footnotesize{$\mathrmbf{DOM}$}}}
\put(240,80){\makebox(0,0){\footnotesize{$\mathrmbf{CLS}$}}}
\put(0,0){\makebox(0,0){\footnotesize{$\mathrmbf{Set}^{\!\scriptscriptstyle{\Uparrow}}$}}}
\put(120,0){\makebox(0,0){\footnotesize{$\mathrmbf{LIST}$}}}
\put(240,0){\makebox(0,0){\footnotesize{$\mathrmbf{Set}^{\!\scriptscriptstyle{\Uparrow}}$}}}
\put(55,92){\makebox(0,0){\scriptsize{$\mathring{\mathrmbfit{arity}}$}}}
\put(55,-12){\makebox(0,0){\scriptsize{$\mathring{\mathrmbfit{arity}}$}}}
\put(185,92){\makebox(0,0){\scriptsize{$\mathring{\mathrmbfit{data}}$}}}
\put(185,-12){\makebox(0,0){\scriptsize{$\mathring{\mathrmbfit{sort}}$}}}
\put(-8,40){\makebox(0,0)[r]{\scriptsize{$\mathrmbfit{1}$}}}
\put(112,40){\makebox(0,0)[r]{\scriptsize{$\mathring{\mathrmbfit{sign}}$}}}
\put(248,40){\makebox(0,0)[l]{\scriptsize{$\mathring{\mathrmbfit{sort}}$}}}
\put(85,80){\vector(-1,0){60}}
\put(85,0){\vector(-1,0){60}}
\put(150,80){\vector(1,0){60}}
\put(150,0){\vector(1,0){60}}
\put(0,65){\vector(0,-1){50}}
\put(120,65){\vector(0,-1){50}}
\put(240,65){\vector(0,-1){50}}
\qbezier(200,30)(210,30)(220,30)\qbezier(200,30)(200,20)(200,10)
\end{picture}
%\end{tabular}}}
%\end{center}
\end{tabular}}}
%%%%%%%%%%%%%%%%%%%%%%%%%%%%%%%%%%%%%%%%%%%%%%%%%%%%%%%%%%%%%%%%%%%%%%%%%%%%%%%%%%%%%%%%%%
\\
\\
%%%%%%%%%%%%%%%%%%%%%%%%%%%%%%%%%%%%%%%%%%%%%%%%%%%%%%%%%%%%%%%%%%%%%%%%%%%%%%%%%%%%%%%%%%
{\footnotesize{$\begin{array}{r@{\hspace{16pt}}r@{\hspace{5pt}=\hspace{5pt}}l@{\hspace{5pt}:\hspace{5pt}}l}
\text{schema}
&
\mathring{\mathrmbfit{sign}} 
& {(\mbox{-})} \circ \mathrmbfit{sign} 
& \mathrmbf{DOM} \rightarrow \mathrmbf{LIST}
\\
\text{data}
&
\mathring{\mathrmbfit{data}} 
& {(\mbox{-})} \circ \mathrmbfit{data} 
& \mathrmbf{DOM} \rightarrow \mathrmbf{CLS}
\end{array}$}}
%%%%%%%%%%%%%%%%%%%%%%%%%%%%%%%%%%%%%%%%%%%%%%%%%%%%%%%%%%%%%%%%%%%%%%%%%%%%%%%%%%%%%%%%%%
\end{tabular}
\end{center}
%\caption{\texttt{FOLE} Schemed Domain Mathematical Context}
\caption{Schemed Domain Context: $\mathrmbf{DOM}$}
%\label{fig:fole:sch:dom:cxt:var}
\label{fig:cxt:sch:dom}
\end{figure}

% not needed
%%%%%%%%%%%%%%%%%%%%%%%%%%%%%%%%%%%%%%%%%%%%%%%%%%%%%%%%%%%%%%%%%%%%%%%%%%%%%%%%%
%%%%%%%%%%%%%%%%%%%%%%%%%%%%%%%%%%%%%%%%%%%%%%%%%%%%%%%%%%%%%%%%%%%%%%%%%%%%%%%%%

%%%%%%%%%%%%%%%%%%%%%%%%%%%%%%%%%%%%%%%%%%%%%%%%%%%%%%%%%%%%%%%%%%%%%%%%%%%%%%%%%%%%%%%%%%
\newpage
\subsection{Type Domain Indexing}
\label{sub:sec:rel:sch:typ:dom}
%:
%$\mathring{\mathrmbf{Dom}}(\mathcal{A})\cong\mathring{\mathrmbf{List}}(X)$,
%$\mathring{\mathrmbf{Dom}}$.}\label{sub:sub:sec:sch:dom:var:shape:typ:dom}
%%%%%%%%%%%%%%%%%%%%%%%%%%%%%%%%%%%%%%%%%%%%%%%%%%%%%%%%%%%%%%%%%%%%%%%%%%%%%%%%%%%%%%%%%%

A type domain, which constrains a signed domain, is an indexed collection of data types.
Here we define schemed domains with fixed type domains.
Schemed domains with fixed signed domains 
(headers plus datatypes)
are trivial.

%%%%%%%%%%%%%%%%%%%%%%%%%%%%%%%%%%%%%%%%%%%%%%%%%%%%%%%%%%%%%%%%%%%%%%%%%%%%%%%%%%%%%%%%%%
%\newpage
\subsubsection{Lower Aspect: $\mathring{\mathrmbf{Dom}}(\mathcal{A})$}
\label{sub:sub:sec:dom:typ:dom:lower}
%%%%%%%%%%%%%%%%%%%%%%%%%%%%%%%%%%%%%%%%%%%%%%%%%%%%%%%%%%%%%%%%%%%%%%%%%%%%%%%%%%%%%%%%%%

Let 
$\mathcal{A} = {\langle{X,Y,\models_{\mathcal{A}}}\rangle}$ be a fixed type domain.
\begin{definition}\label{def:dom:A:oplax:cxt}
The context of $\mathcal{A}$-schemed domains is 
%the lax comma context
$\mathring{\mathrmbf{Dom}}(\mathcal{A}) 
= \mathrmbf{Dom}(\mathcal{A})^{\scriptscriptstyle{\Uparrow}}
= \mathrmbf{List}(X)^{\scriptscriptstyle{\Uparrow}}$,
a diagram context over $\mathcal{A}$-signed domains; i.e. $X$-signatures.
%(Def.~\ref{def:lax:oplax} in \S~\ref{sub:sub:sec:math:context})
\end{definition}
A relational $\mathcal{A}$-schemed domain ${\langle{\mathrmbf{R},\mathrmbfit{S}}\rangle}$ is a diagram of $\mathcal{A}$-signatures,
consisting of 
%a shape context $\mathrmbf{R}$ 
%and an object $\mathrmbfit{T}$ in the fiber context $\mathrmbf{Db}^{\mathrmbf{R}}(\mathcal{A})$
%(Def.~\ref{def:fbr:db:fix:typ:dom} of \S~\ref{sub:sub:sec:rel:db:fix:shape:typ:dom:upper})
%that is,
a shape context $\mathrmbf{R}$ 
and a passage $\mathrmbf{R}\,\xrightarrow{\,\mathrmbfit{S}\;}\mathrmbf{List}(X)$.
A $\mathcal{A}$-schemed domain morphism
${\langle{\mathrmbf{R}_{2},\mathrmbfit{S}_{2}}\rangle} 
\xrightarrow{{\langle{\mathrmbfit{R},\,\varphi}\rangle}}
{\langle{\mathrmbf{R}_{1},\mathrmbfit{S}_{1}}\rangle}$
consists of 
a shape-changing passage $\mathrmbf{R}_{1}\xrightarrow{\;\mathrmbfit{R}\;\,}\mathrmbf{R}_{1}$
and a bridge 
$\mathrmbfit{S}_{2}\xRightarrow{\;\;\varphi\,}
\mathrmbfit{R}{\circ}\mathrmbfit{S}_{1}$.

\begin{figure}
\begin{center}
{{\begin{tabular}{c}
\setlength{\unitlength}{0.56pt}
\begin{picture}(120,80)(8,0)
\put(5,80){\makebox(0,0){\footnotesize{$\mathrmbf{R}_{2}$}}}
\put(125,80){\makebox(0,0){\footnotesize{$\mathrmbf{R}_{1}$}}}
\put(65,0){\makebox(0,0){\footnotesize{$\mathrmbf{List}(X)$}}}
\put(60,92){\makebox(0,0){\scriptsize{$\mathrmbfit{R}$}}}
\put(20,42){\makebox(0,0)[r]{\scriptsize{$\mathrmbfit{S}_{2}$}}}
\put(100,42){\makebox(0,0)[l]{\scriptsize{$\mathrmbfit{S}_{1}$}}}
\put(60,57){\makebox(0,0){\shortstack{\scriptsize{$\varphi$}\\\large{$\Longrightarrow$}}}}
\put(20,80){\vector(1,0){80}}
\put(10,67){\vector(3,-4){38}}
\put(110,68){\vector(-3,-4){38}}
\end{picture}
\end{tabular}}}
\end{center}
\caption{Schemed Domain Morphism: $\mathring{\mathrmbf{Dom}}(\mathcal{A})$}
\label{fig:sch:dom:mor:A}
\end{figure}
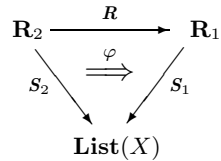
%

%%%%%%%%%%%%%%%%%%%%%%%%%%%%%%%%%%%%%%%%%%%%%%%%%%%%%%%%%%%%%%%%%%%%%%%%%%%%%%%%%%%%%%%%%%
%\mbox{}\newline\rule{120pt}{1pt}{\fbox{\textbf{ Work Zone }}}\rule{120pt}{1pt}\newline
%%%%%%%%%%%%%%%%%%%%%%%%%%%%%%%%%%%%%%%%%%%%%%%%%%%%%%%%%%%%%%%%%%%%%%%%%%%%%%%%%%%%%%%%%%

%
\begin{proposition}\label{prop:dom:A:lim:colim}
The fibered context (Grothendieck construction) of $\mathcal{A}$-schemed domains
$\mathring{\mathrmbf{Dom}}(\mathcal{A}) 
= \mathrmbf{Dom}(\mathcal{A})^{\scriptscriptstyle{\Uparrow}}
= \mathrmbf{List}(X)^{\scriptscriptstyle{\Uparrow}}$
is complete and cocomplete 
and the projection 
$\mathring{\mathrmbf{Dom}}(\mathcal{A}) = 
\mathrmbf{Dom}(\mathcal{A})^{\scriptscriptstyle{\Downarrow}} =
\mathrmbf{List}(X)^{\scriptscriptstyle{\Uparrow}}
\rightarrow\mathrmbf{Cxt}
:{\langle{\mathrmbf{R},\mathrmbfit{S}}\rangle}\mapsto\mathrmbf{R}$ 
is continuous and cocontinuous.
\end{proposition}
\begin{proof}
By Prop.~\ref{prop:lim:colim:A} of \S\,\ref{append:kan:ext},
since 
the context 
$\mathrmbf{Dom}(\mathcal{A}) = \mathrmbf{List}(X)$
is cocomplete and complete
(see \S\,4.2 in the paper
``The {\ttfamily FOLE} Table''
\cite{kent:fole:era:tbl}).
\hfill\rule{5pt}{5pt}
\end{proof}
%

%%%%%%%%%%%%%%%%%%%%%%%%%%%%%%%%%%%%%%%%%%%%%%%%%%%%%%%%%%%%%%%%%%%%%%%%%%%%%%%%%%%%%%%%%%
%\mbox{}\newline\rule{120pt}{1pt}{\fbox{\textbf{ Work Zone }}}\rule{120pt}{1pt}\newline
%%%%%%%%%%%%%%%%%%%%%%%%%%%%%%%%%%%%%%%%%%%%%%%%%%%%%%%%%%%%%%%%%%%%%%%%%%%%%%%%%%%%%%%%%%

%%%%%%%%%%%%%%%%%%%%%%%%%%%%%%%%%%%%%%%%%%%%%%%%%%%%%%%%%%%%%%%%%%%%%%%%%%%%%%%%%%%%%%%%%%
%
\newpage
\subsubsection{Upper Aspect: $\mathring{\mathrmbf{Dom}}$}
\label{sub:sub:sec:dom:typ:dom:upper}
%%%%%%%%%%%%%%%%%%%%%%%%%%%%%%%%%%%%%%%%%%%%%%%%%%%%%%%%%%%%%%%%%%%%%%%%%%%%%%%%%%%%%%%%%%

The 
%mathematical 
subcontext of 
{\ttfamily FOLE} schemed domains 
(with constant type domains and constant type domain morphisms) 
is denoted by $\mathring{\mathrmbf{Dom}}\subseteq\mathrmbf{DOM}$.

\begin{itemize}
\item 
A schemed domain 
$
%{\langle{\mathrmbf{R},\mathrmbfit{Q}}\rangle}=
{\langle{\mathrmbf{R},\mathrmbfit{S},\mathcal{A}}\rangle}$ 
in $\mathring{\mathrmbf{Dom}}$
consists of a type domain $\mathcal{A}$
and 
a schemed domain 
${\langle{\mathrmbf{R},\mathrmbfit{S}}\rangle}$ 
in $\mathring{\mathrmbf{Dom}}(\mathcal{A})$
with a shape context $\mathrmbf{R}$ and
a passage
$\mathrmbf{R}\xrightarrow{\,\mathrmbfit{S}\;}\mathrmbf{Dom}(\mathcal{A})\cong\mathrmbf{List}(X)$.
%a passage 
%$\mathrmbf{R}\xrightarrow{\,\mathrmbfit{Q}\;}\mathrmbf{Dom}$
%that factors as
%$\mathrmbfit{Q} = \mathrmbfit{S} \circ \mathrmbfit{inc}_{\mathcal{A}}$
%in terms of.
%
\item 
A schemed domain morphism
${\langle{\mathrmbf{R}_{2},\mathrmbfit{S}_{2},\mathcal{A}_{2}}\rangle} 
\xrightarrow{{\langle{\mathrmbfit{R},\hat{\varphi},f,g}\rangle}}
{\langle{\mathrmbf{R}_{1},\mathrmbfit{S}_{1},\mathcal{A}_{1}}\rangle}$
in $\mathring{\mathrmbf{Dom}}$
\comment{\begin{center}
{\normalsize{
$
{\langle{\mathrmbf{R}_{2},\mathrmbfit{Q}_{2}\rangle}} 
=
{\langle{\mathrmbf{R}_{2},\mathrmbfit{S}_{2},\mathcal{A}_{2}}\rangle} 
\xrightarrow{{\langle{\mathrmbfit{R},\hat{\varphi},f,g}\rangle}}
{\langle{\mathrmbf{R}_{1},\mathrmbfit{S}_{1},\mathcal{A}_{1}}\rangle}
=
{\langle{\mathrmbf{R}_{1},\mathrmbfit{Q}_{1}}\rangle}
$}}
%(Fig.\,\ref{fig:sch:dom:mor:typ:dom})
\end{center}}
%in $\mathring{\mathrmbf{Dom}}$
(Fig.\,\ref{fig:sch:dom:mor:typ:dom}) 
with constant type domain morphism
${\langle{f,g}\rangle} : \mathcal{A}_{2} \rightleftarrows \mathcal{A}_{1}$
is a 
\texttt{FOLE} schemed domain morphism
${\langle{\mathrmbfit{R},\varsigma}\rangle} : 
{\langle{\mathrmbf{R}_{2},\mathrmbfit{Q}_{2}}\rangle} \rightarrow 
{\langle{\mathrmbf{R}_{1},\mathrmbfit{Q}_{1}}\rangle}$,
whose 
signed domain interpretation bridge
{\footnotesize{$
\mathrmbfit{Q}_{2}
{\,\xRightarrow{\;\,\varsigma\,}\,}
\mathrmbfit{R}{\,\circ\,}\mathrmbfit{Q}_{1}$}}
factors adjointly
\begin{equation}\label{eqn:sch:dom:mor:def}
(\acute{\varphi} \circ \mathrmbfit{inc}_{\mathcal{A}_{2}}) 
\bullet 
(\mathrmbfit{R} \circ \mathrmbfit{S}_{1}  \circ \grave{\iota}_{{\langle{f,g}\rangle}})
= \varsigma = 
(\mathrmbfit{S}_{2} \circ \acute{\iota}_{{\langle{f,g}\rangle}})
\bullet 
(\grave{\varphi} \circ \mathrmbfit{inc}_{\mathcal{A}_{1}}) 
\end{equation}
through the fiber adjunction 
$\underset{\mathrmbf{List}(X_{2})}
{\underbrace{\mathrmbf{Dom}(\mathcal{A}_{2})}}
\xrightarrow
[{\scriptscriptstyle\sum}_{f}\;\dashv\;{f^{\ast}}]
{\grave{\mathrmbfit{dom}}_{{\langle{f,g}\rangle}}
\;\dashv\;
\acute{\mathrmbfit{dom}}_{{\langle{f,g}\rangle}}}
\underset{\mathrmbf{List}(X_{1})}
{\underbrace{\mathrmbf{Dom}(\mathcal{A}_{1})}}$
using
%(Fib.~\ref{fig:db:mor:typ:dom:var:shp})
the signed domain inclusion bridge adjointness 
%(see below)
$\hat{\iota}_{{\langle{f,g}\rangle}} 
= (\acute{\iota}_{{\langle{f,g}\rangle}},\grave{\iota}_{{\langle{f,g}\rangle}})$
\comment{
in terms of 
some bridge
$\acute{\varphi} : 
\mathrmbfit{R}{\,\circ\,}\mathrmbfit{S}_{1}
\Rightarrow
\mathrmbfit{S}_{2}\circ\acute{\mathrmbfit{dom}}_{{\langle{f,g}\rangle}}$
and 
the inclusion bridge
$\grave{\iota}_{{\langle{f,g}\rangle}} : 
\mathrmbfit{inc}_{\mathcal{A}_{2}}
\Leftarrow
\grave{\mathrmbfit{tbl}}_{{\langle{f,g}\rangle}}\circ\mathrmbfit{inc}_{\mathcal{A}_{1}}$;
%
%%%%%%%%%%%%%%%%%%%%%%%%%%%%%%%%%%%%%%%%%%%%%%%%%%%%%%%%%%%%%%%%%%%%%%%%%%%%%%%%
%%%%%%%%%%%%%%%%%%%%%%%%%%%%%%%%%%%%%%%%%%%%%%%%%%%%%%%%%%%%%%%%%%%%%%%%%%%%%%%%
%\footnote{
or equivalently, 
in terms of their \textbf{levo} bridge adjoints
in Tbl.\,\ref{adjoints:composites} of \S\,\ref{sub:sec:bridges}.
%}
%%%%%%%%%%%%%%%%%%%%%%%%%%%%%%%%%%%%%%%%%%%%%%%%%%%%%%%%%%%%%%%%%%%%%%%%%%%%%%%%
%%%%%%%%%%%%%%%%%%%%%%%%%%%%%%%%%%%%%%%%%%%%%%%%%%%%%%%%%%%%%%%%%%%%%%%%%%%%%%%%
}
%\end{itemize}
%
%\begin{table}
\begin{center}
{{\footnotesize\setlength{\extrarowheight}{4pt}$\begin{array}
{|@{\hspace{5pt}}l@{\hspace{15pt}}l@{\hspace{5pt}}|}
\multicolumn{1}{l}{\text{\bfseries levo}} & \multicolumn{1}{l}{\text{\bfseries dextro}}
\\ \hline
\acute{\iota}_{{\langle{f,g}\rangle}} :
\grave{\mathrmbfit{dom}}_{{\langle{f,g}\rangle}}{\circ\;}\mathrmbfit{inc}_{\mathcal{A}_{1}} 
\Leftarrow
%{\;\acute{\iota}_{{\langle{f,g}\rangle}}} 
\mathrmbfit{inc}_{\mathcal{A}_{2}}
&
\grave{\iota}_{{\langle{f,g}\rangle}} :
\mathrmbfit{inc}_{\mathcal{A}_{1}} 
\Leftarrow
%{\;\grave{\iota}_{{\langle{f,g}\rangle}}} 
\acute{\mathrmbfit{dom}}_{{\langle{f,g}\rangle}}{\circ\;}\mathrmbfit{inc}_{\mathcal{A}_{2}}
\\
\acute{\iota}_{{\langle{f,g}\rangle}} = 
\bigl(\eta_{f}{\;\circ\;}\mathrmbfit{inc}_{\mathcal{A}_{2}}\bigr)
{\;\bullet\;} 
\bigl(\grave{\mathrmbfit{dom}}_{{\langle{f,g}\rangle}}{\;\circ\;}\grave{\iota}_{{\langle{f,g}\rangle}}\bigr)
%\acute{\tau}_{{\langle{f,g}\rangle}}=(\varepsilon_{f}^{\mathrm{op}}{\;\circ\;}\mathrmbfit{tup}_{%\mathcal{A}_{1}})\bullet({f^{\ast}}^{\mathrm{op}}{\;\circ\;}\grave{\tau}_{{\langle{f,g}\rangle}})
&
\grave{\iota}_{{\langle{f,g}\rangle}} = 
\bigl(\acute{\mathrmbfit{dom}}_{{\langle{f,g}\rangle}}{\;\circ\;}\acute{\iota}_{{\langle{f,g}\rangle}}\bigr)
{\;\bullet\;}
\bigl(\varepsilon_{f}{\;\circ\;}\mathrmbfit{inc}_{\mathcal{A}_{1}}\bigr)
%\grave{\tau}_{{\langle{f,g}\rangle}}=({\scriptstyle\sum}_{f}^{\mathrm{op}}{\;\circ\;}\acute{\tau}_{{\langle{f,g}\rangle}})\bullet(\eta_{f}^{\mathrm{op}}{\;\circ\;}\mathrmbfit{tup}_{\mathcal{A}_{2}})
\rule[-7pt]{0pt}{10pt}
\\\hline
\multicolumn{1}{l}{}
\end{array}$}}
\end{center}
%\caption{Adjoints}
%\label{adjoints}
%\end{table}
%The inclusion bridges are 
%defined in detail in \S2.4.2 of the paper
%``The {\ttfamily FOLE} Table''
%\cite{kent:fole:era:tbl}.
abstractly defined in 
%the paper 
%``The \texttt{FOLE} Table'' 
%\S\,\ref{sub:sec:inc:bridge}
\S\,\ref{append:grothen:construct}.
This consists of equivalent morphisms  
%(left $\leftrightarrow$ right of Fig.~\ref{fig:db:fbr:adj}) 
\newline\mbox{}\hfill
$\underset{\in\;\;
\mathring{\mathrmbf{Dom}}(\mathcal{A}_{2})
%\mathrmbf{Db}(\mathcal{A}_{2})
}
{\underbrace{{\langle{\mathrmbf{R}_{2},\mathrmbfit{S}_{2}}\rangle} 
\xrightarrow{\;{\langle{\mathrmbfit{R},\acute{\varphi}}\rangle}\;}
\overset{
{\langle{
\mathrmbf{R}_{1},\mathrmbfit{S}_{1}{\circ\,}\acute{\mathrmbfit{dom}}_{{\langle{f,g}\rangle}}
}\rangle}
}
{\overbrace{
\acute{\mathrmbfit{dom}}_{{\langle{f,g}\rangle}}(\mathrmbf{R}_{1},\mathrmbfit{S}_{1})
}}
}}
{\;\;\;\overset{\cong}{\leftrightarrow}\;\;\;}
\underset{\in\;\;
\mathring{\mathrmbf{Dom}}(\mathcal{A}_{1})
%\mathrmbf{Db}(\mathcal{A}_{1})
}
{\underbrace{
\overset{
{\langle{
\mathrmbf{R}_{2},\mathrmbfit{S}_{2}{\circ\,}\grave{\mathrmbfit{dom}}_{{\langle{f,g}\rangle}}
}\rangle}
}
{\overbrace{
\grave{\mathrmbfit{dom}}_{{\langle{f,g}\rangle}}(\mathrmbf{R}_{2},\mathrmbfit{S}_{2}) 
}}
\xrightarrow{\;{\langle{\mathrmbfit{R},\grave{\varphi}}\rangle}\;}
{\langle{\mathrmbf{R}_{1},\mathrmbfit{S}_{1}}\rangle}
}}$
\hfill\mbox{}\newline
We normally just use the bridge restriction 
$\acute{\varphi}$ 
or
$\grave{\varphi}$ 
for the schemed domain morphism.
We use 
%the notation 
$\hat{\varphi}$ to denote either of these equivalent bridges.
The original definition can be computed with the factorization in 
Disp.\,\ref{eqn:sch:dom:mor:def}.
\end{itemize}
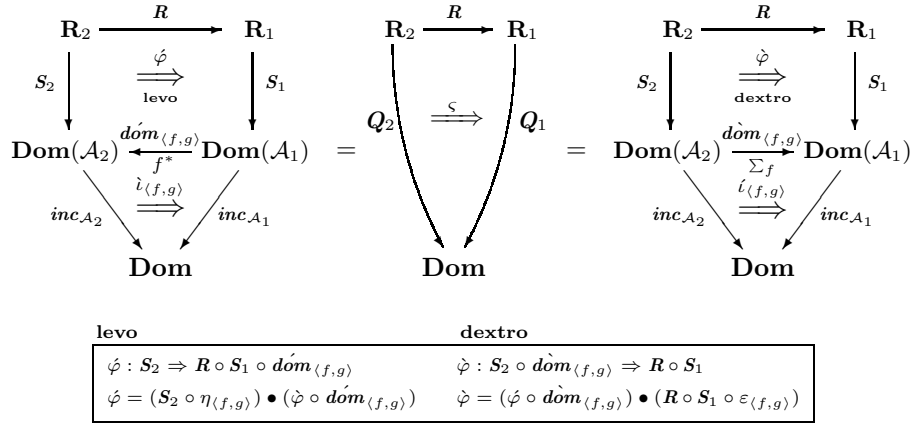
\begin{figure}
\begin{center}
\begin{tabular}{c}
{{
\begin{tabular}{@{\hspace{5pt}}c@{\hspace{50pt}}c@{\hspace{50pt}}c@{\hspace{5pt}}}
{{\begin{tabular}[b]{c}
\setlength{\unitlength}{0.58pt}
\begin{picture}(120,180)(0,0)
\put(5,160){\makebox(0,0){\footnotesize{$\mathrmbf{R}_{2}$}}}
\put(125,160){\makebox(0,0){\footnotesize{$\mathrmbf{R}_{1}$}}}
\put(-2,80){\makebox(0,0){\footnotesize{$\mathrmbf{Dom}(\mathcal{A}_{2})$}}}
\put(122,80){\makebox(0,0){\footnotesize{$\mathrmbf{Dom}(\mathcal{A}_{1})$}}}
\put(60,5){\makebox(0,0){\normalsize{$\mathrmbf{Dom}$}}}
\put(60,172){\makebox(0,0){\scriptsize{$\mathrmbfit{R}$}}}
\put(-8,125){\makebox(0,0)[r]{\scriptsize{$\mathrmbfit{S}_{2}$}}}
\put(128,125){\makebox(0,0)[l]{\scriptsize{$\mathrmbfit{S}_{1}$}}}
\put(60,92){\makebox(0,0){\scriptsize{$\acute{\mathrmbfit{dom}}_{{\langle{f,g}\rangle}}$}}}
\put(62,72){\makebox(0,0){\scriptsize{$f^{\ast}$}}}
\put(24,38){\makebox(0,0)[r]{\scriptsize{$\mathrmbfit{inc}_{\mathcal{A}_{2}}$}}}
\put(97,38){\makebox(0,0)[l]{\scriptsize{$\mathrmbfit{inc}_{\mathcal{A}_{1}}$}}}
\put(60,130){\makebox(0,0){\shortstack{\scriptsize{$\acute{\varphi}$}\\\large{$\Longrightarrow$}
\\\tiny{\textbf{levo}}
}}}
%\put(60,64){\makebox(0,0){\shortstack{\footnotesize{$\xLeftarrow{\acute{\chi}_{{\langle{f,g}\rangle}}}$}}}}
%\put(60,50){\makebox(0,0)
%{\footnotesize{$\xRightarrow{\grave{\iota}_{{\langle{f,g}\rangle}}}$}}}
\put(60,51){\makebox(0,0){\shortstack{\scriptsize{$\grave{\iota}_{{\langle{f,g}\rangle}}$}\\\large{$\Longrightarrow$}}}}
\put(20,160){\vector(1,0){80}}
\put(80,80){\vector(-1,0){40}}
\put(0,145){\vector(0,-1){50}}
\put(120,145){\vector(0,-1){50}}
\put(10,68){\vector(3,-4){38}}
\put(111,68){\vector(-3,-4){38}}
\end{picture}
\end{tabular}}}
%%%%%%%%%%%%%%%%%%%%%%%%%%%%%%%%%%%%%%%%%%%%%%%%%%%%%%%%%%%%
%%%%%%%%%%%%%%%%%%%%%%%%%%%%%%%%%%%%%%%%%%%%%%%%%%%%%%%%%%%%
&
{{\begin{tabular}[b]{c}
\setlength{\unitlength}{0.58pt}
\begin{picture}(80,180)(5,0)
\put(5,160){\makebox(0,0){\footnotesize{$\mathrmbf{R}_{2}$}}}
\put(85,160){\makebox(0,0){\footnotesize{$\mathrmbf{R}_{1}$}}}
\put(40,5){\makebox(0,0){\normalsize{$\mathrmbf{Dom}$}}}
\put(40,172){\makebox(0,0){\scriptsize{$\mathrmbfit{R}$}}}
\put(2,100){\makebox(0,0)[r]{\footnotesize{$\mathrmbfit{Q}_{2}$}}}
\put(81,100){\makebox(0,0)[l]{\footnotesize{$\mathrmbfit{Q}_{1}$}}}
\put(42,100){\makebox(0,0){\shortstack{\normalsize{$\xRightarrow{\;\;\varsigma\;\;}$}}}}
\put(20,160){\vector(1,0){45}}
\qbezier(0,150)(0,80)(33,18)\put(33,18){\vector(2,-3){0}}
\qbezier(80,150)(80,80)(47,18)\put(47,18){\vector(-2,-3){0}}
%\put(-26,80){\makebox(0,0){\shortstack{\normalsize{$=$}}}}
%\put(113,80){\makebox(0,0){\shortstack{\normalsize{$=$}}}}
%%%%%%%%%%%%%%%%%%%%%%%%%%%%%%%%%%%%%%%%%%%%%%%%%%%%%%%%%%%%
\put(-30,80){\makebox(0,0){\shortstack{\normalsize{$=$}}}}
\put(120,80){\makebox(0,0){\shortstack{\normalsize{$=$}}}}
\end{picture}
\end{tabular}}}
&
%%%%%%%%%%%%%%%%%%%%%%%%%%%%%%%%%%%%%%%%%%%%%%%%%%%%%%%%%%%%
%%%%%%%%%%%%%%%%%%%%%%%%%%%%%%%%%%%%%%%%%%%%%%%%%%%%%%%%%%%%
{{\begin{tabular}[b]{c}
\setlength{\unitlength}{0.58pt}
\begin{picture}(120,180)(0,0)
\put(5,160){\makebox(0,0){\footnotesize{$\mathrmbf{R}_{2}$}}}
\put(125,160){\makebox(0,0){\footnotesize{$\mathrmbf{R}_{1}$}}}
\put(-2,80){\makebox(0,0){\footnotesize{$\mathrmbf{Dom}(\mathcal{A}_{2})$}}}
\put(122,80){\makebox(0,0){\footnotesize{$\mathrmbf{Dom}(\mathcal{A}_{1})$}}}
\put(60,5){\makebox(0,0){\normalsize{$\mathrmbf{Dom}$}}}
\put(60,172){\makebox(0,0){\scriptsize{$\mathrmbfit{R}$}}}
\put(-8,125){\makebox(0,0)[r]{\scriptsize{$\mathrmbfit{S}_{2}$}}}
\put(128,125){\makebox(0,0)[l]{\scriptsize{$\mathrmbfit{S}_{1}$}}}
\put(60,92){\makebox(0,0){\scriptsize{$\grave{\mathrmbfit{dom}}_{{\langle{f,g}\rangle}}$}}}
\put(60,70){\makebox(0,0){\scriptsize{${\scriptstyle\sum}_{f}$}}}
\put(24,38){\makebox(0,0)[r]{\scriptsize{$\mathrmbfit{inc}_{\mathcal{A}_{2}}$}}}
\put(97,38){\makebox(0,0)[l]{\scriptsize{$\mathrmbfit{inc}_{\mathcal{A}_{1}}$}}}
\put(60,130){\makebox(0,0){\shortstack{
\scriptsize{$\grave{\varphi}$}\\\large{$\Longrightarrow$}
\\\tiny{\textbf{dextro}}
}}}
%\put(60,50){\makebox(0,0){\shortstack{\footnotesize{$\xRightarrow{\acute{\iota}_{{\langle{f,g}\rangle}}}$}}}}
\put(60,50){\makebox(0,0){\shortstack{\scriptsize{$\acute{\iota}_{{\langle{f,g}\rangle}}$}\\\large{$\Longrightarrow$}}}}
\put(20,160){\vector(1,0){80}}
\put(40,80){\vector(1,0){40}}
\put(0,145){\vector(0,-1){50}}
\put(120,145){\vector(0,-1){50}}
\put(9,68){\vector(3,-4){38}}
\put(111,68){\vector(-3,-4){38}}
\end{picture}
\end{tabular}}}
\end{tabular}}}
%%%%%%%%%%%%%%%%%%%%%%%%%%%%%%%%%%%%%%%%%%%%%%%%%%%%%%%%%%%%
\\\\
%%%%%%%%%%%%%%%%%%%%%%%%%%%%%%%%%%%%%%%%%%%%%%%%%%%%%%%%%%%%
{\scriptsize\setlength{\extrarowheight}{4pt}$\begin{array}{|@{\hspace{5pt}}l@{\hspace{15pt}}l@{\hspace{5pt}}|}
\multicolumn{1}{l}{\text{\bfseries levo}}
& 
\multicolumn{1}{l}{\text{\bfseries dextro}} 
\\ \hline
\acute{\varphi} : 
\mathrmbfit{S}_{2}
\Rightarrow
\mathrmbfit{R}{\,\circ\,}\mathrmbfit{S}_{1}\circ\acute{\mathrmbfit{dom}}_{{\langle{f,g}\rangle}}
&
\grave{\varphi} : 
\mathrmbfit{S}_{2}\circ\grave{\mathrmbfit{dom}}_{{\langle{f,g}\rangle}}
\Rightarrow
\mathrmbfit{R}{\,\circ\,}\mathrmbfit{S}_{1}
\\
\acute{\varphi} =
(\mathrmbfit{S}_{2}\circ\eta_{{\langle{f,g}\rangle}})
\bullet
(\grave{\varphi}\circ\acute{\mathrmbfit{dom}}_{{\langle{f,g}\rangle}})
&
\grave{\varphi} =
(\acute{\varphi}\circ\grave{\mathrmbfit{dom}}_{{\langle{f,g}\rangle}})
\bullet
(\mathrmbfit{R}\,{\circ\,}\mathrmbfit{S}_{1}\circ\varepsilon_{{\langle{f,g}\rangle}})
\rule[-7pt]{0pt}{10pt}
\\\hline
\end{array}$}
\end{tabular}
\end{center}
\caption{Schemed Domain Morphism: $\mathring{\mathrmbf{Dom}}$ }
\label{fig:sch:dom:mor:typ:dom}
\end{figure}
\begin{proposition}\label{prop:fib:cxt:sch:dom:var:typ:dom:sh}
Hence,
the fibered context 
%of schemed domains
$\mathring{\mathrmbf{Dom}}$
%$\mathring{\mathrmbf{Dom}}\xrightarrow{\mathrmbfit{data}}\mathrmbf{Adj}$
(with constant type domain) 
is the Grothendieck construction 
of the indexed adjunction
\newline\mbox{}\hfill
$\mathring{\mathrmbf{Dom}}
= \int_\mathrmit{data}:
%\mathrmbf{Cls}\,\xrightarrow{\,\hat{\mathrmbfit{dom}}\;}\mathrmbf{Adj}$.
\mathrmbf{Cls}\xrightarrow{\,\hat{\mathrmbfit{dom}}\;}\mathrmbf{Adj}$.
\hfill\mbox{}
\newline
%
%\begin{itemize}
%\item 
%the type domain indexed adjunction
%\end{itemize}
%
\end{proposition}
\begin{proof}
For any type domain morphism
$\mathcal{A}_{2}\xrightleftharpoons{{\langle{f,g}\rangle}}\mathcal{A}_{1}$,
there is a schemed domain fiber adjunction
%\newline
%{{\footnotesize
%$\mathring{\mathrmbf{Dom}}(\mathcal{A}_{2})
%\xrightarrow
%{{\langle{\acute{\mathrmbfit{dom}}_{{\langle{f,g}\rangle}}
%{\,\dashv\,}
%\grave{\mathrmbfit{dom}}_{{\langle{f,g}\rangle}}}\rangle}}
%\mathring{\mathrmbf{Dom}}(\mathcal{A}_{1})
%$}}
%\newline
$\underset{\mathrmbf{List}(X_{2})}
{\underbrace{\mathrmbf{Dom}(\mathcal{A}_{2})}}
\xrightarrow
[{\scriptscriptstyle\sum}_{f}\;\dashv\;{f^{\ast}}]
{\grave{\mathrmbfit{dom}}_{{\langle{f,g}\rangle}}
\;\dashv\;
\acute{\mathrmbfit{dom}}_{{\langle{f,g}\rangle}}}
\underset{\mathrmbf{List}(X_{1})}
{\underbrace{\mathrmbf{Dom}(\mathcal{A}_{1})}}$
%
%\newline
that defines the context of schemed domains $\mathring{\mathrmbf{Dom}}$
via the Grothendieck construction.
\mbox{}\hfill\rule{5pt}{5pt}
\end{proof}
%
%{\fbox{The directions are in reverse here. Please change!}}
%\newpage
%\mbox{}\newline
%...................................
%\newline

%
\begin{proposition}\label{prop:dom:typ:dom:lim:colim}
$\mathring{\mathrmbf{Dom}}$ is a complete and cocomplete context.
\end{proposition}
\begin{proof}
By Fact\,\ref{fact:groth:adj:lim:colim}
of \S\,\ref{append:grothen:construct},
since
the indexing context $\mathrmbf{Cls}$ is complete and cocomplete
(\cite{barwise:seligman:97}),
%the indexing context $\mathrmbf{Cxt}$ is complete and cocomplete,
and
the fiber context 
$\mathring{\mathrmbf{Dom}}(\mathcal{A}) 
= \mathrmbf{Dom}(\mathcal{A})^{\scriptscriptstyle{\Uparrow}}
= \mathrmbf{List}(X)^{\scriptscriptstyle{\Uparrow}}$ 
is complete and cocomplete
for each type domain $\mathcal{A}$
by
Prop.\,\ref{prop:dom:A:lim:colim}
above.
\hfill\rule{5pt}{5pt}
\end{proof}
%

%
%%%%%%%%%%%%%%%%%%%%%%%%%%%%%%%%%%%%%%%%%%%%%%%%%%%%%%%%%%%%%%%%%%%%%%%%%%%%%%%%%%%%%%%%%%
%%%%%%%%%%%%%%%%%%%%%%%%%%%%%%%%%%%%%%%%%%%%%%%%%%%%%%%%%%%%%%%%%%%%%%%%%%%%%%%%%%%%%%%%%%
%\newpage
%\subsection{Tuple Passages}\label{sub:sec:append:tuple}
%%%%%%%%%%%%%%%%%%%%%%%%%%%%%%%%%%%%%%%%%%%%%%%%%%%%%%%%%%%%%%%%%%%%%%%%%%%%%%%%%%%%%%%%%%
%%%%%%%%%%%%%%%%%%%%%%%%%%%%%%%%%%%%%%%%%%%%%%%%%%%%%%%%%%%%%%%%%%%%%%%%%%%%%%%%%%%%%%%%%%

%
%%%%%%%%%%%%%%%%%%%%%%%%%%%%%%%%%%%%%%%%%%%%%%%%%%%%%%%%%%%%%%%%%%%%%%%%%%%%%%%%%%%%%%%%%%
%%%%%%%%%%%%%%%%%%%%%%%%%%%%%%%%%%%%%%%%%%%%%%%%%%%%%%%%%%%%%%%%%%%%%%%%%%%%%%%%%%%%%%%%%%
%\newpage
%\subsection{Comma Contexts.}
%\label{sub:sec:comma:cxt:tup:pass}
%%%%%%%%%%%%%%%%%%%%%%%%%%%%%%%%%%%%%%%%%%%%%%%%%%%%%%%%%%%%%%%%%%%%%%%%%%%%%%%%%%%%%%%%%%
%%%%%%%%%%%%%%%%%%%%%%%%%%%%%%%%%%%%%%%%%%%%%%%%%%%%%%%%%%%%%%%%%%%%%%%%%%%%%%%%%%%%%%%%%%
%

%%%%%%%%%%%%%%%%%%%%%%%%%%%%%%%%%%%%%%%%%%%%%%%%%%%%%%%%%%%%%%%%%%%%%%%%%%%%%%%%%%%%%%%%%%
%%%%%%%%%%%%%%%%%%%%%%%%%%%%%%%%%%%%%%%%%%%%%%%%%%%%%%%%%%%%%%%%%%%%%%%%%%%%%%%%%%%%%%%%%%
%%%%%%%%%%%%%%%%%%%%%%%%%%%%%%%%%%%%%%%%%%%%%%%%%%%%%%%%%%%%%%%%%%%%%%%%%%%%%%%%%%%%%%%%%%
\newpage
\section{\texttt{FOLE} Databases}\label{sec:rel:db}
%%%%%%%%%%%%%%%%%%%%%%%%%%%%%%%%%%%%%%%%%%%%%%%%%%%%%%%%%%%%%%%%%%%%%%%%%%%%%%%%%%%%%%%%%%
%%%%%%%%%%%%%%%%%%%%%%%%%%%%%%%%%%%%%%%%%%%%%%%%%%%%%%%%%%%%%%%%%%%%%%%%%%%%%%%%%%%%%%%%%%
%%%%%%%%%%%%%%%%%%%%%%%%%%%%%%%%%%%%%%%%%%%%%%%%%%%%%%%%%%%%%%%%%%%%%%%%%%%%%%%%%%%%%%%%%%

The relational database contexts defined in this paper
are listed in Tbl.~\ref{tbl:db:cxts}.
%They exist in pairs,
%standard versus atypical $(\widehat{\cdot})$, 
%linked by a reflection.
%We use the standard definition when defining the classification-interpretation equivalence. 
%
%Each atypical 
%%relational 
%database context is a comma context
%that is characterized 
%(Tbl.~\ref{tbl:adb:char})
%by an associated tuple passage.
%
\begin{table}
\begin{center}
{\fbox{\scriptsize{\setlength{\extrarowheight}{2pt}{\begin{tabular}
{|@{\hspace{3pt}}r@{\hspace{5pt}=\hspace{5pt}}l@{\hspace{5pt}:\hspace{5pt}}l@{\hspace{20pt}}
l@{\hspace{5pt}\text{in}\hspace{5pt}}l|}
\hline
\rule{0pt}{9pt}
$\mathrmbf{DB}$
&
$\mathrmbf{Tbl}^{\scriptscriptstyle{\Downarrow}}$
&
\textit{diagram context} 
& 
Def.~\ref{def:db:oplax:cxt}
&
\S~\ref{sub:sec:rel:db:gen}
\\\hline
\rule{0pt}{10pt}
${\mathrmbf{Db}(\mathcal{A})}$
&
${\mathrmbf{Tbl}(\mathcal{A})^{\scriptscriptstyle{\Downarrow}}}$ 
&
\textit{diagram context}
& 
Def.~\ref{def:db:A:oplax:cxt}
&
\S~\ref{sub:sub:sec:rel:db:typ:dom:lower}
\\
$\mathrmbf{Db}$
&
$\int
%_\mathrmit{data}:
\mathrmbf{Cls}^{\mathrm{op}}\!\xrightarrow{\,\hat{\mathrmbfit{db}}\;}\mathrmbf{Adj}$
&
\textit{Grothendieck construction}
&
Prop.~\ref{prop:fib:cxt:db:var:typ:dom:sh}
&
\S~\ref{sub:sub:sec:rel:db:typ:dom:upper}
\\\hline
\end{tabular}
}}}}
\end{center}
\caption{Relational Database Contexts}
\label{tbl:db:cxts}
\end{table}
%

%%%%%%%%%%%%%%%%%%%%%%%%%%%%%%%%%%%%%%%%%%%%%%%%%%%%%%%%%%%%%%%%%%%%%%%%%%%%%%%%
%%%%%%%%%%%%%%%%%%%%%%%%%%%%%%%%%%%%%%%%%%%%%%%%%%%%%%%%%%%%%%%%%%%%%%%%%%%%%%%%
\comment{
\begin{table}
\begin{center}
{{\footnotesize{\setlength{\extrarowheight}{3pt}
{\begin{tabular}{|@{\hspace{3pt}}l@{\hspace{10pt}}l@{\hspace{5pt}}|}
\multicolumn{1}{l}{\text{\bfseries atypical context}} & \multicolumn{1}{l}{\text{\bfseries tuple passage}}
\\ \hline
$\widehat{\mathrmbf{DB}}
= \bigl(\mathrmbf{SET}{\;\downarrow\,}\mathring{\mathrmbfit{tup}}\bigr)$
&
$\mathrmbf{DOM}^{\mathrm{op}}
\!\xrightarrow{\mathring{\mathrmbfit{tup}}}
\mathrmbf{SET}$
\\\hline
$\widehat{\mathrmbf{Db}}(\mathcal{A})
= \bigl(\mathrmbf{SET}{\;\downarrow\,}\mathring{\mathrmbfit{tup}}_{\mathcal{A}}\bigr)$
&
$\mathring{\mathrmbf{Dom}}(\mathcal{A})^{\mathrm{op}} 
\xrightarrow{\mathring{\mathrmbfit{tup}}_{\mathcal{A}}}
\mathrmbf{SET}$
\\\hline
$\widehat{\mathrmbf{Db}}
= \bigl(\mathrmbf{SET}{\;\downarrow\,}\mathring{\mathrmbfit{tup}}\bigr)$
&
$\mathring{\mathrmbf{Dom}}^{\mathrm{op}}
\!\xrightarrow{\mathring{\mathrmbfit{tup}}}
\mathrmbf{SET}$
\\\hline
\end{tabular}}}}}
\end{center}
\caption{Characterization}
\label{tbl:adb:char}
\end{table}
}
%%%%%%%%%%%%%%%%%%%%%%%%%%%%%%%%%%%%%%%%%%%%%%%%%%%%%%%%%%%%%%%%%%%%%%%%%%%%%%%%
%%%%%%%%%%%%%%%%%%%%%%%%%%%%%%%%%%%%%%%%%%%%%%%%%%%%%%%%%%%%%%%%%%%%%%%%%%%%%%%%
%

%%%%%%%%%%%%%%%%%%%%%%%%%%%%%%%%%%%%%%%%%%%%%%%%%%%%%%%%%%%%%%%%%%%%%%%%%%%%%%%%%%%%%%%%%%
%%%%%%%%%%%%%%%%%%%%%%%%%%%%%%%%%%%%%%%%%%%%%%%%%%%%%%%%%%%%%%%%%%%%%%%%%%%%%%%%%%%%%%%%%%
%\newpage
%\subsection{\texttt{FOLE} Databases.}
\subsection{General Case: $\mathrmbf{DB}$}
\label{sub:sec:rel:db:gen}
%%%%%%%%%%%%%%%%%%%%%%%%%%%%%%%%%%%%%%%%%%%%%%%%%%%%%%%%%%%%%%%%%%%%%%%%%%%%%%%%%%%%%%%%%%
%%%%%%%%%%%%%%%%%%%%%%%%%%%%%%%%%%%%%%%%%%%%%%%%%%%%%%%%%%%%%%%%%%%%%%%%%%%%%%%%%%%%%%%%%%

In this section
we discuss the context of relational databases $\mathrmbf{DB}$.
% with variable shape.
%Shape variability is handled by the use of Kan extensions.
which mirrors 
%the context of signed domains $\mathrmbf{Dom}$ 
%(Fig.~4 in \cite{kent:fole:era:tbl}) 
the context of tables 
%$\mathrmbf{Dom}$,
%which is the comma context
%\[\mbox
{\footnotesize{$
%\mathrmbf{Set}\xleftarrow{\mathrmbfit{arity}}
\mathrmbf{Tbl}
% = {\bigl(\mathrmbf{Set}{\,\downarrow\,}\mathrmbfit{sort}\bigr)}
%\xrightarrow{\mathrmbfit{data}}\mathrmbf{Cls}
$}\normalsize}
%\]
%
at a higher dimension. 
%
%%%%%%%%%%%%%%%%%%%%%%%%%%%%%%%%%%%%%%%%%%%%%%%%%%%%%%%%%%%%
%%%%%%%%%%%%%%%%%%%%%%%%%%%%%%%%%%%%%%%%%%%%%%%%%%%%%%%%%%%%
\footnote{
The database 
$\mathrmbf{R}^{\mathrm{op}}\!\xrightarrow{\,\mathrmbfit{T}\;}\mathrmbf{Tbl}$
as developed in this paper is a derived concept built up from the basic concept of 
the table 
$\mathcal{T} = {\langle{\mathrmbfit{D},K,t}\rangle} \in \mathrmbf{Tbl}$
as defined in the paper ``The {\ttfamily FOLE} Table'' \cite{kent:fole:era:tbl}.
We replace tables $\mathcal{T} \in \mathrmbf{Tbl}$ 
with diagrams (passages) 
$\mathrmbf{R}^{\mathrm{op}}\!\xrightarrow{\,\mathrmbfit{T}\;}\mathrmbf{Tbl}$, and 
replace table morphisms 
$\mathcal{T}_{2}\xrightarrow{{\langle{{\langle{h,f,g}\rangle},k}\rangle}}
\mathcal{T}_{1}$ in $\mathrmbf{Tbl}$ 
with bridges 
%$\mathrmbfit{Q}_{2}\xRightarrow{\,\zeta\;\,}\mathrmbfit{Q}_{1}$ over $\mathrmbf{Dom}$
%$\varsigma : \mathrmbfit{T}_{2}\Rightarrow\mathrmbfit{R}{\,\circ\,}\mathrmbfit{T}_{1}$.
$\xi : \mathrmbfit{T}_{2}\Leftarrow\mathrmbfit{R}^{\mathrm{op}}{\circ}\mathrmbfit{T}_{1}$.
}
%%%%%%%%%%%%%%%%%%%%%%%%%%%%%%%%%%%%%%%%%%%%%%%%%%%%%%%%%%%%
%%%%%%%%%%%%%%%%%%%%%%%%%%%%%%%%%%%%%%%%%%%%%%%%%%%%%%%%%%%%
%
A relational database ${\langle{\mathrmbf{R},\mathrmbfit{T}}\rangle}$ is a diagram of tables,
consisting of 
%a shape context $\mathrmbf{R}$ 
%and an object $\mathrmbfit{T}$ in the fiber context $\mathrmbf{DB}(\mathrmbf{R})$
%(\S~\ref{sub:sec:rel:db:fix:shape});
%that is,
a shape context $\mathrmbf{R}$ 
and a passage $\mathrmbf{R}^{\mathrm{op}}\!\xrightarrow{\,\mathrmbfit{T}\;}\mathrmbf{Tbl}$.
A relational database morphism 
(Fig.~\ref{fig:db:mor:gen})
${\langle{\mathrmbf{R}_{2},\mathrmbfit{T}_{2}}\rangle} 
\xleftarrow{{\langle{\mathrmbfit{R},\,\xi}\rangle}}
{\langle{\mathrmbf{R}_{1},\mathrmbfit{T}_{1}}\rangle}$
consists of 
a shape-changing passage $\mathrmbf{R}_{2}\xrightarrow{\;\mathrmbfit{R}\;\,}\mathrmbf{R}_{1}$
and a bridge 
$\mathrmbfit{T}_{2}\xLeftarrow{\;\xi\,}
\mathrmbfit{R}^{\mathrm{op}}{\circ}\mathrmbfit{T}_{1}
%=\mathrmbf{Tbl}^{\mathrmbfit{R}^{\mathrm{op}}}\!(\mathrmbfit{T}_{1})
$.
\begin{figure}
\begin{center}
{{\begin{tabular}{c}
\setlength{\unitlength}{0.56pt}
\begin{picture}(120,80)(8,0)
\put(5,80){\makebox(0,0){\footnotesize{$\mathrmbf{R}_{2}^{\mathrm{op}}$}}}
\put(125,80){\makebox(0,0){\footnotesize{$\mathrmbf{R}_{1}^{\mathrm{op}}$}}}
\put(65,0){\makebox(0,0){\footnotesize{$\mathrmbf{Tbl}$}}}
\put(65,92){\makebox(0,0){\scriptsize{$\mathrmbfit{R}^{\mathrm{op}}$}}}
\put(20,42){\makebox(0,0)[r]{\scriptsize{$\mathrmbfit{T}_{2}$}}}
\put(100,42){\makebox(0,0)[l]{\scriptsize{$\mathrmbfit{T}_{1}$}}}
\put(60,57){\makebox(0,0){\shortstack{\scriptsize{$\xi$}\\\large{$\Longleftarrow$}}}}
\put(20,80){\vector(1,0){80}}
\put(10,67){\vector(3,-4){38}}
\put(110,68){\vector(-3,-4){38}}
\end{picture}
\end{tabular}}}
\end{center}
\caption{Database Morphism: $\mathrmbf{DB}$}
\label{fig:db:mor:gen}
\end{figure}
\begin{definition}\label{def:db:oplax:cxt}
The context of databases is the 
oplax comma 
context
$\mathrmbf{DB} = \mathrmbf{Tbl}^{\scriptscriptstyle{\Downarrow}}$,
a diagram context
over tables.
(Def.\,\ref{def:lax:oplax} in \S\,\ref{sub:sub:sec:lax:comma:cxt})
\end{definition}
\begin{proposition}\label{prop:lim:colim:db:tbl}
The context of databases 
$\mathrmbf{DB}$
%is the fibered context (Grothendieck construction) 
%$\mathrmbf{Db} = 
%\mathrmbf{Tbl}^{\scriptscriptstyle{\Downarrow}} = \int\hat{\mathrmbf{Tbl}}$, 
%and hence
is complete and cocomplete 
and the projection 
$\mathrmbf{DB}
%\mathrmbf{Tbl}^{\scriptscriptstyle{\Downarrow}}
\rightarrow\mathrmbf{Cxt}
:{\langle{\mathrmbf{R},\mathrmbfit{T}}\rangle}\mapsto\mathrmbf{R}$ 
is continuous and cocontinuous.
%is complete and cocomplete
%and the projection 
%$\mathrmbf{Db}\rightarrow\mathrmbf{Cls}$
%is continuous and cocontinuous.
\end{proposition}
\begin{proof}
Use the oplax parts of 
Prop.~\ref{prop:lax:fibered:A}
and
Prop.~\ref{prop:lim:colim:A} 
in \S\,\ref{append:kan:ext},
since
the context of databases 
$\mathrmbf{DB}$
is the diagram context 
%(Grothendieck construction) 
$\mathrmbf{DB} = 
%\mathrmbf{Tbl}^{\scriptscriptstyle{\Downarrow}}
% = \int\hat{\mathrmbf{Tbl}}
\mathrmbf{Tbl}^{\scriptscriptstyle{\Downarrow}} = \int\hat{\mathrmbf{Tbl}}$ 
%$\mathrmbf{DB} = \mathrmbf{Tbl}^{\scriptscriptstyle{\Downarrow}}$
and
$\mathrmbf{Tbl}$ is complete and cocomplete.
\mbox{}\hfill\rule{5pt}{5pt}
\end{proof}

%%%%%%%%%%%%%%%%%%%%%%%%%%%%%%%%%%%%%%%%%%%%%%%%%%%%%%%%%%%%%%%%%%%%%%%%%%%%%%%%%%%%%%%%%%
%\newpage
\paragraph{Examples.}
%%%%%%%%%%%%%%%%%%%%%%%%%%%%%%%%%%%%%%%%%%%%%%%%%%%%%%%%%%%%%%%%%%%%%%%%%%%%%%%%%%%%%%%%%%

Since $\mathrmbf{DB}$
is complete and cocomplete, 
so is $\mathrmbf{DB}^{\mathrm{op}}$.
Here we illustrate specific limits and colimits in $\mathrmbf{DB}^{\mathrm{op}}$:
initial and terminal objects,
binary coproducts and binary products, etc.
\begin{description}
\item[initial:]
${\langle{\emptyset^{\mathrm{op}},\mathrmbfit{0}_{\mathrmbf{Tbl}}}\rangle}$
is the initial database in $\mathrmbf{DB}^{\mathrm{op}}$.
To any
database ${\langle{\mathrmbf{R},\mathrmbfit{T}}\rangle}$
there is a unique database morphism
%\newline
${\langle{\emptyset^{\mathrm{op}},\mathrmbfit{0}_{\mathrmbf{Tbl}}}\rangle}
\xrightarrow{{\langle{0^{\mathrm{op}},\,1_{\mathrmbfit{0}_{\mathrmbf{Tbl}}}}\rangle}}
{\langle{\mathrmbf{R},\mathrmbfit{T}}\rangle}$
 in $\mathrmbf{DB}^{\mathrm{op}}$.
\newline
\item[terminal:]
${\langle{\mathrmbf{Tbl}^{\mathrm{op}},\mathrmbfit{1}_{\mathrmbf{Tbl}}}\rangle}$
is the terminal database in $\mathrmbf{DB}^{\mathrm{op}}$.
From any
database ${\langle{\mathrmbf{R},\mathrmbfit{T}}\rangle}$
there is a unique database morphism
%\newline
${\langle{\mathrmbf{R},\mathrmbfit{T}}\rangle} 
\xrightarrow{{\langle{\mathrmbfit{T}^{\mathrm{op}},\,1_{\mathrmbfit{T}}}\rangle}}
{\langle{\mathrmbf{Tbl}^{\mathrm{op}},\mathrmbfit{1}_{\mathrmbf{Tbl}}}\rangle}$
 in $\mathrmbf{DB}^{\mathrm{op}}$.
\newline
\item[coproduct:]
Any pair of databases
${\langle{\mathrmbf{R}_{1},\mathrmbfit{T}_{1}}\rangle}$
and
${\langle{\mathrmbf{R}_{2},\mathrmbfit{T}_{2}}\rangle}$
has
the \underline{binary coproduct} database 
${\langle{\mathrmbf{R}_{12},\mathrmbfit{T}_{12}}\rangle}$
in $\mathrmbf{DB}^{\mathrm{op}}$
that combines the two databases in an independent union.
The binary coproduct
%${\langle{\mathrmbf{R}_{2},\mathrmbfit{T}_{12}}\rangle}$
%in $\mathrmbf{DB}^{\mathrm{op}}$
consists of 
\begin{itemize}
\item 
the (opposite of the) coproduct context
$\mathrmbf{R}_{12}^{\mathrm{op}} = 
\mathrmbf{R}_{1}^{\mathrm{op}}{+\,}\mathrmbf{R}_{2}^{\mathrm{op}}$
%$\mathrmbf{R}_{12} = \mathrmbf{R}_{1}{+\,}\mathrmbf{R}_{2}$
and 
\item 
the copairing passage 
{\footnotesize{$
\mathrmbf{R}_{12}^{\mathrm{op}} 
%= \mathrmbf{R}_{1}^{\mathrm{op}}{+\,}\mathrmbf{R}_{2}^{\mathrm{op}}
\!\xrightarrow
[{[\mathrmbfit{T}_{1},\mathrmbfit{T}_{2}]}]
{\,\mathrmbfit{T}_{12}\;}
\mathrmbf{Tbl}$}}.
\end{itemize}
There is an
opspan of \underline{injection} database morphisms in $\mathrmbf{DB}^{\mathrm{op}}$
\[\mbox{\footnotesize{$
{\langle{\mathrmbf{R}_{1},\mathrmbfit{T}_{1}}\rangle}
\xrightarrow{{\langle{\mathrmbfit{inc}_{1},\,1}\rangle}}
{\langle{\mathrmbf{R}_{1}{+\,}\mathrmbf{R}_{2},\mathrmbfit{T}_{12}}\rangle}
\xleftarrow{{\langle{\mathrmbfit{inc}_{2},\,1}\rangle}}
{\langle{\mathrmbf{R}_{2},\mathrmbfit{T}_{2}}\rangle}
$}\normalsize}\]
with
the opspan of 
inclusion passages 
$\mathrmbf{R}_{1}\xhookrightarrow{\mathrmbfit{inc}_{1}}
\mathrmbf{R}_{1}{\,+\,}\mathrmbf{R}_{2}
\xhookleftarrow{\mathrmbfit{inc}_{2}}\mathrmbf{R}_{2}$
and the pair of identity bridges 
$
\mathrmbfit{T}_{1}
\xLeftarrow{\;1\,}
\mathrmbfit{inc}_{1}^{\mathrm{op}}{\circ}\mathrmbfit{T}_{12}
$
and
$
\mathrmbfit{inc}_{2}^{\mathrm{op}}{\circ}\mathrmbfit{T}_{12}
\xRightarrow{\;1\,}
\mathrmbfit{T}_{2}
$.
\newline
\item[product:]
Any pair of databases
${\langle{\mathrmbf{R}_{1},\mathrmbfit{T}_{1}}\rangle}$
and
${\langle{\mathrmbf{R}_{2},\mathrmbfit{T}_{2}}\rangle}$
with the
opspan of passages 
$\mathrmbf{R}_{1}^{\mathrm{op}}\xrightarrow{\;\mathrmbfit{T}_{1}\;\,}
\mathrmbf{Tbl}
\xleftarrow{\;\mathrmbfit{T}_{2}\;\,}\mathrmbf{R}_{2}^{\mathrm{op}}$
has
a \underline{binary product} database
in $\mathrmbf{DB}^{\mathrm{op}}.$
%that combines the two databases in an independent union.
The binary product
${\langle{\mathrmbf{R}_{12},\mathrmbfit{T}_{12}}\rangle}$
consists of 
\begin{itemize}
\item 
the (opposite of the) pullback context
$\mathrmbf{R}_{12}^{\mathrm{op}}
 = \mathrmbf{R}_{1}^{\mathrm{op}}{\times_{\mathrmbf{Tbl}}\,}\mathrmbf{R}_{2}^{\mathrm{op}}$
and 
\item 
the pairing passage 
{\footnotesize{$
\mathrmbf{R}_{12}^{\mathrm{op}} 
%= \mathrmbf{R}_{1}^{\mathrm{op}}{\times_{\mathrmbf{Tbl}}\,}\mathrmbf{R}_{2}^{\mathrm{op}}
\!\xrightarrow[{(\mathrmbfit{T}_{1},\mathrmbfit{T}_{2})}]{\,\mathrmbfit{T}_{12}\;}
\mathrmbf{Tbl}$}}.
\end{itemize}
There is a
span of \underline{projection} database morphisms in $\mathrmbf{DB}^{\mathrm{op}}$
\[\mbox{\footnotesize{$
{\langle{\mathrmbf{R}_{1},\mathrmbfit{T}_{1}}\rangle}
\xleftarrow{{\langle{\mathrmbfit{pr}_{1},\,1}\rangle}}
{\langle{\mathrmbf{R}_{12},\mathrmbfit{T}_{12}}\rangle}
\xrightarrow{{\langle{\mathrmbfit{pr}_{2},\,1}\rangle}}
{\langle{\mathrmbf{R}_{2},\mathrmbfit{T}_{2}}\rangle}
$}\normalsize}\]
with
the span of 
projection passages 
$\mathrmbf{R}_{1}^{\mathrm{op}}\xleftarrow{\mathrmbfit{pr}_{1}^{\mathrm{op}}}
\mathrmbf{R}_{12}^{\mathrm{op}}
\xrightarrow{\mathrmbfit{pr}_{2}^{\mathrm{op}}}\mathrmbf{R}_{2}^{\mathrm{op}}$
and the pair of identity bridges 
$\mathrmbfit{pr}_{1}^{\mathrm{op}}{\circ}\mathrmbfit{T}_{1}
\xLeftarrow{\;1\,}
\mathrmbfit{T}_{12}$
and
$\mathrmbfit{pr}_{2}^{\mathrm{op}}{\circ}\mathrmbfit{T}_{12}
\xRightarrow{\;1\,}
\mathrmbfit{T}_{2}$.
\end{description}
%

%%%%%%%%%%%%%%%%%%%%%%%%%%%%%%%%%%%%%%%%%%%%%%%%%%%%%%%%%%%%%%%%%%%%%%%%%%%%%%%%%%%%%%%%%%
%\newpage
%\subsubsection{Projections in $\mathrmbf{DB}$.}
%\label{sub:sub:sec:rel:db:proj}
%%%%%%%%%%%%%%%%%%%%%%%%%%%%%%%%%%%%%%%%%%%%%%%%%%%%%%%%%%%%%%%%%%%%%%%%%%%%%%%%%%%%%%%%%%

%%%%%%%%%%%%%%%%%%%%%%%%%%%%%%%%%%%%%%%%%%%%%%%%%%%%%%%%%%%%%%%%%%%%%%%%%%%%%%%%%%%%%%%%%%
%
\newpage
\paragraph{Projections.}
%%%%%%%%%%%%%%%%%%%%%%%%%%%%%%%%%%%%%%%%%%%%%%%%%%%%%%%%%%%%%%%%%%%%%%%%%%%%%%%%%%%%%%%%%%

%Composition with table projection passages 
%(Expo.~\ref{expo:intro:tbl} in \S~\ref{sec:intro}) 
%define database projection passages.
%Projections offer an alternate representation,
%defining the three primary components of 
%databases and database morphisms: 
%diagram shapes, schemed domains and key diagrams.
Projections offer an alternate representation,
defining the three primary components of 
databases and database morphisms: 
diagram shapes, schemed domains and key diagrams.
Diagram shapes are direct projections,
whereas schemed domains and key diagrams
are indirect,
coming from composition with table projection passages
Compare the following figure 
(from the {\ttfamily FOLE} table paper \cite{kent:fole:era:tbl})
to 
Fig.\,\ref{fig:fole:db:cxt:var} below. 
%(Expo.~\ref{expo:intro:tbl} in \S~\ref{sub:sub:sec:tables}).
%\begin{figure}
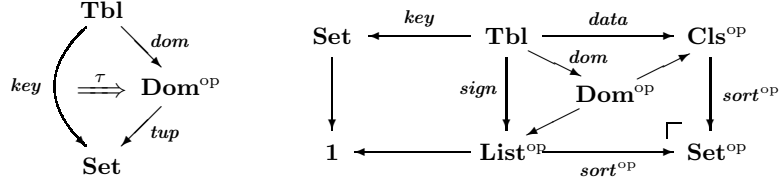
\begin{center}
{{\begin{tabular}{c@{\hspace{30pt}}c}
{{\begin{tabular}{c}
\setlength{\unitlength}{0.5pt}
\begin{picture}(100,120)(-20,0)
\put(0,118){\makebox(0,0){\footnotesize{$\mathrmbf{Tbl}$}}}
\put(60,60){\makebox(0,0){\footnotesize{$\mathrmbf{Dom}^{\mathrm{op}}$}}}
\put(0,0){\makebox(0,0){\footnotesize{$\mathrmbf{Set}$}}}
\put(-45,60){\makebox(0,0)[r]{\scriptsize{$\mathrmbfit{key}$}}}
\put(36,95){\makebox(0,0)[l]{\scriptsize{$\mathrmbfit{dom}$}}}
\put(36,24){\makebox(0,0)[l]{\scriptsize{$\mathrmbfit{tup}$}}}
\put(0,55){\makebox(0,0){{$\xRightarrow{\;\;\tau\;}$}}}
\put(15,105){\vector(1,-1){30}}
\put(45,45){\vector(-1,-1){30}}
\qbezier(-12,105)(-60,60)(-12,15)\put(-12,15){\vector(1,-1){0}}
\end{picture}
\end{tabular}}}
%%%%%%%%%%%%%%%%%%%%%%%%%%%%%%%%%%%%%%%%
&
%%%%%%%%%%%%%%%%%%%%%%%%%%%%%%%%%%%%%%%%
{{\begin{tabular}{c}
\setlength{\unitlength}{0.55pt}
\begin{picture}(320,160)(0,-35)
\put(20,80){\makebox(0,0){\footnotesize{$\mathrmbf{Set}$}}}
\put(20,0){\makebox(0,0){\footnotesize{$\mathrmbf{1}$}}}
\put(140,80){\makebox(0,0){\footnotesize{$\mathrmbf{Tbl}$}}}
\put(285,80){\makebox(0,0){\footnotesize{$\mathrmbf{Cls}^{\mathrm{op}}$}}}
\put(145,0){\makebox(0,0){\footnotesize{$\mathrmbf{List}^{\mathrm{op}}$}}}
\put(285,0){\makebox(0,0){\footnotesize{$\mathrmbf{Set}^{\mathrm{op}}$}}}
\put(215,40){\makebox(0,0){\footnotesize{$\mathrmbf{Dom}^{\mathrm{op}}$}}}
\put(80,92){\makebox(0,0){\scriptsize{$\mathrmbfit{key}$}}}
\put(182,66){\makebox(0,0)[l]{\scriptsize{$\mathrmbfit{dom}$}}}
\put(210,92){\makebox(0,0){\scriptsize{$\mathrmbfit{data}$}}}
\put(134,43){\makebox(0,0)[r]{\scriptsize{$\mathrmbfit{sign}$}}}
\put(210,-12){\makebox(0,0){\scriptsize{$\mathrmbfit{sort}^{\mathrm{op}}$}}}
\put(288,40){\makebox(0,0)[l]{\scriptsize{$\mathrmbfit{sort}^{\mathrm{op}}$}}}
\put(115,80){\vector(-1,0){70}}
\put(165,80){\vector(1,0){90}}
\put(115,0){\vector(-1,0){80}}
\put(165,0){\vector(1,0){90}}
\put(20,65){\vector(0,-1){50}}
\put(140,65){\vector(0,-1){50}}
\put(280,65){\vector(0,-1){50}}
\put(155,70){\vector(2,-1){36}}
\put(190,30){\vector(-2,-1){36}}
\put(230,50){\vector(2,1){36}}
\qbezier(250,20)(255,20)(260,20)
\qbezier(250,20)(250,15)(250,10)
\end{picture}
\end{tabular}}}
\end{tabular}}}
\end{center}
%\caption{\texttt{FOLE} Table Mathematical Context}
%\label{fig:tbl:cxt}
%\end{figure}
%
%Composition yields
%signature/type-domain projection passages
%$\mathrmbf{LIST}^{\mathrm{op}}\xleftarrow{\mathring{\mathrmbfit{sign}}}
%\mathrmbf{DB}\xrightarrow{\mathring{\mathrmbfit{data}}}\mathrmbf{CLS}^{\mathrm{op}}$.
%We can have four indexing contexts for databasess 
%(Fig.\,\ref{fig:fole:db:cxt:var}):
%% (Figure~\ref{fig:proj:func}):
%schema $\mathrmbf{LIST}$, 
%type domain diagrams $\mathrmbf{CLS}$,
%schemed domains $\mathrmbf{DOM}$,
%and key diagrams $\mathrmbf{SET}$.
%Hence,
The context of schemed domains
$\mathrmbf{DOM}$ can also be defined as the fibered product
\[\mbox{\footnotesize{$
\mathrmbf{LIST}\xleftarrow{\mathring{\mathrmbfit{sign}}}
\mathrmbf{DOM}=\mathrmbf{LIST}{\times_{\mathrmbf{SET}}}\mathrmbf{CLS}
\xrightarrow{\mathring{\mathrmbfit{data}}}\mathrmbf{CLS}
$,}\normalsize}\]
for the opspan of passages
$\mathrmbf{LIST}
\xrightarrow{\mathring{\mathrmbfit{sort}}}
\mathrmbf{SET}\xleftarrow{\mathring{\mathrmbfit{sort}}}
\mathrmbf{CLS}$.
The database projections are described in Fig.\,\ref{fig:fole:db:cxt:var}
and are defined as follows.
\begin{itemize}
\item 
The schemed domain  projection 
$\mathring{\mathrmbfit{dom}} 
= {(\mbox{-})}^{\mathrm{op}} \circ \mathrmbfit{dom} 
: \mathrmbf{DB}^{\mathrm{op}} \rightarrow \mathrmbf{DOM}$
%\begin{itemize}
%\item 
maps
a relational database ${\langle{\mathrmbf{R},\mathrmbfit{T}}\rangle}$ 
to the schemed domain
$\mathring{\mathrmbfit{dom}}(\mathrmbf{R},\mathrmbfit{T})
={\langle{\mathrmbf{R},\mathrmbfit{Q}}\rangle}$
with signed domain diagram
$\mathrmbf{R}\xrightarrow[\,\mathrmbfit{T}^{\mathrm{op}}\!{\circ\,}\mathrmbfit{dom}]{\;\mathrmbfit{Q}\;}\mathrmbf{Dom}$,
%\item 
and maps 
a relational database morphism 
${\langle{\mathrmbf{R}_{2},\mathrmbfit{T}_{2}}\rangle} 
\xleftarrow{{\langle{\mathrmbfit{R},\,\xi}\rangle}}
{\langle{\mathrmbf{R}_{1},\mathrmbfit{T}_{1}}\rangle}$
to the schemed domain morphism
{\footnotesize{$
%\mathring{\mathrmbfit{dom}}(\mathcal{R}_{2})=
\mathring{\mathrmbfit{dom}}(\mathrmbfit{R},\xi) = {\langle{\mathrmbfit{R},\varsigma}\rangle}
:
{\langle{\mathrmbf{R}_{2},\mathrmbfit{Q}_{2}}\rangle}
\rightarrow
%\xrightarrow[\mathring{\mathrmbfit{dom}}(\mathrmbfit{R},\xi)]
%{{\langle{\mathrmbfit{R},\varsigma}\rangle}} 
{\langle{\mathrmbf{R}_{1},\mathrmbfit{Q}_{1}}\rangle}
%=\mathring{\mathrmbfit{dom}}(\mathcal{R}_{1})
%in $\mathrmbf{DOM}$
%with equivalent bridge pair
%$\hat{\zeta} = {\langle{\acute{\zeta},\grave{\zeta}}\rangle}
%= \xi^{\mathrm{op}}\!{\circ\,}\mathrmbfit{dom}$,
%:
%\mathrmbfit{K}_{2} \Leftarrow
%%[\;\,\psi{\;\circ\;}\mathrmbfit{key}_{\mathcal{A}_{1}}]
%%{\;\,\kappa\;}
%\mathrmbfit{R}^{\mathrm{op}}{\circ\;}\mathrmbfit{K}_{1}
$}}
%\newline
with schemed domain bridge
$
%\varsigma = 
%\xi^{\mathrm{op}}{\,\circ\,}\mathrmbfit{dom} :
\mathrmbfit{Q}_{2} \xRightarrow
[\,
\xi^{\mathrm{op}}{\circ\,}\mathrmbfit{dom}
]
{\;\,\varsigma\;}
\mathrmbfit{R}{\;\circ\;}\mathrmbfit{Q}_{1}$.
%\end{itemize}
%
\item 
The key projection 
$\mathring{\mathrmbfit{key}} 
= {(\mbox{-})} \circ \mathrmbfit{key} 
: \mathrmbf{DB} \rightarrow \mathrmbf{SET}$
%=\mathrmbf{Set}^{\!\scriptscriptstyle{\Downarrow}}$ 
%\begin{itemize}
%\item 
maps
a relational database ${\langle{\mathrmbf{R},\mathrmbfit{T}}\rangle}$ 
to the key set diagram
$\mathring{\mathrmbfit{key}}(\mathrmbf{R},\mathrmbfit{T})
={\langle{\mathrmbf{R},\mathrmbfit{K}}\rangle}$ 
with the key passage
$\mathrmbf{R}^{\mathrm{op}}\!\xrightarrow[\,\mathrmbfit{T}{\,\circ\,}\mathrmbfit{key}]{\,\mathrmbfit{K}\;}\mathrmbf{Set}$,
%\item 
and maps 
a relational database morphism 
${\langle{\mathrmbf{R}_{2},\mathrmbfit{T}_{2}}\rangle} 
\xleftarrow{{\langle{\mathrmbfit{R},\,\xi}\rangle}}
{\langle{\mathrmbf{R}_{1},\mathrmbfit{T}_{1}}\rangle}$
to the $\mathrmbf{SET}$-morphism 
%\newline\mbox{}\hfill
{\footnotesize{$
\mathring{\mathrmbfit{key}}(\mathrmbfit{R},\,\xi)
=
{\langle{\mathrmbfit{R},\kappa}\rangle} :
{\langle{\mathrmbf{R}_{2},\mathrmbfit{K}_{2}}\rangle}
\leftarrow
{\langle{\mathrmbf{R}_{1},\mathrmbfit{K}_{1}}\rangle}$}}
%\hfill\mbox{}\newline
with a bridge
$\mathrmbfit{K}_{2}
\xLeftarrow[\,\xi{\,\circ\,}\mathrmbfit{key}]{\kappa}
\mathrmbfit{R} \circ \mathrmbfit{K}_{1}$.
%between (key) set diagrams.
%\end{itemize}
%
\end{itemize}
\begin{figure}
\begin{center}
{{\begin{tabular}{c@{\hspace{30pt}}c}
{{\begin{tabular}{c}
\setlength{\unitlength}{0.5pt}
\begin{picture}(100,120)(-20,0)
\put(0,118){\makebox(0,0){\footnotesize{$\mathrmbf{DB}$}}}
\put(65,60){\makebox(0,0){\footnotesize{$\mathrmbf{DOM}^{\mathrm{op}}$}}}
\put(2,-2){\makebox(0,0){\footnotesize{$\mathrmbf{SET}$}}}
\put(-40,60){\makebox(0,0)[r]{\scriptsize{$\mathring{\mathrmbfit{key}}$}}}
\put(40,95){\makebox(0,0)[l]{\scriptsize{$\mathring{\mathrmbfit{dom}}^{\mathrm{op}}$}}}
\put(36,24){\makebox(0,0)[l]{\scriptsize{$\mathring{\mathrmbfit{tup}}$}}}
\put(0,58){\makebox(0,0){{$\xRightarrow{\,\;\mathring{\tau}\,}$}}}
\put(15,105){\vector(1,-1){30}}
\put(45,45){\vector(-1,-1){30}}
\qbezier(-12,105)(-60,60)(-12,15)\put(-12,15){\vector(1,-1){0}}
\end{picture}
\end{tabular}}}
%%%%%%%%%%%%%%%%%%%%%%%%%%%%%%%%%%%%%%%%
&
%%%%%%%%%%%%%%%%%%%%%%%%%%%%%%%%%%%%%%%%
{{\begin{tabular}{c}
\setlength{\unitlength}{0.55pt}
\begin{picture}(320,160)(0,-35)
\put(20,80){\makebox(0,0){\footnotesize{$\mathrmbf{SET}$}}}
\put(20,0){\makebox(0,0){\footnotesize{$\mathrmbf{1}$}}}
\put(140,80){\makebox(0,0){\footnotesize{$\mathrmbf{DB}$}}}
\put(290,80){\makebox(0,0){\footnotesize{$\mathrmbf{CLS}^{\mathrm{op}}$}}}
\put(145,0){\makebox(0,0){\footnotesize{$\mathrmbf{LIST}^{\mathrm{op}}$}}}
\put(290,0){\makebox(0,0){\footnotesize{$
{\mathrmbf{Set}^{\!\scriptscriptstyle{\Uparrow}}}^{\mathrm{op}}$}}}
\put(218,40){\makebox(0,0){\footnotesize{$\mathrmbf{DOM}^{\mathrm{op}}$}}}
\put(80,92){\makebox(0,0){\scriptsize{$\mathring{\mathrmbfit{key}}$}}}
\put(182,66){\makebox(0,0)[l]{\scriptsize{$\mathring{\mathrmbfit{dom}}^{\mathrm{op}}$}}}
\put(215,92){\makebox(0,0){\scriptsize{$\mathring{\mathrmbfit{data}}^{\mathrm{op}}$}}}
\put(134,43){\makebox(0,0)[r]{\scriptsize{$\mathring{\mathrmbfit{sch}}^{\mathrm{op}}$}}}
\put(220,-12){\makebox(0,0){\scriptsize{$\mathring{\mathrmbfit{sort}}^{\mathrm{op}}$}}}
\put(288,40){\makebox(0,0)[l]{\scriptsize{$\mathring{\mathrmbfit{sort}}^{\mathrm{op}}$}}}
\put(105,80){\vector(-1,0){60}}
\put(95,0){\vector(-1,0){60}}
\put(175,80){\vector(1,0){70}}
\put(180,0){\vector(1,0){75}}
\put(20,65){\vector(0,-1){50}}
\put(140,65){\vector(0,-1){50}}
\put(280,65){\vector(0,-1){50}}
\put(155,70){\vector(2,-1){36}}
\put(190,30){\vector(-2,-1){36}}
\put(230,50){\vector(2,1){36}}
\qbezier(250,20)(255,20)(260,20)
\qbezier(250,20)(250,15)(250,10)
\end{picture}
\end{tabular}}}
%%%%%%%%%%%%%%%%%%%%%%%%%%%%%%%%%%%%%%%%%%%%%%%%%%%%%%%%%%%%%%%%%%%%%%%%%%%%%%%%%%%%%%%%%%%%%%%%%%%%
\\
%%%%%%%%%%%%%%%%%%%%%%%%%%%%%%%%%%%%%%%%%%%%%%%%%%%%%%%%%%%%%%%%%%%%%%%%%%%%%%%%%%%%%%%%%%%%%%%%%%%%
\multicolumn{2}{c}{
{\footnotesize{$\begin{array}{r@{\hspace{16pt}}r@{\hspace{5pt}=\hspace{5pt}}l@{\hspace{5pt}:\hspace{5pt}}l}
\text{schemed domain}
&
\mathring{\mathrmbfit{dom}} 
& {(\mbox{-})}^{\mathrm{op}} \circ \mathrmbfit{dom} 
& \mathrmbf{DB}^{\mathrm{op}} \rightarrow \mathrmbf{DOM}
\\
\text{data}
&
\mathring{\mathrmbfit{data}} 
& {(\mbox{-})}^{\mathrm{op}} \circ \mathrmbfit{data} 
& \mathrmbf{DB}^{\mathrm{op}} \rightarrow \mathrmbf{CLS}
\\
\text{schema}
&
\mathring{\mathrmbfit{sch}} 
& {(\mbox{-})}^{\mathrm{op}} \circ \mathrmbfit{sign} 
& \mathrmbf{DB}^{\mathrm{op}} \rightarrow \mathrmbf{LIST}
\\
\text{key}
&
\mathring{\mathrmbfit{key}} 
& {(\mbox{-})} \circ \mathrmbfit{key} 
& \mathrmbf{DB} \rightarrow \mathrmbf{SET}=\mathrmbf{Set}^{\!\scriptscriptstyle{\Downarrow}}
\end{array}$}}}
\end{tabular}}}
\end{center}
\caption{Database Mathematical Context: $\mathrmbf{DB}$}
\label{fig:fole:db:cxt:var}
\end{figure}

\newpage

\begin{proposition}\label{prop:db:proj}
Using projections 
%mathematical context 
$\mathrmbf{DB}$ can be described as follows.
\begin{itemize}
\item  
A database 
$\mathcal{R}={\langle{\mathrmbf{R},\mathrmbfit{Q},\mathrmbfit{K},\tau}\rangle}$
%in $\mathrmbf{DB}$
%=\mathrmbf{Tbl}^{\scriptscriptstyle{\Downarrow}}$
consists of
a shape context $\mathrmbf{R}$,
%\begin{itemize}
%\item 
%(i)
%an object 
%$\mathring{\mathrmbfit{key}}(\mathcal{R})={\langle{\mathrmbf{R},\mathrmbfit{K}}\rangle}$
%in $\mathrmbf{SET}=\mathrmbf{Set}^{\scriptscriptstyle{\Downarrow}}$
%with shape context $\mathrmbf{R}$
%and 
a schemed domain
$\mathring{\mathrmbfit{dom}}(\mathrmbf{R},\mathrmbfit{T})
={\langle{\mathrmbf{R},\mathrmbfit{Q}}\rangle}$
%in $\mathrmbf{DOM}$
%=\mathrmbf{Dom}^{\scriptscriptstyle{\Uparrow}}$
with 
%shape context $\mathrmbf{R}$ and 
the signed domain diagram
$\mathrmbf{R}\xrightarrow[\,\mathrmbfit{T}^{\mathrm{op}}\!{\circ\,}\mathrmbfit{dom}]{\;\mathrmbfit{Q}\;}\mathrmbf{Dom}$,
a 
key set diagram
$\mathring{\mathrmbfit{key}}(\mathrmbf{R},\mathrmbfit{T})
={\langle{\mathrmbf{R},\mathrmbfit{K}}\rangle}$
with the key passage
$\mathrmbf{R}^{\mathrm{op}}\!\xrightarrow[\,\mathrmbfit{T}{\,\circ\,}\mathrmbfit{key}]{\,\mathrmbfit{K}\;}\mathrmbf{Set}$,
%\item 
%(ii)
%an object (
and
%\item 
%(iii)
a tuple bridge
$\mathrmbfit{K}
\xRightarrow[\mathrmbfit{T}{\,\circ\,}\tau_{\mathrm{Tbl}}]{\;\tau\;}
\mathrmbfit{Q}^{\mathrm{op}}{\circ\,}\mathrmbfit{tup}$.
%\newline
%\end{itemize}
%
\item 
%\item[database morphism:] 
A database morphism 
(Fig.~\ref{fig:db:morph:proj})
%is a morphism
\newline\mbox{}\hfill
{\footnotesize{$\mathcal{R}_{2}=
{\langle{\mathrmbf{R}_{2},\mathrmbfit{Q}_{2},\mathrmbfit{K}_{2},\tau_{2}}\rangle} 
\xleftarrow{{\langle{\mathrmbfit{R},\varsigma,\kappa}\rangle}}
{\langle{\mathrmbf{R}_{1},\mathrmbfit{Q}_{1},\mathrmbfit{K}_{1},\tau_{1}}\rangle}
=\mathcal{R}_{1}$}}
\hfill\mbox{}\newline
%in $\mathrmbf{DB}$
consists of
a shape passage
$\mathrmbf{R}_{2} \xrightarrow{\mathrmbfit{R}\;} \mathrmbf{R}_{1}$,
%\begin{itemize}
%
%\item 
%a shape passage
%$\mathrmbf{R}_{2}\xrightarrow{\mathrmbfit{R}}\mathrmbf{R}_{1}$,
%\item 
%a $\mathrmbf{DOM}$-morphism 
a schemed domain morphism 
{\footnotesize{$
%\mathring{\mathrmbfit{dom}}(\mathcal{R}_{2})=
\mathring{\mathrmbfit{dom}}(\mathrmbfit{R},\xi) = {\langle{\mathrmbfit{R},\varsigma}\rangle}
:
{\langle{\mathrmbf{R}_{2},\mathrmbfit{Q}_{2}}\rangle}
%\rightarrow
\xrightarrow[\mathring{\mathrmbfit{dom}}(\mathrmbfit{R},\xi)]
{{\langle{\mathrmbfit{R},\varsigma}\rangle}} 
{\langle{\mathrmbf{R}_{1},\mathrmbfit{Q}_{1}}\rangle},
%=\mathring{\mathrmbfit{dom}}(\mathcal{R}_{1})
%in $\mathrmbf{DOM}$
%with equivalent bridge pair
%$\hat{\zeta} = {\langle{\acute{\zeta},\grave{\zeta}}\rangle}
%= \xi^{\mathrm{op}}\!{\circ\,}\mathrmbfit{dom}$,
%:
%\mathrmbfit{K}_{2} \Leftarrow
%%[\;\,\psi{\;\circ\;}\mathrmbfit{key}_{\mathcal{A}_{1}}]
%%{\;\,\kappa\;}
%\mathrmbfit{R}^{\mathrm{op}}{\circ\;}\mathrmbfit{K}_{1}
$}}
%\newline
%with schemed domain bridge
%$\varsigma = 
%\xi^{\mathrm{op}}{\,\circ\,}\mathrmbfit{dom}
%:
%\mathrmbfit{Q}_{2} \Rightarrow
%%[\;\,\psi{\;\circ\;}\mathrmbfit{key}_{\mathcal{A}_{1}}]
%%{\;\,\kappa\;}
%\mathrmbfit{R}{\;\circ\;}\mathrmbfit{Q}_{1}$,
and
%\item 
%a $\mathrmbf{SET}$-morphism 
a key diagram morphism
%\newline\mbox{}\hfill
{\footnotesize{$
\mathring{\mathrmbfit{key}}(\mathrmbfit{R},\xi) = 
{{\langle{\mathrmbfit{R},\kappa}\rangle}}
:
%\mathcal{K}_{2}=
{\langle{\mathrmbf{R}_{2},\mathrmbfit{K}_{2}}\rangle} 
\leftarrow
{\langle{\mathrmbf{R}_{1},\mathrmbfit{K}_{1}}\rangle}
$}}
%\hfill\mbox{}\newline
with the
key bridge
$\kappa = \xi{\;\circ\;}\mathrmbfit{key}
:
\mathrmbfit{K}_{2} \Leftarrow
%[\;\,\psi{\;\circ\;}\mathrmbfit{key}_{\mathcal{A}_{1}}]
%{\;\,\kappa\;}
\mathrmbfit{R}^{\mathrm{op}}{\circ\;}\mathrmbfit{K}_{1}$,
%a morphism
%{\footnotesize{$\mathring{\mathrmbfit{key}}(\mathcal{R}_{2})=
%{\langle{\mathrmbf{R}_{2},\mathrmbfit{K}_{2}}\rangle}
%%\xleftarrow{{\langle{\mathrmbfit{R},\hat{\kappa}}\rangle}}
%\xleftarrow[{\langle{\mathrmbfit{R},\hat{\kappa}}\rangle}]
%{\mathring{\mathrmbfit{key}}(\mathrmbfit{R},\hat{\zeta},\kappa)}
%{\langle{\mathrmbf{R}_{1},\mathrmbfit{K}_{1}}\rangle}
%=\mathring{\mathrmbfit{key}}(\mathcal{R}_{1})$}}
%in $\mathrmbf{SET}$
%with equivalent bridge pair
%$\kappa = {\langle{\acute{\kappa},\grave{\kappa}}\rangle}
%= \xi{\,\circ\,}\mathrmbfit{key}$, and
%\end{itemize}
%
which satisfy the condition
\begin{equation}\label{expo:db:mor:var:shape}
{{\begin{tabular}{c}
{{\footnotesize{\setlength{\extrarowheight}{2pt}
$\begin{array}{@{\hspace{-15pt}}
%r@{\hspace{10pt}}
%r@{\hspace{5pt}=\hspace{5pt}}
l@{\hspace{5pt}=\hspace{5pt}}l}
%&
%\grave{\xi}{\,\circ\,}\tau
%&
(\underset{\textstyle{{\kappa}}}{\underbrace{{\xi}{\,\circ\,}\mathrmbfit{key}}}
){\,\bullet\,}
%(
%\underset{\textstyle{
\tau_{2}
%}}{\underbrace{\mathrmbfit{T}_{2}{\,\circ\,}\tau}})
&
(\mathrmbfit{R}^{\mathrm{op}}{\circ\,}\tau_{1})
{\,\bullet\,}(
\underset{\textstyle{{\varsigma}^{\mathrm{op}}}}
{\underbrace{{\xi}{\,\circ\,}\mathrmbfit{dom}^{\mathrm{op}}}}
{\circ\,}\mathrmbfit{tup})
\end{array}$.
}}}
\end{tabular}}}
\end{equation}
\end{itemize}
\end{proposition}
%

%\newpage

%
\begin{figure}
\begin{center}
{{
\begin{tabular}{@{\hspace{-20pt}}c@{\hspace{20pt}}c}
\begin{tabular}{c}
\setlength{\unitlength}{0.5pt}
\begin{picture}(240,180)(100,0)
%\put(0,0){\framebox(240,180){}}
\put(125,190){\makebox(0,0){\scriptsize{$\mathrmbfit{R}^{\mathrm{op}}$}}}
%\put(120,120){\makebox(0,0){\scriptsize{$\mathrmbfit{id}$}}}
\put(120,60){\makebox(0,0){\scriptsize{$\mathrmbfit{id}$}}}
\put(120,0){\makebox(0,0){\scriptsize{$\mathrmbfit{id}$}}}
\put(30,180){\vector(1,0){180}}
%\put(25,120){\vector(1,0){190}}
\put(95,60){\vector(1,0){50}}\put(95,60){\vector(-1,0){0}}
\put(25,0){\vector(1,0){190}}\put(25,0){\vector(-1,0){0}}
\put(120,146){\makebox(0,0){\shortstack{\scriptsize{$\;\varsigma^{\mathrm{op}}$}\\\large{$\Longleftarrow$}}}}
\put(120,92){\makebox(0,0){\large{$\overset{\rule[-2pt]{0pt}{5pt}\kappa}{\Longleftarrow}$}}}
\qbezier[60](42,140)(110,140)(180,140)
\qbezier[150](-75,85)(120,85)(300,85)
\put(-8,0){\begin{picture}(0,0)(0,0)
\put(8,180){\makebox(0,0){\footnotesize{$\mathrmbf{R}_{2}^{\mathrm{op}}$}}}
\put(0,118){\makebox(0,0){\footnotesize{$\mathrmbf{Tbl}$}}}
\put(68,60){\makebox(0,0){\footnotesize{$\mathrmbf{Dom}^{\mathrm{op}}$}}}
\put(0,0){\makebox(0,0){\footnotesize{$\mathrmbf{Set}$}}}
\put(-6,148){\makebox(0,0)[r]{\scriptsize{$\mathrmbfit{T}_{2}$}}}
\put(-75,95){\makebox(0,0)[r]{\scriptsize{$\mathrmbfit{K}_{2}$}}}
\put(55,145){\makebox(0,0)[l]{\scriptsize{$\mathrmbfit{Q}_{2}^{\mathrm{op}}$}}}
\put(-37,72){\makebox(0,0)[r]{\scriptsize{$\mathrmbfit{key}$}}}
\put(38,93){\makebox(0,0)[l]{\scriptsize{$\mathrmbfit{dom}^{\mathrm{op}}$}}}
\put(36,26){\makebox(0,0)[l]{\scriptsize{$\mathrmbfit{tup}$}}}
\put(0,60){\makebox(0,0){\shortstack{\scriptsize{$\;\tau_{\scriptscriptstyle\mathrmbf{Tbl}}$}\\\large{$\Longrightarrow$}}}}
\put(0,165){\vector(0,-1){34}}
\put(15,105){\vector(1,-1){30}}
\put(45,45){\vector(-1,-1){30}}
\qbezier(-18,167)(-120,90)(-20,13)\put(-20,13){\vector(1,-1){0}}
\qbezier(-12,105)(-60,60)(-12,15)\put(-12,15){\vector(1,-1){0}}
\qbezier(18,167)(70,140)(66,76)\put(66,76){\vector(0,-1){0}}
\end{picture}}
\put(233,0){\begin{picture}(0,0)(0,0)
\put(8,180){\makebox(0,0){\footnotesize{$\mathrmbf{R}_{1}^{\mathrm{op}}$}}}
\put(0,118){\makebox(0,0){\footnotesize{$\mathrmbf{Tbl}$}}}
\put(-52,60){\makebox(0,0){\footnotesize{$\mathrmbf{Dom}^{\mathrm{op}}$}}}
\put(0,0){\makebox(0,0){\footnotesize{$\mathrmbf{Set}$}}}
\put(6,148){\makebox(0,0)[l]{\scriptsize{$\mathrmbfit{T}_{1}$}}}
\put(75,95){\makebox(0,0)[l]{\scriptsize{$\mathrmbfit{K}_{1}$}}}
\put(-50,145){\makebox(0,0)[r]{\scriptsize{$\mathrmbfit{Q}_{1}^{\mathrm{op}}$}}}
\put(37,72){\makebox(0,0)[l]{\scriptsize{$\mathrmbfit{key}$}}}
\put(-32,93){\makebox(0,0)[r]{\scriptsize{$\mathrmbfit{dom}^{\mathrm{op}}$}}}
\put(-36,26){\makebox(0,0)[r]{\scriptsize{$\mathrmbfit{tup}$}}}
\put(0,60){\makebox(0,0){\shortstack{\scriptsize{$\;\tau_{\scriptscriptstyle\mathrmbf{Tbl}}$}\\\large{$\Longleftarrow$}}}}
\put(0,165){\vector(0,-1){34}}
\put(-15,105){\vector(-1,-1){30}}
\put(-45,45){\vector(1,-1){30}}
\qbezier(18,167)(120,90)(20,13)\put(20,13){\vector(-1,-1){0}}
\qbezier(12,105)(60,60)(12,15)\put(12,15){\vector(-1,-1){0}}
\qbezier(-18,167)(-70,140)(-66,76)\put(-66,76){\vector(0,-1){0}}
\end{picture}}
\end{picture}
\end{tabular}
&
\begin{tabular}{c}
%
%$\Bigg\Downarrow \bigg\Downarrow \Big\Downarrow \big\Downarrow$
\setlength{\unitlength}{0.6pt}
\begin{picture}(120,100)(0,-10)
%\put(0,0){\framebox(120,80){}}
\put(0,80){\makebox(0,0){\scriptsize{$\mathrmbfit{R}^{\mathrm{op}} \circ \mathrmbfit{K}_{1}$}}}
%\put(125,80){\makebox(0,0){\scriptsize{$\mathrmbfit{R}^{\mathrm{op}}\circ\mathrmbfit{Q}_{1}^{\mathrm{op}}\circ\mathrmbfit{tup}$}}}
\put(125,94){\makebox(0,0){\scriptsize{$\overset{\textstyle\mathrmbfit{R}^{\mathrm{op}}\circ\mathring{\mathrmbfit{tup}}
(\mathrmbf{R}_{1},\mathrmbfit{Q}_{1})}{\overbrace{\mathrmbfit{R}^{\mathrm{op}}\circ\mathrmbfit{Q}_{1}^{\mathrm{op}}\circ\mathrmbfit{tup}}}$}}}
\put(185,80){\makebox(0,0)[l]{\scriptsize{$=$}}}
\put(205,80){\makebox(0,0)[l]{\scriptsize{$\mathrmbfit{R}^{\mathrm{op}}\circ\mathrmbfit{T}_{1}\circ\mathrmbfit{dom}^{\mathrm{op}}\circ\mathrmbfit{tup}$}}}
\put(0,0){\makebox(0,0){\scriptsize{$\mathrmbfit{K}_{2}$}}}
%\put(125,0){\makebox(0,0){\scriptsize{$\mathrmbfit{Q}_{2}^{\mathrm{op}}\circ\mathrmbfit{tup}$}}}
\put(125,-13){\makebox(0,0){\scriptsize{$\overset{\textstyle\underbrace{\mathrmbfit{Q}_{2}^{\mathrm{op}}\circ\mathrmbfit{tup}}}{\mathring{\mathrmbfit{tup}}(\mathrmbf{R}_{2},\mathrmbfit{Q}_{2})}$}}}
\put(185,0){\makebox(0,0)[l]{\scriptsize{$=$}}}
\put(205,0){\makebox(0,0)[l]{\scriptsize{$\mathrmbfit{T}_{2}\circ\mathrmbfit{dom}^{\mathrm{op}}\circ\mathrmbfit{tup}$}}}
\put(145,45){\makebox(0,0)[l]{\scriptsize{$\overset{\textstyle\mathring{\mathrmbfit{tup}}(\mathrmbfit{R},\varsigma)}{\overbrace{\varsigma^{\mathrm{op}}\circ\mathrmbfit{tup}}}$}}}
\put(265,40){\makebox(0,0)[l]{\scriptsize{$\xi\circ\mathrmbfit{dom}^{\mathrm{op}}\circ\mathrmbfit{tup}$}}}
\put(55,95){\makebox(0,0){\scriptsize{$
\mathrmbfit{R}^{\mathrm{op}}{\circ\,}\tau_{1}
$}}}
\put(55,-15){\makebox(0,0){\scriptsize{$\tau_{2}$}}}
\put(-10,40){\makebox(0,0)[r]{\scriptsize{$\kappa$}}}
\put(55,80){\makebox(0,0){\large{$\Longrightarrow$}}}
\put(55,0){\makebox(0,0){\large{$\Longrightarrow$}}}
\put(0,40){\makebox(0,0){\large{$\bigg\Downarrow$}}}
\put(125,40){\makebox(0,0){\large{$\bigg\Downarrow$}}}
\put(245,40){\makebox(0,0){\large{$\bigg\Downarrow$}}}
\end{picture}
\end{tabular}
\\ 
\end{tabular}
}}
\end{center}
\caption{\texttt{FOLE} Database Morphism (proj): $\mathrmbf{DB}$}
\label{fig:db:morph:proj}
\end{figure}
%

%\newpage

\comment{
\begin{proposition}\label{prop:comma:cxt:DB}
The context of relational databases is 
the comma context 
\begin{center}
{\footnotesize{$
\mathrmbf{SET}\xleftarrow{\;\mathring{\mathrmbfit{key}}\;}
\mathrmbf{DB} = \bigl(\mathrmbf{SET}{\;\downarrow\,}\mathring{\mathrmbfit{tup}}\bigr)
\xrightarrow{\;\mathring{\mathrmbfit{dom}}^{\mathrm{op}}}\mathrmbf{DOM}^{\mathrm{op}}
$}}
\end{center}
for the tuple passage 
$\mathrmbf{DOM}^{\mathrm{op}}
%={\mathrmbf{Dom}^{\scriptscriptstyle{\Uparrow}}}^{\mathrm{op}}
\!\xrightarrow{\mathring{\mathrmbfit{tup}}}
\mathrmbf{SET}
%=\mathrmbf{Set}^{\!\scriptscriptstyle{\Downarrow}}
$.
%(Def.~\ref{def:tup:pass}
%in
%\S~\ref{sec:sch:dom}).
%\S~\ref{sub:sub:sec:sch:dom:var:shape:gen}).
%(Def.~\ref{def:tup:pass} in \S~\ref{sub:sub:sec:sch:dom:var:shape:gen})
%$\mathrmbf{DOM}^{\mathrm{op}}\xrightarrow{\mathring{\mathrmbfit{tup}}}\mathrmbf{SET}$.
\end{proposition}
\begin{proof}
\mbox{}
%\newline
The descriptions
for
$\mathrmbf{DB}$
given in
Prop.\,\ref{prop:db:proj} above
and
$\bigl(\mathrmbf{SET}{\;\downarrow\,}\mathring{\mathrmbfit{tup}}\bigr)$
given in
Def.\,\ref{def:comma:cxt:DB} in \S\,\ref{sub:sub:sec:sch:dom:tup}
are equivalent.
\mbox{}\hfill\rule{5pt}{5pt}
\end{proof}
}

.

%picture below.

%
\comment{% not needed
\begin{center}
{{\begin{tabular}{c}
\setlength{\unitlength}{0.56pt}
\begin{picture}(120,160)(8,0)
\put(65,160){\makebox(0,0){\footnotesize{$
\mathrmbf{R}_{1}^{\mathrm{op}}{+\,}\mathrmbf{R}_{2}^{\mathrm{op}}$}}}
\put(5,80){\makebox(0,0){\footnotesize{$\mathrmbf{R}_{1}^{\mathrm{op}}$}}}
\put(125,80){\makebox(0,0){\footnotesize{$\mathrmbf{R}_{2}^{\mathrm{op}}$}}}
\put(65,0){\makebox(0,0){\footnotesize{$\mathrmbf{Tbl}$}}}
\put(15,132){\makebox(0,0){\scriptsize{$\mathrmbfit{inc}_{1}^{\mathrm{op}}$}}}
\put(115,132){\makebox(0,0){\scriptsize{$\mathrmbfit{inc}_{2}^{\mathrm{op}}$}}}
\put(20,42){\makebox(0,0)[r]{\scriptsize{$\mathrmbfit{T}_{1}$}}}
\put(100,42){\makebox(0,0)[l]{\scriptsize{$\mathrmbfit{T}_{2}$}}}
\put(65,80){\makebox(0,0)[l]{\scriptsize{$\mathrmbfit{T}_{12}$}}}
%
%\put(38,82){\makebox(0,0){\shortstack{\scriptsize{$=$}\\\large{$\Longleftarrow$}}}}
%\put(82,82){\makebox(0,0){\shortstack{\scriptsize{$=$}\\\large{$\Longrightarrow$}}}}
%
%\put(75,147){\vector(3,-4){38}}
\put(114,96){\vector(-3,4){38}}
%\put(45,148){\vector(-3,-4){38}}
\put(7,98){\vector(3,4){38}}
\put(60,140){\vector(0,-1){120}}
\put(10,67){\vector(3,-4){38}}
\put(110,68){\vector(-3,-4){38}}
\end{picture}
\end{tabular}}}
\end{center}
}% not needed
%

%%%%%%%%%%%%%%%%%%%%%%%%%%%%%%%%%%%%%%%%%%%%%%%%%%%%%%%%%%%%%%%%%%%%%%%%%%%%%%%%%%%%%%%%%%
%%%%%%%%%%%%%%%%%%%%%%%%%%%%%%%%%%%%%%%%%%%%%%%%%%%%%%%%%%%%%%%%%%%%%%%%%%%%%%%%%%%%%%%%%%
\newpage
\subsection{Type Domain Indexing}
\label{sub:sec:rel:db:typ:dom}
%%%%%%%%%%%%%%%%%%%%%%%%%%%%%%%%%%%%%%%%%%%%%%%%%%%%%%%%%%%%%%%%%%%%%%%%%%%%%%%%%%%%%%%%%%
%%%%%%%%%%%%%%%%%%%%%%%%%%%%%%%%%%%%%%%%%%%%%%%%%%%%%%%%%%%%%%%%%%%%%%%%%%%%%%%%%%%%%%%%%%

A type domain, which constrains the body of a relational table, is an indexed
collection of data types from which a table's tuples are chosen.
Here we define databases with fixed type domain (datatypes).
These are more generalized than databases with fixed signed domain (header plus datatypes).
%
%%%%%%%%%%%%%%%%%%%%%%%%%%%%%%%%%%%%%%%%%%%%%%%%%%%%%%%%%%%%%%%%%%%%%%%%%%%%%%%%
%%%%%%%%%%%%%%%%%%%%%%%%%%%%%%%%%%%%%%%%%%%%%%%%%%%%%%%%%%%%%%%%%%%%%%%%%%%%%%%%
\footnote{These are not defined here.
We could also define databases with fixed schema (header).}
%%%%%%%%%%%%%%%%%%%%%%%%%%%%%%%%%%%%%%%%%%%%%%%%%%%%%%%%%%%%%%%%%%%%%%%%%%%%%%%%
%%%%%%%%%%%%%%%%%%%%%%%%%%%%%%%%%%%%%%%%%%%%%%%%%%%%%%%%%%%%%%%%%%%%%%%%%%%%%%%%
%

%%%%%%%%%%%%%%%%%%%%%%%%%%%%%%%%%%%%%%%%%%%%%%%%%%%%%%%%%%%%%%%%%%%%%%%%%%%%%%%%%%%%%%%%%%
%\newpage
\subsubsection{Lower Aspect: $\mathrmbf{Db}(\mathcal{A})$}
\label{sub:sub:sec:rel:db:typ:dom:lower}
%%%%%%%%%%%%%%%%%%%%%%%%%%%%%%%%%%%%%%%%%%%%%%%%%%%%%%%%%%%%%%%%%%%%%%%%%%%%%%%%%%%%%%%%%%

Let $\mathcal{A}$ be a fixed type domain.
\begin{definition}\label{def:db:A:oplax:cxt}
The context of $\mathcal{A}$-databases is the 
oplax comma 
context
$\mathrmbf{Db}(\mathcal{A}) = \mathrmbf{Tbl}(\mathcal{A})^{\scriptscriptstyle{\Downarrow}}$,
a diagram context
over $\mathcal{A}$-tables.
(Def.~\ref{def:lax:oplax} in \S~\ref{sub:sub:sec:math:context})
\end{definition}
A relational $\mathcal{A}$-database ${\langle{\mathrmbf{R},\mathrmbfit{T}}\rangle}$ is a diagram of $\mathcal{A}$-tables,
consisting of 
%a shape context $\mathrmbf{R}$ 
%and an object $\mathrmbfit{T}$ in the fiber context $\mathrmbf{Db}^{\mathrmbf{R}}(\mathcal{A})$
%(Def.~\ref{def:fbr:db:fix:typ:dom} of \S~\ref{sub:sub:sec:rel:db:fix:shape:typ:dom:upper})
%that is,
a shape context $\mathrmbf{R}$ 
and a passage $\mathrmbf{R}^{\mathrm{op}}\!\xrightarrow{\,\mathrmbfit{T}\;}\mathrmbf{Tbl}(\mathcal{A})$.
A relational $\mathcal{A}$-database morphism
${\langle{\mathrmbf{R}_{2},\mathrmbfit{T}_{2}}\rangle} 
\xleftarrow{{\langle{\mathrmbfit{R},\,\psi}\rangle}}
{\langle{\mathrmbf{R}_{1},\mathrmbfit{T}_{1}}\rangle}$
(Fig.\,\ref{fig:db:A:mor:adj})
consists of 
a shape-changing passage $\mathrmbf{R}_{1}\xrightarrow{\;\mathrmbfit{R}\;\,}\mathrmbf{R}_{1}$
and a bridge 
$\mathrmbfit{T}_{2}\xLeftarrow{\;\;\psi\,}
\mathrmbfit{R}^{\mathrm{op}}{\circ}\mathrmbfit{T}_{1}
%=\mathrmbf{Tbl}(\mathcal{A})^\mathrmbfit{R}(\mathrmbfit{T}_{1})
$.
\begin{figure}
\begin{center}
%{{\begin{tabular}{c}
%\begin{tabular}{c@{\hspace{30pt}}c}
%
{{\begin{tabular}[b]{c}
\setlength{\unitlength}{0.56pt}
\begin{picture}(120,80)(8,0)
\put(5,80){\makebox(0,0){\footnotesize{$\mathrmbf{R}_{2}^{\mathrm{op}}$}}}
\put(125,80){\makebox(0,0){\footnotesize{$\mathrmbf{R}_{1}^{\mathrm{op}}$}}}
\put(65,0){\makebox(0,0){\footnotesize{$\mathrmbf{Tbl}(\mathcal{A})$}}}
\put(65,92){\makebox(0,0){\scriptsize{$\mathrmbfit{R}^{\mathrm{op}}$}}}
\put(20,42){\makebox(0,0)[r]{\scriptsize{$\mathrmbfit{T}_{2}$}}}
\put(100,42){\makebox(0,0)[l]{\scriptsize{$\mathrmbfit{T}_{1}$}}}
\put(60,57){\makebox(0,0){\shortstack{\scriptsize{$\psi$}\\\large{$\Longleftarrow$}}}}
\put(20,80){\vector(1,0){80}}
\put(10,67){\vector(3,-4){38}}
\put(110,68){\vector(-3,-4){38}}
\end{picture}
\end{tabular}}}
\end{center}
\caption{Database Morphism: $\mathrmbf{Db}(\mathcal{A})$}
\label{fig:db:A:mor:adj}
\end{figure}
%

%\newpage

%
\begin{proposition}\label{prop:db:A:lim:colim}
The fibered context (Grothendieck construction) of $\mathcal{A}$-databases
$\mathrmbf{Db}(\mathcal{A}) 
%= \int\hat{\mathrmbf{db}}_{\mathcal{A}} 
= \mathrmbf{Tbl}(\mathcal{A})^{\scriptscriptstyle{\Downarrow}}$
is complete and cocomplete 
and the projection 
$\mathrmbf{Db}(\mathcal{A}) = 
\mathrmbf{Tbl}(\mathcal{A})^{\scriptscriptstyle{\Downarrow}}\rightarrow\mathrmbf{Cxt}
:{\langle{\mathrmbf{R},\mathrmbfit{T}}\rangle}\mapsto\mathrmbf{R}$ 
is continuous and cocontinuous.
\end{proposition}
\begin{proof}
By Prop.~\ref{prop:lim:colim:A} of \S\ref{append:kan:ext},
since 
%context 
$\mathrmbf{Tbl}(\mathcal{A})$ is complete/cocomplete \cite{kent:fole:era:tbl}.
\hfill\rule{5pt}{5pt}
\end{proof}
%

%%%%%%%%%%%%%%%%%%%%%%%%%%%%%%%%%%%%%%%%%%%%%%%%%%%%%%%%%%%%%%%%%%%%%%%%%%%%%%%%%%%%%%%%%%
%\newpage
%\subsubsection{Projections in $\mathrmbf{Db}(\mathcal{A})$.}
%\label{sub:sub:sec:rel:db:typ:dom:lower:proj}
%%%%%%%%%%%%%%%%%%%%%%%%%%%%%%%%%%%%%%%%%%%%%%%%%%%%%%%%%%%%%%%%%%%%%%%%%%%%%%%%%%%%%%%%%%

%%%%%%%%%%%%%%%%%%%%%%%%%%%%%%%%%%%%%%%%%%%%%%%%%%%%%%%%%%%%%%%%%%%%%%%%%%%%%%%%%%%%%%%%%%
%\newpage
\paragraph{Projections.}
%%%%%%%%%%%%%%%%%%%%%%%%%%%%%%%%%%%%%%%%%%%%%%%%%%%%%%%%%%%%%%%%%%%%%%%%%%%%%%%%%%%%%%%%%%

Projections offer an alternate representation,
defining the three primary components of 
$\mathcal{A}$-databases and $\mathcal{A}$-database morphisms: 
diagram shapes, schemed domains and key diagrams.
Diagram shapes are direct projections,
whereas schemed domains and key diagrams
are indirect,
coming from composition with $\mathcal{A}$-table projection passages. 
%(Expo.~\ref{expo:intro:tbl:A} in \S~\ref{sub:sub:sec:tables:A}). 
%
The database projections are described in Fig.\,\ref{fig:fole:db:cxt:fbr:var:A}
and are defined as follows.

\begin{figure}
\begin{center}
{{\begin{tabular}{c@{\hspace{30pt}}c}
{{\begin{tabular}{c}
\setlength{\unitlength}{0.5pt}
\begin{picture}(100,120)(-20,0)
\put(0,118){\makebox(0,0){\footnotesize{$\mathrmbf{Db}(\mathcal{A})$}}}
\put(70,60){\makebox(0,0){\footnotesize{${\mathring{\mathrmbf{Dom}}(\mathcal{A})}^{\mathrm{op}}$}}}
\put(0,0){\makebox(0,0){\footnotesize{$\mathrmbf{SET}$}}}
\put(-45,60){\makebox(0,0)[r]{\scriptsize{$\mathring{\mathrmbfit{key}}_{\mathcal{A}}$}}}
\put(36,95){\makebox(0,0)[l]{\scriptsize{$\mathring{\mathrmbfit{dom}_{\mathcal{A}}}^{\mathrm{op}}$}}}
\put(36,24){\makebox(0,0)[l]{\scriptsize{$\mathring{\mathrmbfit{tup}}_{\mathcal{A}}$}}}
\put(0,60){\makebox(0,0){{$\xRightarrow{\mathring{\tau}_{\mathcal{A}}}$}}}
\put(15,105){\vector(1,-1){30}}
\put(45,45){\vector(-1,-1){30}}
\qbezier(-12,105)(-60,60)(-12,15)\put(-12,15){\vector(1,-1){0}}
\end{picture}
\end{tabular}}}
%%%%%%%%%%%%%%%%%%%%%%%%%%%%%%%%%%%%%%%%
&
%%%%%%%%%%%%%%%%%%%%%%%%%%%%%%%%%%%%%%%%
{{\begin{tabular}{c}
\setlength{\unitlength}{0.55pt}
\begin{picture}(320,160)(0,-35)
\put(20,80){\makebox(0,0){\footnotesize{$\mathrmbf{SET}$}}}
\put(20,0){\makebox(0,0){\footnotesize{$\mathrmbf{1}$}}}
\put(140,80){\makebox(0,0){\footnotesize{$\mathrmbf{Db}(\mathcal{A})$}}}
%\put(280,80){\makebox(0,0){\footnotesize{$\mathrmbf{1}$}}}
\put(145,0){\makebox(0,0){\footnotesize{${\mathring{\mathrmbf{List}}(X)}^{\mathrm{op}}$}}}
%\put(280,0){\makebox(0,0){\footnotesize{$\mathrmbf{1}$}}}
\put(218,40){\makebox(0,0){\footnotesize{${\mathring{\mathrmbf{Dom}}(\mathcal{A})}^{\mathrm{op}}$}}}
\put(175,20){\makebox(0,0){\footnotesize{$\cong$}}}
\put(80,92){\makebox(0,0){\scriptsize{$\mathring{\mathrmbfit{key}}_{\mathcal{A}}$}}}
\put(182,66){\makebox(0,0)[l]{\scriptsize{$\mathring{\mathrmbfit{dom}}_{\mathcal{A}}^{\mathrm{op}}$}}}
%\put(215,92){\makebox(0,0){\scriptsize{$\mathring{\mathrmbfit{data}}^{\mathrm{op}}$}}}
\put(134,43){\makebox(0,0)[r]{\scriptsize{$\mathring{\mathrmbfit{sign}}_{\mathcal{A}}^{\mathrm{op}}$}}}
%\put(220,-12){\makebox(0,0){\scriptsize{$\mathring{\mathrmbfit{sort}}^{\mathrm{op}}$}}}
%\put(288,40){\makebox(0,0)[l]{\scriptsize{$\mathring{\mathrmbfit{sort}}^{\mathrm{op}}$}}}
%
\put(105,80){\vector(-1,0){60}}
\put(95,0){\vector(-1,0){60}}
%\put(175,80){\vector(1,0){70}}
%\put(180,0){\vector(1,0){75}}
\put(20,65){\vector(0,-1){50}}
\put(140,65){\vector(0,-1){50}}
%\put(280,65){\vector(0,-1){50}}
\put(155,70){\vector(2,-1){36}}
%\put(190,30){\vector(-2,-1){36}}
%\put(230,50){\vector(2,1){36}}
%\qbezier(250,20)(255,20)(260,20)
%\qbezier(250,20)(250,15)(250,10)
%
\end{picture}
\end{tabular}}}
\end{tabular}}}
\end{center}
\caption{Database Mathematical Context: $\mathrmbf{Db}(\mathcal{A})$}
\label{fig:fole:db:cxt:fbr:var:A}
\end{figure}
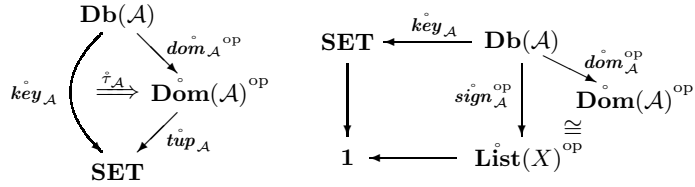
%

%
%We use the tuple passage
%$\mathring{\mathrmbf{Dom}}(\mathcal{A})^{\mathrm{op}} 
%\!\xrightarrow{\mathring{\mathrmbfit{tup}}_{\mathcal{A}}}
%\mathrmbf{SET}$
%(Def.~\ref{def:tup:pass:lax:A} in \S~\ref{sub:sub:sec:tup:pass:db}).
%

%$\mathrmbf{Dom}(\mathcal{A})\xhookrightarrow{\mathrmbfit{inc}_{\mathcal{A}}}\mathrmbf{Dom}$

%\newpage

%\newpage

%
\begin{itemize}
\item 
The schemed domain projection
$\mathring{\mathrmbfit{dom}}_{\mathcal{A}} :
\mathrmbf{Db}(\mathcal{A})
\rightarrow
\mathring{\mathrmbf{Dom}}(\mathcal{A})$ 
%= \mathrmbf{Dom}(\mathcal{A})^{\scriptscriptstyle{\Uparrow}}$
maps 
\begin{itemize}
\item 
a relational $\mathcal{A}$-database
${\langle{\mathrmbf{R},\mathrmbfit{T}}\rangle}$
in $\mathrmbf{Db}(\mathcal{A})$
to the $\mathcal{A}$-schemed domain 
$\mathring{\mathrmbfit{dom}}_{\mathcal{A}}(\mathrmbf{R},\mathrmbfit{T})
={\langle{\mathrmbf{R},\mathrmbfit{S}}\rangle}$
in $\mathring{\mathrmbf{Dom}}(\mathcal{A})$ 
%= \mathrmbf{Dom}(\mathcal{A})^{\scriptscriptstyle{\Uparrow}}$
with signed domain diagram (schema)
$\mathrmbf{R}
\xrightarrow
[\,\mathrmbfit{T}^{\mathrm{op}}\!{\circ\,}\mathrmbfit{dom}_{\mathcal{A}}]
{\;\mathrmbfit{S}\;}
\mathrmbf{Dom}(\mathcal{A})\cong\mathrmbf{List}(X)$, and
%blah-blah with signed domain diagram
%blah-blah, and 
\item 
maps 
a relational $\mathcal{A}$-database morphism
%\newline\mbox{}\hfill
{\footnotesize{$
{\langle{\mathrmbf{R}_{2},\mathrmbfit{T}_{2}}\rangle} 
\xleftarrow{{\langle{\mathrmbfit{R},\psi}\rangle}}
{\langle{\mathrmbf{R}_{1},\mathrmbfit{T}_{1}}\rangle}
$}}
%\hfill\mbox{}\newline
%in $\mathrmbf{Db}(\mathcal{A})$
to the schemed domain morphism 
$
%\mathring{\mathrmbfit{dom}}_{\mathcal{A}}(\mathcal{R}_{2})=
{\langle{\mathrmbf{R}_{2},\mathrmbfit{S}_{2}}\rangle}
\xrightarrow
[{\langle{\mathrmbfit{R},\varphi}\rangle}]
{\mathring{\mathrmbfit{dom}}_{\mathcal{A}}(\mathrmbfit{R},\psi)} 
{\langle{\mathrmbf{R}_{1},\mathrmbfit{S}_{1}}\rangle}
%=\mathring{\mathrmbfit{dom}}_{\mathcal{A}}(\mathcal{R}_{1})
$
in 
$\mathring{\mathrmbf{Dom}}(\mathcal{A})$ 
with bridge 
$\varphi = \psi^{\mathrm{op}}\!{\circ\,}\mathrmbfit{dom}_{\mathcal{A}} :
\mathrmbfit{S}_{2}\Rightarrow\mathrmbfit{R}{\circ}\mathrmbfit{S}_{1}$.
\end{itemize}
\mbox{}
\item 
The key projection 
$\mathring{\mathrmbfit{key}}_{\mathcal{A}} :
\mathrmbf{Db}(\mathcal{A})
\rightarrow
\mathrmbf{SET}$
%=\mathrmbf{Set}^{\scriptscriptstyle{\Downarrow}}$ 
%
\begin{itemize}
\item maps a relational database 
${\langle{\mathrmbf{R},\mathrmbfit{T}}\rangle}$
in $\mathrmbf{Db}(\mathcal{A})$
to the key set diagram 
$\mathring{\mathrmbfit{key}}_{\mathcal{A}}(\mathrmbf{R},\mathrmbfit{T})=
{\langle{\mathrmbf{R},\mathrmbfit{K}}\rangle}$
in $\mathrmbf{SET}$
%=\mathrmbf{Set}^{\scriptscriptstyle{\Downarrow}}$
with 
key passage
$\mathrmbf{R}^{\mathrm{op}}\!\xrightarrow[\,\mathrmbfit{T}{\,\circ\,}\mathrmbfit{key}_{\mathcal{A}}]{\,\mathrmbfit{K}\;}\mathrmbf{Set}$, 
and 
\item 
maps 
a relational $\mathcal{A}$-database morphism
%\newline\mbox{}\hfill
{\footnotesize{$
{\langle{\mathrmbf{R}_{2},\mathrmbfit{T}_{2}}\rangle} 
\xleftarrow{{\langle{\mathrmbfit{R},\psi}\rangle}}
{\langle{\mathrmbf{R}_{1},\mathrmbfit{T}_{1}}\rangle}
$}}
%\hfill\mbox{}\newline
%in $\mathrmbf{Db}(\mathcal{A})$
to the $\mathrmbf{SET}$-morphism 
{\footnotesize{$
{\langle{\mathrmbf{R}_{2},\mathrmbfit{K}_{2}}\rangle} 
\xleftarrow
[{\langle{\mathrmbfit{R},\kappa}\rangle}]
{\mathring{\mathrmbfit{key}}_{\mathcal{A}}(\mathrmbfit{R},\psi)} 
{\langle{\mathrmbf{R}_{1},\mathrmbfit{K}_{1}}\rangle}
$}}
with a bridge 
$\kappa = \psi{\,\circ\,}\mathrmbfit{key}_{\mathcal{A}}
:
\mathrmbfit{K}_{2}
\Leftarrow
%\xLeftarrow[\psi{\,\circ\,}\mathrmbfit{key}_{\mathcal{A}}]{\;\;\kappa\,}
\mathrmbfit{R}^{\mathrm{op}}{\circ}\mathrmbfit{K}_{1}
$.
\end{itemize}
%\newline
%
\end{itemize}

\newpage

%\begin{proposition}\label{prop:incl:db:A}

%
\begin{proposition}\label{prop:db:proj:A}
Using projections 
%mathematical context 
$\mathrmbf{Db}(\mathcal{A})$ can be described as follows.
%The context of relational $\mathcal{A}$-databases is the comma context 
%$\mathrmbf{Db}(\mathcal{A})
%=
%\bigl(\mathrmbf{SET}{\;\downarrow\,}\mathring{\mathrmbfit{tup}}_{\mathcal{A}}\bigr)$
%for the tuple passage mentioned above.
%(Def.~\ref{def:tup:pass:lax:A} in \S~\ref{sub:sub:sec:sch:dom:var:shape:typ:dom})
%%$\mathrmbf{List}(X)^{\mathrm{op}}\xrightarrow{\mathring{\mathrmbfit{tup}}_{\mathcal{A}}}\mathrmbf{SET}$
%$\mathring{\mathrmbf{Dom}}(\mathcal{A})^{\mathrm{op}} 
%%={\mathrmbf{Dom}(\mathcal{A})^{\!\scriptscriptstyle{\Uparrow}}}^{\mathrm{op}}
%\xrightarrow{\mathring{\mathrmbfit{tup}}_{\mathcal{A}}}
%\mathrmbf{SET}$.
%=\mathrmbf{Set}^{\!\scriptscriptstyle{\Downarrow}}$
%
\begin{itemize}
\item 
An $\mathcal{A}$-database is an object 
$\mathcal{R}={\langle{\mathrmbf{R},\mathrmbfit{S},\mathrmbfit{K},\tau}\rangle}$
in $\mathrmbf{Db}(\mathcal{A})$
% = \mathrmbf{Tbl}(\mathcal{A})^{\scriptscriptstyle{\Downarrow}}$
consisting of
(i) an object
$\mathring{\mathrmbfit{key}}_{\mathcal{A}}(\mathcal{R})=
{\langle{\mathrmbf{R},\mathrmbfit{K}}\rangle}$
in $\mathrmbf{SET}=\mathrmbf{Set}^{\scriptscriptstyle{\Downarrow}}$
with shape context $\mathrmbf{R}$
and key diagram
$\mathrmbf{R}^{\mathrm{op}}\!\xrightarrow[\,\mathrmbfit{T}{\,\circ\,}\mathrmbfit{key}_{\mathcal{A}}]{\,\mathrmbfit{K}\;}\mathrmbf{Set}$,
(ii) an object 
$\mathring{\mathrmbfit{dom}}_{\mathcal{A}}(\mathcal{R})=
{\langle{\mathrmbf{R},\mathrmbfit{S}}\rangle}$
in $\mathring{\mathrmbf{Dom}}(\mathcal{A}) 
= \mathrmbf{Dom}(\mathcal{A})^{\scriptscriptstyle{\Uparrow}}$
with shape context $\mathrmbf{R}$
and schema
$\mathrmbf{R}\xrightarrow[\,\mathrmbfit{T}^{\mathrm{op}}\!{\circ\,}\mathrmbfit{dom}_{\mathcal{A}}]{\;\mathrmbfit{S}\;}\mathrmbf{Dom}(\mathcal{A})\cong\mathrmbf{List}(X)$, and
(iii) a tuple bridge
$\mathrmbfit{K}\xRightarrow[\mathrmbfit{T}{\,\circ\,}\tau_{\mathcal{A}}]{\;\tau\;}
\mathrmbfit{S}^{\mathrm{op}}{\circ\,}\mathrmbfit{tup}_{\mathcal{A}}$.
%=\mathring{\mathrmbfit{tup}}_{\mathcal{A}}(\mathrmbfit{S})
%
Hence, a database in $\mathrmbf{Db}(\mathcal{A})$ is
a $\mathrmbf{SET}$-morphism 
\newline\mbox{}\hfill
{\footnotesize{$\mathring{\mathrmbfit{key}}_{\mathcal{A}}(\mathcal{R})=
{\langle{\mathrmbf{R},\mathrmbfit{K}}\rangle}
\xrightarrow[{\langle{\mathrmbfit{1},\tau}\rangle}]{\mathring{\tau}_{\mathcal{R}}}
{\langle{\mathrmbf{R},\mathrmbfit{S}^{\mathrm{op}}{\circ\,}\mathrmbfit{tup}_{\mathcal{A}}}\rangle}
=\mathring{\mathrmbfit{tup}}_{\mathcal{A}}(\mathrmbf{R},\mathrmbfit{S})
=\mathring{\mathrmbfit{tup}}_{\mathcal{A}}(\mathring{\mathrmbfit{dom}}_{\mathcal{A}}(\mathcal{R}))$.}}
\hfill\mbox{}\newline
%and thus an object 
%(RHS Fig.~\ref{fig:rel:db:mor:var:A})
%in the comma context
%$\bigl(\mathrmbf{SET}{\;\downarrow\,}\mathring{\mathrmbfit{tup}_{\mathcal{A}}}\bigr)$.
%\newline
%
\item 
%\item[database morphism:] 
%{\fbox{\textbf{Eliminate adjointness!}}}
An $\mathcal{A}$-database morphism is a morphism 
\newline\mbox{}\hfill
{\footnotesize{$\mathcal{R}_{2}=
{\langle{\mathrmbf{R}_{2},\mathrmbfit{S}_{2},\mathrmbfit{K}_{2},\tau_{2}}\rangle} 
\xleftarrow{{\langle{\mathrmbfit{R},\varphi,\kappa}\rangle}}
{\langle{\mathrmbf{R}_{1},\mathrmbfit{S}_{1},\mathrmbfit{K}_{1},\tau_{1}}\rangle}
=\mathcal{R}_{1}$}}
\hfill\mbox{}\newline
in $\mathrmbf{Db}(\mathcal{A})$
consisting of 
(i)
a morphism
{\footnotesize{$\mathring{\mathrmbfit{key}}_{\mathcal{A}}(\mathcal{R}_{2})=
{\langle{\mathrmbf{R}_{2},\mathrmbfit{K}_{2}}\rangle}
\xleftarrow[{\langle{\mathrmbfit{R},\kappa}\rangle}]
{\mathring{\mathrmbfit{key}}_{\mathcal{A}}(\mathrmbfit{R},\hat{\varphi},\kappa)}
{\langle{\mathrmbf{R}_{1},\mathrmbfit{K}_{1}}\rangle}
=\mathring{\mathrmbfit{key}}_{\mathcal{A}}(\mathcal{R}_{1})$}}
in $\mathrmbf{SET}$
%\hfill\mbox{}\newline
with the bridge
$\mathrmbfit{K}_{2}
\xLeftarrow{\;\kappa\,}
\hat{\mathrmbfit{R}}^{\mathrm{op}}{\circ\;}\mathrmbfit{K}_{1}$,
%with equivalent bridge pair
%$\kappa = {\langle{\acute{\kappa},\grave{\kappa}}\rangle}
%={\langle{\acute{\psi}{\,\circ\,}\mathrmbfit{key}_{\mathcal{A}},
%\grave{\psi}{\,\circ\,}\mathrmbfit{key}_{\mathcal{A}}}\rangle}$,
and
(ii)
%a morphism
an $\mathcal{A}$-schema morphism
$\mathring{\mathrmbfit{dom}}_{\mathcal{A}}(\mathcal{R}_{2})=
{\langle{\mathrmbf{R}_{2},\mathrmbfit{S}_{2}}\rangle}
\xrightarrow[{\langle{\mathrmbfit{R},\varphi}\rangle}]
{\mathring{\mathrmbfit{dom}}_{\mathcal{A}}(\mathrmbfit{R},\hat{\varphi},\kappa)} 
{\langle{\mathrmbf{R}_{1},\mathrmbfit{S}_{1}}\rangle}
=\mathring{\mathrmbfit{dom}}_{\mathcal{A}}(\mathcal{R}_{1})$
in 
$\mathring{\mathrmbf{Dom}}(\mathcal{A}) 
=\mathrmbf{Dom}(\mathcal{A})^{\!\scriptscriptstyle{\Uparrow}}$
with the bridge
$\mathrmbfit{S}_{2} \xRightarrow{\;\varphi\;} \mathrmbfit{R}{\;\circ\;}\mathrmbfit{S}_{1}$,
%$\mathrmbf{DOM}(\mathcal{A})\cong\mathrmbf{LIST}(X)=\mathrmbf{List}(X)^{\scriptscriptstyle{\Uparrow}}$
%with equivalent bridge pair
%$\hat{\varphi} = {\langle{\acute{\varphi},\grave{\varphi}}\rangle}
%={\langle{\acute{\psi}^{\mathrm{op}}\!{\circ\,}\mathrmbfit{sign}_{\mathcal{A}},
%\grave{\psi}^{\mathrm{op}}\!{\circ\,}\mathrmbfit{sign}_{\mathcal{A}}}\rangle}$,
which satisfy the condition
{\footnotesize{${\langle{\hat{\mathrmbf{R}},\kappa}\rangle} \circ 
{\langle{\mathrmbfit{R}_{2},\tau_{2}}\rangle}
=
{\langle{\mathrmbfit{R}_{1},\tau_{1}}\rangle} \circ 
%{\mathring{\mathrmbfit{tup}}_{\mathcal{A}}(\mathrmbfit{R},\varphi)}
{\langle{\mathrmbfit{R},\varphi^{\mathrm{op}}{\!\circ\,}\mathrmbfit{tup}_{\mathcal{A}}}\rangle}
$}}
in $\mathrmbf{SET}$. 
%\end{itemize}
Hence,
{\footnotesize{$
\mathrmbfit{R}_{2}{\;\circ\;}\hat{\mathrmbfit{R}}
=
\mathrmbfit{R}{\;\circ\;}\mathrmbfit{R}_{1}$}}
and
%such that 
${\footnotesize{\kappa{\,\bullet\,}\tau_{2}
=
(\mathrmbfit{R}^{\mathrm{op}}{\circ\;}\tau_{1})
{\,\bullet\,}
(\varphi^{\mathrm{op}}{\!\circ\;}\mathrmbfit{tup}_{\mathcal{A}})}}$.
%
%Hence,
%a database morphism 
%in $\mathrmbf{Db}(\mathcal{A})$
%is a morphism
%(RHS Fig.~\ref{fig:rel:db:mor:var:A}) 
%in the comma context
%$\bigl(\mathrmbf{SET}{\;\downarrow\,}\mathring{\mathrmbfit{tup}}_{\mathcal{A}}\bigr)$.
%
%%%%%%%%%%%%%%%%%%%%%%%%%%%%%%%%%%%%%%%%%%%%%%%%%%%%%%%%%%%%%%%%%%%%%%%%%%%%%%%%
%%%%%%%%%%%%%%%%%%%%%%%%%%%%%%%%%%%%%%%%%%%%%%%%%%%%%%%%%%%%%%%%%%%%%%%%%%%%%%%%
\footnote{When the shape context is $\mathrmbf{1}$,
an $\mathcal{A}$-database is an $\mathcal{A}$-table
${\langle{I,X,K,t}\rangle}\in\mathrmbf{Tbl}(\mathcal{A})$
consisting of
an $X$-sorted signature
${\langle{I,X}\rangle}\in\mathrmbf{List}(X)$,
a key set 
$K\in\mathrmbf{Set}$,
and
a tuple function
$K\xrightarrow{t}\mathrmbfit{tup}_{\mathcal{S}}(I,X)$.
When the shape passage is $\mathrmbf{1}\xrightarrow{\mathrmbfit{1}}\mathrmbf{1}$,
an $\mathcal{A}$-database morphism is an $\mathcal{A}$-table morphism
${\langle{I_{2},X_{2},K_{2},t_{2}}\rangle}
\xleftarrow{\langle{h,k}\rangle}
{\langle{I_{1},X_{1},K_{1},t_{1}}\rangle}$
in $\mathrmbf{Tbl}(\mathcal{A})$
consisting of
a $X$-sorted signature morphism ${\langle{I_{2},s_{2}}\rangle}\xrightarrow{h}{\langle{I_{1},s_{1}}\rangle}$
in $\mathrmbf{List}(X)$,
and
a key function $K_{2}\xleftarrow{k}K_{1}$,
which satisfy the condition
$k{\;\cdot\;}t_{2} = \tau_{1}{\;\cdot\;}\mathrmbfit{tup}_{\mathcal{A}}(h)$
in $\mathrmbf{Set}$.}
%%%%%%%%%%%%%%%%%%%%%%%%%%%%%%%%%%%%%%%%%%%%%%%%%%%%%%%%%%%%%%%%%%%%%%%%%%%%%%%%
%%%%%%%%%%%%%%%%%%%%%%%%%%%%%%%%%%%%%%%%%%%%%%%%%%%%%%%%%%%%%%%%%%%%%%%%%%%%%%%%
%
\end{itemize}
\end{proposition}
\begin{figure}
\begin{center}
{{\begin{tabular}{c@{\hspace{90pt}}c}
{{\begin{tabular}{c}
\setlength{\unitlength}{0.6pt}
\begin{picture}(120,120)(0,-28)
\put(0,80){\makebox(0,0){\scriptsize{$\mathrmbfit{K}_{2}$}}}
\put(130,80){\makebox(0,0){\scriptsize{$
\mathrmbfit{R}^{\mathrm{op}}{\circ\;}\mathrmbfit{K}_{1}$}}}
\put(0,0){\makebox(0,0){\scriptsize{$
\mathrmbfit{S}_{2}^{\mathrm{op}}{\circ\;}\mathrmbfit{tup}_{\mathcal{A}}$}}}
\put(135,0){\makebox(0,0){\scriptsize{$
\mathrmbfit{R}^{\mathrm{op}}{\circ\,}\mathrmbfit{S}_{1}^{\mathrm{op}}{\circ\;}\mathrmbfit{tup}_{\mathcal{A}}$}}}
\put(-7,43){\makebox(0,0)[r]{\scriptsize{$\tau_{2}$}}}
\put(140,43){\makebox(0,0)[l]{\scriptsize{$\mathrmbfit{R}^{\mathrm{op}}{\circ\;}\tau_{1}$}}}
\put(63,93){\makebox(0,0){\scriptsize{$\kappa$}}}
\put(63,-14){\makebox(0,0){\scriptsize{$\varphi^{\mathrm{op}}{\!\circ\;}\mathrmbfit{tup}_{\mathcal{A}}$}}}
\put(60,75){\makebox(0,0){\large{$\xLeftarrow{\;\;\;\;\;\;\;\;\;\;\;\;\;\;\;\;\;}$}}}
\put(60,-5){\makebox(0,0){\large{$\xLeftarrow{\;\;\;\;\;\;\;\;\;}$}}}
\put(0,40){\makebox(0,0){\large{$\bigg\Downarrow$}}}
\put(130,40){\makebox(0,0){\large{$\bigg\Downarrow$}}}
\end{picture}
\end{tabular}}}
%%%%%%%%%%%%%%%%%%%%%%%%%%%%%%%%%%%%%%%%%%%%%%%%%%%%%%%%%%%%%%%%%%%%%%%%%%%%%%%%
&
%%%%%%%%%%%%%%%%%%%%%%%%%%%%%%%%%%%%%%%%%%%%%%%%%%%%%%%%%%%%%%%%%%%%%%%%%%%%%%%%
{{\begin{tabular}{c}
\setlength{\unitlength}{0.63pt}
\begin{picture}(140,80)(0,-10)
\put(0,80){\makebox(0,0){\scriptsize{${\langle{\mathrmbf{R}_{2},\mathrmbfit{K}_{2}}\rangle}$}}}
\put(140,80){\makebox(0,0){\scriptsize{${\langle{\mathrmbf{R}_{1},\mathrmbfit{K}_{1}}\rangle}$}}}
\put(0,-10){\makebox(0,0){\scriptsize{$
\underset{\mathring{\mathrmbfit{tup}}(\mathrmbf{R}_{2},\mathrmbfit{S}_{2})}
{\underbrace{\langle{\mathrmbf{R}_{2},\mathrmbfit{S}_{2}^{\mathrm{op}}{\circ}\mathrmbfit{tup}_{\mathcal{A}}}\rangle}}$}}}
\put(140,-10){\makebox(0,0){\scriptsize{$
\underset{\mathring{\mathrmbfit{tup}}(\mathrmbf{R}_{1},\mathrmbfit{S}_{1})}
{\underbrace{\langle{\mathrmbf{R}_{1},\mathrmbfit{S}_{1}^{\mathrm{op}}{\circ}\mathrmbfit{tup}_{\mathcal{A}}}\rangle}}$}}}
\put(-7,40){\makebox(0,0)[r]{\scriptsize{${\langle{\mathrmbfit{1},\tau_{2}}\rangle}$}}}
\put(147,40){\makebox(0,0)[l]{\scriptsize{${\langle{\mathrmbfit{1},\tau_{1}}\rangle}$}}}
\put(70,90){\makebox(0,0){\scriptsize{${\langle{\mathrmbfit{R},\kappa}\rangle}$}}}
\put(70,-30){\makebox(0,0){\scriptsize{$
\underset{\mathring{\mathrmbfit{tup}}_{\mathcal{A}}(\mathrmbfit{R},\varphi)}
{\underbrace{\langle{\mathrmbfit{R},\varphi^{\mathrm{op}}{\!\circ\,}\mathrmbfit{tup}_{\mathcal{A}}}\rangle}}
%{\underbrace{\langle{\mathrmbfit{R},\mathrmbfit{Q}_{2}^{\mathrm{op}}{\circ}\mathrmbfit{tup}}\rangle}}
%{\langle{\mathrmbfit{R},\mathring{\mathrmbfit{tup}}(\acute{\zeta})}\rangle}
$}}}
\put(105,80){\vector(-1,0){70}}
\put(90,0){\vector(-1,0){40}}
\put(0,65){\vector(0,-1){50}}
\put(140,65){\vector(0,-1){50}}
\put(70,40){\makebox(0,0){\footnotesize{$\mathrmbf{SET}=\mathrmbf{Set}^{\!\scriptscriptstyle{\Downarrow}}$}}}
\end{picture}
\end{tabular}}}
\end{tabular}}}
\end{center}
\caption{Database Morphism (proj): $\mathrmbf{Db}(\mathcal{A})$}
\label{fig:rel:db:mor:var:A}
\end{figure}
%

%%%%%%%%%%%%%%%%%%%%%%%%%%%%%%%%%%%%%%%%%%%%%%%%%%%%%%%%%%%%%%%%%%%%%%%%%%%%%%%%
%%%%%%%%%%%%%%%%%%%%%%%%%%%%%%%%%%%%%%%%%%%%%%%%%%%%%%%%%%%%%%%%%%%%%%%%%%%%%%%%
\comment{ % not needed
\begin{table}
\begin{center}
{\fbox{\footnotesize{\begin{tabular}{l@{\hspace{20pt}}l}
$\mathrmbf{Db}(\mathcal{A})
= \mathrmbf{Tbl}(\mathcal{A})^{\scriptscriptstyle{\Downarrow}}$
%= \bigl({(\mbox{-})}^{\mathrm{op}}{\Downarrow\,}\mathrmbf{Tbl}\bigr)$
%\cong \bigl(\mathrmbf{Cxt}{\,\Downarrow\,}\mathrmbf{Tbl}\bigr)$
&
\textit{oplax comma context (original)}
\\
$\mathrmbf{Db}(\mathcal{A})
\subseteq
\bigl(\mathrmbf{SET}{\;\Downarrow\,}\mathring{\mathrmbfit{tup}}_{\mathcal{A}}\bigr)$
&
\textit{comma context}
%\\
%$\mathrmbf{Db}(\mathcal{A})=\int\hat{\mathrmbfit{db}}_{\mathcal{A}}\rightarrow\mathrmbf{Cxt}$
%& 
%\textit{Grothendieck construction}
%\\{\fbox{\textbf{problem}}}
%$\mathrmbf{DB} 
%= {\bigl(\mathring{\mathrmbf{Set}}{\,\Downarrow\,}\mathring{\mathrmbfit{tup}}\bigr)}^{\mathrm{op}}$
%&
%\textit{lax comma context (projections)}
\end{tabular}}}}
\end{center}
\caption{Alternate Definitions of $\mathrmbf{Db}(\mathcal{A})$}
\label{tbl:alt:defs:db:A}
\end{table}
}
%%%%%%%%%%%%%%%%%%%%%%%%%%%%%%%%%%%%%%%%%%%%%%%%%%%%%%%%%%%%%%%%%%%%%%%%%%%%%%%%
%%%%%%%%%%%%%%%%%%%%%%%%%%%%%%%%%%%%%%%%%%%%%%%%%%%%%%%%%%%%%%%%%%%%%%%%%%%%%%%%
%

%\newpage
%yada yada yada

%%%%%%%%%%%%%%%%%%%%%%%%%%%%%%%%%%%%%%%%%%%%%%%%%%%%%%%%%%%%%%%%%%%%%%%%%%%%%%%%%%%%%%%%%%
\newpage
\subsubsection{Upper Aspect: $\mathrmbf{Db}$}\label{sub:sub:sec:rel:db:typ:dom:upper}
%%%%%%%%%%%%%%%%%%%%%%%%%%%%%%%%%%%%%%%%%%%%%%%%%%%%%%%%%%%%%%%%%%%%%%%%%%%%%%%%%%%%%%%%%%

%
%%%%%%%%%%%%%%%%%%%%%%%%%%%%%%%%%%%%%%%%%%%%%%%%%%%%%%%%%%%%%%%%%%%%%%%%%%%%%%%%%%%%%%%%%%
%\mbox{}\newline
%\rule{142pt}{1pt}{\fbox{\textbf{Work Zone}}}\rule{142pt}{1pt}
%\newline\newline
%%%%%%%%%%%%%%%%%%%%%%%%%%%%%%%%%%%%%%%%%%%%%%%%%%%%%%%%%%%%%%%%%%%%%%%%%%%%%%%%%%%%%%%%%%
%

The 
%mathematical 
subcontext of 
{\ttfamily FOLE} databases 
(with constant type domains and constant type domain morphisms) 
is denoted by $\mathrmbf{Db}\subseteq\mathrmbf{DB}$.

\begin{itemize}
\item 
A relational database 
${\langle{\mathrmbf{R},\mathrmbfit{T},\mathcal{A}}\rangle}$ 
in $\mathrmbf{Db}$
consists of a type domain $\mathcal{A}$
and a relational database 
${\langle{\mathrmbf{R},\mathrmbfit{T}}\rangle}$ in $\mathrmbf{Db}(\mathcal{A})$
with a shape context $\mathrmbf{R}$ 
and  a passage 
$\mathrmbf{R}^{\mathrm{op}}\!\xrightarrow{\,\mathrmbfit{T}\;}\mathrmbf{Tbl}(\mathcal{A})$.
\item 
A relational database morphism 
${\langle{\mathrmbf{R}_{2},\mathrmbfit{T}_{2},\mathcal{A}_{2}}\rangle} 
\xleftarrow{\;{\langle{\mathrmbfit{R},\hat{\psi},f,g}\rangle}\;}
{\langle{\mathrmbf{R}_{1},\mathrmbfit{T}_{1},\mathcal{A}_{1}}\rangle}$
in $\mathrmbf{Db}$
(Fig.~\ref{fig:db:mor:Db})
%in $\mathrmbf{Db}$,
with constant type domain morphism
${\langle{f,g}\rangle} : \mathcal{A}_{2} \rightleftarrows \mathcal{A}_{1}$, 
is a 
%diagram morphism 
\texttt{FOLE} database morphism
${\langle{\mathrmbfit{R},\xi}\rangle} : 
%\mathcal{D}_{2} = 
{\langle{\mathrmbf{R}_{2},\mathrmbfit{T}_{2}}\rangle} \leftarrow 
{\langle{\mathrmbf{R}_{1},\mathrmbfit{T}_{1}}\rangle}$,
% = \mathcal{D}_{1}
%in $\mathrmbf{DB}$
%consisting of
%a 
%shape-changing 
%relation passage 
%$\mathrmbfit{R} : \mathrmbf{R}_{2} \rightarrow \mathrmbf{R}_{1}$,
%and
whose tabular interpretation bridge
{\footnotesize{${\mathrmbfit{T}_{2}
{\,\xLeftarrow{\;\,\xi\,}\,}
\mathrmbfit{R}{\,\circ\,}\mathrmbfit{T}_{1}}$}}
factors adjointly
\begin{equation}\label{eqn:db:mor:def}
%{\footnotesize
(\mathrmbfit{R}^{\mathrm{op}} \circ \mathrmbfit{T}_{1})
\circ \acute{\chi}_{{\langle{f,g}\rangle}}
\bullet
(\acute{\psi} \circ \mathrmbfit{inc}_{\mathcal{A}_{2}})
=
\xi 
=
(\grave{\psi} \circ \mathrmbfit{inc}_{\mathcal{A}_{1}})
\bullet
(\mathrmbfit{T}_{2} \circ \grave{\chi}_{{\langle{f,g}\rangle}})
\normalsize
%}
\end{equation}
through the fiber adjunction 
$\mathrmbf{Tbl}(\mathcal{A}_{2})
\xleftarrow{\acute{\mathrmbfit{tbl}}_{{\langle{f,g}\rangle}}\;\dashv\;\grave{\mathrmbfit{tbl}}_{{\langle{f,g}\rangle}}}
\mathrmbf{Tbl}(\mathcal{A}_{1})$
%
%%%%%%%%%%%%%%%%%%%%%%%%%%%%%%%%%%%%%%%%%%%%%%%%%%%%%%%%%%%%%%%%%%%%%%%%%%%%%%%%
%%%%%%%%%%%%%%%%%%%%%%%%%%%%%%%%%%%%%%%%%%%%%%%%%%%%%%%%%%%%%%%%%%%%%%%%%%%%%%%%
\comment{
Defined in 
%Upper Aspect, 
Type Domain Indexing of
the paper
``The {\ttfamily FOLE} Table''.
\cite{kent:fole:era:tbl}.
}
%%%%%%%%%%%%%%%%%%%%%%%%%%%%%%%%%%%%%%%%%%%%%%%%%%%%%%%%%%%%%%%%%%%%%%%%%%%%%%%%
%%%%%%%%%%%%%%%%%%%%%%%%%%%%%%%%%%%%%%%%%%%%%%%%%%%%%%%%%%%%%%%%%%%%%%%%%%%%%%%%
%
%(Def.\,\ref{tbl:fbr:adj}
%of
%\S\,\ref{sub:sub:sec:tables})
%
%in terms of a bridge pair
%$\hat{\psi}=(\acute{\psi},\grave{\psi})$.
using
%(Fib.~\ref{fig:db:mor:typ:dom:var:shp})
the table inclusion bridge adjointness 
%(see below)
$\hat{\chi}_{{\langle{f,g}\rangle}} 
= (\acute{\chi}_{{\langle{f,g}\rangle}},\grave{\chi}_{{\langle{f,g}\rangle}})$
%
%
%%%%%%%%%%%%%%%%%%%%%%%%%%%%%%%%%%%%%%%%%%%%%%%%%%%%%%%%%%%%%%%%%%%%%%%%%%%%%%%%
%%%%%%%%%%%%%%%%%%%%%%%%%%%%%%%%%%%%%%%%%%%%%%%%%%%%%%%%%%%%%%%%%%%%%%%%%%%%%%%%
\footnote{
The table bridge pair projects to the signed domain bridge pair
%\begin{figure}
\begin{center}
\begin{tabular}{c@{\hspace{65pt}}c}
{\scriptsize{
$\acute{\chi}_{{\langle{f,g}\rangle}}{\;\circ\;}\mathrmbfit{dom}^{\mathrm{op}}
=\mathrmbfit{dom}_{\mathcal{A}_{1}}^{\mathrm{op}}{\;\circ\;}
{\grave{\iota}_{{\langle{f,g}\rangle}}}^{\mathrm{op}}$}}
&
{\scriptsize{
$\grave{\chi}_{{\langle{f,g}\rangle}}{\;\circ\;}\mathrmbfit{dom}^{\mathrm{op}}
=\mathrmbfit{dom}_{\mathcal{A}_{2}}^{\mathrm{op}}{\;\circ\;}
{\acute{\iota}_{{\langle{f,g}\rangle}}}^{\mathrm{op}}$}}
\\
%&\\
\textbf{levo}
&
\textbf{dextro}
%\\&\\
\end{tabular}
\end{center}
}
%%%%%%%%%%%%%%%%%%%%%%%%%%%%%%%%%%%%%%%%%%%%%%%%%%%%%%%%%%%%%%%%%%%%%%%%%%%%%%%%
%%%%%%%%%%%%%%%%%%%%%%%%%%%%%%%%%%%%%%%%%%%%%%%%%%%%%%%%%%%%%%%%%%%%%%%%%%%%%%%%

%
%\begin{figure}
\begin{center}
\begin{tabular}{c}
{\footnotesize\setlength{\extrarowheight}{4pt}$\begin{array}{|@{\hspace{5pt}}l@{\hspace{15pt}}l@{\hspace{5pt}}|}
\multicolumn{1}{l}{\text{\bfseries levo}}
& 
\multicolumn{1}{l}{\text{\bfseries dextro}} 
\\ \hline
\acute{\chi}_{{\langle{f,g}\rangle}} : 
\acute{\mathrmbfit{tbl}}_{{\langle{f,g}\rangle}}\circ\mathrmbfit{inc}_{\mathcal{A}_{2}}\Leftarrow\mathrmbfit{inc}_{\mathcal{A}_{1}}
&
\grave{\chi}_{{\langle{f,g}\rangle}} : 
\mathrmbfit{inc}_{\mathcal{A}_{2}}\Leftarrow\grave{\mathrmbfit{tbl}}_{{\langle{f,g}\rangle}}\circ\mathrmbfit{inc}_{\mathcal{A}_{1}}
\\
\acute{\chi}_{{\langle{f,g}\rangle}} =
(\eta_{{\langle{f,g}\rangle}}\circ\mathrmbfit{inc}_{\mathcal{A}_{1}})
\bullet
(\acute{\mathrmbfit{tbl}}_{{\langle{f,g}\rangle}}\circ\grave{\chi}_{{\langle{f,g}\rangle}})
&
\grave{\chi}_{{\langle{f,g}\rangle}} =
(\grave{\mathrmbfit{tbl}}_{{\langle{f,g}\rangle}}\circ\acute{\chi}_{{\langle{f,g}\rangle}})
\bullet
(\varepsilon_{{\langle{f,g}\rangle}}\circ\mathrmbfit{inc}_{\mathcal{A}_{2}})
\rule[-7pt]{0pt}{10pt}
\\\hline
\end{array}$}
\end{tabular}
\end{center}
%\caption{Inclusion Bridge: Tables}
%\label{fig:incl:bridge:tbl}
%\end{figure}
%
%abstractly 
defined in 
%the paper 
%``The \texttt{FOLE} Table'' 
%\S\,\ref{sub:sec:inc:bridge}.
(Kent \cite{kent:fole:era:tbl}). 
This consists of equivalent morphisms  
%(left $\leftrightarrow$ right of Fig.~\ref{fig:db:fbr:adj}) 
\newline\mbox{}\hfill
$\underset{\in\;\;\mathrmbf{Db}(\mathcal{A}_{2})}
{\underbrace{{\langle{\mathrmbf{R}_{2},\mathrmbfit{T}_{2}}\rangle} 
\xleftarrow{\;{\langle{\mathrmbfit{R},\acute{\psi}}\rangle}\;}
\overset{
{\langle{
\mathrmbf{R}_{1},\mathrmbfit{T}_{1}{\circ\,}\acute{\mathrmbfit{tbl}}_{{\langle{f,g}\rangle}}
}\rangle}
}
{\overbrace{
\acute{\mathrmbfit{db}}_{{\langle{f,g}\rangle}}(\mathrmbf{R}_{1},\mathrmbfit{T}_{1})
}}
}}
{\;\;\;\overset{\cong}{\leftrightarrow}\;\;\;}
\underset{\in\;\;\mathrmbf{Db}(\mathcal{A}_{1})}
{\underbrace{
\overset{
{\langle{
\mathrmbf{R}_{2},\mathrmbfit{T}_{2}{\circ\,}\grave{\mathrmbfit{tbl}}_{{\langle{f,g}\rangle}}
}\rangle}
}
{\overbrace{
\grave{\mathrmbfit{db}}_{{\langle{f,g}\rangle}}(\mathrmbf{R}_{2},\mathrmbfit{T}_{2}) 
}}
\xleftarrow{\;{\langle{\mathrmbfit{R},\grave{\psi}}\rangle}\;}
{\langle{\mathrmbf{R}_{1},\mathrmbfit{T}_{1}}\rangle}
}}$
\hfill\mbox{}\newline
%	
%\end{description}
%
We normally just use the bridge restriction 
$\acute{\psi}$ 
or
$\grave{\psi}$ 
for the database morphism.
We use 
%the notation 
$\hat{\psi}$ to denote either of these equivalent bridges.
The original definition can be computed with the factorization in 
Disp.\,\ref{eqn:db:mor:def}.
\end{itemize}
%

%%%%%%%%%%%%%%%%%%%%%%%%%%%%%%%%%%%%%%%%%%%%%%%%%%%%%%%%%%%%%%%%%%%%%%%%%%%%%%%%%%%%%%%%%%
%\mbox{}\newline
%\rule{142pt}{1pt}{\fbox{\textbf{Work Zone}}}\rule{142pt}{1pt}
%\newline\newline
%%%%%%%%%%%%%%%%%%%%%%%%%%%%%%%%%%%%%%%%%%%%%%%%%%%%%%%%%%%%%%%%%%%%%%%%%%%%%%%%%%%%%%%%%%
%

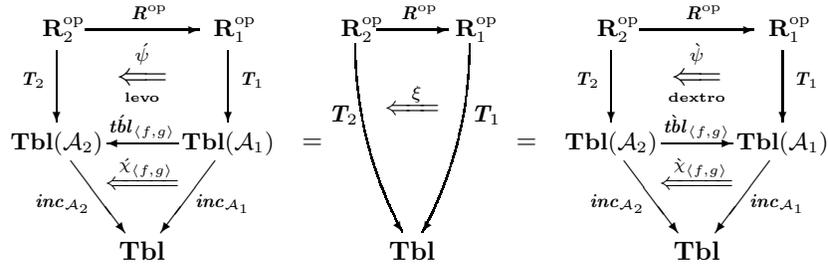
\begin{figure}
\begin{center}
{{
\begin{tabular}{@{\hspace{5pt}}c@{\hspace{45pt}}c@{\hspace{45pt}}c@{\hspace{5pt}}}
{{\begin{tabular}[b]{c}
\setlength{\unitlength}{0.54pt}
\begin{picture}(120,180)(0,0)
\put(5,160){\makebox(0,0){\footnotesize{$\mathrmbf{R}_{2}^{\mathrm{op}}$}}}
\put(125,160){\makebox(0,0){\footnotesize{$\mathrmbf{R}_{1}^{\mathrm{op}}$}}}
\put(0,80){\makebox(0,0){\footnotesize{$\mathrmbf{Tbl}(\mathcal{A}_{2})$}}}
\put(120,80){\makebox(0,0){\footnotesize{$\mathrmbf{Tbl}(\mathcal{A}_{1})$}}}
\put(60,5){\makebox(0,0){\normalsize{$\mathrmbf{Tbl}$}}}
\put(65,172){\makebox(0,0){\scriptsize{$\mathrmbfit{R}^{\mathrm{op}}$}}}
\put(-8,125){\makebox(0,0)[r]{\scriptsize{$\mathrmbfit{T}_{2}$}}}
\put(128,125){\makebox(0,0)[l]{\scriptsize{$\mathrmbfit{T}_{1}$}}}
\put(60,92){\makebox(0,0){\scriptsize{$\acute{\mathrmbfit{tbl}}_{{\langle{f,g}\rangle}}$}}}
\put(24,38){\makebox(0,0)[r]{\scriptsize{$\mathrmbfit{inc}_{\mathcal{A}_{2}}$}}}
\put(97,38){\makebox(0,0)[l]{\scriptsize{$\mathrmbfit{inc}_{\mathcal{A}_{1}}$}}}
\put(60,130){\makebox(0,0){\shortstack{\scriptsize{$\acute{\psi}$}\\\large{$\Longleftarrow$}
\\\tiny{\textbf{levo}}
}}}
\put(60,54){\makebox(0,0){\shortstack{\footnotesize{$\xLeftarrow{\acute{\chi}_{{\langle{f,g}\rangle}}}$}}}}
\put(20,160){\vector(1,0){80}}
\put(85,80){\vector(-1,0){50}}
\put(0,145){\vector(0,-1){50}}
\put(120,145){\vector(0,-1){50}}
\put(10,68){\vector(3,-4){38}}
\put(111,68){\vector(-3,-4){38}}
\end{picture}
\end{tabular}}}
%%%%%%%%%%%%%%%%%%%%%%%%%%%%%%%%%%%%%%%%%%%%%%%%%%%%%%%%%%%%
%%%%%%%%%%%%%%%%%%%%%%%%%%%%%%%%%%%%%%%%%%%%%%%%%%%%%%%%%%%%
&
{{\begin{tabular}[b]{c}
\setlength{\unitlength}{0.54pt}
\begin{picture}(80,180)(5,0)
\put(5,160){\makebox(0,0){\footnotesize{$\mathrmbf{R}_{2}^{\mathrm{op}}$}}}
\put(85,160){\makebox(0,0){\footnotesize{$\mathrmbf{R}_{1}^{\mathrm{op}}$}}}
\put(40,5){\makebox(0,0){\normalsize{$\mathrmbf{Tbl}$}}}
\put(45,172){\makebox(0,0){\scriptsize{$\mathrmbfit{R}^{\mathrm{op}}$}}}
\put(2,100){\makebox(0,0)[r]{\footnotesize{$\mathrmbfit{T}_{2}$}}}
\put(81,100){\makebox(0,0)[l]{\footnotesize{$\mathrmbfit{T}_{1}$}}}
\put(40,105){\makebox(0,0){\shortstack{\normalsize{$\xLeftarrow{\;\;\xi\;\;}$}}}}
\put(15,160){\vector(1,0){50}}
\qbezier(0,150)(0,80)(33,18)\put(33,18){\vector(2,-3){0}}
\qbezier(80,150)(80,80)(47,18)\put(47,18){\vector(-2,-3){0}}
%\put(-26,80){\makebox(0,0){\shortstack{\normalsize{$=$}}}}
%\put(113,80){\makebox(0,0){\shortstack{\normalsize{$=$}}}}
%%%%%%%%%%%%%%%%%%%%%%%%%%%%%%%%%%%%%%%%%%%%%%%%%%%%%%%%%%%%
\put(-30,80){\makebox(0,0){\shortstack{\normalsize{$=$}}}}
\put(120,80){\makebox(0,0){\shortstack{\normalsize{$=$}}}}
\end{picture}
\end{tabular}}}
&
%%%%%%%%%%%%%%%%%%%%%%%%%%%%%%%%%%%%%%%%%%%%%%%%%%%%%%%%%%%%
%%%%%%%%%%%%%%%%%%%%%%%%%%%%%%%%%%%%%%%%%%%%%%%%%%%%%%%%%%%%
{{\begin{tabular}[b]{c}
\setlength{\unitlength}{0.54pt}
\begin{picture}(120,180)(0,0)
\put(5,160){\makebox(0,0){\footnotesize{$\mathrmbf{R}_{2}^{\mathrm{op}}$}}}
\put(125,160){\makebox(0,0){\footnotesize{$\mathrmbf{R}_{1}^{\mathrm{op}}$}}}
\put(0,80){\makebox(0,0){\footnotesize{$\mathrmbf{Tbl}(\mathcal{A}_{2})$}}}
\put(120,80){\makebox(0,0){\footnotesize{$\mathrmbf{Tbl}(\mathcal{A}_{1})$}}}
\put(60,5){\makebox(0,0){\normalsize{$\mathrmbf{Tbl}$}}}
\put(65,172){\makebox(0,0){\scriptsize{$\mathrmbfit{R}^{\mathrm{op}}$}}}
\put(-8,125){\makebox(0,0)[r]{\scriptsize{$\mathrmbfit{T}_{2}$}}}
\put(128,125){\makebox(0,0)[l]{\scriptsize{$\mathrmbfit{T}_{1}$}}}
\put(60,92){\makebox(0,0){\scriptsize{$\grave{\mathrmbfit{tbl}}_{{\langle{f,g}\rangle}}$}}}
\put(24,38){\makebox(0,0)[r]{\scriptsize{$\mathrmbfit{inc}_{\mathcal{A}_{2}}$}}}
\put(97,38){\makebox(0,0)[l]{\scriptsize{$\mathrmbfit{inc}_{\mathcal{A}_{1}}$}}}
\put(60,130){\makebox(0,0){\shortstack{
\scriptsize{$\grave{\psi}$}\\\large{$\Longleftarrow$}
\\\tiny{\textbf{dextro}}
}}}
\put(60,54){\makebox(0,0){\shortstack{\footnotesize{$
\xLeftarrow{\grave{\chi}_{{\langle{f,g}\rangle}}}$}}}}
\put(20,160){\vector(1,0){80}}
\put(35,80){\vector(1,0){50}}
\put(0,145){\vector(0,-1){50}}
\put(120,145){\vector(0,-1){50}}
\put(9,68){\vector(3,-4){38}}
\put(111,68){\vector(-3,-4){38}}
\end{picture}
\end{tabular}}}
\end{tabular}}}
\end{center}
\caption{Database Morphism: $\mathrmbf{Db}$ }
\label{fig:db:mor:Db}
\end{figure}
%

%%%%%%%%%%%%%%%%%%%%%%%%%%%%%%%%%%%%%%%%%%%%%%%%%%%%%%%%%%%%%%%%%%%%%%%%%%%%%%%%%%%%%%%%%%
%\newpage
%\paragraph{{Fibered context $\mathrmbf{Db}$.}}
%%%%%%%%%%%%%%%%%%%%%%%%%%%%%%%%%%%%%%%%%%%%%%%%%%%%%%%%%%%%%%%%%%%%%%%%%%%%%%%%%%%%%%%%%%

%
\begin{proposition}\label{prop:fib:cxt:db:var:typ:dom:sh}
Hence,
the fibered context $\mathrmbf{Db}$ (with constant type domain) 
is the Grothendieck construction of
%can be developed as 
%follows.
%\begin{description}
%
%\item[ 
the indexed adjunction
\newline\mbox{}\hfill
$\mathrmbf{Db}
= \int_\mathrmit{data}:
\mathrmbf{Cls}^{\mathrm{op}}\!\xrightarrow{\,\hat{\mathrmbfit{db}}\;}\mathrmbf{Adj}$.
\hfill\mbox{}
%\newline
%
\end{proposition}
\begin{proof}
By composing on the bottom (type domain),
define the following 
database fiber adjunction
for any type domain morphisms
$\mathcal{A}_{2}\xrightleftharpoons{{\langle{f,g}\rangle}}\mathcal{A}_{1}$
%\begin{figure}
\begin{center}
{{\begin{tabular}[b]{c}
\setlength{\unitlength}{0.9pt}
\begin{picture}(160,30)(-16,-70)
\put(12,-52){\makebox(0,0){\normalsize{$
\textstyle{\mathrmbf{Db}(\mathcal{A}_{2})}$}}}
\put(120,-52){\makebox(0,0){\normalsize{$
\textstyle{\mathrmbf{Db}(\mathcal{A}_{1})}$.}}}
\put(250,-52){\makebox(0,0){\normalsize{\mbox{}\hfill\rule{5pt}{5pt}}}}
\put(70,-70){\makebox(0,0){\scriptsize{$\acute{\mathrmbfit{db}}_{{\langle{f,g}\rangle}}$}}}
\put(70,-35){\makebox(0,0){\scriptsize{$\grave{\mathrmbfit{db}}_{{\langle{f,g}\rangle}}$}}}
\put(50,-52){\makebox(0,0){\tiny{$\eta_{{\langle{f,g}\rangle}}$}}}
\put(68,-50){\makebox(0,0){\tiny{${\dashv}$}}}
\put(86,-52){\makebox(0,0){\tiny{$\varepsilon_{{\langle{f,g}\rangle}}$}}}
\put(90,-60){\vector(-1,0){50}}
\put(40,-44){\vector(1,0){50}}
\end{picture}
\end{tabular}}}
\end{center}
%\caption{Database Fiber Adjunctions: $\mathrmbf{Db}(\mathcal{A})$}
%\label{fig:db:fbr:adj}
%\end{figure}
%
%\newline\mbox{}\hfill
%{\footnotesize{$\hat{\mathrmbfit{db}}_{{\langle{f,g}\rangle}} = 
%{\langle{\acute{\mathrmbfit{db}}_{{\langle{f,g}\rangle}}
%{\!\dashv\,}
%\grave{\mathrmbfit{db}}_{{\langle{f,g}\rangle}}}\rangle} :
%\mathrmbf{Db}(\mathcal{A}_{2})
%\leftarrow
%\mathrmbf{Db}(\mathcal{A}_{1})$}}
%\hfill\mbox{}\newline
%Similar to
%$\mathrmbf{Db}(\mathrmbf{R})$.
%\newline
%\end{description}
%\hfill\rule{5pt}{5pt}
%
\end{proof}
\begin{proposition}\label{prop:db:typ:dom:lim:colim}
The fibered context (Grothendieck construction) 
$\mathrmbf{Db} = \int_\mathrmit{data}:
\mathrmbf{Cls}^{\mathrm{op}}\!\xrightarrow{\,\hat{\mathrmbfit{db}}\;}\mathrmbf{Adj}$
is complete and cocomplete
and the projection 
$\mathrmbf{Db}\rightarrow\mathrmbf{Cls}$
is continuous and cocontinuous.
\end{proposition}
\begin{proof}
By Fact\,\ref{fact:groth:adj:lim:colim}
of \S\,\ref{append:grothen:construct},
since
the indexing context $\mathrmbf{Cls}$ is complete and cocomplete
(\cite{barwise:seligman:97}),
%the indexing context $\mathrmbf{Cxt}$ is complete and cocomplete,
and
the fiber context $\mathrmbf{Db}(\mathcal{A})$ is complete and cocomplete
for each type domain $\mathcal{A}$
by Prop.\,\ref{prop:db:A:lim:colim} above. 
\hfill\rule{5pt}{5pt}
\end{proof}
%

%%%%%%%%%%%%%%%%%%%%%%%%%%%%%%%%%%%%%%%%%%%%%%%%%%
%\newpage
%\paragraph{Inclusion.}
%%%%%%%%%%%%%%%%%%%%%%%%%%%%%%%%%%%%%%%%%%%%%%%%%%
%

%(Horizontal) composition of relational database morphisms is defined component-wise.
%$\mathrmbf{Db}$
%denote the context of relational databases (with constant type domain), which 
%is a subcontext of the full context of relational databases
%$\mathrmbf{Db}{\;\subsetneq\;}\mathrmbf{DB}$.

%
\begin{proposition}\label{prop:incl:db:var:typ:dom}
$\mathrmbf{Db}$ is a subcontext of the context of databases
$\mathrmbf{Db}{\;\subsetneq\;}\mathrmbf{DB}$.
\end{proposition}
\begin{proof}
There is a (vertical) composition passage 
%\[\mbox
{\footnotesize$
\mathrmbf{Db}\xhookrightarrow{\mathring{\mathrmbfit{inc}}}\mathrmbf{DB}
$,\normalsize}
%\]
%
which maps a database
${\langle{\mathrmbf{R},\mathrmbfit{T},\mathcal{A}}\rangle}$
to the composite passage
$\mathring{\mathrmbfit{inc}}(\mathrmbf{R},\mathrmbfit{T},\mathcal{A})
= {\langle{\mathrmbf{R},\mathrmbfit{T}{\,\circ\,}\mathrmbfit{inc}_{\mathcal{A}}}\rangle}$
and maps a database morphism
${\langle{\mathrmbf{R}_{2},\mathrmbfit{T}_{2},\mathcal{A}_{2}}\rangle}
\xleftarrow{{\langle{\mathrmbfit{R},\hat{\psi},f,g}\rangle}}
{\langle{\mathrmbf{R}_{1},\mathrmbfit{T}_{1},\mathcal{A}_{1}}\rangle}$
to the composite bridge 
%\[\mbox{\footnotesize
$\mathring{\mathrmbfit{inc}}(\mathrmbf{R}_{2},\mathrmbfit{T}_{2},\mathcal{A}_{2})
\xLeftarrow[{\langle{\mathrmbfit{R},\;\xi\;=\;\hat{\psi}{\;\circ\;}\hat{\chi}_{{\langle{f,g}\rangle}}}\rangle}]
{\;\,\mathring{\mathrmbfit{inc}}(\mathrmbfit{R},\hat{\psi},f,g)\;}
\mathring{\mathrmbfit{inc}}(\mathrmbf{R}_{1},\mathrmbfit{T}_{1},\mathcal{A}_{1})$.
%\normalsize}\]
%in Fig.~\ref{inc:typ:dom}
%
\hfill\rule{5pt}{5pt}
\end{proof}
%

%
%%%%%%%%%%%%%%%%%%%%%%%%%%%%%%%%%%%%%%%%%%%%%%%%%%%%%%%%%%%%%%%%%%%%%%%%%%%%%%%%%%%%%%%%%%
%\mbox{}\newpage
%\rule{142pt}{1pt}{\fbox{\textbf{Work Zone}}}\rule{142pt}{1pt}
%\newline\newline
%%%%%%%%%%%%%%%%%%%%%%%%%%%%%%%%%%%%%%%%%%%%%%%%%%%%%%%%%%%%%%%%%%%%%%%%%%%%%%%%%%%%%%%%%%
%

%%%%%%%%%%%%%%%%%%%%%%%%%%%%%%%%%%%%%%%%%%%%%%%%%%%%%%%%%%%%%%%%%%%%%%%%%%%%%%%%%%%%%%%%%%
%\newpage
%\subsubsection{Projections in $\mathrmbf{Db}$.}
%\label{sub:sub:sec:rel:db:typ:dom:proj}
%%%%%%%%%%%%%%%%%%%%%%%%%%%%%%%%%%%%%%%%%%%%%%%%%%%%%%%%%%%%%%%%%%%%%%%%%%%%%%%%%%%%%%%%%%

%%%%%%%%%%%%%%%%%%%%%%%%%%%%%%%%%%%%%%%%%%%%%%%%%%
%\newpage
\paragraph{Projections.} 
%%%%%%%%%%%%%%%%%%%%%%%%%%%%%%%%%%%%%%%%%%%%%%%%%%

Composition with table projection passages define database projection passages.
These projections offer an alternate representation,
defining the three primary components of 
databases and their morphisms in $\mathrmbf{Db}$: 
diagram shapes, schemed domains and key diagrams.
%
%We use the tuple passage
%$\mathring{\mathrmbf{Dom}}^{\mathrm{op}}
%\!\xrightarrow{\mathring{\mathrmbfit{tup}}\:}
%\mathrmbf{SET}=\mathrmbf{Set}^{\!\scriptscriptstyle{\Downarrow}}$
%(Def.~\ref{def:tup:pass:lax:db} 
%in
%\S~\ref{sub:sub:sec:tup:pass:db}).
%\newpage
%
%Projections offer an alternate representation,
%defining the three primary components of 
%databases and database morphisms: 
%diagram shapes, schemed domains and key diagrams.
%Diagram shapes are direct projections,
%whereas schemed domains and key diagrams
%are indirect,
%coming from composition with table projection passages 
%(Expo.~\ref{expo:intro:tbl} in \S~\ref{sub:sub:sec:tables}). 
%
The database projections are described in 
Fig.\,\ref{fig:fole:db:cxt:typ:dom:var}
and are defined as follows.
\begin{itemize}
\item 
The schemed domain  projection 
$\mathring{\mathrmbfit{dom}} 
= {(\mbox{-})}^{\mathrm{op}} \circ \mathrmbfit{dom} 
: \mathrmbf{Db}^{\mathrm{op}} \rightarrow \mathring{\mathrmbf{Dom}}$
%\begin{itemize}
%\item 
maps
a relational database 
$\mathcal{R} = {\langle{\mathrmbf{R},\mathrmbfit{T},\mathcal{A}}\rangle}$ 
to the schemed domain
$\mathring{\mathrmbfit{dom}}(\mathcal{R})
={\langle{\mathrmbf{R},\mathrmbfit{S},\mathcal{A}}\rangle}$
with schema
$\mathrmbf{R}
\xrightarrow[\mathrmbfit{T}^{\mathrm{op}}{\!\circ}\mathrmbfit{dom}_{\mathcal{A}}]
{\;\mathrmbfit{S}\;}\mathrmbf{Dom}(\mathcal{A})\cong\mathrmbf{List}(X)$,
%\item 
and maps 
a relational database morphism 
${\langle{\mathrmbfit{R},\hat{\psi},f,g}\rangle}
: \mathcal{R}_{2}
%{\langle{\mathrmbf{R}_{2},\mathrmbfit{T}_{2},\mathcal{A}_{2}}\rangle} 
\leftarrow
%{\langle{\mathrmbf{R}_{1},\mathrmbfit{T}_{1},\mathcal{A}_{1}}\rangle}
\mathcal{R}_{1}$
%\newline...........................................................\newline
to the schema morphism 
{\footnotesize{$
\mathring{\mathrmbfit{dom}}(\mathrmbfit{R},\hat{\psi},f,g)
 = {\langle{\mathrmbfit{R},\hat{\varphi},f,g}\rangle}
:
\mathring{\mathrmbfit{dom}}(\mathcal{R}_{2})
%={\langle{\mathrmbf{R}_{2},\mathrmbfit{S}_{2},\mathcal{A}_{2}}\rangle}
\rightarrow
%{\;{\langle{\mathrmbfit{R},\hat{\varphi},f,g}\rangle}\;}
%{\langle{\mathrmbf{R}_{1},\mathrmbfit{S}_{1},\mathcal{A}_{1}}\rangle}=
\mathring{\mathrmbfit{dom}}(\mathcal{R}_{1})$}}
with equivalent bridge pair
$\hat{\varphi} = {\langle{\acute{\varphi},\grave{\varphi}}\rangle}
={\langle{\acute{\psi}^{\mathrm{op}}\!{\circ\,}\mathrmbfit{sign}_{\mathcal{A}_{2}},
\grave{\psi}^{\mathrm{op}}\!{\circ\,}\mathrmbfit{sign}_{\mathcal{A}_{1}}}\rangle}$.
%\end{itemize}
\newline
\item 
The key projection 
$\mathring{\mathrmbfit{key}} 
= {(\mbox{-})} \circ \mathrmbfit{key} 
: \mathrmbf{Db} \rightarrow \mathrmbf{SET}$ 
%\newline...........................................................\newline
%\begin{itemize}
%\item 
maps
a relational database $\mathcal{R}={\langle{\mathrmbf{R},\mathrmbfit{T},\mathcal{A}}\rangle}$ 
to the key set diagram
$\mathring{\mathrmbfit{key}}(\mathcal{R})
={\langle{\mathrmbf{R},\mathrmbfit{K}}\rangle}$ 
in
$\mathrmbf{SET}$
with the key passage
$\mathrmbf{R}^{\mathrm{op}}\!\xrightarrow[\mathrmbfit{T} \circ \mathrmbfit{key}_{\mathcal{A}}]{\mathrmbfit{K}\;}\mathrmbf{Set}$,
%\item 
and maps 
a relational database morphism 
${\langle{\mathrmbfit{R},\hat{\psi},f,g}\rangle}
: \mathcal{R}_{2}
%{\langle{\mathrmbf{R}_{2},\mathrmbfit{T}_{2},\mathcal{A}_{2}}\rangle} 
\leftarrow
%{\langle{\mathrmbf{R}_{1},\mathrmbfit{T}_{1},\mathcal{A}_{1}}\rangle}
\mathcal{R}_{1}$
to the $\mathrmbf{SET}$-morphism 
%\newline\mbox{}\hfill
{\footnotesize{$
\mathring{\mathrmbfit{key}}(\mathrmbfit{R},\hat{\psi},f,g)
=
{\langle{\mathrmbfit{R},\kappa}\rangle} :
{\langle{\mathrmbf{R}_{2},\mathrmbfit{K}_{2}}\rangle}
\leftarrow
{\langle{\mathrmbf{R}_{1},\mathrmbfit{K}_{1}}\rangle}$}}
%\hfill\mbox{}\newline
with key bridge
%$\mathrmbfit{K}_{2}
%\xLeftarrow[\,\xi{\,\circ\,}\mathrmbfit{key}]{\kappa}
%\mathrmbfit{R} \circ \mathrmbfit{K}_{1}$.
$\mathrmbfit{K}_{2} 
\xLeftarrow{\,\kappa\;} 
\mathrmbfit{R}^{\mathrm{op}}{\circ\;}\mathrmbfit{K}_{1}$.
%between (key) set diagrams.
%\end{itemize}
%
\end{itemize}
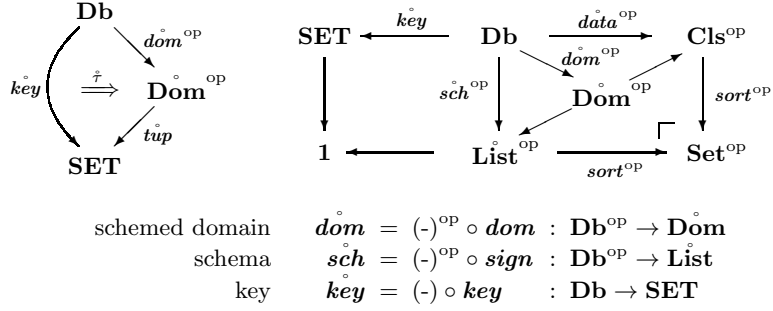
\begin{figure}
\begin{center}
{{\begin{tabular}{c@{\hspace{30pt}}c}
{{\begin{tabular}{c}
\setlength{\unitlength}{0.5pt}
\begin{picture}(100,120)(-20,0)
\put(0,118){\makebox(0,0){\footnotesize{$\mathrmbf{Db}$}}}
\put(70,60){\makebox(0,0){\footnotesize{$\mathring{\mathrmbf{Dom}}^{\mathrm{op}}$}}}
\put(0,0){\makebox(0,0){\footnotesize{$\mathrmbf{SET}$}}}
\put(-40,60){\makebox(0,0)[r]{\scriptsize{$\mathring{\mathrmbfit{key}}$}}}
\put(36,97){\makebox(0,0)[l]{\scriptsize{$\mathring{\mathrmbfit{dom}}^{\mathrm{op}}$}}}
\put(36,24){\makebox(0,0)[l]{\scriptsize{$\mathring{\mathrmbfit{tup}}$}}}
\put(4,57){\makebox(0,0){{$\xRightarrow{\;\mathring{\tau}\,}$}}}
\put(15,105){\vector(1,-1){30}}
\put(45,45){\vector(-1,-1){30}}
\qbezier(-12,105)(-60,60)(-12,15)\put(-12,15){\vector(1,-1){0}}
\end{picture}
\end{tabular}}}
%%%%%%%%%%%%%%%%%%%%%%%%%%%%%%%%%%%%%%%%
&
%%%%%%%%%%%%%%%%%%%%%%%%%%%%%%%%%%%%%%%%
{{\begin{tabular}{c}
\setlength{\unitlength}{0.55pt}
\begin{picture}(320,160)(0,-35)
\put(20,80){\makebox(0,0){\footnotesize{$\mathrmbf{SET}$}}}
\put(20,0){\makebox(0,0){\footnotesize{$\mathrmbf{1}$}}}
\put(140,80){\makebox(0,0){\footnotesize{$\mathrmbf{Db}$}}}
\put(290,80){\makebox(0,0){\footnotesize{$\mathrmbf{Cls}^{\mathrm{op}}$}}}
\put(145,0){\makebox(0,0){\footnotesize{$\mathring{\mathrmbf{List}}^{\mathrm{op}}$}}}
\put(290,0){\makebox(0,0){\footnotesize{$\mathrmbf{Set}^{\mathrm{op}}$}}}
\put(218,40){\makebox(0,0){\footnotesize{$\mathring{\mathrmbf{Dom}}^{\mathrm{op}}$}}}
\put(80,92){\makebox(0,0){\scriptsize{$\mathring{\mathrmbfit{key}}$}}}
\put(182,66){\makebox(0,0)[l]{\scriptsize{$\mathring{\mathrmbfit{dom}}^{\mathrm{op}}$}}}
\put(215,92){\makebox(0,0){\scriptsize{$\mathring{\mathrmbfit{data}}^{\mathrm{op}}$}}}
\put(134,43){\makebox(0,0)[r]{\scriptsize{$\mathring{\mathrmbfit{sch}}^{\mathrm{op}}$}}}
\put(220,-12){\makebox(0,0){\scriptsize{$\mathrmbfit{sort}^{\mathrm{op}}$}}}
\put(288,40){\makebox(0,0)[l]{\scriptsize{$\mathrmbfit{sort}^{\mathrm{op}}$}}}
\put(105,80){\vector(-1,0){60}}
\put(95,0){\vector(-1,0){60}}
\put(175,80){\vector(1,0){70}}
\put(180,0){\vector(1,0){75}}
\put(20,65){\vector(0,-1){50}}
\put(140,65){\vector(0,-1){50}}
\put(280,65){\vector(0,-1){50}}
\put(155,70){\vector(2,-1){36}}
\put(190,30){\vector(-2,-1){36}}
\put(230,50){\vector(2,1){36}}
\qbezier(250,20)(255,20)(260,20)
\qbezier(250,20)(250,15)(250,10)
\end{picture}
\end{tabular}}}
%%%%%%%%%%%%%%%%%%%%%%%%%%%%%%%%%%%%%%%%%%%%%%%%%%%%%%%%%%%%%%%%%%%%%%%%%%%%%%%%%%%%%%%%%%%%%%%%%%%%
\\
%%%%%%%%%%%%%%%%%%%%%%%%%%%%%%%%%%%%%%%%%%%%%%%%%%%%%%%%%%%%%%%%%%%%%%%%%%%%%%%%%%%%%%%%%%%%%%%%%%%%
\multicolumn{2}{c}{
{\footnotesize{$\begin{array}{r@{\hspace{16pt}}r@{\hspace{5pt}=\hspace{5pt}}l@{\hspace{5pt}:\hspace{5pt}}l}
\text{schemed domain}
&
\mathring{\mathrmbfit{dom}} 
& {(\mbox{-})}^{\mathrm{op}} \circ \mathrmbfit{dom} 
& \mathrmbf{Db}^{\mathrm{op}} \rightarrow \mathring{\mathrmbf{Dom}}
\\
\text{schema}
&
\mathring{\mathrmbfit{sch}} 
& {(\mbox{-})}^{\mathrm{op}} \circ \mathrmbfit{sign} 
& \mathrmbf{Db}^{\mathrm{op}} \rightarrow \mathring{\mathrmbf{List}}
\\
\text{key}
&
\mathring{\mathrmbfit{key}} 
& {(\mbox{-})} \circ \mathrmbfit{key} 
& \mathrmbf{Db} \rightarrow \mathrmbf{SET}
%=\mathrmbf{Set}^{\!\scriptscriptstyle{\Downarrow}}
\end{array}$}}}
\end{tabular}}}
\end{center}
\caption{Database Mathematical Context: $\mathrmbf{Db}$}
\label{fig:fole:db:cxt:typ:dom:var}
\end{figure}

\newpage

%To here 8-26-2022

%
\begin{proposition}\label{prop:db:fixed:proj}
Using projections, 
%mathematical context 
$\mathrmbf{Db}$ can be described as follows.
\begin{itemize}
\item 
A database 
(with constant type domain) 
$\mathcal{R} =
{\langle{\mathrmbf{R},\mathrmbfit{S},\mathcal{A},\mathrmbfit{K},\tau}\rangle}$
%is
%an object
%in $\mathrmbf{Db}$
%(with constant type domain)
consists of
an object
$\mathring{\mathrmbfit{key}}(\mathcal{R}) = {\langle{\mathrmbf{R},\mathrmbfit{K}}\rangle}$
in $\mathrmbf{SET}$
%=\mathrmbf{Set}^{\!\scriptscriptstyle{\Downarrow}}$
%%%%%%%%%%%%%%%%%%%%%%%%%%%%%%%%%%%%%%%%%%%%%%%%%%
%%%%%%%%%%%%%%%%%%%%%%%%%%%%%%%%%%%%%%%%%%%%%%%%%%
\footnote{A shape context $\mathrmbf{R}$
and a key diagram
$\mathrmbf{R}^{\mathrm{op}}\!\xrightarrow[\mathrmbfit{T}{\circ}\mathrmbfit{key}_{\mathcal{A}}]{\mathrmbfit{K}\;}\mathrmbf{Set}$.}, 
%%%%%%%%%%%%%%%%%%%%%%%%%%%%%%%%%%%%%%%%%%%%%%%%%%
%%%%%%%%%%%%%%%%%%%%%%%%%%%%%%%%%%%%%%%%%%%%%%%%%%
an object 
$\mathring{\mathrmbfit{dom}}(\mathcal{R}) = {\langle{\mathrmbf{R},\mathrmbfit{S},\mathcal{A}}\rangle}$
in $\mathring{\mathrmbf{Dom}}$
%
%%%%%%%%%%%%%%%%%%%%%%%%%%%%%%%%%%%%%%%%%%%%%%%%%%%%%%%%%%%%%%%%%%%%%%%%%%%%%%%%
%%%%%%%%%%%%%%%%%%%%%%%%%%%%%%%%%%%%%%%%%%%%%%%%%%%%%%%%%%%%%%%%%%%%%%%%%%%%%%%%
\footnote{A shape context $\mathrmbf{R}$
and a schema
$\mathrmbf{R}
\xrightarrow[\mathrmbfit{T}^{\mathrm{op}}{\!\circ}\mathrmbfit{dom}_{\mathcal{A}}]
{\;\mathrmbfit{S}\;}\mathrmbf{Dom}(\mathcal{A})\cong\mathrmbf{List}(X)$.}
%%%%%%%%%%%%%%%%%%%%%%%%%%%%%%%%%%%%%%%%%%%%%%%%%%%%%%%%%%%%%%%%%%%%%%%%%%%%%%%%
%%%%%%%%%%%%%%%%%%%%%%%%%%%%%%%%%%%%%%%%%%%%%%%%%%%%%%%%%%%%%%%%%%%%%%%%%%%%%%%%
%
and a tuple bridge
$\mathrmbfit{K}\xRightarrow[\mathrmbfit{T}{\circ}\tau_{\mathcal{A}}]{\;\tau\;}\mathrmbfit{S}^{\mathrm{op}}\!{\,\circ\,}\mathrmbfit{tup}_{\mathcal{A}}
=\mathring{\mathrmbfit{tup}}_{\mathcal{A}}(\mathrmbfit{S})$.
Hence, a database in $\mathrmbf{Db}$
%can be view as a
is
a $\mathrmbf{SET}$-morphism 
\newline\mbox{}\hfill
{\footnotesize{$\mathring{\mathrmbfit{key}}(\mathcal{R})=
{\langle{\mathrmbf{R},\mathrmbfit{K}}\rangle}
\xrightarrow
%[
{\langle{\mathrmbfit{1},\tau}\rangle}
%]{\mathring{\tau}_{\mathcal{R}}}
{\langle{\mathrmbf{R},\mathring{\mathrmbfit{tup}}_{\mathcal{A}}(\mathrmbfit{S})}\rangle}
=\mathring{\mathrmbfit{tup}}(\mathrmbf{R},\mathrmbfit{S},\mathcal{A})
=\mathring{\mathrmbfit{tup}}(\mathring{\mathrmbfit{dom}}(\mathcal{R}))$.}}
\hfill\mbox{}\newline
%and thus an object 
%(Fig.~\ref{fig:rel:db:mor:var:typ:dom})
%in the comma context
%$\bigl(\mathrmbf{SET}{\;\Downarrow\,}\mathring{\mathrmbfit{tup}}\bigr)$.
\newline
\item 
A database morphism (with constant type domain morphism) 
%is a morphism 
\newline\mbox{}\hfill
{\footnotesize{$\mathcal{R}_{2}=
{\langle{\mathrmbf{R}_{2},\mathrmbfit{S}_{2},\mathcal{A}_{2},\mathrmbfit{K}_{2},\tau_{2}}\rangle}
\xleftarrow{{\langle{\mathrmbfit{R},\hat{\varphi},f,g,\kappa}\rangle}}
{\langle{\mathrmbf{R}_{1},\mathrmbfit{S}_{1},\mathcal{A}_{1},\mathrmbfit{K}_{1},\tau_{1}}\rangle}
=\mathcal{R}_{1}$}}
\hfill\mbox{}\newline
%in $\mathrmbf{Db}$
%(with constant type domain morphism)
consists of 
a morphism
{\footnotesize{$\mathring{\mathrmbfit{key}}(\mathcal{R}_{2})=
{\langle{\mathrmbf{R}_{2},\mathrmbfit{K}_{2}}\rangle}
\xleftarrow{{\langle{\mathrmbfit{R},\kappa}\rangle}}
{\langle{\mathrmbf{R}_{1},\mathrmbfit{K}_{1}}\rangle}
=\mathring{\mathrmbfit{key}}(\mathcal{R}_{1})$}}
in $\mathrmbf{SET}$
%=\mathrmbf{Set}^{\!\scriptscriptstyle{\Downarrow}}$
%%%%%%%%%%%%%%%%%%%%%%%%%%%%%%%%%%%%%%%%%%%%%%%%%%
%%%%%%%%%%%%%%%%%%%%%%%%%%%%%%%%%%%%%%%%%%%%%%%%%%
%\footnote{A shape-changing passage 
%$\mathrmbf{R}_{2}\xrightarrow{\;\mathrmbfit{R}\;\,}\mathrmbf{R}_{1}$
%and 
%an 
with bridge 
%$\kappa = {\langle{\acute{\kappa},\grave{\kappa}}\rangle}
%= \xi{\,\circ\,}\mathrmbfit{key}$,
$\mathrmbfit{K}_{2} 
\xLeftarrow{\,\kappa\;} 
\mathrmbfit{R}^{\mathrm{op}}{\circ\;}\mathrmbfit{K}_{1}
%%= {\langle{\acute{\kappa},\grave{\kappa}}\rangle}
%={\langle{\psi{\,\circ\,}\mathrmbfit{key}_{\mathcal{A}_{2}},
%\psi{\,\circ\,}\mathrmbfit{key}_{\mathcal{A}_{1}}}\rangle}
$,
%.}
%%%%%%%%%%%%%%%%%%%%%%%%%%%%%%%%%%%%%%%%%%%%%%%%%%
%%%%%%%%%%%%%%%%%%%%%%%%%%%%%%%%%%%%%%%%%%%%%%%%%%
and
a morphism 
{\footnotesize{$\mathring{\mathrmbfit{dom}}(\mathcal{R}_{2})=
{\langle{\mathrmbf{R}_{2},\mathrmbfit{S}_{2},\mathcal{A}_{2}}\rangle}
\xrightarrow{\;{\langle{\mathrmbfit{R},\hat{\varphi},f,g}\rangle}\;}
{\langle{\mathrmbf{R}_{1},\mathrmbfit{S}_{1},\mathcal{A}_{1}}\rangle}
=\mathring{\mathrmbfit{dom}}(\mathcal{R}_{1})$}}
in 
%the fibered 
%context of schemas 
%(with constant infomorphism)
$\mathring{\mathrmbf{Dom}}$
%\xrightarrow{\mathrmbfit{shape}}\mathrmbf{Cxt}$
%%%%%%%%%%%%%%%%%%%%%%%%%%%%%%%%%%%%%%%%%%%%%%%%%%
%%%%%%%%%%%%%%%%%%%%%%%%%%%%%%%%%%%%%%%%%%%%%%%%%%
%\footnote{A schema morphism 
%$\mathrmbfit{S}_{2}\xrightarrow{\;{\langle{\hat{\varphi},f,g}\rangle}\;}
%\mathring{\mathrmbfit{dom}}_{\mathrmbfit{R}}(\mathrmbfit{S}_{1})$
%in the fiber context $\mathring{\mathrmbf{Dom}}(\mathrmbf{R}_{2})$
with equivalent bridge pair
$\hat{\varphi} = {\langle{\acute{\varphi},\grave{\varphi}}\rangle}
={\langle{\acute{\psi}^{\mathrm{op}}\!{\circ\,}\mathrmbfit{sign}_{\mathcal{A}_{2}},
\grave{\psi}^{\mathrm{op}}\!{\circ\,}\mathrmbfit{sign}_{\mathcal{A}_{1}}}\rangle}$,
%$\hat{\varphi}{\,\circ\,}\hat{\iota}_{{\langle{f,g}\rangle}} 
%= \hat{\zeta}
%%{\langle{\acute{\zeta},\grave{\zeta}}\rangle}
%= \xi^{\mathrm{op}}\!{\circ\,}\mathrmbfit{dom}$,
%%.},
%%%%%%%%%%%%%%%%%%%%%%%%%%%%%%%%%%%%%%%%%%%%%%%%%%
%%%%%%%%%%%%%%%%%%%%%%%%%%%%%%%%%%%%%%%%%%%%%%%%%% 
%
which satisfy the condition
%
%%%%%%%%%%%%%%%%%%%%%%%%%%%%%%%%%%%%%%%%%%%%%%%%%%
%%%%%%%%%%%%%%%%%%%%%%%%%%%%%%%%%%%%%%%%%%%%%%%%%%
\comment{It is strict or trim when the underlying schema morphism is strict or trim ($\acute{\varphi} = 1$).}
%%%%%%%%%%%%%%%%%%%%%%%%%%%%%%%%%%%%%%%%%%%%%%%%%%
%%%%%%%%%%%%%%%%%%%%%%%%%%%%%%%%%%%%%%%%%%%%%%%%%%
%
\begin{equation}\label{expo:db:mor:typ:dom:var:shape}
{{\begin{picture}(120,10)(0,-4)
\put(50,0){\makebox(0,0){\footnotesize{$
{\langle{\mathrmbfit{R},\kappa}\rangle}{\;\circ\;}{\langle{\mathrmbfit{1},\tau_{2}}\rangle}=
{\langle{\mathrmbfit{1},\tau_{1}}\rangle}{\;\circ\;}\mathring{\mathrmbfit{tup}}(\mathrmbfit{R},\hat{\varphi},f,g)
$}}}
\end{picture}}}
\end{equation}
%
%Hence,
%a database morphism 
%in $\mathrmbf{Db}$
%is a morphism
%(Fig.~\ref{fig:rel:db:mor:var:typ:dom}) 
%in the comma context
%$\bigl(\mathrmbf{SET}{\;\Downarrow\,}\mathring{\mathrmbfit{tup}}\bigr)$.
%
The condition (Expo.~\ref{expo:db:mor:typ:dom:var:shape}) resolves into the adjoint conditions
\begin{center}
\begin{tabular}{@{\hspace{-25pt}}c@{\hspace{-10pt}}c@{\hspace{20pt}}c}
{{\scriptsize{\setlength{\extrarowheight}{2pt}
$\begin{array}[t]{r@{\hspace{5pt}=\hspace{5pt}}l}
\kappa
{\;\bullet\;}
\tau_{2}
&
(\mathrmbfit{R}^{\mathrm{op}}{\circ\;}\tau_{1})
{\;\bullet\;}
(\mathrmbfit{R}^{\mathrm{op}}{\circ\;}\mathrmbfit{S}_{1}^{\mathrm{op}}{\!\circ\;}\acute{\tau}_{{\langle{f,g}\rangle}})
{\;\bullet\;}
(\acute{\varphi}^{\mathrm{op}}{\!\circ\;}\mathrmbfit{tup}_{\mathcal{A}_{2}})
\\
& 
(\mathrmbfit{R}^{\mathrm{op}}{\circ\;}\tau_{1})
{\;\bullet\;}
\underset{\mathring{\mathrmbfit{tup}}_{{\langle{f,g}\rangle}}(\hat{\varphi})}
{\underbrace{(\hat{\varphi}^{\mathrm{op}}{\;\circ\;}\hat{\tau}_{{\langle{f,g}\rangle}})}}
\end{array}$}}}
%%%%%%%%%%%%%%%%%%%%%%%%%%%%%%%%%%%%%%%%%%%%%%%%%%%%%%%%%%%%%%%%%%%%%%%%%%%%%%%%
%%%%%%%%%%%%%%%%%%%%%%%%%%%%%%%%%%%%%%%%%%%%%%%%%%%%%%%%%%%%%%%%%%%%%%%%%%%%%%%%
&&
%%%%%%%%%%%%%%%%%%%%%%%%%%%%%%%%%%%%%%%%%%%%%%%%%%%%%%%%%%%%%%%%%%%%%%%%%%%%%%%%
%%%%%%%%%%%%%%%%%%%%%%%%%%%%%%%%%%%%%%%%%%%%%%%%%%%%%%%%%%%%%%%%%%%%%%%%%%%%%%%%
{{\scriptsize{\setlength{\extrarowheight}{2pt}
$\begin{array}[t]{r@{\hspace{5pt}=\hspace{5pt}}l}
\kappa
{\;\bullet\;}
\tau_{2}
& 
(\mathrmbfit{R}^{\mathrm{op}}{\circ\;}\tau_{1})
{\;\bullet\;}
(\grave{\varphi}^{\mathrm{op}}{\!\circ\;}
\mathrmbfit{tup}_{\mathcal{A}_{1}})
{\;\bullet\;}
(\mathrmbfit{S}_{2}^{\mathrm{op}}{\!\circ\;}\grave{\tau}_{{\langle{f,g}\rangle}})
\\
& 
(\mathrmbfit{R}^{\mathrm{op}}{\circ\;}\tau_{1})
{\;\bullet\;}
\underset{\mathring{\mathrmbfit{tup}}_{{\langle{f,g}\rangle}}(\hat{\varphi})}
{\underbrace{(\hat{\varphi}^{\mathrm{op}}{\!\circ\;}\hat{\tau}_{{\langle{f,g}\rangle}})}}
\end{array}$}}}
\\&&\\
\textbf{levo}
&&
\textbf{dextro}
\\
%&&\\
\end{tabular}
\end{center}
%
%%%%%%%%%%%%%%%%%%%%%%%%%%%%%%%%%%%%%%%%%%%%%%%%%%%%%%%%%%%%%%%%%%%%%%%%%%%%%%%%%%%%%%%%%%%%%%%%%%%%
\end{itemize}
\end{proposition}
\begin{proposition}
The schemed domain  projection 
$\mathring{\mathrmbfit{dom}} 
%= {(\mbox{-})}^{\mathrm{op}} \circ \mathrmbfit{dom} 
: \mathrmbf{DB}^{\mathrm{op}} \rightarrow \mathrmbf{DOM}$
is continuous and cocontinuous.
\end{proposition}
%

%%%%%%%%%%%%%%%%%%%%%%%%%%%%%%%%%%%%%%%%%%%%%%%%%%%%%%%%%%%%%%%%%%%%%%%%%%%%%%%%
%%%%%%%%%%%%%%%%%%%%%%%%%%%%%%%%%%%%%%%%%%%%%%%%%%%%%%%%%%%%%%%%%%%%%%%%%%%%%%%%
%\comment{% needed???
%
\begin{figure}
\begin{center}
\begin{tabular}{@{\hspace{-10pt}}c@{\hspace{45pt}}c}
\begin{tabular}{c}
\setlength{\unitlength}{0.58pt}
\begin{picture}(240,240)(90,-20)
\put(125,190){\makebox(0,0){\scriptsize{$\mathrmbfit{R}^{\mathrm{op}}$}}}
\put(120,120){\makebox(0,0){\scriptsize{$\acute{\mathrmbfit{tbl}}_{{\langle{f,g}\rangle}}$}}}
\put(120,60){\makebox(0,0){\scriptsize{${(f^{\ast})^{\mathrm{op}}}$}}}
\put(120,0){\makebox(0,0){\scriptsize{$\mathrmbfit{id}$}}}
\put(30,180){\vector(1,0){180}}
%\put(25,120){\vector(1,0){190}}
\put(140,60){\vector(-1,0){40}}
\put(25,0){\vector(1,0){190}}\put(25,0){\vector(-1,0){0}}
\put(120,146){\makebox(0,0){\shortstack{\scriptsize{$\;\acute{\varphi}^{\mathrm{op}}$}\\\large{$\Longleftarrow$}}}}
\put(120,92){\makebox(0,0){\large{$\overset{\rule[-2pt]{0pt}{5pt}\kappa}{\Longleftarrow}$}}}
\put(120,35){\makebox(0,0){\large{$\overset{\acute{\tau}_{{\langle{f,g}\rangle}}}{\Longleftarrow}$}}}
\qbezier[60](42,140)(110,140)(180,140)
\qbezier[150](-75,85)(120,85)(300,85)
\put(-8,0){\begin{picture}(0,0)(0,0)
\put(8,180){\makebox(0,0){\footnotesize{$\mathrmbf{R}_{2}^{\mathrm{op}}$}}}
\put(0,118){\makebox(0,0){\footnotesize{$\mathrmbf{Tbl}(\mathcal{A}_{2})$}}}
\put(72,60){\makebox(0,0){\footnotesize{$\mathrmbf{List}(X_{2})^{\mathrm{op}}$}}}
\put(0,0){\makebox(0,0){\footnotesize{$\mathrmbf{Set}$}}}
\put(-6,148){\makebox(0,0)[r]{\scriptsize{$\mathrmbfit{T}_{2}$}}}
\put(-75,95){\makebox(0,0)[r]{\scriptsize{$\mathrmbfit{K}_{2}$}}}
\put(55,145){\makebox(0,0)[l]{\scriptsize{$\mathrmbfit{S}_{2}^{\mathrm{op}}$}}}
\put(-37,72){\makebox(0,0)[r]{\scriptsize{$\mathrmbfit{key}_{\mathcal{A}_{2}}$}}}
\put(38,93){\makebox(0,0)[l]{\scriptsize{$\mathrmbfit{sign}_{\mathcal{A}_{2}}^{\mathrm{op}}$}}}
\put(36,26){\makebox(0,0)[l]{\scriptsize{$\mathrmbfit{tup}_{\mathcal{A}_{2}}$}}}
\put(0,60){\makebox(0,0){\shortstack{\scriptsize{$\;\tau_{\mathcal{A}_{2}}$}\\\large{$\Longrightarrow$}}}}
\put(0,165){\vector(0,-1){34}}
\put(15,105){\vector(1,-1){30}}
\put(45,45){\vector(-1,-1){30}}
\qbezier(-18,167)(-120,90)(-20,13)\put(-20,13){\vector(1,-1){0}}
\qbezier(-12,105)(-60,60)(-12,15)\put(-12,15){\vector(1,-1){0}}
\qbezier(18,167)(70,140)(66,76)\put(66,76){\vector(0,-1){0}}
\end{picture}}
\put(233,0){\begin{picture}(0,0)(0,0)
\put(8,180){\makebox(0,0){\footnotesize{$\mathrmbf{R}_{1}^{\mathrm{op}}$}}}
\put(0,118){\makebox(0,0){\footnotesize{$\mathrmbf{Tbl}(\mathcal{A}_{1})$}}}
\put(-48,60){\makebox(0,0){\footnotesize{$\mathrmbf{List}(X_{1})^{\mathrm{op}}$}}}
\put(0,0){\makebox(0,0){\footnotesize{$\mathrmbf{Set}$}}}
\put(6,148){\makebox(0,0)[l]{\scriptsize{$\mathrmbfit{T}_{1}$}}}
\put(75,95){\makebox(0,0)[l]{\scriptsize{$\mathrmbfit{K}_{1}$}}}
\put(-50,145){\makebox(0,0)[r]{\scriptsize{$\mathrmbfit{S}_{1}^{\mathrm{op}}$}}}
\put(37,72){\makebox(0,0)[l]{\scriptsize{$\mathrmbfit{key}_{\mathcal{A}_{1}}$}}}
\put(-32,93){\makebox(0,0)[r]{\scriptsize{$\mathrmbfit{sign}_{\mathcal{A}_{1}}^{\mathrm{op}}$}}}
\put(-36,26){\makebox(0,0)[r]{\scriptsize{$\mathrmbfit{tup}_{\mathcal{A}_{1}}$}}}
\put(0,60){\makebox(0,0){\shortstack{\scriptsize{$\;\tau_{\mathcal{A}_{1}}$}\\\large{$\Longleftarrow$}}}}
\put(0,165){\vector(0,-1){34}}
\put(-15,105){\vector(-1,-1){30}}
\put(-45,45){\vector(1,-1){30}}
\qbezier(18,167)(120,90)(20,13)\put(20,13){\vector(-1,-1){0}}
\qbezier(12,105)(60,60)(12,15)\put(12,15){\vector(-1,-1){0}}
\qbezier(-18,167)(-70,140)(-66,76)\put(-66,76){\vector(0,-1){0}}
\end{picture}}
\end{picture}
\end{tabular}
%%%%%%%%%%%%%%%%%%%%%%%%%%%%%%%%%%%%%%%%%%%%%%%%%%%%%%%%%%%%%%%%%%%%%%%%%%%%%%%%
&
%%%%%%%%%%%%%%%%%%%%%%%%%%%%%%%%%%%%%%%%%%%%%%%%%%%%%%%%%%%%%%%%%%%%%%%%%%%%%%%%
{{\begin{tabular}{c}
\setlength{\unitlength}{0.78pt}
\begin{picture}(120,80)(0,0)
\put(0,80){\makebox(0,0){\scriptsize{$\mathrmbfit{K}_{2}$}}}
\put(130,80){\makebox(0,0){\scriptsize{$
\mathrmbfit{R}^{\mathrm{op}}{\circ\;}\mathrmbfit{K}_{1}$}}}
\put(-5,0){\makebox(0,0){\scriptsize{$\overset{\textstyle\underbrace{
\mathring{\mathrmbfit{tup}}(\mathrmbf{R}_{2},\mathrmbfit{S}_{2},\mathcal{A}_{2})}}{
\mathrmbfit{S}_{2}^{\mathrm{op}}\circ\mathrmbfit{tup}_{\mathcal{A}_{2}}}$}}}
\put(140,0){\makebox(0,0){\scriptsize{$\overset{\textstyle\mathrmbfit{R}^{\mathrm{op}}\circ\mathring{\mathrmbfit{tup}}(\mathrmbf{R}_{1},\mathrmbfit{S}_{1},\mathcal{A}_{1})}{\overbrace{\mathrmbfit{R}^{\mathrm{op}}\circ\mathrmbfit{S}_{1}^{\mathrm{op}}\circ\mathrmbfit{tup}_{\mathcal{A}_{1}}}}$}}}
\put(-7,43){\makebox(0,0)[r]{\scriptsize{$\tau_{2}$}}}
\put(140,43){\makebox(0,0)[l]{\scriptsize{$\mathrmbfit{R}^{\mathrm{op}}{\circ\;}\tau_{1}$}}}
\put(60,90){\makebox(0,0){\scriptsize{$\kappa$}}}
%\put(65,20){\makebox(0,0){\scriptsize{$\mathring{\mathrmbfit{tup}}(\mathrmbfit{R},\acute{\varphi},f,g)$}}}
\put(65,-20){\makebox(0,0){\scriptsize{$\overset{\textstyle
\mathring{\mathrmbfit{tup}}_{{\langle{f,g}\rangle}}(\hat{\varphi})
%\mathring{\mathrmbfit{tup}}(\mathrmbfit{R},\acute{\varphi},f,g)
}{\overbrace{(\mathrmbfit{R}^{\mathrm{op}}\circ\mathrmbfit{S}_{1}^{\mathrm{op}}\circ\acute{\tau}_{{\langle{f,g}\rangle}})\bullet(\acute{\varphi}^{\mathrm{op}}\circ\mathrmbfit{tup}_{\mathcal{A}_{2}})}}$}}}
\put(60,75){\makebox(0,0){\large{$\xLeftarrow{\;\;\;\;\;\;\;\;\;\;\;\;\;\;\;\;\;\;\;}$}}}
\put(60,0){\makebox(0,0){\large{$\xLeftarrow{\;\;\;\;\;\;\;\;\;\;}$}}}
\put(0,42){\makebox(0,0){\large{$\bigg\Downarrow$}}}
\put(130,42){\makebox(0,0){\large{$\bigg\Downarrow$}}}
\end{picture}
\end{tabular}}}
\end{tabular}
\end{center}
%
%%%%%%%%%%%%%%%%%%%%%%%%%%%%%%%%%%%%%%%%%%%%%%%%%%%%%%%%%%%%%%%%%%%%%%%%%%%%%%%%%%%%%%%%%%
%%%%%%%%%%%%%%%%%%%%%%%%%%%%%%%%%%%%%%%%%%%%%%%%%%%%%%%%%%%%%%%%%%%%%%%%%%%%%%%%%%%%%%%%%%
%
\caption{Database Morphism (proj): $\mathrmbf{Db}$}
\label{fig:rel:db:mor:var:typ:dom}
\end{figure}
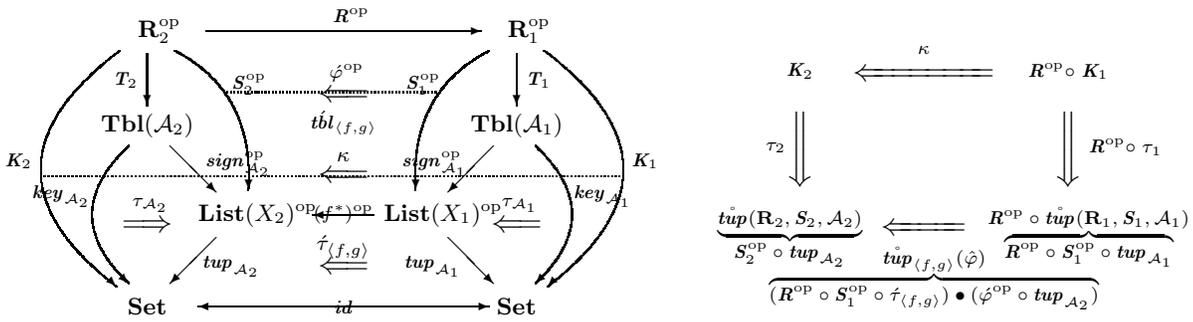

\appendix

%%%%%%%%%%%%%%%%%%%%%%%%%%%%%%%%%%%%%%%%%%%%%%%%%%%%%%%%%%%%%%%%%%%%%%%%%%%%%%%%%%%%%%%%%%
%%%%%%%%%%%%%%%%%%%%%%%%%%%%%%%%%%%%%%%%%%%%%%%%%%%%%%%%%%%%%%%%%%%%%%%%%%%%%%%%%%%%%%%%%%
%%%%%%%%%%%%%%%%%%%%%%%%%%%%%%%%%%%%%%%%%%%%%%%%%%%%%%%%%%%%%%%%%%%%%%%%%%%%%%%%%%%%%%%%%%
%\newpage
\section{Appendix}\label{sec:append}
%%%%%%%%%%%%%%%%%%%%%%%%%%%%%%%%%%%%%%%%%%%%%%%%%%%%%%%%%%%%%%%%%%%%%%%%%%%%%%%%%%%%%%%%%%
%%%%%%%%%%%%%%%%%%%%%%%%%%%%%%%%%%%%%%%%%%%%%%%%%%%%%%%%%%%%%%%%%%%%%%%%%%%%%%%%%%%%%%%%%%
%%%%%%%%%%%%%%%%%%%%%%%%%%%%%%%%%%%%%%%%%%%%%%%%%%%%%%%%%%%%%%%%%%%%%%%%%%%%%%%%%%%%%%%%%%

%%%%%%%%%%%%%%%%%%%%%%%%%%%%%%%%%%%%%%%%%%%%%%%%%%%%%%%%%%%%%%%%%%%%%%%%%%%%%%%%%%%%%%%%%%
%%%%%%%%%%%%%%%%%%%%%%%%%%%%%%%%%%%%%%%%%%%%%%%%%%%%%%%%%%%%%%%%%%%%%%%%%%%%%%%%%%%%%%%%%%
%\newpage
\subsection{General Theory}\label{sub:sec:append:gen:th}
%%%%%%%%%%%%%%%%%%%%%%%%%%%%%%%%%%%%%%%%%%%%%%%%%%%%%%%%%%%%%%%%%%%%%%%%%%%%%%%%%%%%%%%%%%
%%%%%%%%%%%%%%%%%%%%%%%%%%%%%%%%%%%%%%%%%%%%%%%%%%%%%%%%%%%%%%%%%%%%%%%%%%%%%%%%%%%%%%%%%%

%%%%%%%%%%%%%%%%%%%%%%%%%%%%%%%%%%%%%%%%%%%%%%%%%%%%%%%%%%%%%%%%%%%%%%%%%%%%%%%%
%\subsubsection{Triangle Identities.}
%\label{sub:sub:sec:triangle:identities}
%%%%%%%%%%%%%%%%%%%%%%%%%%%%%%%%%%%%%%%%%%%%%%%%%%%%%%%%%%%%%%%%%%%%%%%%%%%%%%%%

\comment{There is an equivalence between the two contexts 
$\widehat{\mathrmbf{DB}}$ and $\mathrmbf{DB}$: 
when the inverse passages
$\widehat{\mathrmbf{DB}}\;\xrightarrow{\mathrmbfit{db}}\;\mathrmbf{DB}$
and
$\widehat{\mathrmbf{DB}}\;\xleftarrow{\mathrmbfit{inc}}\;\mathrmbf{DB}$
are natural isomorphisms:
%\begin{itemize}
%\item 
$\mathrmbfit{db} \circ \mathrmbfit{inc} \cong \mathrmbfit{1}_{\widehat{\mathrmbf{DB}}}$ 
and
%\item 
$\mathrmbfit{inc} \circ \mathrmbfit{db} = \mathrmbfit{1}_{\mathrmbf{DB}}$.
%\end{itemize}
%
This is an adjoint equivalence when 
%the natural isomorphisms 
%above satisfy the triangle identities, thus exhibiting 
$\mathrmbfit{db}$ and $\mathrmbfit{inc}$ as a pair of adjoint passages.}

\begin{fact}
There is an adjunction 
$\mathrmbf{A}
\xrightarrow{{\langle{\mathrmbfit{F}{\;\dashv\;}\mathrmbfit{G}}\rangle}}
\mathrmbf{B}$
with unit 
$1_{\mathrmbf{A}}\xRightarrow{\,\eta\;\,}
\mathrmbfit{F}{\,\circ\,}\mathrmbfit{G}$
and counit 
$1_{\mathrmbf{B}}\xLeftarrow{\;\,\varepsilon\,}\mathrmbfit{G}{\,\circ\,}\mathrmbfit{F}$
when the following two \underline{triangle identities} hold
%\begin{center}
%$(\eta \circ \mathrmbfit{F}) \bullet (\mathrmbfit{F} \circ \epsilon) = 1_{\mathrmbfit{F}} 
%\text{ and }
%(\mathrmbfit{G} \circ \eta) \bullet (\epsilon \circ \mathrmbfit{G}) = 1_{\mathrmbfit{G}}$ 
%\end{center}
%
\begin{center}
$\mathrmbfit{F}
\xRightarrow{(\eta \circ \mathrmbfit{F})}
\mathrmbfit{F}\circ\mathrmbfit{G}\circ\mathrmbfit{F}
\xRightarrow{(\mathrmbfit{F} \circ \epsilon)} 
\mathrmbfit{F}
= 
\mathrmbfit{F}
\xRightarrow{1_{\mathrmbfit{F}}} 
\mathrmbfit{F}$

$\mathrmbfit{G}
\xRightarrow{(\mathrmbfit{G} \circ \eta)} 
\mathrmbfit{G}\circ\mathrmbfit{F}\circ\mathrmbfit{G}
\xRightarrow{(\epsilon \circ \mathrmbfit{G})}
\mathrmbfit{G}
= 
\mathrmbfit{G}\xRightarrow{1_{\mathrmbfit{G}}}\mathrmbfit{G}$
\end{center}
\end{fact}
\begin{definition}
%Let
%$\mathrmbf{A}\xhookleftarrow{\mathrmbfit{G}}\mathrmbf{B}$
%be an inclusion passage. 
%\begin{itemize}
%\item 
A left adjoint 
$\mathrmbf{A}\xrightarrow{\mathrmbfit{F}}\mathrmbf{B}$
to an inclusion passage 
$\mathrmbf{A}\xhookleftarrow{\mathrmbfit{G}}\mathrmbf{B}$
(of a full context $\mathrmbf{B}$) is called a \underline{reflection}.
%\item 
A full subcontext $\mathrmbf{B}$ of a context $\mathrmbf{A}$ 
is said to be \underline{reflective} in $\mathrmbf{A}$ 
when the inclusion passage 
$\mathrmbf{A}\xhookleftarrow{\mathrmbfit{G}}\mathrmbf{B}$
has a left adjoint.
%\item 
A passage 
$\mathrmbf{A}\xrightarrow{\mathrmbfit{F}}\mathrmbf{B}$
is \underline{full} when
to every pair $a$ and $a'$ of objects of $\mathrmbf{A}$
and to every morphism 
$\mathrmbfit{F}(a)\xrightarrow{g}\mathrmbfit{F}(a')$ of $\mathrmbf{B}$
there is a morphism $a\xrightarrow{f}a'$ of $\mathrmbf{A}$
such that $\mathrmbfit{F}(f) = g$.
%\end{itemize}
%
\end{definition}
\begin{fact}
For a reflection
$1_{\mathrmbf{B}}\xLeftarrow
[=]{\;\varepsilon\;}
\mathrmbfit{G}{\,\circ\,}\mathrmbfit{F}$
the triangle identities become
\begin{center}
$\mathrmbfit{F}
\xRightarrow{(\eta \circ \mathrmbfit{F})}
\mathrmbfit{F}
= 
\mathrmbfit{F}
\xRightarrow{1_{\mathrmbfit{F}}} 
\mathrmbfit{F}$
\text{and}
$\mathrmbfit{G}
\xRightarrow{(\mathrmbfit{G} \circ \eta)} 
\mathrmbfit{G}
= 
\mathrmbfit{G}\xRightarrow{1_{\mathrmbfit{G}}}\mathrmbfit{G}$
\end{center}
\end{fact}

% does not seem to be needed
%%%%%%%%%%%%%%%%%%%%%%%%%%%%%%%%%%%%%%%%%%%%%%%%%%%%%%%%%%%%%%%%%%%%%%%%%%%%%%%%
%%%%%%%%%%%%%%%%%%%%%%%%%%%%%%%%%%%%%%%%%%%%%%%%%%%%%%%%%%%%%%%%%%%%%%%%%%%%%%%%
%%%%%%%%%%%%%%%%%%%%%%%%%%%%%%%%%%%%%%%%%%%%%%%%%%%%%%%%%%%%%%%%%%%%%%%%%%%%%%%%

%%%%%%%%%%%%%%%%%%%%%%%%%%%%%%%%%%%%%%%%%%%%%%%%%%%%%%%%%%%%%%%%%%%%%%%%%%%%%%%%%%%%%%%%%%
\newpage
\subsubsection{Grothendieck Construction}\label{append:grothen:construct}

% figure not needed here
%%%%%%%%%%%%%%%%%%%%%%%%%%%%%%%%%%%%%%%%%%%%%%%%%%%%%%%%%%%%%%%%%%%%%%%%%%%%%%%%
%%%%%%%%%%%%%%%%%%%%%%%%%%%%%%%%%%%%%%%%%%%%%%%%%%%%%%%%%%%%%%%%%%%%%%%%%%%%%%%%

%
\begin{description}
\item[fibration:] 
A fibration (fibered context) $\int{\!\grave{\mathrmbfit{C}}}$
is the Grothendieck construction of a contravariant pseudo-passage (indexed context)
$\mathrmbf{I}^{\mathrm{op}}\xrightarrow{\grave{\mathrmbfit{C}}}\mathrmbf{Cxt}$,
where the action on any indexing object $i$ in $\mathrmbf{I}$ is the fiber context
$\grave{\mathrmbfit{C}}_{i}$
and
the action on any indexing morphism $i\xrightarrow{a}i'$ is the fiber passage
$\mathrmbfit{C}_{i}\xleftarrow{\grave{\mathrmbfit{C}}_{a}}\mathrmbfit{C}_{i'}$. 
An object in $\int{\!\grave{\mathrmbfit{C}}}$ is a pair ${\langle{i,A}\rangle}$,
where $i$ is an indexing object in $\mathrmbf{I}$ and $A$ is an object in the fiber context $\grave{\mathrmbfit{C}}_{i}$.
A morphism in $\int{\!\grave{\mathrmbfit{C}}}$ is a pair 
${\langle{i,A}\rangle}\xrightarrow{{\langle{a,\acute{f}}\rangle}}{\langle{i',A'}\rangle}$,
where $i\xrightarrow{a}i'$ is an indexing morphism in $\mathrmbf{I}$ 
and $A\xrightarrow{\acute{f}}\grave{\mathrmbfit{C}}_{a}(A')$ is a fiber morphism in $\grave{\mathrmbfit{C}}_{i}$. 
%Given another $\int{\!\grave{\mathrmbfit{C}}}$-morphism
%${\langle{a',f'}\rangle} : {\langle{i',A'}\rangle}\rightarrow{\langle{i'',A''}\rangle}$, 
%define their composition as
%${\langle{a,f}\rangle}{\,\cdot\,}{\langle{a',f'}\rangle}
%\doteq {\langle{a{\,\cdot\,}a', f{\,\cdot\,}\acute{\mathrmbfit{C}_{a}}(f')}\rangle} 
%: {\langle{i,A}\rangle}\rightarrow{\langle{i'',A''}\rangle}$.
There is a projection passage
$\int{\!\grave{\mathrmbfit{C}}}\rightarrow\mathrmbf{I}$.
\newline
\item[opfibration:] 
An opfibration $\int{\!\acute{\mathrmbfit{C}}}$
is the Grothendieck construction of a covariant pseudo-passage (indexed context)
$\mathrmbf{I}\xrightarrow{\acute{\mathrmbfit{C}}}\mathrmbf{Cxt}$,
where the action on any indexing object $i$ in $\mathrmbf{I}$ is the fiber context
$\acute{\mathrmbfit{C}}_{i}$
and
the action on any indexing morphism $i\xrightarrow{a}i'$ is the fiber passage
$\mathrmbfit{C}_{i}\xrightarrow{\acute{\mathrmbfit{C}}_{a}}\mathrmbfit{C}_{i'}$. 
An object in $\int{\!\acute{\mathrmbfit{C}}}$ is a pair ${\langle{i,A}\rangle}$,
where $i$ is an indexing object in $\mathrmbf{I}$ and $A$ is an object in the fiber context $\acute{\mathrmbfit{C}}_{i}$.
A morphism in $\int{\!\acute{\mathrmbfit{C}}}$ is a pair 
${\langle{i,A}\rangle}\xrightarrow{{\langle{a,\grave{f}}\rangle}}{\langle{i',A'}\rangle}$,
where $i\xrightarrow{a}i'$ is an indexing morphism in $\mathrmbf{I}$ 
and $\acute{\mathrmbfit{C}}_{a}(A)\xrightarrow{\grave{f}}A'$ is a fiber morphism in $\acute{\mathrmbfit{C}}_{i'}$. 
There is a projection passage
$\int{\!\acute{\mathrmbfit{C}}}\rightarrow\mathrmbf{I}$.
\newline
\item[bifibration:] 
A bifibration $\int{\!{\mathrmbfit{C}}}$
%(Fig.\ref{fig:bi:fbr})
is the Grothendieck construction of an indexed adjunction
$\mathrmbf{I}\xrightarrow{{\mathrmbfit{C}}}\mathrmbf{Adj}$
consisting of 
a left adjoint covariant pseudo-passage
$\mathrmbf{I}\xrightarrow{\acute{\mathrmbfit{C}}}\mathrmbf{Cxt}$
and a right adjoint contravariant pseudo-passage
$\mathrmbf{I}^{\mathrm{op}}\xrightarrow{\grave{\mathrmbfit{C}}}\mathrmbf{Cxt}$. 
The action on any indexing object $i$ in $\mathrmbf{I}$ is the fiber context
$\mathrmbfit{C}_{i} = \grave{\mathrmbfit{C}}_{i} = \acute{\mathrmbfit{C}}_{i}$
and
the action on any indexing morphism $i\xrightarrow{a}i'$ is the fiber adjunction
$\bigl(\mathrmbfit{C}_{i}\xrightarrow{\mathrmbfit{C}_{a}}\mathrmbfit{C}_{i'}\bigr) = 
\bigl(\mathrmbfit{C}_{i}\xrightarrow{{\langle{\acute{\mathrmbfit{C}}_{a}{\;\dashv\;}\grave{\mathrmbfit{C}}_{a}}\rangle}}\mathrmbfit{C}_{i'}\bigr)$. 
The Grothendieck constructions of component fibration and component opfibration are isomorphic
$\int{\!\grave{\mathrmbfit{C}}}{\;\cong\;}\int{\!\acute{\mathrmbfit{C}}}$
\[\mbox{\footnotesize{$
\bigl({\langle{i,A}\rangle}\xrightarrow{{\langle{a,\acute{f}}\rangle}}{\langle{i',A'}\rangle}\bigr)
{\;\;\;\overset{\cong}{\rightleftarrows}\;\;\;}
\bigl({\langle{i,A}\rangle}\xrightarrow{{\langle{a,\grave{f}}\rangle}}{\langle{i',A'}\rangle}\bigr)
$}\normalsize}\]
via 
%(Fig.~\ref{fig:bi:fbr}) 
the adjoint pair
$A\xrightarrow{\acute{f}}\grave{\mathrmbfit{C}}_{a}(A'){\;\cong\;}
\acute{\mathrmbfit{C}}_{a}(A)\xrightarrow{\grave{f}}A'$.
%\mbox{}\newline\noindent\rule{140pt}{1pt}\newline\mbox{}
Define the Grothendieck construction of the bifibration to be the Grothendieck construction of component fibration
$\int{\!\mathrmbfit{C}}{\;\doteq}\int{\!\grave{\mathrmbfit{C}}}$
with projection 
$\int{\!\mathrmbfit{C}}\rightarrow\mathrmbf{I}$.
\end{description}
%

%%%%%%%%%%%%%%%%%%%%%%%%%%%%%%%%%%%%%%%%%%%%%%%%%%%%%%%%%%%%%%%%%%%%%%%%%%%%%%%%%%%%%%%%%%%
%%%%%%%%%%%%%%%%%%%%%%%%%%%%%%%%%%%%%%%%%%%%%%%%%%%%%%%%%%%%%%%%%%%%%%%%%%%%%%%%%%%%%%%%%%%
%\newpage
%\subsubsection{Inclusion Bridge: Fibered Context.}\label{sub:sec:inc:bridge}
%%%%%%%%%%%%%%%%%%%%%%%%%%%%%%%%%%%%%%%%%%%%%%%%%%%%%%%%%%%%%%%%%%%%%%%%%%%%%%%%%%%%%%%%%%%
%%%%%%%%%%%%%%%%%%%%%%%%%%%%%%%%%%%%%%%%%%%%%%%%%%%%%%%%%%%%%%%%%%%%%%%%%%%%%%%%%%%%%%%%%%%

\comment{% Overview terminology
A \emph{Grothendieck fibration} (also called a fibered context or just a fibration) 
is a passage $E \xrightarrow{\;p\;} I$ such that the fibers $E_i = p^{-1}(i)$ depend (contravariantly) pseudofunctorially on $i \in I$. 
One also says that $E$ is a fibered context over $I$. 
Dually, in a \emph{(Grothendieck) opfibration} the dependence is covariant.
There is an equivalence of 2-contexts
$∫ Fib(I) \stackrel{\simeq}{\leftrightarrow} [I^{op}, Cat] : \int$
between the 2-context of fibrations over $I$ and the 2-context $[I^{op},Cxt]$ of contravariant pseudo-passages from $I$ to $Cxt$, 
also called $I$-indexed contexts.
The construction 
$\int : [I^{op}, Cat] \to Fib(I) : F \mapsto \int F$ 
of a fibration from a pseudo-passage is sometimes called the \emph{Grothendieck construction}. 
A less ambiguous term for $\int F$ is the oplax colimit of $F$.
%
%%%%%%%%%%%%%%%%%%%%%%%%%%%%%%%%%%%%%%%%%%%%%%%%%%%%%%%%%%%%
%\newpage
%\paragraph{Original Terminology.}
%%%%%%%%%%%%%%%%%%%%%%%%%%%%%%%%%%%%%%%%%%%%%%%%%%%%%%%%%%%%
}% Overview terminology

%\grave{\mathrmbfit{C}}

\comment{
A bifibration $\mathrmbf{E}=\int{\!{\mathrmbfit{C}}}\xrightarrow{\;\mathrmbfit{P}\;}\mathrmbf{I}$
(see the paper ''The \texttt{FOLE} Table'' \cite{kent:fole:era:tbl})
is the Grothendieck construction of an indexed adjunction
$\mathrmbf{I}\xrightarrow{{\mathrmbfit{C}}}\mathrmbf{Adj}$,
consisting of 
a left adjoint covariant pseudo-passage
$\mathrmbf{I}\xrightarrow{\acute{\mathrmbfit{C}}}\mathrmbf{Cxt}$
and a right adjoint contravariant pseudo-passage
$\mathrmbf{I}^{\mathrm{op}}\xrightarrow{\grave{\mathrmbfit{C}}}\mathrmbf{Cxt}$. 
}

Given an $\mathrmbf{I}$-morphism $i_{2}\xleftarrow{\,a\,}i_{1}$,
the fiber passage
$\mathrmbf{C}_{i_{2}}\xrightarrow{\;\grave{\mathrmbfit{C}}_{a}\;}\mathrmbf{C}_{i_{1}}$
and injection bridge
$\grave{\mathrmbfit{C}}_{a}{\,\circ\,}\mathrmbfit{inc}_{i_{1}}
\xRightarrow{\,\grave{\iota}_{a}\;\,}
\mathrmbfit{inc}_{i_{2}}$
(Fig.~\ref{fig:incl:bridge:fbr:cxt} {\footnotesize{\textbf{dextro}}})
have adjoints,
the fiber passage
$\mathrmbfit{C}_{i_{2}}\xleftarrow{\;\acute{\mathrmbfit{C}}_{a}\;}\mathrmbfit{C}_{i_{1}}$
and injection bridge
$\mathrmbfit{inc}_{i_{1}}
\xRightarrow{\,\acute{\iota}_{a}\;\,}
\acute{\mathrmbfit{C}}_{a}{\,\circ\,}\mathrmbfit{inc}_{i_{2}}$
(Fig.~\ref{fig:incl:bridge:fbr:cxt} {\footnotesize{\textbf{levo}}}).
%
%%%%%%%%%%%%%%%%%%%%%%%%%%%%%%%%%%%%%%%%%%%%%%%%%%%%%%%%%%%%%%%%%%%%%%%%%%%%%%%%%%%%%%%%%%%%%%%%%%%%
%%%%%%%%%%%%%%%%%%%%%%%%%%%%%%%%%%%%%%%%%%%%%%%%%%%%%%%%%%%%%%%%%%%%%%%%%%%%%%%%%%%%%%%%%%%%%%%%%%%%
\footnote{
The fiber adjunction
$\mathrmbfit{C}_{i_{2}}
\xleftarrow{{\langle{\acute{\mathrmbfit{C}}_{a}{\;\dashv\;}\grave{\mathrmbfit{C}}_{a}}\rangle}}
\mathrmbfit{C}_{i_{1}}$
has unit $\mathrmbfit{1}_{\mathrmbfit{C}_{i_{1}}}\xRightarrow{\eta_{a}}\acute{\mathrmbfit{C}}_{a}{\;\circ\;}\grave{\mathrmbfit{C}}_{a}$
with the $\mathrmbfit{C}_{i}$-morphism $A\xrightarrow{\eta_{a}(A)}\grave{\mathrmbfit{C}}_{a}(\acute{\mathrmbfit{C}}_{a}(A))$
as its $A^{\mathrm{th}}$ component, and
has counit $\grave{\mathrmbfit{C}}_{a}{\;\circ\;}\acute{\mathrmbfit{C}}_{a}\xRightarrow{\varepsilon_{a}}\mathrmbfit{1}_{\mathrmbfit{C}_{i_{2}}}$
with the $\mathrmbfit{C}_{i_{2}}$-morphism $\acute{\mathrmbfit{C}}_{a}(\grave{\mathrmbfit{C}}_{a}(A'))\xrightarrow{\varepsilon_{a}(A')}A'$
as its $A'^{\mathrm{th}}$ component.
\begin{center}
{\scriptsize\setlength{\extrarowheight}{3pt}$\begin{array}{r@{\hspace{5pt}=\hspace{5pt}}l}
A_{1}\xrightarrow{\grave{f}}\grave{\mathrmbfit{C}}_{a}(A_{2})
&
A_{1}\xrightarrow{\eta_{a}(A_{1})}\grave{\mathrmbfit{C}}_{a}(\acute{\mathrmbfit{C}}_{a}(A_{1}))
\xrightarrow{\grave{\mathrmbfit{C}}_{a}(\acute{f})}\grave{\mathrmbfit{C}}_{a}(A_{2})
\\
\acute{\mathrmbfit{C}}_{a}(A_{1})\xrightarrow{\acute{f}}A_{2}
&
\acute{\mathrmbfit{C}}_{a}(A_{1})\xrightarrow{\acute{\mathrmbfit{C}}_{a}(\grave{f})}
\acute{\mathrmbfit{C}}_{a}(\grave{\mathrmbfit{C}}_{a}(A_{2}))\xrightarrow{\varepsilon_{a}(A_{2})}A_{2}
\end{array}$}
\end{center}
}
%%%%%%%%%%%%%%%%%%%%%%%%%%%%%%%%%%%%%%%%%%%%%%%%%%%%%%%%%%%%%%%%%%%%%%%%%%%%%%%%%%%%%%%%%%%%%%%%%%%%
%%%%%%%%%%%%%%%%%%%%%%%%%%%%%%%%%%%%%%%%%%%%%%%%%%%%%%%%%%%%%%%%%%%%%%%%%%%%%%%%%%%%%%%%%%%%%%%%%%%%
%
%
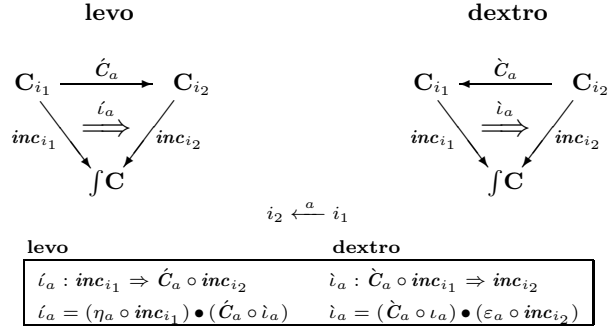
\begin{figure}
\begin{center}
\begin{tabular}{c}
{{\begin{tabular}{c@{\hspace{30pt}}c}
\textbf{levo} & \textbf{dextro}
\\&\\
{{\begin{tabular}[b]{c}
\setlength{\unitlength}{0.5pt}
\begin{picture}(230,100)(-50,0)
\put(5,80){\makebox(0,0){\footnotesize{$\mathrmbf{C}_{i_{1}}$}}}
\put(125,80){\makebox(0,0){\footnotesize{$\mathrmbf{C}_{i_{2}}$}}}
\put(60,5){\makebox(0,0){\footnotesize{$\int\!\mathrmbf{C}$}}}
\put(60,92){\makebox(0,0){\scriptsize{$\acute{\mathrmbfit{C}}_{a}$}}}
\put(24,38){\makebox(0,0)[r]{\scriptsize{$\mathrmbfit{inc}_{i_{1}}$}}}
\put(97,38){\makebox(0,0)[l]{\scriptsize{$\mathrmbfit{inc}_{i_{2}}$}}}
\put(60,54){\makebox(0,0){\shortstack{\scriptsize{$\acute{\iota}_{a}$}\\\large{$\Longrightarrow$}}}}
\put(25,80){\vector(1,0){70}}
\put(10,67){\vector(3,-4){38}}
\put(111,68){\vector(-3,-4){38}}
\end{picture}
\end{tabular}}}
%%%%%%%%%%
&
%%%%%%%%%%
{{\begin{tabular}[b]{c}
\setlength{\unitlength}{0.5pt}
\begin{picture}(230,100)(-50,0)
\put(5,80){\makebox(0,0){\footnotesize{$\mathrmbf{C}_{i_{1}}$}}}
\put(125,80){\makebox(0,0){\footnotesize{$\mathrmbf{C}_{i_{2}}$}}}
\put(60,5){\makebox(0,0){\footnotesize{$\int\!\mathrmbf{C}$}}}
\put(60,92){\makebox(0,0){\scriptsize{$\grave{\mathrmbfit{C}}_{a}$}}}
\put(24,38){\makebox(0,0)[r]{\scriptsize{$\mathrmbfit{inc}_{i_{1}}$}}}
\put(97,38){\makebox(0,0)[l]{\scriptsize{$\mathrmbfit{inc}_{i_{2}}$}}}
\put(60,54){\makebox(0,0){\shortstack{\scriptsize{$\grave{\iota}_{a}$}\\\large{$\Longrightarrow$}}}}
\put(95,80){\vector(-1,0){70}}
\put(9,68){\vector(3,-4){38}}
\put(111,68){\vector(-3,-4){38}}
\end{picture}
\end{tabular}}}
\\
\multicolumn{2}{c}{{\scriptsize{$i_{2}\xleftarrow{\;a\;}i_{1}$}}}
\end{tabular}}}
%%%%%%%%%%%%%%%%%%%%%%%%%%%%%%%%%%%%%%%%%%%%%%%%%%%%%%%%%%%%
\\
%%%%%%%%%%%%%%%%%%%%%%%%%%%%%%%%%%%%%%%%%%%%%%%%%%%%%%%%%%%%
{\scriptsize\setlength{\extrarowheight}{4pt}$\begin{array}{|@{\hspace{5pt}}l@{\hspace{15pt}}l@{\hspace{5pt}}|}
\multicolumn{1}{l}{\text{\bfseries levo}} & \multicolumn{1}{l}{\text{\bfseries dextro}}
\\ \hline
\acute{\iota}_{a} : 
\mathrmbfit{inc}_{i_{1}}
\Rightarrow
\acute{\mathrmbfit{C}}_{a}{\,\circ\,}\mathrmbfit{inc}_{i_{2}}
&
\grave{\iota}_{a} : 
\grave{\mathrmbfit{C}}_{a}{\,\circ\,}\mathrmbfit{inc}_{i_{1}}
\Rightarrow
\mathrmbfit{inc}_{i_{2}}
\\
\acute{\iota}_{a} = (\eta_{a}{\,\circ\,}\mathrmbfit{inc}_{i_{1}}){\,\bullet\,}(\acute{\mathrmbfit{C}}_{a}{\,\circ\,}\grave{\iota}_{a})
&
\grave{\iota}_{a} = (\grave{\mathrmbfit{C}}_{a}{\,\circ\,}{\iota}_{a}){\,\bullet\,}(\varepsilon_{a}{\,\circ\,}\mathrmbfit{inc}_{i_{2}})
\\\hline
\end{array}$}
%%%%%%%%%%%%%%%%%%%%%%%%%%%%%%%%%%%%%%%%%%%%%%%%%%%%%%%%%%%%
%\\\\
%%%%%%%%%%%%%%%%%%%%%%%%%%%%%%%%%%%%%%%%%%%%%%%%%%%%%%%%%%%%
%{{\scriptsize\begin{tabular}{p{200pt}}
%For any $\mathrmbf{C}_{i_{2}}$-object $A_{2}$,
%the $A_{2}^{\text{th}}$-component of
%the inclusion bridge
%$\mathrmbfit{inc}_{i_{2}}\xLeftarrow{\;\grave{\chi}_{a}\;}\grave{\mathrmbfit{C}}_{a}{\;\circ\;}\mathrmbfit{inc}_{i_{1}}$
%%of the $\mathrmbf{I}$-morphism $i_{2}\xleftarrow{\;a\;}i_{1}$,
%is 
%$A_{2}\xleftarrow{\,\grave{\chi}_{a}(A_{2})\,}\grave{\mathrmbfit{C}}_{a}(A_{2})$,
%the ``cartesian lifting'' of $\mathrmbf{I}$-morphism $i_{2}\xleftarrow{\;a\;}i_{1}$ to $A_{2}$.
%\end{tabular}}}
%
\end{tabular}
\end{center}
\caption{Inclusion Bridge: Fibered Context}
\label{fig:incl:bridge:fbr:cxt}
\end{figure}
%

%\begin{lemma}\label{lem:fibr:mor:fact}
Any morphism
${\langle{i_{1},A_{1}}\rangle}
\xrightarrow{{\langle{a,f}\rangle}}
%[{\langle{a,\grave{f}}\rangle}]
%{{\langle{a,\acute{f}}\rangle}}
{\langle{1_{2},A_{2}}\rangle}$ 
in the fibered context $\mathrmbf{E}=\int{\!{\mathrmbfit{C}}}$
consists of
a morphism $i_{1}\xrightarrow{a}i_{2}$ in the indexing context $\mathrmbf{I}$
\underline{and}
the adjoint fiber morphisms, 
$\acute{\mathrmbfit{C}}_{a}(A_{1})\xrightarrow{\acute{f}}A_{2}$ in $\mathrmbf{C}_{i_{2}}$ 
or 
$A_{1}\xrightarrow{\grave{f}}\grave{\mathrmbfit{C}}_{a}(A_{2})$ in $\mathrmbf{C}_{i_{1}}$,
with the 
factorization
%in Fig.\ref{fig:groth:fibr:mor:fact}.
%\end{lemma}
%
%\begin{figure}
\begin{center}
{{\begin{tabular}{c}
\setlength{\unitlength}{0.35pt}
\begin{picture}(320,180)(-40,15)
\put(0,172){\makebox(0,0){\footnotesize{$
\overset{\langle{i_{1},A_{1}}\rangle}
{\mathrmbfit{inc}_{i_{1}}(A_{1})}
$}}}
\put(240,175){\makebox(0,0){\footnotesize{$
\overset{\langle{i_{2},\acute{\mathrmbfit{C}}_{a}(A_{1})}\rangle}
{\mathrmbfit{inc}_{i_{2}}(\acute{\mathrmbfit{C}}_{a}(A_{1}))}
$}}}
\put(0,-13){\makebox(0,0){\footnotesize{$
\underset{\langle{i_{1},\grave{\mathrmbfit{C}}_{a}(A_{2})}\rangle}
{\mathrmbfit{inc}_{i_{1}}(\grave{\mathrmbfit{C}}_{a}(A_{2}))}
$}}}
\put(240,-13){\makebox(0,0){\footnotesize{$
\underset{\langle{i_{2},A_{2}}\rangle}
{\mathrmbfit{inc}_{i_{2}}(A_{2})}
$}}}
\put(110,180){\makebox(0,0){\scriptsize{${\acute{\iota}_{a}(A_{1})}$}}}
\put(130,20){\makebox(0,0){\scriptsize{${\grave{\iota}_{a}(A_{2})}$}}}
\put(-8,75){\makebox(0,0)[r]{\scriptsize{$
\underset{\langle{\mathrm{1}_{i_{1}},\grave{f}}\rangle}
{\mathrmbfit{inc}_{i_{1}}(\grave{f})}
$}}}
\put(250,75){\makebox(0,0)[l]{\scriptsize{$
\underset{\langle{\mathrm{1}_{i_{2}},\acute{f}}\rangle}
{\mathrmbfit{inc}_{i_{2}}(\acute{f})}
$}}}
\put(140,95){\makebox(0,0){\scriptsize{${\langle{a,f}\rangle}$}}}
%\put(140,95){\makebox(0,0){\scriptsize{${\langle{a,\acute{f}}\rangle}$}}}
%\put(96,67){\makebox(0,0){\scriptsize{${\langle{a,\grave{f}}\rangle}$}}}
%
\put(70,160){\vector(1,0){80}}
\put(90,0){\vector(1,0){80}}
\put(0,130){\vector(0,-1){100}}
\put(240,130){\vector(0,-1){100}}
\put(40,133){\vector(3,-2){160}}
\end{picture}
\end{tabular}}}
\end{center}
%\caption{Fibration Morphism Factorization}
%\label{fig:groth:fibr:mor:fact}
%\end{figure}
%

\newpage

\begin{fact}\label{fact:groth:lim}
If\, $\mathrmbf{I}^{\mathrm{op}}\!\xrightarrow{\mathrmbfit{C}}\mathrmbf{Cxt}$ 
is a contravariant pseudo-passage (indexed context)
s.t.
\begin{enumerate}
\item 
the indexing context $\mathrmbf{I}$ is complete,
\item 
the fiber context $\mathrmbf{C}_{i}$ is complete for each $i\in\mathrmbf{I}$, and
\item 
the fiber passage $\mathrmbf{C}_{i}\xleftarrow{\mathrmbfit{C}_{a}}\mathrmbf{C}_{j}$ is continuous for each $i\xrightarrow{a}j$ in $\mathrmbf{I}$,
\end{enumerate}
then the fibered context (Grothendieck construction) $\int\mathrmbf{C}$ is complete 
and the projection $\int\mathrmbf{C}\xrightarrow{\mathrmbfit{P}}\mathrmbf{I}$ is continuous.
\end{fact}
\begin{proof}
Tarlecki, Burstall and Goguen~\cite{tarlecki:burstall:goguen:91}.
\end{proof}

% big proof
%%%%%%%%%%%%%%%%%%%%%%%%%%%%%%%%%%%%%%%%%%%%%%%%%%%%%%%%%%%%%%%%%%%%%%%%%%%%%%%%%%%%%%%%%%
%%%%%%%%%%%%%%%%%%%%%%%%%%%%%%%%%%%%%%%%%%%%%%%%%%%%%%%%%%%%%%%%%%%%%%%%%%%%%%%%%%%%%%%%%%
%%%%%%%%%%%%%%%%%%%%%%%%%%%%%%%%%%%%%%%%%%%%%%%%%%%%%%%%%%%%%%%%%%%%%%%%%%%%%%%%%%%%%%%%%%
%%%%%%%%%%%%%%%%%%%%%%%%%%%%%%%%%%%%%%%%%%%%%%%%%%%%%%%%%%%%%%%%%%%%%%%%%%%%%%%%%%%%%%%%%%

%
\begin{fact}\label{fact:groth:colim}
If\, $\mathrmbf{I}\!\xrightarrow{\mathrmbfit{C}}\mathrmbf{Cxt}$ 
is a covariant pseudo-passage (indexed context)
s.t.
\begin{enumerate}
\item 
the indexing context $\mathrmbf{I}$ is cocomplete,
\item 
the fiber context $\mathrmbf{C}_{i}$ is cocomplete for each $i\in\mathrmbf{I}$, and
\item 
the fiber passage $\mathrmbf{C}_{i}\xrightarrow{\mathrmbfit{C}_{a}}\mathrmbf{C}_{j}$ is cocontinuous for each $i\xrightarrow{a}j$ in $\mathrmbf{I}$,
\end{enumerate}
then the fibered context (Grothendieck construction) $\int\mathrmbf{C}$ is cocomplete 
and the projection $\int\mathrmbf{C}\xrightarrow{\mathrmbfit{P}}\mathrmbf{I}$ is cocontinuous.
\end{fact}
\begin{proof}
Dual to the above.
\mbox{}\hfill\rule{5pt}{5pt}
\end{proof}
%

%\newpage

%
\begin{fact}\label{fact:groth:adj:lim:colim}
If\, $\mathrmbf{I}\!\xrightarrow{\mathrmbfit{C}}\mathrmbf{Adj}$ 
is an indexed adjunction consisting of 
a contravariant pseudo-passage 
%(indexed context)
$\mathrmbf{I}^{\mathrm{op}}\xrightarrow{\grave{\mathrmbfit{C}}}\mathrmbf{Cxt}$
and
a covariant pseudo-passage
$\mathrmbf{I}\!\xrightarrow{\acute{\mathrmbfit{C}}}\mathrmbf{Cxt}$
that are locally adjunctive
%\newline\mbox{}\hfill
$\bigl(
\mathrmbfit{C}_{i}\xrightarrow{{\langle{\acute{\mathrmbfit{C}}_{a}{\;\dashv\;}\grave{\mathrmbfit{C}}_{a}}\rangle}}\mathrmbfit{C}_{i'}
\bigr)$
%\hfill\mbox{}\newline
for each $i\xrightarrow{a}j$ in $\mathrmbf{I}$,
s.t.
\begin{enumerate}
\item 
the indexing context $\mathrmbf{I}$ is complete and cocomplete,
\item 
the fiber context $\mathrmbf{C}_{i}$ is complete and cocomplete for each $i\in\mathrmbf{I}$,
\end{enumerate}
then the fibered context (Grothendieck construction) $\int\mathrmbf{C}\rightarrow\mathrmbf{I}$ is complete and cocomplete
and the projection
$\int\mathrmbf{C}\rightarrow\mathrmbf{I}$
is continuous and cocontinuous.
\end{fact}
\begin{proof}
Use Facts.~\ref{fact:groth:lim}~\&~\ref{fact:groth:colim},
since the fiber passage $\mathrmbf{C}_{i}\xleftarrow{\grave{\mathrmbfit{C}}_{a}}\mathrmbf{C}_{i'}$ is continuous (being right adjoint)
and the fiber passage $\mathrmbf{C}_{i}\xrightarrow{\acute{\mathrmbfit{C}}_{a}}\mathrmbf{C}_{i'}$ is cocontinuous (being left adjoint)
for each $i\xrightarrow{a}j$ in $\mathrmbf{I}$.
\mbox{}\hfill\rule{5pt}{5pt}
\end{proof}
%

%%%%%%%%%%%%%%%%%%%%%%%%%%%%%%%%%%%%%%%%%%%%%%%%%%%%%%%%%%%%%%%%%%%%%%%%%%%%%%%%%%%%%%%%%%
\newpage
\subsubsection{Diagram Contexts}\label{sub:sub:sec:lax:comma:cxt}

% rehash
%

For any base mathematical context $\mathrmbf{A}$,
%where $\mathrmbf{A}$ is used as a base context,
there are two contexts of lax diagrams 
with 
%varying 
indexing (shape) contexts.
\begin{definition}\label{def:lax:oplax}\mbox{}
\begin{itemize}
\item[$\mathrmbf{DOM}:$] 
The lax comma context over $\mathrmbf{A}$ is defined to be
the co-variant lax diagram context
$\mathrmbf{A}^{\!\scriptscriptstyle{\Uparrow}} 
= \bigl({(\mbox{-})}{\,\Uparrow\,}\mathrmbf{A}\bigr)
= \bigl(\mathrmbf{Cxt}{\,\Uparrow\,}\mathrmbf{A}\bigr)$.
An object of $\mathrmbf{A}^{\!\scriptscriptstyle{\Uparrow}}$ 
is a pair ${\langle{\mathrmbf{B},\mathrmbfit{F}}\rangle}$ consisting of a
context $\mathrmbf{B}$ and a passage 
(diagram)
$\mathrmbf{B}\xrightarrow{\;\mathrmbfit{F}\;}\mathrmbf{A}$.
A morphism of $\mathrmbf{A}^{\!\scriptscriptstyle{\Uparrow}}$ 
is a pair 
${\langle{\mathrmbf{B}_{2},\mathrmbfit{F}_{2}}\rangle}
\xrightarrow{\;{\langle{\mathrmbfit{G},\alpha}\rangle}\;}
{\langle{\mathrmbf{B}_{1},\mathrmbfit{F}_{1}}\rangle}$
consisting of shape varying passage
$\mathrmbf{B}_{2}\xrightarrow{\;\mathrmbfit{G}\;}\mathrmbf{B}_{1}$
and a bridge
$\mathrmbfit{F}_{2}\xRightarrow{\;\alpha\;}\mathrmbfit{G}{\;\circ\;}\mathrmbfit{F}_{1}$.
The fiber at shape context
$\mathrmbf{B}$ is the context of diagrams
$\mathrmbf{A}^{\mathrmbf{B}}
= \bigl[\mathrmbf{B},\mathrmbf{A}\bigr]$.
\newline
\begin{center}
{{\begin{tabular}{c}
\setlength{\unitlength}{0.5pt}
\begin{picture}(120,80)(0,-10)
\put(0,80){\makebox(0,0){\footnotesize{$\mathrmbf{B}_{2}$}}}
\put(120,80){\makebox(0,0){\footnotesize{$\mathrmbf{B}_{1}$}}}
\put(62,5){\makebox(0,0){\footnotesize{$\mathrmbf{A}$}}}
\put(60,92){\makebox(0,0){\scriptsize{$\mathrmbfit{G}$}}}
\put(20,42){\makebox(0,0)[r]{\scriptsize{$\mathrmbfit{F}_{2}$}}}
\put(100,42){\makebox(0,0)[l]{\scriptsize{$\mathrmbfit{F}_{1}$}}}
\put(60,57){\makebox(0,0){\shortstack{\scriptsize{$\alpha$}\\\large{$\Longrightarrow$}}}}
\put(20,80){\vector(1,0){80}}
\put(10,67){\vector(3,-4){38}}
\put(110,68){\vector(-3,-4){38}}
\end{picture}
\end{tabular}}}
\end{center}
\item[$\mathrmbf{DB}:$]  
The oplax comma context over $\mathrmbf{A}$ is defined to be
the contra-variant lax diagram context
$\mathrmbf{A}^{\!\scriptscriptstyle{\Downarrow}} 
= \bigl({(\mbox{-})}^{\mathrm{op}}{\,\Downarrow\,}\mathrmbf{A}\bigr)
\cong \bigl(\mathrmbf{Cxt}{\,\Downarrow\,}\mathrmbf{A}\bigr)$.
An object of $\mathrmbf{A}^{\!\scriptscriptstyle{\Downarrow}}$ 
is a pair ${\langle{\mathrmbf{B},\mathrmbfit{F}}\rangle}$ consisting of a
context $\mathrmbf{B}$ and a passage (diagram) 
$\mathrmbf{B}^{\mathrm{op}}\xrightarrow{\;\mathrmbfit{F}\;}\mathrmbf{A}$.
A morphism of $\mathrmbf{A}^{\!\scriptscriptstyle{\Downarrow}}$ 
is a pair ${\langle{\mathrmbf{B}_{2},\mathrmbfit{F}_{2}}\rangle}
\xleftarrow{\;{\langle{\mathrmbfit{G},\alpha}\rangle}\;}
{\langle{\mathrmbf{B}_{1},\mathrmbfit{F}_{1}}\rangle}$
consisting of a shape passage
$\mathrmbf{B}_{2}\xrightarrow{\;\mathrmbfit{G}\;}\mathrmbf{B}_{1}$
and a bridge
$\mathrmbfit{F}_{2}
\xLeftarrow{\;\alpha\;}\mathrmbfit{G}^{\mathrm{op}}{\;\circ\;}\mathrmbfit{F}_{1}$.
%
%%%%%%%%%%%%%%%%%%%%%%%%%%%%%%%%%%%%%%%%%%%%%%%%%%%%%%%%%%%%%%%%%%%%%%%%%%%%%%%%
%%%%%%%%%%%%%%%%%%%%%%%%%%%%%%%%%%%%%%%%%%%%%%%%%%%%%%%%%%%%%%%%%%%%%%%%%%%%%%%%
\footnote{Notice the direction of the morphisms.
This is opposite the definition of the ``super-comma'' context in
Mac~Lane~\cite{maclane:71}.}
%%%%%%%%%%%%%%%%%%%%%%%%%%%%%%%%%%%%%%%%%%%%%%%%%%%%%%%%%%%%%%%%%%%%%%%%%%%%%%%%
%%%%%%%%%%%%%%%%%%%%%%%%%%%%%%%%%%%%%%%%%%%%%%%%%%%%%%%%%%%%%%%%%%%%%%%%%%%%%%%%
%
The fiber at shape context
$\mathrmbf{B}$ is the context of diagrams
$\mathrmbf{A}^{\mathrmbf{B}^{\mathrm{op}}}
= \bigl[\mathrmbf{B}^{\mathrm{op}},\mathrmbf{A}\bigr]$.
\newline
\begin{center}
{{\begin{tabular}{c}
\setlength{\unitlength}{0.5pt}
\begin{picture}(120,80)(0,0)
\put(5,80){\makebox(0,0){\footnotesize{$\mathrmbf{B}_{2}^{\mathrm{op}}$}}}
\put(125,80){\makebox(0,0){\footnotesize{$\mathrmbf{B}_{1}^{\mathrm{op}}$}}}
\put(62,5){\makebox(0,0){\footnotesize{$\mathrmbf{A}$}}}
\put(65,92){\makebox(0,0){\scriptsize{$\mathrmbfit{G}^{\mathrm{op}}$}}}
\put(20,42){\makebox(0,0)[r]{\scriptsize{$\mathrmbfit{F}_{2}$}}}
\put(100,42){\makebox(0,0)[l]{\scriptsize{$\mathrmbfit{F}_{1}$}}}
\put(62,57){\makebox(0,0){\shortstack{\scriptsize{$\alpha$}\\\large{$\Longleftarrow$}}}}
\put(20,80){\vector(1,0){80}}
\put(10,67){\vector(3,-4){38}}
\put(110,68){\vector(-3,-4){38}}
\end{picture}
\end{tabular}}}
\end{center}
\end{itemize}
\end{definition}
Composition is component-wise.
%Note that
%$\bigl(\mathrmbf{A}^{\!\scriptscriptstyle{\Downarrow}}\bigr)^{\propto} 
%{\,\cong\,} 
%\bigl(\mathrmbf{A}^{\mathrm{op}}\bigr)^{\!\scriptscriptstyle{\Uparrow}}$
%and
%$(\mathrmbf{A}^{\!\scriptscriptstyle{\Uparrow}})^{\propto}
%{\,\cong\,}
%(\mathrmbf{A}^{\!\mathrm{op}})^{\scriptscriptstyle{\Downarrow}}$.
%
%%%%%%%%%%%%%%%%%%%%%%%%%%%%%%%%%%%%%%%%%%%%%%%%%%%%%%%%%%%%%%%%%%%%%%%%%%%%%%%%
%%%%%%%%%%%%%%%%%%%%%%%%%%%%%%%%%%%%%%%%%%%%%%%%%%%%%%%%%%%%%%%%%%%%%%%%%%%%%%%%
\footnote{The operator $(\text{-})^{\propto}$
applies the opposite to all elements of a lax diagram
(contexts, passages and bridges).
Hence,
%\newline
$\bigl(\mathrmbf{A}^{\!\scriptscriptstyle{\Downarrow}}\bigr)^{\propto} 
= \bigl({(\mbox{-})}^{\mathrm{op}}{\,\Downarrow\,}\mathrmbf{A}\bigr)^{\propto}
{\triangleq\;}
\bigl({(\mbox{-})}{\,\Uparrow\,}\mathrmbf{A}^{\mathrm{op}}\bigr)
= \bigl(\mathrmbf{A}^{\mathrm{op}}\bigr)^{\!\scriptscriptstyle{\Uparrow}}$
and
%\newline
$\bigl(\mathrmbf{A}^{\!\scriptscriptstyle{\Uparrow}}\bigr)^{\propto}
= \bigl({(\mbox{-})}{\,\Uparrow\,}\mathrmbf{A}\bigr)^{\propto}
{\triangleq\;}
{(\mbox{-})}^{\mathrm{op}}{\,\Downarrow\,}(\mathrmbf{A}^{\mathrm{op}})
= (\mathrmbf{A}^{\mathrm{op}})^{\!\scriptscriptstyle{\Downarrow}}$.
\begin{center}
{{\begin{tabular}{c@{\hspace{20pt}}c@{\hspace{15pt}}c}
{{\begin{tabular}{c}
\setlength{\unitlength}{0.5pt}
\begin{picture}(100,80)(-45,0)
\put(-38,70){\makebox(0,0){\footnotesize{$\mathrmbf{B}_{2}^{\mathrm{op}}$}}}
\put(42,70){\makebox(0,0){\footnotesize{$\mathrmbf{B}_{1}^{\mathrm{op}}$}}}
\put(0,0){\makebox(0,0){\footnotesize{$\mathrmbf{A}$}}}
\put(-30,40){\makebox(0,0)[r]{\scriptsize{$\mathrmbfit{H}_{2}$}}}
\put(30,40){\makebox(0,0)[l]{\scriptsize{$\mathrmbfit{H}_{1}$}}}
\put(5,80){\makebox(0,0){\scriptsize{$\mathrmbfit{G}^{\mathrm{op}}$}}}
\put(-33,56){\vector(1,-2){22}}
\put(33,56){\vector(-1,-2){22}}
\put(-23,70){\vector(1,0){46}}
\put(0,48){\makebox(0,0){\shortstack{\normalsize{$\xLeftarrow{\;\;\alpha\;}$}}}}
\put(8,35){\makebox(0,0){\large{$\Biggl(\;\;\;\;\;\;\;\;\;\;\;\;\;\;\;\;\;\;\Biggr)^{\!\propto}$}}}
\end{picture}
\end{tabular}}}
&
{{\begin{tabular}{c}
\setlength{\unitlength}{0.55pt}
\begin{picture}(0,80)(0,0)
\put(0,40){\makebox(0,0){\normalsize{$\triangleq$}}}
\end{picture}
\end{tabular}}}
&
{{\begin{tabular}{c}
\setlength{\unitlength}{0.52pt}
\begin{picture}(100,80)(-45,0)
\put(-38,70){\makebox(0,0){\footnotesize{$\mathrmbf{B}_{2}$}}}
\put(42,70){\makebox(0,0){\footnotesize{$\mathrmbf{B}_{1}$}}}
\put(10,0){\makebox(0,0){\footnotesize{$\mathrmbf{A}^{\mathrm{op}}$.}}}
\put(-30,40){\makebox(0,0)[r]{\scriptsize{$\mathrmbfit{H}_{2}^{\mathrm{op}}$}}}
\put(30,40){\makebox(0,0)[l]{\scriptsize{$\mathrmbfit{H}_{1}^{\mathrm{op}}$}}}
\put(0,80){\makebox(0,0){\scriptsize{$\mathrmbfit{G}$}}}
\put(-33,56){\vector(1,-2){22}}
\put(33,56){\vector(-1,-2){22}}
\put(-23,70){\vector(1,0){46}}
\put(0,48){\makebox(0,0){\shortstack{\normalsize{$\xRightarrow{\;\alpha^{\mathrm{op}}\!}$}}}}
\end{picture}
\end{tabular}}}
\\
{\scriptsize{$\Bigl({\langle{\mathrmbf{B}_{2},\mathrmbfit{H}_{2}}\rangle}
\xleftarrow{{\langle{\mathrmbfit{G},\alpha}\rangle}}
{\langle{\mathrmbf{B}_{1},\mathrmbfit{H}_{1}}\rangle}\Bigr)^{\propto}$}}
& $\triangleq$ & 
{\scriptsize{${\langle{\mathrmbf{B}_{2},\mathrmbfit{H}^{\mathrm{op}}_{2}}\rangle}
\xrightarrow{{\langle{\mathrmbfit{G},\alpha^{\mathrm{op}}}\rangle}}
{\langle{\mathrmbf{B}_{1},\mathrmbfit{H}^{\mathrm{op}}_{1}}\rangle}$}}
\end{tabular}}}
\end{center}}
%%%%%%%%%%%%%%%%%%%%%%%%%%%%%%%%%%%%%%%%%%%%%%%%%%%%%%%%%%%%%%%%%%%%%%%%%%%%%%%%
%%%%%%%%%%%%%%%%%%%%%%%%%%%%%%%%%%%%%%%%%%%%%%%%%%%%%%%%%%%%%%%%%%%%%%%%%%%%%%%%

%%%%%%%%%%%%%%%%%%%%%%%%%%%%%%%%%%%%%%%%%%%%%%%%%%%%%%%%%%%%
%\newpage
\paragraph{Example.}
%%%%%%%%%%%%%%%%%%%%%%%%%%%%%%%%%%%%%%%%%%%%%%%%%%%%%%%%%%%%
%
The context 
$\mathrmbf{SET}
= \mathrmbf{Set}^{\scriptscriptstyle{\Downarrow}}
= \bigl({(\mbox{-})}^{\mathrm{op}}{\,\Downarrow\,}\mathrmbf{Set}\bigr)
\cong \bigl(\mathrmbf{Cxt}{\,\Downarrow\,}\mathrmbf{Set}\bigr)
$
is the oplax comma context 
over
$\mathrmbf{Set}$.
Hence,
an object
${\langle{\mathrmbf{B},\mathrmbfit{F}}\rangle}$
in $\mathrmbf{SET}$
is a diagram
$\mathrmbf{B}^{\mathrm{op}} \xrightarrow{\,\mathrmbfit{F}\;} \mathrmbf{A}$,
and
a morphism
${\langle{\mathrmbf{B}_{2},\mathrmbfit{F}_{2}}\rangle}
\xleftarrow{\;{\langle{\mathrmbfit{G},\alpha}\rangle}\;}
{\langle{\mathrmbf{B}_{1},\mathrmbfit{F}_{1}}\rangle}$in $\mathrmbf{SET}$
is
a diagram morphism 
with shape passage
$\mathrmbf{B}_{2}\xrightarrow{\;\mathrmbfit{G}\;}\mathrmbf{B}_{1}$
and bridge
$\mathrmbfit{F}_{2}
\xLeftarrow{\;\alpha\;}\mathrmbfit{G}^{\mathrm{op}}{\;\circ\;}\mathrmbfit{F}_{1}$.
%Composition is component-wise.
%
%%%%%%%%%%%%%%%%%%%%%%%%%%%%%%%%%%%%%%%%%%%%%%%%%%%%%%%%%%%%%%%%%%%%%%%%%%%%%%%%
%%%%%%%%%%%%%%%%%%%%%%%%%%%%%%%%%%%%%%%%%%%%%%%%%%%%%%%%%%%%%%%%%%%%%%%%%%%%%%%%
\comment{% not needed
\begin{figure}
\begin{center}
{{\begin{tabular}{c}
\setlength{\unitlength}{0.56pt}
\begin{picture}(120,80)(8,0)
\put(5,80){\makebox(0,0){\footnotesize{$\mathrmbf{B}_{2}^{\mathrm{op}}$}}}
\put(125,80){\makebox(0,0){\footnotesize{$\mathrmbf{B}_{1}^{\mathrm{op}}$}}}
\put(65,0){\makebox(0,0){\footnotesize{$\mathrmbf{Set}$}}}
\put(60,92){\makebox(0,0){\scriptsize{$\mathrmbfit{G}^{\mathrm{op}}$}}}
\put(20,42){\makebox(0,0)[r]{\scriptsize{$\mathrmbfit{F}_{2}$}}}
\put(100,42){\makebox(0,0)[l]{\scriptsize{$\mathrmbfit{F}_{1}$}}}
\put(60,57){\makebox(0,0){\shortstack{\scriptsize{$\alpha$}\\\large{$\Longleftarrow$}}}}
\put(20,80){\vector(1,0){80}}
\put(10,67){\vector(3,-4){38}}
\put(110,68){\vector(-3,-4){38}}
\end{picture}
\end{tabular}}}
\end{center}
\caption{Morphism: $\mathrmbf{SET}$}
\label{fig:SET:mor}
\end{figure}
}% not needed
%%%%%%%%%%%%%%%%%%%%%%%%%%%%%%%%%%%%%%%%%%%%%%%%%%%%%%%%%%%%%%%%%%%%%%%%%%%%%%%%
%%%%%%%%%%%%%%%%%%%%%%%%%%%%%%%%%%%%%%%%%%%%%%%%%%%%%%%%%%%%%%%%%%%%%%%%%%%%%%%%
%

%
%%%%%%%%%%%%%%%%%%%%%%%%%%%%%%%%%%%%%%%%%%%%%%%%%%%%%%%%%%%%
%\newpage
\paragraph{Change of Base.}
%%%%%%%%%%%%%%%%%%%%%%%%%%%%%%%%%%%%%%%%%%%%%%%%%%%%%%%%%%%%
%
Any passage $\mathrmbf{A}\xrightarrow{\mathrmbfit{H}}\mathrmbf{B}$
defines lax diagram morphisms.
\begin{center}
{{\footnotesize{\begin{tabular}{|@{\hspace{10pt}}c
%@{\hspace{20pt}}l
@{\hspace{10pt}}|}
\hline
$\mathrmbf{A}^{\!\scriptscriptstyle{\Downarrow}} 
= \bigl({(\mbox{-})}^{\mathrm{op}}{\,\Downarrow\,}\mathrmbf{A}\bigr)
\xrightarrow[{(\text{-})}{\,\circ\,}\mathrmbfit{H}]{\;\mathrmbfit{H}^{\!\scriptscriptstyle{\Downarrow}}\;}
\bigl({(\mbox{-})}^{\mathrm{op}}{\,\Downarrow\,}\mathrmbf{B}\bigr) = 
\mathrmbf{B}^{\!\scriptscriptstyle{\Downarrow}}$
%&
%$\mathrmbf{A}\xrightarrow{\mathrmbfit{H}}\mathrmbf{B}$
\\
$\bigl(\mathrmbf{A}^{\!\scriptscriptstyle{\Downarrow}}\bigr)^{\propto}
\cong \bigl(\mathrmbf{A}^{\mathrm{op}}\bigr)^{\!\scriptscriptstyle{\Uparrow}}
\xrightarrow
[{(\text{-})^{\mathrm{op}}}{\,\circ\,}\mathrmbfit{H}^{\mathrm{op}}]
{\bigl(\mathrmbfit{H}^{\!\scriptscriptstyle{\Downarrow}}\bigr)^{\propto}
=\;\bigl(\mathrmbfit{H}^{\mathrm{op}}\bigr)^{\!\scriptscriptstyle{\Uparrow}}}
\bigl(\mathrmbf{B}^{\mathrm{op}}\bigr)^{\!\scriptscriptstyle{\Uparrow}} \cong 
\bigl(\mathrmbf{B}^{\!\scriptscriptstyle{\Downarrow}}\bigr)^{\propto}$
%&
%$\mathrmbf{A}\xrightarrow{\mathrmbfit{H}}\mathrmbf{B}$
\\\hline
$\mathrmbf{A}^{\!\scriptscriptstyle{\Uparrow}} 
= \bigl({(\mbox{-})}{\,\Uparrow\,}\mathrmbf{A}\bigr)
\xrightarrow[{(\text{-})}{\,\circ\,}\mathrmbfit{H}]{\;\mathrmbfit{H}^{\!\scriptscriptstyle{\Uparrow}}\;}
\bigl({(\mbox{-})}{\,\Uparrow\,}\mathrmbf{B}\bigr) = 
\mathrmbf{B}^{\!\scriptscriptstyle{\Uparrow}}$
%&
%$\mathrmbf{A}\xrightarrow{\mathrmbfit{H}}\mathrmbf{B}$
\\
$
(\mathrmbf{A}^{\!\scriptscriptstyle{\Uparrow}})^{\propto} 
{\,\cong\,}
(\mathrmbf{A}^{\!\mathrm{op}})^{\scriptscriptstyle{\Downarrow}} 
%= \bigl({(\mbox{-})}{\,\Uparrow\,}\mathrmbf{A}\bigr)^{\mathrm{op}}
\xrightarrow
[{(\text{-})}{\,\circ\,}\mathrmbfit{H}^{\mathrm{op}}]
%({(\text{-})}{\,\circ\,}\mathrmbfit{H})^{\mathrm{op}}{\,\cong\,}
{\;(\mathrmbfit{H}^{\!\scriptscriptstyle{\Uparrow}})^{\propto}
{\,\cong\,}(\mathrmbfit{H}^{\mathrm{op}})^{\!\scriptscriptstyle{\Downarrow}}\;}
%\bigl({(\mbox{-})}{\,\Uparrow\,}\mathrmbf{B}\bigr)^{\mathrm{op}} = 
(\mathrmbf{B}^{\!\mathrm{op}})^{\scriptscriptstyle{\Downarrow}}
{\,\cong\,}
(\mathrmbf{B}^{\!\scriptscriptstyle{\Uparrow}})^{\propto}
$
%&
%$\mathrmbf{A}\xrightarrow{\mathrmbfit{H}}\mathrmbf{B}$
\\\hline
\end{tabular}}}}
\end{center}
The passage 
$\mathrmbf{A}^{\!\scriptscriptstyle{\Downarrow}} 
\xrightarrow[{(\text{-})}{\,\circ\,}\mathrmbfit{H}]{\;\mathrmbfit{H}^{\!\scriptscriptstyle{\Downarrow}}\;}
\mathrmbf{B}^{\!\scriptscriptstyle{\Downarrow}}$
is defined as follows.
$\mathrmbf{A}^{\!\scriptscriptstyle{\Uparrow}} 
%= \bigl({(\mbox{-})}{\,\Uparrow\,}\mathrmbf{A}\bigr)
\xrightarrow[{(\text{-})}{\,\circ\,}\mathrmbfit{H}]{\;\mathrmbfit{H}^{\!\scriptscriptstyle{\Uparrow}}\;}
%\bigl({(\mbox{-})}{\,\Uparrow\,}\mathrmbf{B}\bigr) = 
\mathrmbf{B}^{\!\scriptscriptstyle{\Uparrow}}$
is defined dually.
%similarly.
%
\begin{itemize}
\item 
$\mathrmbfit{H}^{\!\scriptscriptstyle{\Downarrow}}$
maps an object 
$\mathrmbf{I}^{\mathrm{op}}\xrightarrow{\;\mathrmbfit{F}\;}\mathrmbf{A} \in
\mathrmbf{A}^{\!\scriptscriptstyle{\Downarrow}}$
%consisting of a context $\mathrmbf{C}$ and a passage (diagram) 
%$\mathrmbf{C}^{\mathrm{op}}\xrightarrow{\;\mathrmbfit{F}\;}\mathrmbf{A}$
\underline{to}
the object
$\mathrmbf{I}^{\mathrm{op}}\xrightarrow{\;\mathrmbfit{F}{\,\circ\,}\mathrmbfit{H}
\;}\mathrmbf{B} \in
\mathrmbf{B}^{\!\scriptscriptstyle{\Downarrow}}$.
\item 
$\mathrmbfit{H}^{\!\scriptscriptstyle{\Downarrow}}$
maps  
a morphism 
${\langle{\mathrmbf{C}_{2},\mathrmbfit{F}_{2}}\rangle}
\xleftarrow{\;{\langle{\mathrmbfit{G},\alpha}\rangle}\;}
{\langle{\mathrmbf{C}_{1},\mathrmbfit{F}_{1}}\rangle}$
of $\mathrmbf{A}^{\!\scriptscriptstyle{\Downarrow}}$, 
a pair consisting of a shape passage
$\mathrmbf{C}_{2}\xrightarrow{\;\mathrmbfit{G}\;}\mathrmbf{C}_{1}$
and a bridge
$\mathrmbfit{F}_{2}\xLeftarrow{\;\alpha\;}\mathrmbfit{G}^{\mathrm{op}}{\;\circ\;}\mathrmbfit{F}_{1}$,
\underline{to} 
the morphism 
${\langle{\mathrmbf{C}_{2},\mathrmbfit{F}_{2}{\circ}\mathrmbfit{H}}\rangle}
\xleftarrow{\;{\langle{\mathrmbfit{G},\alpha{\circ}\mathrmbfit{H}}\rangle}\;}
{\langle{\mathrmbf{C}_{1},\mathrmbfit{F}_{1}{\circ}\mathrmbfit{H}}\rangle}$
of $\mathrmbf{B}^{\!\scriptscriptstyle{\Downarrow}}$, 
a pair
consisting of a shape passage
$\mathrmbf{C}_{2}\xrightarrow{\;\mathrmbfit{G}\;}\mathrmbf{C}_{1}$
and a bridge
$\mathrmbfit{F}_{2}{\circ}\mathrmbfit{H}
\xLeftarrow{\;\alpha{\circ}\mathrmbfit{H}\;}
\mathrmbfit{G}^{\mathrm{op}}{\;\circ\;}\mathrmbfit{F}_{1}{\circ}\mathrmbfit{H}$.
\end{itemize}

.

% definitions of major contexts
%%%%%%%%%%%%%%%%%%%%%%%%%%%%%%%%%%%%%%%%%%%%%%%%%%%%%%%%%%%%%%%%%%%%%%%%%%%%%%%%
%%%%%%%%%%%%%%%%%%%%%%%%%%%%%%%%%%%%%%%%%%%%%%%%%%%%%%%%%%%%%%%%%%%%%%%%%%%%%%%%

%$\mathrmbf{DOM}=\mathrmbf{Dom}^{\!\scriptscriptstyle{\Uparrow}}$ (schemed domains),
%$\mathrmbf{LIST}=\mathrmbf{List}^{\!\scriptscriptstyle{\Uparrow}}$ (schemas),
%$\mathrmbf{CLS}=\mathrmbf{Cls}^{\!\scriptscriptstyle{\Uparrow}}$, and
%$\mathrmbf{SET}=\mathrmbf{Set}^{\!\scriptscriptstyle{\Downarrow}}$
%\triangleq\bigl({(\mbox{-})}^{\mathrm{op}}{\Downarrow\,}\mathrmbf{Set}\bigr)
%\triangleq\bigl(\mathrmbf{Cxt}{\,\Uparrow\,}\mathrmbf{Cls}\bigr)
%\triangleq\bigl(\mathrmbf{Cxt}{\,\Uparrow\,}\mathrmbf{Dom}\bigr)
%\triangleq\bigl(\mathrmbf{Cxt}{\,\Uparrow\,}\mathrmbf{List}\bigr)

%%%%%%%%%%%%%%%%%%%%%%%%%%%%%%%%%%%%%%%%%%%%%%%%%%%%%%%%%%%%
\newpage
\paragraph{Complete/Cocomplete Diagram Contexts.}
%%%%%%%%%%%%%%%%%%%%%%%%%%%%%%%%%%%%%%%%%%%%%%%%%%%%%%%%%%%%

%%%%%%%%%%%%%%%%%%%%%%%%%%%%%%%%%%%%%%%%%%%%%%%%%%%%%%%%%%%%%%%%%%%%%%%%%%%%%%%%%%%%%%%%%%
%\mbox{}\newline\rule{120pt}{1pt}{\fbox{\textbf{ Work Zone }}}\rule{120pt}{1pt}\newline
%%%%%%%%%%%%%%%%%%%%%%%%%%%%%%%%%%%%%%%%%%%%%%%%%%%%%%%%%%%%%%%%%%%%%%%%%%%%%%%%%%%%%%%%%%

%\begin{itemize}
%\item 
A cone over a diagram $\mathrmbf{I}\xrightarrow{\;\mathrmbfit{F}\;}\mathrmbf{A}$
is a bridge $\mathrmbfit{F}\xLeftarrow{\;\pi}\Delta(a)$,
whose vertex is an element $a{\,\in\,}\mathrmbf{A}$.
Hence, 
a cone is a morphism
${\langle{\mathrmbf{I},\mathrmbfit{F}}\rangle}
\xrightarrow{{\langle{\Delta,\pi}\rangle}}
{\langle{\mathrmbf{1},\mathrmbfit{a}}\rangle}$
in 
$\mathrmbf{A}^{\!\scriptscriptstyle{\Downarrow}} 
= \bigl({(\mbox{-})}^{\mathrm{op}}{\,\Downarrow\,}\mathrmbf{A}\bigr)$
consisting of a shape passage
$\mathrmbf{I}\xrightarrow{\;\Delta\;}\mathrmbf{1}$
and a bridge
$\mathrmbfit{I}\xLeftarrow{\;\pi\;}\Delta^{\mathrm{op}}{\;\circ\;}\mathrmbfit{a}$.
\begin{center}
{{\begin{tabular}{c}
\setlength{\unitlength}{0.5pt}
\begin{picture}(120,80)(0,0)
\put(5,80){\makebox(0,0){\footnotesize{$\mathrmbf{I}^{\mathrm{op}}$}}}
\put(125,80){\makebox(0,0){\footnotesize{$\mathrmbf{1}^{\mathrm{op}}$}}}
\put(62,5){\makebox(0,0){\footnotesize{$\mathrmbf{A}$}}}
\put(65,92){\makebox(0,0){\scriptsize{$\Delta^{\mathrm{op}}$}}}
\put(20,42){\makebox(0,0)[r]{\scriptsize{$\mathrmbfit{F}$}}}
\put(100,42){\makebox(0,0)[l]{\scriptsize{$\mathrmbfit{a}$}}}
\put(62,57){\makebox(0,0){\shortstack{\scriptsize{$\pi$}\\\large{$\Longleftarrow$}}}}
\put(20,80){\vector(1,0){80}}
\put(10,67){\vector(3,-4){38}}
\put(110,68){\vector(-3,-4){38}}
\end{picture}
\end{tabular}}}
\end{center}
%\item 
A limiting cone for diagram
$\mathrmbf{I}\xrightarrow{\;\mathrmbfit{F}\;}\mathrmbf{A}$
is a universal cone:
for any other cone
$\mathrmbfit{F}\xLeftarrow{\;\,\pi'}\Delta(a')$,
there is a unique $\mathrmbf{A}$-morphism
$a\xleftarrow{f}a'$
such that $\Delta(f){\,\cdot\,}\pi = \pi'$.
%\end{itemize}
%
A context $\mathrmbf{A}$ is complete
when any diagram
$\mathrmbf{I}\xrightarrow{\;\mathrmbfit{F}\;}\mathrmbf{A}$
in $\mathrmbf{A}$
has a limiting cone
$\mathrmbfit{F}\xLeftarrow{\;\pi}\Delta(a)$.
The dual notions are colimiting cocones and cocompleteness.
%
%%%%%%%%%%%%%%%%%%%%%%%%%%%%%%%%%%%%%%%%%%%%%%%%%%%%%%%%%%%%%%%%%%%%%%%%%%%%%%%%
%%%%%%%%%%%%%%%%%%%%%%%%%%%%%%%%%%%%%%%%%%%%%%%%%%%%%%%%%%%%%%%%%%%%%%%%%%%%%%%%
\footnote{In any complete context,
the limits of arbitrary diagrams
can be constructed by using only 
the terminal object and (binary) pullbacks.
Dually,
in any cocomplete context,
the colimits of arbitrary diagrams
can be constructed by using only the initial object and (binary) pushouts.}
%%%%%%%%%%%%%%%%%%%%%%%%%%%%%%%%%%%%%%%%%%%%%%%%%%%%%%%%%%%%%%%%%%%%%%%%%%%%%%%%
%%%%%%%%%%%%%%%%%%%%%%%%%%%%%%%%%%%%%%%%%%%%%%%%%%%%%%%%%%%%%%%%%%%%%%%%%%%%%%%%
%
\begin{proposition}\label{prop:context:lim:colim}
For any cocomplete context $\mathrmbf{A}$
there is a passage
$\mathrmbf{A}^{\!\scriptscriptstyle{\Uparrow}}
\xrightarrow{\mathrmbfit{colim}}
\mathrmbf{A}$.
Dually,
For any complete context $\mathrmbf{A}$
there is a passage
$\mathrmbf{A}^{\!\scriptscriptstyle{\Downarrow}}
\xrightarrow{\mathrmbfit{lim}}
\mathrmbf{A}$.
\end{proposition}
\begin{proof}
We give an argument for limits in a complete context.
A dual argument holds for colimits in a cocomplete context.
The limit passage
$\mathrmbf{A}^{\!\scriptscriptstyle{\Downarrow}}\xrightarrow{\mathrmbfit{lim}}\mathrmbf{A}$
maps a diagram
${\langle{\mathrmbf{I},\mathrmbfit{F}}\rangle}$
to the limiting cone vertex $a$,
and maps a diagram morphism 
${\langle{\mathrmbf{B}_{2},\mathrmbfit{F}_{2}}\rangle}
\xleftarrow{\;{\langle{\mathrmbfit{G},\alpha}\rangle}\;}
{\langle{\mathrmbf{B}_{1},\mathrmbfit{F}_{1}}\rangle}$
to the unique morphism
$a_{2}\xleftarrow{g}a_{1}$
from the vertex of the cone
$\mathrmbfit{F}_{2}
\xLeftarrow{\;(\mathrmbfit{G}^{\mathrm{op}}\circ\pi_{1})\bullet\alpha\;}
\Delta(a_{1})$
to vertex of the limiting cone
$\mathrmbfit{F}_{2}\xLeftarrow{\;\pi_{2}}\Delta(a_{2})$ 
of diagram $\mathrmbf{I}_{2}\xrightarrow{\;\mathrmbfit{F}_{2}\;}\mathrmbf{A}$
such that $\Delta(g){\,\cdot\,}\pi_{2} = (\mathrmbfit{G}^{\mathrm{op}}\circ\pi_{1})\bullet\alpha$.
\hfill\rule{5pt}{5pt}
\end{proof}
%

%%%%%%%%%%%%%%%%%%%%%%%%%%%%%%%%%%%%%%%%%%%%%%%%%%%%%%%%%%%%%%%%%%%%%%%%%%%%%%%%%%%%%%%%%%
%\mbox{}\newline\rule{120pt}{1pt}{\fbox{\textbf{ Work Zone }}}\rule{120pt}{1pt}\newline
%%%%%%%%%%%%%%%%%%%%%%%%%%%%%%%%%%%%%%%%%%%%%%%%%%%%%%%%%%%%%%%%%%%%%%%%%%%%%%%%%%%%%%%%%%

%%%%%%%%%%%%%%%%%%%%%%%%%%%%%%%%%%%%%%%%%%%%%%%%%%%%%%%%%%%%
\newpage
\paragraph{Continuity/Cocontinuity.}
%%%%%%%%%%%%%%%%%%%%%%%%%%%%%%%%%%%%%%%%%%%%%%%%%%%%%%%%%%%%
%

%
\begin{definition}\label{def:contin:cocontin}
%\mbox{}
A passage 
$\mathrmbf{A}\xrightarrow{\;\mathrmbfit{H}\;}\mathrmbf{B}$ is \underline{continuous} when
it maps any limiting cone to a limiting cone.
%$\mathrmbfit{F}\xLeftarrow{\;\pi}\Delta(a)$
%$\Delta(a)\xRightarrow{\pi}\mathrmbfit{F}$
%for a diagram $\mathrmbf{I}\xrightarrow{\mathrmbfit{F}}\mathrmbf{A}$
%$\mathrmbfit{F}{\circ}\mathrmbfit{H}\xLeftarrow{\;\pi{\circ}\mathrmbfit{H}}\Delta(\mathrmbfit{H}(a))$. 
%$\Delta(\mathrmbfit{H}(a))\xRightarrow{\pi{\circ}\mathrmbfit{H}}\mathrmbfit{F}{\circ}\mathrmbfit{H}$.
Dually, 
it is cocontinuous when it maps any colimiting cocone to a colimiting cocone.
\end{definition}
\begin{proposition}\label{prop:passage:lim:colim}
For any continuous passage 
$\mathrmbf{A}\xrightarrow{\;\mathrmbfit{H}\;}\mathrmbf{B}$
between complete contexts $\mathrmbf{A}$ and $\mathrmbf{B}$,
the passage
$\mathrmbf{A}^{\!\scriptscriptstyle{\Downarrow}} 
\xrightarrow{\;\mathrmbfit{H}^{\!\scriptscriptstyle{\Downarrow}}\;}
\mathrmbf{B}^{\!\scriptscriptstyle{\Downarrow}}$
satisfies the condition
$\mathrmbfit{lim}{\;\circ\;}\mathrmbfit{H}
\cong
\mathrmbfit{H}^{\!\scriptscriptstyle{\Downarrow}}{\;\circ\;}\mathrmbfit{lim}$
expressing continuity.
Dually,
for any cocontinuous passage 
$\mathrmbf{A}\xrightarrow{\;\mathrmbfit{H}\;}\mathrmbf{B}$
between cocomplete contexts $\mathrmbf{A}$ and $\mathrmbf{B}$,
the passage
$\mathrmbf{A}^{\!\scriptscriptstyle{\Uparrow}}
\xrightarrow{\;\mathrmbfit{H}^{\!\scriptscriptstyle{\Uparrow}}\;}
\mathrmbf{B}^{\!\scriptscriptstyle{\Uparrow}}$
satisfies the condition
$\mathrmbfit{colim}{\;\circ\;}\mathrmbfit{H}
\cong
\mathrmbfit{H}^{\!\scriptscriptstyle{\Uparrow}}{\;\circ\;}\mathrmbfit{colim}$
expressing cocontinuity.
\begin{center}
{{\begin{tabular}{@{\hspace{20pt}}c@{\hspace{20pt}}c@{\hspace{20pt}}}
%%%%%%%%%%%%%%%%%%%%%%%%%%%%%%%%%%%%%%%%%%%%%%%%%%%%%%%%%%%%%%%%%%%%%%%%%%%%%%%%
{{\begin{tabular}{c}
\setlength{\unitlength}{0.45pt}
\begin{picture}(300,130)(-100,-10)
\put(-25,100){\makebox(0,0){\footnotesize{$
\mathrmbf{A}^{\!\scriptscriptstyle{\Downarrow}} 
\cong \bigl({(\mbox{-})}{\,\Downarrow\,}\mathrmbf{A}\bigr)$}}}
\put(-28,0){\makebox(0,0){\footnotesize{$
\mathrmbf{B}^{\!\scriptscriptstyle{\Downarrow}}
= \bigl({(\mbox{-})}{\,\Downarrow\,}\mathrmbf{B}\bigr)$}}}
\put(160,100){\makebox(0,0){\footnotesize{$\mathrmbf{A}$}}}
\put(165,0){\makebox(0,0){\footnotesize{$\mathrmbf{B}$}}}
\put(85,110){\makebox(0,0){\scriptsize{$\mathrmbfit{lim}$}}}
\put(85,-10){\makebox(0,0){\scriptsize{$\mathrmbfit{lim}$}}}
\put(-15,50){\makebox(0,0)[r]{\scriptsize{$\mathrmbfit{H}^{\!\scriptscriptstyle{\Downarrow}}
\cong \bigl({(\mbox{-})}{\,\Downarrow\,}\mathrmbfit{H}\bigr)$}}}
\put(170,50){\makebox(0,0)[l]{\scriptsize{$\mathrmbfit{H}$}}}
\put(45,100){\vector(1,0){90}}
\put(45,0){\vector(1,0){90}}
\put(0,85){\vector(0,-1){70}}
\put(160,85){\vector(0,-1){70}}
\put(80,50){\makebox(0,0){\footnotesize{$\cong$}}}
\end{picture}
\end{tabular}}}
%%%%%%%%%%%%%%%%%%%%%%%%%%%%%%%%%%%%%%%%%%%%%%%%%%%%%%%%%%%%%%%%%%%%%%%%%%%%%%%%
&
%%%%%%%%%%%%%%%%%%%%%%%%%%%%%%%%%%%%%%%%%%%%%%%%%%%%%%%%%%%%%%%%%%%%%%%%%%%%%%%%
{{\begin{tabular}{c}
\setlength{\unitlength}{0.45pt}
\begin{picture}(300,130)(-100,-10)
\put(-25,100){\makebox(0,0){\footnotesize{$
\mathrmbf{A}^{\!\scriptscriptstyle{\Uparrow}} 
\cong \bigl({(\mbox{-})}{\,\Uparrow\,}\mathrmbf{A}\bigr)$}}}
\put(-28,0){\makebox(0,0){\footnotesize{$
\mathrmbf{B}^{\!\scriptscriptstyle{\Uparrow}} 
= \bigl({(\mbox{-})}{\,\Uparrow\,}\mathrmbf{B}\bigr)$}}}
\put(160,100){\makebox(0,0){\footnotesize{$\mathrmbf{A}$}}}
\put(165,0){\makebox(0,0){\footnotesize{$\mathrmbf{B}$}}}
\put(85,110){\makebox(0,0){\scriptsize{$\mathrmbfit{colim}$}}}
\put(85,-10){\makebox(0,0){\scriptsize{$\mathrmbfit{colim}$}}}
\put(-15,50){\makebox(0,0)[r]{\scriptsize{$\mathrmbfit{H}^{\!\scriptscriptstyle{\Uparrow}}
\cong \bigl({(\mbox{-})}{\,\Uparrow\,}\mathrmbfit{H}\bigr)$}}}
\put(170,50){\makebox(0,0)[l]{\scriptsize{$\mathrmbfit{H}$}}}
\put(45,100){\vector(1,0){90}}
\put(45,0){\vector(1,0){90}}
\put(0,85){\vector(0,-1){70}}
\put(160,85){\vector(0,-1){70}}
\put(80,50){\makebox(0,0){\footnotesize{$\cong$}}}
\end{picture}
\end{tabular}}}
%%%%%%%%%%%%%%%%%%%%%%%%%%%%%%%%%%%%%%%%%%%%%%%%%%%%%%%%%%%%%%%%%%%%%%%%%%%%%%%%
\end{tabular}}}
\end{center}
\end{proposition}
\begin{proof}
%
%\begin{itemize}
%\item  
We give an argument for limits in a complete context:
%
%{\fbox{\bf{Does this need an argument or a proof?}}}
%$\mathrmbfit{H}(\mathrmbfit{lim}(\mathrmbfit{F})) = 
%\mathrmbfit{lim}(\mathrmbfit{F}{\,\circ\,}\mathrmbfit{H})$:
%\newline
%\begin{itemize}
%\item  
%
the passage 
$\mathrmbf{A}^{\!\scriptscriptstyle{\Downarrow}} 
\xrightarrow[{(\text{-})}{\,\circ\,}\mathrmbfit{H}]{\;\mathrmbfit{H}^{\!\scriptscriptstyle{\Downarrow}}\;}
\mathrmbf{B}^{\!\scriptscriptstyle{\Downarrow}}$
maps a limiting cone
${\langle{\mathrmbf{I},\mathrmbfit{F}}\rangle}
\xrightarrow{{\langle{\Delta,\pi}\rangle}}
{\langle{\mathrmbf{1},a}\rangle}$
in
$\mathrmbf{A}^{\!\scriptscriptstyle{\Downarrow}}$ 
with vertex (limit object)
$a = \mathrmbfit{lim}(\mathrmbfit{F}) \in \mathrmbf{A}$
to the limiting cone
${\langle{\mathrmbf{I},\mathrmbfit{F}{\circ}\mathrmbfit{H}}\rangle}
\xrightarrow{{\langle{\Delta,\pi{\circ}\mathrmbfit{H}}\rangle}}
{\langle{\mathrmbf{1},\mathrmbfit{H}(a)}\rangle}$
in
$\mathrmbf{B}^{\!\scriptscriptstyle{\Downarrow}}$
with vertex (limit object)
$\mathrmbfit{H}(a) = \mathrmbfit{lim}(\mathrmbfit{F}{\,\circ\,}\mathrmbfit{H}) \in
\mathrmbf{B}$.
%\newline
%\item  
%
A dual argument holds for colimits in a cocomplete context.
%\end{itemize}
\hfill\rule{5pt}{5pt}
\end{proof}

\comment{
that is,
maps a bridge
$\mathrmbfit{F}\xLeftarrow{\;\pi}\Delta(a)$
to the bridge
$\mathrmbfit{F}{\circ}\mathrmbfit{H}\xLeftarrow{\;\pi{\circ}\mathrmbfit{H}}\Delta(\mathrmbfit{H}(a))$. 
\item  
When $\mathrmbfit{H}$ is continuous,
it maps the limit
$a_{0} = \mathrmbfit{lim}(\mathrmbfit{F}) \in
\mathrmbf{A}$
to the limit 
$\mathrmbfit{H}(a_{0}) = \mathrmbfit{lim}(\mathrmbfit{F}{\,\circ\,}\mathrmbfit{H}) \in
\mathrmbf{B}$.
\item 
A cone is a morphism
${\langle{\mathrmbf{I},\mathrmbfit{F}}\rangle}
\xrightarrow{{\langle{\Delta,\pi}\rangle}}
{\langle{\mathrmbf{1},\mathrmbfit{a}}\rangle}$
in 
$\mathrmbf{A}^{\!\scriptscriptstyle{\Downarrow}} 
= \bigl({(\mbox{-})}^{\mathrm{op}}{\,\Downarrow\,}\mathrmbf{A}\bigr)$
consisting of a shape passage
$\mathrmbf{I}\xrightarrow{\;\Delta\;}\mathrmbf{1}$
and a bridge
$\mathrmbfit{I}\xLeftarrow{\;\pi\;}\Delta^{\mathrm{op}}{\;\circ\;}\mathrmbfit{a}$.
\item 
A limiting cone for diagram
$\mathrmbf{I}\xrightarrow{\;\mathrmbfit{F}\;}\mathrmbf{A}$
is a universal cone:
for any other cone
$\mathrmbfit{F}\xLeftarrow{\;\,\pi'}\Delta(a')$,
there is a unique $\mathrmbf{A}$-morphism
$a\xleftarrow{f}a'$
such that $\Delta(f){\,\cdot\,}\pi = \pi'$.
}

%%%%%%%%%%%%%%%%%%%%%%%%%%%%%%%%%%%%%%%%%%%%%%%%%%%%%%%%%%%%%%%%%%%%%%%%%%%%%%%%%%%%%%%%%%
\newpage
\subsubsection{Kan Extensions}\label{append:kan:ext}
%%%%%%%%%%%%%%%%%%%%%%%%%%%%%%%%%%%%%%%%%%%%%%%%%%%%%%%%%%%%%%%%%%%%%%%%%%%%%%%%
%%%%%%%%%%%%%%%%%%%%%%%%%%%%%%%%%%%%%%%%%%%%%%%%%%%%%%%%%%%%%%%%%%%%%%%%%%%%%%%%
\footnote{Kan extensions are specific cases of the Grothendieck construction:
they are used to describe the general contexts of 
schemed domains $\mathrmbf{DOM}$ and databases $\mathrmbf{DB}$.
The more specialized contexts
$\mathring{\mathrmbf{Dom}}$ and $\mathrmbf{Db}$
use the more general notion of Grothendieck construction.}
%%%%%%%%%%%%%%%%%%%%%%%%%%%%%%%%%%%%%%%%%%%%%%%%%%%%%%%%%%%%%%%%%%%%%%%%%%%%%%%%
%%%%%%%%%%%%%%%%%%%%%%%%%%%%%%%%%%%%%%%%%%%%%%%%%%%%%%%%%%%%%%%%%%%%%%%%%%%%%%%%
%%%%%%%%%%%%%%%%%%%%%%%%%%%%%%%%%%%%%%%%%%%%%%%%%%%%%%%%%%%%%%%%%%%%%%%%%%%%%%%%%%%%%%%%%%

%
\begin{definition}\label{kan-extensions}
Given a passage $\mathrmbfit{K} : \mathrmbf{C}_{2} \rightarrow \mathrmbf{C}_{1}$ 
and a context $\mathrmbf{A}$,
if $\mathrmbf{A}^\mathrmbf{C} = [\mathrmbf{C},\mathrmbf{A}]$ is the hom context, 
then
the passage
$\mathrmbf{A}^\mathrmbfit{K} : \mathrmbf{A}^{\mathrmbf{C}_{2}} \leftarrow \mathrmbf{A}^{\mathrmbf{C}_{1}}$
is defined by composition
$(\mathrmbfit{S}_{1}\xRightarrow{\,\alpha\;}\mathrmbfit{S}'_{1}) 
\mapsto 
(\mathrmbfit{K}{\;\circ\;}\mathrmbfit{S}_{1}
\xRightarrow{\;\mathrmbfit{K}{\;\circ\;}\alpha\;}
\mathrmbfit{K}{\;\circ\;}\mathrmbfit{S}_{1}')$.
\begin{itemize}
%
%\item[$\mathrmbf{DOM}:$]  
\item
Given a passage $\mathrmbfit{K} : \mathrmbf{C}_{2} \rightarrow \mathrmbf{C}_{1}$ and 
an element 
%(passage) 
$\mathrmbfit{S}_{2} \in \mathrmbf{A}^{\mathrmbf{C}_{2}}$,
a \emph{left Kan extension} of $\mathrmbfit{S}_{2}$ along $\mathrmbfit{K}$ 
%(Mac~Lane\cite{maclane:71})
%``Categories for the Working Mathematician''
is a pair $(\mathrmbfit{L},\eta)$
consisting of 
an element $\mathrmbfit{L} \in \mathrmbf{A}^{\mathrmbf{C}_{1}}$ and 
a bridge 
$\mathrmbfit{S}_{2}
\xRightarrow{\;\eta\;\,}
\mathrmbfit{K}{\;\circ\;}\mathrmbfit{L}$, 
%which is universal as an morphism from 
%$\mathrmbf{A}^\mathrmbfit{K} : \mathrmbf{A}^{\mathrmbf{C}_{2}} \leftarrow \mathrmbf{A}^{\mathrmbf{C}_{1}}$
%to $\mathrmbfit{S}_{2} \in \mathrmbf{A}^{\mathrmbf{C}_{2}}$.
%
%Universality determines the passage $\mathrmbfit{R}$ uniquely up to isomorphism.
%
%In detail,
%given an 
such that
for each 
element $\mathrmbfit{S}_{1} \in \mathrmbf{A}^{\mathrmbf{C}_{1}}$
%,
%\newline
%for each 
and bridge
$\mathrmbfit{S}_{2}
\xRightarrow{\;\alpha\;\,} 
\mathrmbf{A}^\mathrmbfit{K}(\mathrmbfit{S}_{1}) = \mathrmbfit{K} \circ \mathrmbfit{S}_{1}$,
%\newline
there is a unique bridge
$\mathrmbfit{L}
%=\mathrmbfit{ran}_{\mathrmbfit{K}}(\mathrmbfit{S}_{2})
\xRightarrow{\;\beta\;\,}
\mathrmbfit{S}_{1}$ 
with
\newline\mbox{}\hfill
{\footnotesize{$
\bigl(
\mathrmbfit{S}_{2}
\xRightarrow{\;\alpha\;\,} 
\mathrmbf{A}^\mathrmbfit{K}(\mathrmbfit{S}_{1})
\bigr)
{\;=\;}
\bigl(
\mathrmbfit{S}_{2}
\xRightarrow{\;\eta\;\,}
\mathrmbf{A}^\mathrmbfit{K}(\mathrmbfit{L})
\xRightarrow{\;\,\mathrmbf{A}^\mathrmbfit{K}(\beta)\;}
\mathrmbf{A}^\mathrmbfit{K}(\mathrmbfit{S}_{1})
\bigr)
$}}
\hfill\mbox{}
\begin{center}
%{{\begin{tabular}{c}
%\\\\
%
%{{\begin{tabular}{cc}
{{\begin{tabular}[b]{c}
\setlength{\unitlength}{0.56pt}
\begin{picture}(120,90)(-40,0)
\put(0,80){\makebox(0,0){\footnotesize{$\mathrmbf{C}_{2}$}}}
\put(120,80){\makebox(0,0){\footnotesize{$\mathrmbf{C}_{1}$}}}
\put(60,0){\makebox(0,0){\footnotesize{$\mathrmbf{A}$}}}
\put(60,88){\makebox(0,0){\scriptsize{$\mathrmbfit{K}$}}}
\put(24,42){\makebox(0,0)[r]{\scriptsize{$\mathrmbfit{S}_{2}$}}}
\put(78,57){\makebox(0,0)[r]{\scriptsize{$\mathrmbfit{L}$}}}
\put(106,39){\makebox(0,0)[l]{\scriptsize{$\mathrmbfit{S}_{1}$}}}
\put(50,63){\makebox(0,0){\normalsize{$\xRightarrow{\;\eta\;\;}$}}}
\put(65,40){\makebox(0,0){\normalsize{$\xRightarrow{\;\;\alpha\;\;\,}$}}}
\put(92,51){\makebox(0,0){\normalsize{$\xRightarrow{\beta\;}$}}}
\put(20,80){\vector(1,0){80}}
\put(10,67){\vector(3,-4){38}}
\qbezier(106,71)(72,56)(66,21)\put(66,21){\vector(-1,-3){0}}
\qbezier(114,68)(106,32)(74,15)\put(74,15){\vector(-2,-1){0}}
\put(-100,40){\makebox(0,0){\footnotesize{$\alpha = \eta{\;\bullet\;}(\mathrmbfit{K}{\;\circ\;}\beta)$}}}
\end{picture}
\end{tabular}}}
%\end{tabular}}}
\end{center}
We use the notation $\mathrmbfit{L}=\mathrmbfit{lan}_{\mathrmbfit{K}}(\mathrmbfit{S}_{2})$.
\newline
\item  
%\item[$\mathrmbf{DB}:$]  
%\begin{definition}
The (dual) \emph{right Kan extension} 
$\mathrmbfit{ran}_{\mathrmbfit{K}}(\mathrmbfit{S}_{2})$
is defined by reversing the bridges.
Given 
a passage $\mathrmbfit{K} : \mathrmbf{C}_{2} \rightarrow \mathrmbf{C}_{1}$ and 
an element 
%(passage) 
$\mathrmbfit{S}_{2} \in \mathrmbf{A}^{\mathrmbf{C}_{2}}$,
a \emph{right Kan extension} of $\mathrmbfit{S}_{2}$ along $\mathrmbfit{K}$ 
%(Mac~Lane\cite{maclane:71})
%``Categories for the Working Mathematician''
is a pair $(\mathrmbfit{R},\varepsilon)$
consisting of 
an element $\mathrmbfit{R} \in \mathrmbf{A}^{\mathrmbf{C}_{1}}$ and 
a bridge 
$\mathrmbfit{S}_{2}
\xLeftarrow{\;\,\varepsilon\;}
\mathrmbfit{K}{\;\circ\;}\mathrmbfit{R}$, 
%which is universal as an morphism from 
%$\mathrmbf{A}^\mathrmbfit{K} : \mathrmbf{A}^{\mathrmbf{C}_{2}} \leftarrow \mathrmbf{A}^{\mathrmbf{C}_{1}}$
%to $\mathrmbfit{S}_{2} \in \mathrmbf{A}^{\mathrmbf{C}_{2}}$.
%
%Universality determines the passage $\mathrmbfit{R}$ uniquely up to isomorphism.
%
%In detail,
%given an 
such that
for each 
element $\mathrmbfit{S}_{1} \in \mathrmbf{A}^{\mathrmbf{C}_{1}}$
%,
%\newline
%for each 
and bridge
$\mathrmbfit{S}_{2}
\xLeftarrow{\;\,\alpha\;} 
\mathrmbf{A}^\mathrmbfit{K}(\mathrmbfit{S}_{1}) = \mathrmbfit{K} \circ \mathrmbfit{S}_{1}$,
%\newline
there is a unique bridge
$\mathrmbfit{R}
%=\mathrmbfit{ran}_{\mathrmbfit{K}}(\mathrmbfit{S}_{2})
\xLeftarrow{\;\,\beta\;}
\mathrmbfit{S}_{1}$ 
with
\newline\mbox{}\hfill
{\footnotesize{$
\bigl(
\mathrmbfit{S}_{2}
\xLeftarrow{\;\,\alpha\;} 
\mathrmbf{A}^\mathrmbfit{K}(\mathrmbfit{S}_{1})
\bigr)
{\;=\;}
\bigl(
\mathrmbfit{S}_{2}
\xLeftarrow{\;\,\varepsilon\;}
\mathrmbf{A}^\mathrmbfit{K}(
\mathrmbfit{R}
%\mathrmbfit{ran}_{\mathrmbfit{K}}(\mathrmbfit{S}_{2})
)
\xLeftarrow{\;\,\mathrmbf{A}^\mathrmbfit{K}(\beta)\;}
\mathrmbf{A}^\mathrmbfit{K}(\mathrmbfit{S}_{1})
\bigr)
$}}
\hfill\mbox{}
\begin{center}
%{{\begin{tabular}{c}
%\\\\
%
%{{\begin{tabular}{cc}
{{\begin{tabular}[b]{c}
\setlength{\unitlength}{0.56pt}
\begin{picture}(120,90)(-40,0)
\put(0,80){\makebox(0,0){\footnotesize{$\mathrmbf{C}_{2}$}}}
\put(120,80){\makebox(0,0){\footnotesize{$\mathrmbf{C}_{1}$}}}
\put(60,0){\makebox(0,0){\footnotesize{$\mathrmbf{A}$}}}
\put(60,88){\makebox(0,0){\scriptsize{$\mathrmbfit{K}$}}}
\put(24,42){\makebox(0,0)[r]{\scriptsize{$\mathrmbfit{S}_{2}$}}}
\put(78,57){\makebox(0,0)[r]{\scriptsize{$\mathrmbfit{R}$}}}
\put(106,39){\makebox(0,0)[l]{\scriptsize{$\mathrmbfit{S}_{1}$}}}
\put(50,63){\makebox(0,0){\normalsize{$\xLeftarrow{\;\;\varepsilon\;}$}}}
\put(65,40){\makebox(0,0){\normalsize{$\xLeftarrow{\;\;\,\alpha\;\;}$}}}
\put(92,52){\makebox(0,0){\normalsize{$\xLeftarrow{\;\beta}$}}}
\put(20,80){\vector(1,0){80}}
\put(10,67){\vector(3,-4){38}}
\qbezier(106,71)(72,56)(66,21)\put(66,21){\vector(-1,-3){0}}
\qbezier(114,68)(106,32)(74,15)\put(74,15){\vector(-2,-1){0}}
\put(-100,40){\makebox(0,0){\footnotesize{$\alpha = (\mathrmbfit{K}{\;\circ\;}\beta){\;\bullet\;}\epsilon$}}}
\end{picture}
\end{tabular}}}
%\end{tabular}}}
\end{center}
We use the notation $\mathrmbfit{R}=\mathrmbfit{ran}_{\mathrmbfit{K}}(\mathrmbfit{S}_{2})$.
\end{itemize}
\end{definition}

% Kan extensions and base adjunctions
%%%%%%%%%%%%%%%%%%%%%%%%%%%%%%%%%%%%%%%%%%%%%%%%%%%%%%%%%%%%%%%%%%%%%%%%%%%%%%%%
%%%%%%%%%%%%%%%%%%%%%%%%%%%%%%%%%%%%%%%%%%%%%%%%%%%%%%%%%%%%%%%%%%%%%%%%%%%%%%%%
%

\newpage

\begin{fact}\label{fact:rt:kan:adj}
(Mac~Lane\cite{maclane:71})
\begin{itemize}
\item 
%\item[$\mathrmbf{DOM}$:] 
For any cocomplete context $\mathrmbf{A}$
and any passage $\mathrmbf{C}_{2}\xrightarrow{\;\mathrmbfit{K}\;\,}\mathrmbf{C}_{1}$, 
the left Kan extension passage
$\mathrmbf{A}^{\mathrmbf{C}_{2}}
\xrightarrow{\;\,\mathrmbfit{lan}_{\mathrmbfit{K}}\;}
\mathrmbf{A}^{\mathrmbf{C}_{1}}$ 
is left adjoint
to the composition passage
$\mathrmbf{A}^{\mathrmbf{C}_{2}}
\xleftarrow{\;\mathrmbf{A}^{\mathrmbfit{K}}\;\,}
\mathrmbf{A}^{\mathrmbf{C}_{1}}$.
Since the left Kan operation is preserved under composition
$\mathrmbfit{lan}_{(\mathrmbfit{K}_{2}{\;\circ\;}\mathrmbfit{K}_{1})}{\;\cong\;}
\mathrmbfit{lan}_{\mathrmbfit{K}_{2}}{\;\circ\;}\mathrmbfit{lan}_{\mathrmbfit{K}_{1}}$,
this forms a fiber (an indexed) adjunction
$\mathrmbf{Cxt}\xrightarrow{\tilde{\mathrmbf{A}}}\mathrmbf{Adj}$
with local adjoint pair
$\mathrmbf{A}^{\mathrmbf{C}_{2}}
\xrightarrow[{\langle{\mathrmbfit{lan}_{\mathrmbfit{K}}{\;\dashv\;}\mathrmbf{A}^{\mathrmbfit{K}}}\rangle}]
{{\langle{\acute{\mathrmbf{A}}{\;\dashv\;}\grave{\mathrmbf{A}}}\rangle}}
\mathrmbf{A}^{\mathrmbf{C}_{1}}$.
\newline
\item
%\item[$\mathrmbf{DB}$:] 
For any complete context $\mathrmbf{A}$
and any passage $\mathrmbf{C}_{2}\xrightarrow{\;\mathrmbfit{K}\;\,}\mathrmbf{C}_{1}$, 
the right Kan extension passage
$\mathrmbf{A}^{\mathrmbf{C}_{2}}
\xrightarrow{\;\,\mathrmbfit{ran}_{\mathrmbfit{K}}\;}
\mathrmbf{A}^{\mathrmbf{C}_{1}}$ 
is right adjoint
to the composition passage
$\mathrmbf{A}^{\mathrmbf{C}_{2}}
\xleftarrow{\;\mathrmbf{A}^{\mathrmbfit{K}}\;\,}
\mathrmbf{A}^{\mathrmbf{C}_{1}}$.
Since the right Kan operation is preserved under composition
$\mathrmbfit{ran}_{(\mathrmbfit{K}_{2}{\;\circ\;}\mathrmbfit{K}_{1})}{\;\cong\;}
\mathrmbfit{ran}_{\mathrmbfit{K}_{2}}{\;\circ\;}\mathrmbfit{ran}_{\mathrmbfit{K}_{1}}$,
this forms a fiber (an indexed) adjunction
$\mathrmbf{Cxt}^{\mathrm{op}}\!\xrightarrow{\hat{\mathrmbf{A}}}\mathrmbf{Adj}$
with local adjoint pair
$\mathrmbf{A}^{\mathrmbf{C}_{2}}
\xleftarrow[{\langle{\mathrmbf{A}^{\mathrmbfit{K}}{\;\dashv\;}\mathrmbfit{ran}_{\mathrmbfit{K}}}\rangle}]
{{\langle{\acute{\mathrmbf{A}}{\;\dashv\;}\grave{\mathrmbf{A}}}\rangle}}
\mathrmbf{A}^{\mathrmbf{C}_{1}}$.
\end{itemize}
\end{fact}
%

%%%%%%%%%%%%%%%%%%%%%%%%%%%%%%%%%%%%%%%%%%%%%%%%%%%%%%%%%%%%%%%%%%%%%%%%%%%%%%%%
%%%%%%%%%%%%%%%%%%%%%%%%%%%%%%%%%%%%%%%%%%%%%%%%%%%%%%%%%%%%%%%%%%%%%%%%%%%%%%%%
\comment{% redundant
\fbox{This defines a bifibration
$\int{\!{\mathrmbfit{A}}}\xrightarrow{\;\;}\mathrmbf{Cxt}$.}
%
%%%%%%%%%%%%%%%%%%%%%%%%%%%%%%%%%%%%%%%%%%%%%%%%%%%%%%%%%%%%%%%%%%%%%%%%%%%%%%%%
%%%%%%%%%%%%%%%%%%%%%%%%%%%%%%%%%%%%%%%%%%%%%%%%%%%%%%%%%%%%%%%%%%%%%%%%%%%%%%%%
\footnote{
A bifibration $\int{\!{\mathrmbfit{C}}}\xrightarrow{\;\;}\mathrmbf{I}$
is the Grothendieck construction of an indexed adjunction
$\mathrmbf{I}\xrightarrow{{\mathrmbfit{C}}}\mathrmbf{Adj}$,
consisting of 
a left adjoint covariant pseudo-passage
$\mathrmbf{I}\xrightarrow{\acute{\mathrmbfit{C}}}\mathrmbf{Cxt}$
and a right adjoint contravariant pseudo-passage
$\mathrmbf{I}^{\mathrm{op}}\xrightarrow{\grave{\mathrmbfit{C}}}\mathrmbf{Cxt}$. }
%%%%%%%%%%%%%%%%%%%%%%%%%%%%%%%%%%%%%%%%%%%%%%%%%%%%%%%%%%%%%%%%%%%%%%%%%%%%%%%%
%%%%%%%%%%%%%%%%%%%%%%%%%%%%%%%%%%%%%%%%%%%%%%%%%%%%%%%%%%%%%%%%%%%%%%%%%%%%%%%%
\newline
}% redundant
%%%%%%%%%%%%%%%%%%%%%%%%%%%%%%%%%%%%%%%%%%%%%%%%%%%%%%%%%%%%%%%%%%%%%%%%%%%%%%%%
%%%%%%%%%%%%%%%%%%%%%%%%%%%%%%%%%%%%%%%%%%%%%%%%%%%%%%%%%%%%%%%%%%%%%%%%%%%%%%%%

\newpage

\begin{proposition}\label{prop:lax:fibered:A}\mbox{}
\begin{itemize}
\item
%[$\mathrmbf{DOM}$:] 
By Fact\,\ref{fact:rt:kan:adj},
the diagram (lax comma)
context over a cocomplete context $\mathrmbf{A}$
is the bifibration 
%fibered context 
(Grothendieck construction) 
$\mathrmbf{A}^{\scriptscriptstyle{\Uparrow}} = \int\hat{\mathrmbf{A}}$ 
of the indexed adjunction
$\mathrmbf{Cxt}\xrightarrow{\hat{\mathrmbf{A}}}\mathrmbf{Adj}$.
A morphism
${\langle{\mathrmbf{C}_{2},\mathrmbfit{S}_{2}}\rangle} 
\xrightarrow{{\langle{\mathrmbfit{K},\,\hat{\alpha}}\rangle}}
{\langle{\mathrmbf{C}_{1},\mathrmbfit{S}_{1}}\rangle}$
in $\mathrmbf{A}^{\scriptscriptstyle{\Uparrow}}$
consists of
a passage $\mathrmbf{C}_{2}\xrightarrow{\;\mathrmbfit{K}\;\,}\mathrmbf{C}_{1}$ and 
a pair $\hat{\alpha} = {\langle{\acute{\alpha},\grave{\alpha}}\rangle}$ of equivalent bridges.
\begin{equation}
%\label{def:tbl:cxt}
{{\begin{picture}(120,10)(0,-4)
\put(60,0){\makebox(0,0){\footnotesize{$
\underset{\textstyle{\text{in}\;\mathrmbf{A}^{\mathrmbf{C}_{2}}}}
{\mathrmbfit{S}_{2}\xRightarrow[\;\mathbf{levo}\;]{\;\acute{\alpha}\;\,}
\mathrmbf{A}^\mathrmbfit{K}(\mathrmbfit{S}_{1})=
\grave{\mathrmbf{A}}(\mathrmbfit{S}_{1})}
{\;\;\;\;\;\;\;\;\rightleftarrows\;\;\;\;\;\;\;\;}
\underset{\textstyle{\text{in}\;\mathrmbf{A}^{\mathrmbf{C}_{1}}}}
{\acute{\mathrmbf{A}}(\mathrmbfit{S}_{2})
=\mathrmbfit{lan}_{\mathrmbfit{K}}(\mathrmbfit{S}_{2})
\xRightarrow[\mathbf{dextro}]{\;\grave{\alpha}\;\,}
\mathrmbfit{S}_{1}}
$}}}
\end{picture}}}
\end{equation}
%
%\begin{figure}
\begin{center}
{{\begin{tabular}{c}
\\
\begin{tabular}{c@{\hspace{20pt}}c}
{{\begin{tabular}[b]{c}
\setlength{\unitlength}{0.56pt}
\begin{picture}(120,80)(8,0)
\put(0,80){\makebox(0,0){\footnotesize{$\mathrmbf{C}_{2}$}}}
\put(120,80){\makebox(0,0){\footnotesize{$\mathrmbf{C}_{1}$}}}
\put(60,0){\makebox(0,0){\footnotesize{$\mathrmbf{A}$}}}
\put(60,92){\makebox(0,0){\scriptsize{$\mathrmbfit{K}$}}}
\put(20,42){\makebox(0,0)[r]{\scriptsize{$\mathrmbfit{S}_{2}$}}}
\put(100,42){\makebox(0,0)[l]{\scriptsize{$\mathrmbfit{S}_{1}$}}}
\put(60,57){\makebox(0,0){\shortstack{\scriptsize{$\acute{\alpha}$}\\\large{$\Longrightarrow$}}}}
\put(20,80){\vector(1,0){80}}
\put(10,67){\vector(3,-4){38}}
\put(110,68){\vector(-3,-4){38}}
\end{picture}
\end{tabular}}}
%%%%%%%%%%
&
%%%%%%%%%%
{{\begin{tabular}[b]{c}
\setlength{\unitlength}{0.56pt}
\begin{picture}(120,80)(8,0)
\put(0,80){\makebox(0,0){\footnotesize{$\mathrmbf{C}_{2}$}}}
\put(120,80){\makebox(0,0){\footnotesize{$\mathrmbf{C}_{1}$}}}
\put(60,0){\makebox(0,0){\footnotesize{$\mathrmbf{A}$}}}
\put(60,92){\makebox(0,0){\scriptsize{$\mathrmbfit{lan}_{\mathrmbfit{K}}(\mathrmbfit{S}_{2})$}}}
\put(20,42){\makebox(0,0)[r]{\scriptsize{$\mathrmbfit{S}_{2}$}}}
\put(100,42){\makebox(0,0)[l]{\scriptsize{$\mathrmbfit{S}_{1}$}}}
\put(60,57){\makebox(0,0){\shortstack{\scriptsize{$\grave{\alpha}$}\\\large{$\Longrightarrow$}}}}
%
%\put(100,80){\vector(-1,0){80}}
\qbezier(105,80)(-10,90)(55,20)\put(55,20){\vector(3,-4){0}}
\put(10,67){\vector(3,-4){38}}
\put(110,68){\vector(-3,-4){38}}
\end{picture}
\end{tabular}}}
\\&\\
%{\scriptsize{$\mathrmbfit{S}_{2}
%\xRightarrow[\;\mathbf{levo}\;]{\;\acute{\alpha}\;\,}
%\acute{\mathrmbfit{R}}{\,\circ\,}\mathrmbfit{S}_{1}
%=\grave{\mathrmbfit{list}}_{\hat{\mathrmbfit{R}}}(\mathrmbfit{S}_{1})$}}
%&
%{\scriptsize{$\acute{\mathrmbfit{list}}_{\hat{\mathrmbfit{R}}}(\mathrmbfit{S}_{2})=
%\grave{\mathrmbfit{R}}{\,\circ\,}\mathrmbfit{S}_{2}
%\xRightarrow[\mathbf{dextro}]{\;\grave{\alpha}\;\,} 
%\mathrmbfit{S}_{1}$}}
%\\&\\
%\textbf{levo} & \textbf{dextro}
%\\&\\
%
\end{tabular}
%%%%%%%%%%%%%%%%%%%%%%%%%%%%%%%%%%%%%%%%%%%%%%%%%%%%%%%%%%%%
\\
%%%%%%%%%%%%%%%%%%%%%%%%%%%%%%%%%%%%%%%%%%%%%%%%%%%%%%%%%%%%
{\scriptsize\setlength{\extrarowheight}{4pt}$\begin{array}{|@{\hspace{5pt}}l@{\hspace{25pt}}l@{\hspace{5pt}}|}
%\multicolumn{1}{c}{\text{\bfseries levo}}
%& 
%\multicolumn{1}{c}{\text{\bfseries dextro}} 
%\\ 
\hline
\acute{\alpha}:
\mathrmbfit{S}_{2}\Rightarrow
%\grave{\mathrmbf{A}}(\mathrmbfit{S}_{1})=
\mathrmbf{A}^\mathrmbfit{K}(\mathrmbfit{S}_{1})
&
\grave{\alpha}:
\mathrmbfit{lan}_{\mathrmbfit{K}}(\mathrmbfit{S}_{2})
%=\acute{\mathrmbf{A}}(\mathrmbfit{S}_{2})
\Rightarrow\mathrmbfit{S}_{1}
\\
\acute{\alpha} =
\eta_{\mathrmbfit{S}_{2}}
\bullet
\mathrmbf{A}^\mathrmbfit{K}(\grave{\alpha})
&
\grave{\alpha} =
\mathrmbfit{lan}_{\mathrmbfit{K}}(\acute{\alpha})
\bullet
\varepsilon_{\mathrmbfit{S}_{1}}
\\\hline
\end{array}$}
\end{tabular}}}
\end{center}
%\caption{Schema Morphism}
%\label{fig:schema:mor:adj}
%\end{figure}
%
\item
%[$\mathrmbf{DB}$:] 
%
By Fact\,\ref{fact:rt:kan:adj},
the diagram (oplax comma) context over a complete context $\mathrmbf{A}$
is the bifibration 
%fibered context 
(Grothendieck construction) 
$\mathrmbf{A}^{\scriptscriptstyle{\Downarrow}} = \int\hat{\mathrmbf{A}}$ 
of the 
indexed adjunction
$\mathrmbf{Cxt}^{\mathrm{op}}\!\xrightarrow{\hat{\mathrmbf{A}}}\mathrmbf{Adj}$.
A morphism
${\langle{\mathrmbf{C}_{2},\mathrmbfit{S}_{2}}\rangle} 
\xleftarrow{{\langle{\mathrmbfit{K},\,\hat{\alpha}}\rangle}}
{\langle{\mathrmbf{C}_{1},\mathrmbfit{S}_{1}}\rangle}$
in $\mathrmbf{A}^{\scriptscriptstyle{\Downarrow}}$
consists of
a passage $\mathrmbf{C}_{2}\xrightarrow{\;\mathrmbfit{K}\;\,}\mathrmbf{C}_{1}$ and 
a pair $\hat{\alpha} = {\langle{\acute{\alpha},\grave{\alpha}}\rangle}$ of equivalent bridges.
\begin{equation}
%\label{def:tbl:cxt}
{{\begin{picture}(120,10)(0,-4)
\put(60,0){\makebox(0,0){\footnotesize{$
\underset{\textstyle{\text{in}\;\mathrmbf{A}^{\mathrmbf{C}_{2}}}}
{\mathrmbfit{S}_{2}
\xLeftarrow[\;\mathbf{levo}\;]{\;\,\acute{\alpha}\;}
%\mathrmbfit{K}^{\mathrm{op}}\!{\circ\;}\mathrmbfit{S}_{1}=
\mathrmbf{A}^\mathrmbfit{K}(\mathrmbfit{S}_{1})
=\acute{\mathrmbf{A}}(\mathrmbfit{S}_{1})}
{\;\;\;\;\;\;\;\;\rightleftarrows\;\;\;\;\;\;\;\;}
\underset{\textstyle{\text{in}\;\mathrmbf{A}^{\mathrmbf{C}_{1}}}}
{\grave{\mathrmbf{A}}(\mathrmbfit{S}_{2})
=\mathrmbfit{ran}_{\mathrmbfit{K}}(\mathrmbfit{S}_{2})
\xLeftarrow[\mathbf{dextro}]{\;\,\grave{\alpha}\;}
\mathrmbfit{S}_{1}}
$}}}
\end{picture}}}
\end{equation}
%
%\begin{figure}
\begin{center}
{{\begin{tabular}{c}
\\
\begin{tabular}{c@{\hspace{20pt}}c}
{{\begin{tabular}[b]{c}
\setlength{\unitlength}{0.56pt}
\begin{picture}(120,80)(8,0)
\put(5,80){\makebox(0,0){\footnotesize{$\mathrmbf{C}_{2}^{\mathrm{op}}$}}}
\put(125,80){\makebox(0,0){\footnotesize{$\mathrmbf{C}_{1}^{\mathrm{op}}$}}}
\put(65,0){\makebox(0,0){\footnotesize{$\mathrmbf{A}$}}}
\put(65,92){\makebox(0,0){\scriptsize{$\mathrmbfit{K}^{\mathrm{op}}$}}}
\put(20,42){\makebox(0,0)[r]{\scriptsize{$\mathrmbfit{S}_{2}$}}}
\put(100,42){\makebox(0,0)[l]{\scriptsize{$\mathrmbfit{S}_{1}$}}}
\put(60,57){\makebox(0,0){\shortstack{\scriptsize{$\acute{\alpha}$}\\\large{$\Longleftarrow$}}}}
\put(20,80){\vector(1,0){80}}
\put(10,67){\vector(3,-4){38}}
\put(110,68){\vector(-3,-4){38}}
\end{picture}
\end{tabular}}}
%%%%%%%%%%
&
%%%%%%%%%%
{{\begin{tabular}[b]{c}
\setlength{\unitlength}{0.56pt}
\begin{picture}(120,80)(8,0)
\put(5,80){\makebox(0,0){\footnotesize{$\mathrmbf{C}_{2}^{\mathrm{op}}$}}}
\put(125,80){\makebox(0,0){\footnotesize{$\mathrmbf{C}_{1}^{\mathrm{op}}$}}}
\put(65,0){\makebox(0,0){\footnotesize{$\mathrmbf{A}$}}}
\put(65,92){\makebox(0,0){\scriptsize{$\mathrmbfit{ran}_{\mathrmbfit{K}}(\mathrmbfit{S}_{2})$}}}
\put(20,42){\makebox(0,0)[r]{\scriptsize{$\mathrmbfit{S}_{2}$}}}
\put(100,42){\makebox(0,0)[l]{\scriptsize{$\mathrmbfit{S}_{1}$}}}
\put(60,57){\makebox(0,0){\shortstack{\scriptsize{$\grave{\alpha}$}\\\large{$\Longleftarrow$}}}}
%
%\put(100,80){\vector(-1,0){80}}
\qbezier(105,80)(-10,90)(55,20)\put(55,20){\vector(3,-4){0}}
\put(10,67){\vector(3,-4){38}}
\put(110,68){\vector(-3,-4){38}}
\end{picture}
\end{tabular}}}
\\&\\
%{\scriptsize{$\mathrmbfit{S}_{2}
%\xLeftarrow[\mathbf{levo}]{\;\,\acute{\alpha}\;}
%\mathrmbf{A}^\mathrmbfit{K}(\mathrmbfit{S}_{1})$}}
%&
%{\scriptsize{$
%\mathrmbfit{ran}_{\mathrmbfit{K}}(\mathrmbfit{S}_{2})
%\xLeftarrow[\mathbf{dextro}]{\;\,\grave{\alpha}\;} 
%\mathrmbfit{S}_{1}$}}
%\\&\\
%\textbf{levo} & \textbf{dextro}
%\\&\\
%
\end{tabular}
%%%%%%%%%%%%%%%%%%%%%%%%%%%%%%%%%%%%%%%%%%%%%%%%%%%%%%%%%%%%
\\
%%%%%%%%%%%%%%%%%%%%%%%%%%%%%%%%%%%%%%%%%%%%%%%%%%%%%%%%%%%%
{\scriptsize\setlength{\extrarowheight}{4pt}$\begin{array}{|@{\hspace{5pt}}l@{\hspace{25pt}}l@{\hspace{5pt}}|}
%\multicolumn{1}{c}{\text{\bfseries levo}}
%& 
%\multicolumn{1}{c}{\text{\bfseries dextro}} 
%\\ 
\hline
\acute{\alpha} : 
\mathrmbfit{S}_{2}\Leftarrow
\mathrmbf{A}^\mathrmbfit{K}(\mathrmbfit{S}_{1})
%=\mathrmbfit{K}^{\mathrm{op}}\!{\circ\;}\mathrmbfit{S}_{1}
&
\grave{\alpha} : 
%\grave{\mathrmbfit{K}}^{\mathrm{op}}\!{\circ\;}\mathrmbfit{S}_{2}
\mathrmbfit{ran}_{\mathrmbfit{K}}(\mathrmbfit{S}_{2})
\Leftarrow
\mathrmbfit{S}_{1}
\\
\acute{\alpha} =
\mathrmbf{A}^\mathrmbfit{K}(\grave{\alpha})
\bullet
\varepsilon_{\mathrmbfit{S}_{2}}
&
\grave{\alpha} =
\eta_{\mathrmbfit{S}_{1}}
\bullet
\mathrmbfit{ran}_{\mathrmbfit{K}}(\acute{\alpha})
%(\underset{\textstyle{\grave{\mathrmbfit{db}}_{\hat{\mathrmbf{R}}}(\acute{\alpha})}}
%{\underbrace{\grave{\mathrmbfit{K}}^{\mathrm{op}}\!{\circ\;}\acute{\alpha}}})
\\\hline
\end{array}$}
\end{tabular}}}
\end{center}
%\caption{Database Morphism}
%\label{fig:db:mor:adj}
%\end{figure}
%
\end{itemize}
\end{proposition}
%
%{\fbox{\bf{We do not necessarily need the opposite notation here.}}}

\newpage

\begin{proposition}\label{prop:lim:colim:A}\mbox{}
\begin{itemize}
\item
%[$\mathrmbf{DOM}$:]  
For any complete and cocomplete context $\mathrmbf{A}$,
the fibered context (Grothendieck construction) 
$\mathrmbf{A}^{\scriptscriptstyle{\Uparrow}} = \int\hat{\mathrmbf{A}}$ is complete and cocomplete 
and the projection 
$\mathrmbf{A}^{\scriptscriptstyle{\Uparrow}}\rightarrow\mathrmbf{Cxt}
:{\langle{\mathrmbf{C},\mathrmbfit{S}}\rangle}\mapsto\mathrmbf{C}$ 
is continuous and cocontinuous.
\newline
\item
%[$\mathrmbf{DB}$:]  
For any complete and cocomplete context $\mathrmbf{A}$,
the fibered context (Grothendieck construction) 
$\mathrmbf{A}^{\scriptscriptstyle{\Downarrow}} = \int\hat{\mathrmbf{A}}$ is complete and cocomplete 
and the projection 
$\mathrmbf{A}^{\scriptscriptstyle{\Downarrow}}\rightarrow\mathrmbf{Cxt}
:{\langle{\mathrmbf{C},\mathrmbfit{S}}\rangle}\mapsto\mathrmbf{C}$ 
is continuous and cocontinuous.
\end{itemize}
\end{proposition}
\begin{proof}
Use Prop.~\ref{prop:lax:fibered:A}
and
Fact.~\ref{fact:groth:adj:lim:colim}
of \S\,\ref{append:grothen:construct},
since 
the indexing context $\mathrmbf{Cxt}$ is complete and cocomplete, and
the fiber context $\mathrmbf{A}^{\mathrmbf{C}}$ is complete and cocomplete 
for each $\mathrmbf{C}\in\mathrmbf{Cxt}$.
\end{proof}
%

%%%%%%%%%%%%%%%%%%%%%%%%%%%%%%%%%%%%%%%%%%%%%%%%%%%%%%%%%%%%%%%%%%%%%%%%%%%%%%%%%%%%%%%%%%
%%%%%%%%%%%%%%%%%%%%%%%%%%%%%%%%%%%%%%%%%%%%%%%%%%%%%%%%%%%%%%%%%%%%%%%%%%%%%%%%%%%%%%%%%%
\newpage
\subsection{Database Applications}\label{sub:sec:append:db:app}
%%%%%%%%%%%%%%%%%%%%%%%%%%%%%%%%%%%%%%%%%%%%%%%%%%%%%%%%%%%%%%%%%%%%%%%%%%%%%%%%%%%%%%%%%%
%%%%%%%%%%%%%%%%%%%%%%%%%%%%%%%%%%%%%%%%%%%%%%%%%%%%%%%%%%%%%%%%%%%%%%%%%%%%%%%%%%%%%%%%%%

%%%%%%%%%%%%%%%%%%%%%%%%%%%%%%%%%%%%%%%%%%%%%%%%%%%%%%%%%%%%%%%%%%%%%%%%%%%%%%%%%%%%%%%%%%
%\newpage
\subsubsection{Contexts}\label{sub:sub:sec:math:context}
%%%%%%%%%%%%%%%%%%%%%%%%%%%%%%%%%%%%%%%%%%%%%%%%%%%%%%%%%%%%%%%%%%%%%%%%%%%%%%%%%%%%%%%%%%

%
Figure \ref{fig:diag:comma:cxts}
shows six important contexts in this paper:
%\begin{itemize}
%\item 
two comma contexts
%{\footnotesize{$\mathrmbf{List}$}}, 
{\footnotesize{$\mathrmbf{Dom}$}} and 
{\footnotesize{$\mathrmbf{Tbl}$}};
%, and {\footnotesize{$\widehat{\mathrmbf{DB}}$}};
%\item 
two diagram 
%(lax comma) 
contexts
%{\footnotesize{$\mathrmbf{LIST}$}}, 
{\footnotesize{$\mathrmbf{DOM}$}} and 
{\footnotesize{$\mathrmbf{DB}$}};
%
%\end{itemize}
%
%The diagram contexts use the diagram contexts:
%%\newline
%$\mathrmbf{Set}^{\!\scriptscriptstyle{\Uparrow}}
%= \bigl(\mathrmbf{Cxt}{\,\Uparrow\,}\mathrmbf{Set}\bigr)$
%and
%$\mathrmbf{SET}
%= \mathrmbf{Set}^{\!\scriptscriptstyle{\Downarrow}}
%=     \bigl({(\mbox{-})}^{\mathrm{op}}{\,\Downarrow\,}\mathrmbf{Set}\bigr)
%\cong \bigl(\mathrmbf{Cxt}{\,\Downarrow\,}\mathrmbf{Set}\bigr)$.
%
%%%%%%%%%%%%%%%%%%%%%%%%%%%%%%%%%%%%%%%%%%%%%%%%%%%%%%%%%%%%%%%%%%%%%%%%%%%%%%%%%%%%%%%%%%
%%%%%%%%%%%%%%%%%%%%%%%%%%%%%%%%%%%%%%%%%%%%%%%%%%%%%%%%%%%%%%%%%%%%%%%%%%%%%%%%%%%%%%%%%%
\footnote{For any context $\mathrmbf{C}$,
the ``super-comma'' context 
$\bigl(\mathrmbf{Cxt}{\,\Downarrow\,}\mathrmbf{C}\bigr)
=\mathrmbf{C}^{\!\scriptscriptstyle{\Downarrow}}$
is defined \cite{maclane:71} as follows:
%\newline
%(1) 
an object is a $\mathrmbf{C}$-diagram ${\langle{\mathrmbf{I},\mathrmbfit{D}}\rangle}$
with indexing context $\mathrmbf{I}$ and passage $\mathrmbf{I}\xrightarrow{\mathrmbfit{D}}\mathrmbf{C}$;
%\newline
%(2) 
a morphism is a $\mathrmbf{C}$-diagram morphism 
${\langle{\mathrmbf{I}_{2},\mathrmbfit{D}_{2}}\rangle}
\xrightarrow{\langle{\mathrmbfit{F},\alpha}\rangle}
{\langle{\mathrmbf{I}_{1},\mathrmbfit{D}_{1}}\rangle}$
with indexing passage $\mathrmbf{I}_{2}\xrightarrow{\mathrmbfit{F}}\mathrmbf{I}_{1}$
and bridge $\mathrmbfit{D}_{2}\xLeftarrow{\;\alpha\;}\mathrmbfit{F}{\;\circ\;}\mathrmbfit{D}_{1}$.
To dualize,
the context 
$\bigl(\mathrmbf{Cxt}{\,\Uparrow\,}\mathrmbf{C}\bigr)
=\mathrmbf{C}^{\!\scriptscriptstyle{\Uparrow}}$
is defined as follows:
%\newline
%(1) 
an object is as above;
%\newline
%(2) 
a morphism is a $\mathrmbf{C}$-diagram morphism 
${\langle{\mathrmbf{I}_{2},\mathrmbfit{D}_{2}}\rangle}
\xrightarrow{\langle{\mathrmbfit{F},\alpha}\rangle}
{\langle{\mathrmbf{I}_{1},\mathrmbfit{D}_{1}}\rangle}$
with indexing passage $\mathrmbf{I}_{2}\xrightarrow{\mathrmbfit{F}}\mathrmbf{I}_{1}$
and bridge $\mathrmbfit{D}_{2}\xRightarrow{\;\alpha\;}\mathrmbfit{F}{\;\circ\;}\mathrmbfit{D}_{1}$.}
%%%%%%%%%%%%%%%%%%%%%%%%%%%%%%%%%%%%%%%%%%%%%%%%%%%%%%%%%%%%%%%%%%%%%%%%%%%%%%%%%%%%%%%%%%
%%%%%%%%%%%%%%%%%%%%%%%%%%%%%%%%%%%%%%%%%%%%%%%%%%%%%%%%%%%%%%%%%%%%%%%%%%%%%%%%%%%%%%%%%%
%
and 
%\item 
two fibered contexts defined by the Grothendieck construction:
$\mathring{\mathrmbf{Dom}}$ and $\mathrmbf{Db}$.
Since 
%$\mathrmbf{List}$, $\mathrmbf{Cls}$, 
$\mathrmbf{Dom}$ and $\mathrmbf{Tbl}$
are (co)complete mathematical contexts,
%schemas, type domain diagrams, schemed domains and databases 
by Prop.\,\ref{prop:context:lim:colim}
(also see the paper 
''The \texttt{FOLE} Table'' 
\cite{kent:fole:era:tbl})
we have the passages
indicated in 
Tbl.\,\ref{tbl:colim:lim:pass}.
%
%{\fbox{$\mathring{\mathrmbf{Dom}}$ and $\mathrmbf{Db}$ are Grothendieck constructions.}}
Tbl.~\ref{tbl:sch:dom:morphs} and
Tbl.~\ref{tbl:rel:db:morphs}
display the definitions for schemed domain and relational database morphisms.

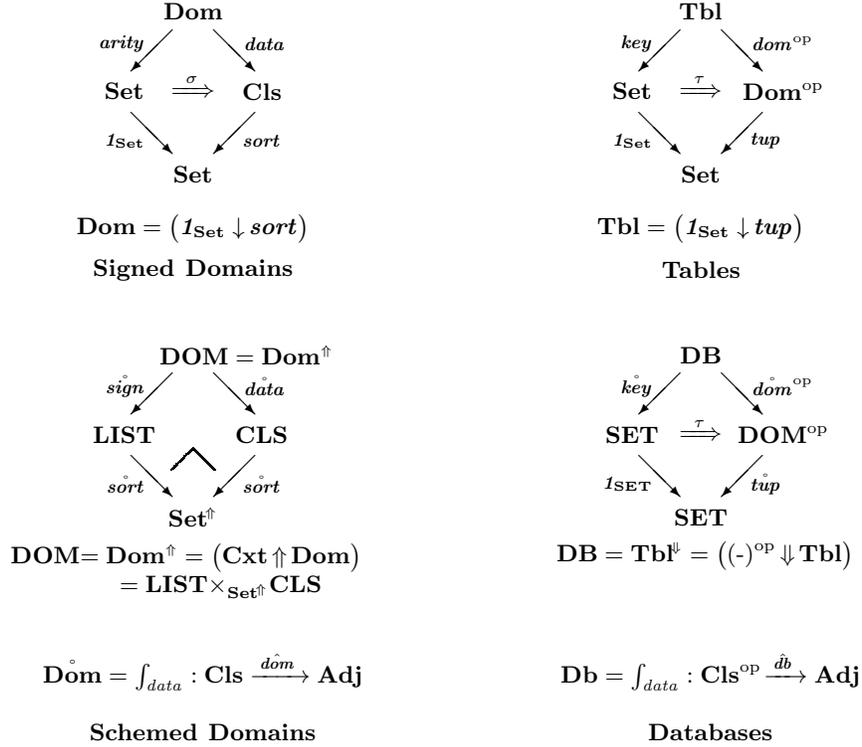
\begin{figure}
\begin{center}
%{{\begin{tabular}{c@{\hspace{45pt}}c@{\hspace{45pt}}c}
{{\begin{tabular}{
%c@{\hspace{45pt}}
c@{\hspace{20pt}}@{\hspace{55pt}}c}
\comment{{\begin{tabular}{c}
\setlength{\unitlength}{0.52pt}
\begin{picture}(100,180)(-40,-70)
\put(0,118){\makebox(0,0){\footnotesize{$\mathrmbf{List}$}}}
\put(50,60){\makebox(0,0){\footnotesize{$\mathrmbf{Set}$}}}
\put(-50,60){\makebox(0,0){\footnotesize{$\mathrmbf{Set}$}}}
\put(0,0){\makebox(0,0){\footnotesize{$\mathrmbf{Set}$}}}
\put(-35,95){\makebox(0,0)[r]{\scriptsize{$\mathrmbfit{arity}$}}}
\put(-35,24){\makebox(0,0)[r]{\scriptsize{$\mathrmbfit{1}_{\mathrmbf{Set}}$}}}
\put(37,95){\makebox(0,0)[l]{\scriptsize{$\mathrmbfit{sort}$}}}
\put(36,24){\makebox(0,0)[l]{\scriptsize{$\mathrmbfit{1}_{\mathrmbf{Set}}$}}}
\put(0,58){\makebox(0,0){{$\xRightarrow{\;\,\alpha\,}$}}}
\put(-15,105){\vector(-1,-1){30}}
\put(15,105){\vector(1,-1){30}}
\put(-45,45){\vector(1,-1){30}}
\put(45,45){\vector(-1,-1){30}}
%\qbezier(-12,105)(-60,60)(-12,15)\put(-12,15){\vector(1,-1){0}}
\put(0,-40){\makebox(0,0){\footnotesize{$\mathrmbf{List}
= \bigl(\mathrmbfit{1}_{\mathrmbf{Set}}{\;\downarrow\,}\mathrmbfit{1}_{\mathrmbf{Set}}\bigr)$}}}
\put(0,-70){\makebox(0,0){\footnotesize{$\mathrmbf{Signatures}$}}}
\end{picture}
\end{tabular}}}
%
%%%%%%%%%%%%%%%%%%%%%%%%%%%%%%%%%%%%%%%%%%%%%%%%%%%%%%%%%%%%%%%%%%%%%%%%%%%%%%%%
%&
%%%%%%%%%%%%%%%%%%%%%%%%%%%%%%%%%%%%%%%%%%%%%%%%%%%%%%%%%%%%%%%%%%%%%%%%%%%%%%%%
%
{{\begin{tabular}{c}
\setlength{\unitlength}{0.52pt}
\begin{picture}(100,180)(-40,-70)
\put(0,118){\makebox(0,0){\footnotesize{$\mathrmbf{Dom}$}}}
\put(50,60){\makebox(0,0){\footnotesize{$\mathrmbf{Cls}$}}}
\put(-50,60){\makebox(0,0){\footnotesize{$\mathrmbf{Set}$}}}
\put(0,0){\makebox(0,0){\footnotesize{$\mathrmbf{Set}$}}}
\put(-35,95){\makebox(0,0)[r]{\scriptsize{$\mathrmbfit{arity}$}}}
\put(-35,24){\makebox(0,0)[r]{\scriptsize{$\mathrmbfit{1}_{\mathrmbf{Set}}$}}}
\put(37,95){\makebox(0,0)[l]{\scriptsize{$\mathrmbfit{data}$}}}
\put(36,24){\makebox(0,0)[l]{\scriptsize{$\mathrmbfit{sort}$}}}
\put(0,58){\makebox(0,0){{$\xRightarrow{\;\,\sigma\,}$}}}
\put(-15,105){\vector(-1,-1){30}}
\put(15,105){\vector(1,-1){30}}
\put(-45,45){\vector(1,-1){30}}
\put(45,45){\vector(-1,-1){30}}
%\qbezier(-12,105)(-60,60)(-12,15)\put(-12,15){\vector(1,-1){0}}
\put(0,-40){\makebox(0,0){\footnotesize{$\mathrmbf{Dom}
= \bigl(\mathrmbfit{1}_{\mathrmbf{Set}}{\;\downarrow\,}\mathrmbfit{sort}\bigr)$}}}
\put(0,-70){\makebox(0,0){\footnotesize{$\mathrmbf{Signed}\,\;\mathrmbf{Domains}$}}}
\end{picture}
\end{tabular}}}
%
%%%%%%%%%%%%%%%%%%%%%%%%%%%%%%%%%%%%%%%%%%%%%%%%%%%%%%%%%%%%%%%%%%%%%%%%%%%%%%%%
&
%%%%%%%%%%%%%%%%%%%%%%%%%%%%%%%%%%%%%%%%%%%%%%%%%%%%%%%%%%%%%%%%%%%%%%%%%%%%%%%%
%
{{\begin{tabular}{c}
\setlength{\unitlength}{0.52pt}
\begin{picture}(100,180)(-40,-70)
\put(0,118){\makebox(0,0){\footnotesize{$\mathrmbf{Tbl}$}}}
\put(60,60){\makebox(0,0){\footnotesize{$\mathrmbf{Dom}^{\mathrm{op}}$}}}
\put(-50,60){\makebox(0,0){\footnotesize{$\mathrmbf{Set}$}}}
\put(0,0){\makebox(0,0){\footnotesize{$\mathrmbf{Set}$}}}
\put(-35,95){\makebox(0,0)[r]{\scriptsize{$\mathrmbfit{key}$}}}
\put(-35,24){\makebox(0,0)[r]{\scriptsize{$\mathrmbfit{1}_{\mathrmbf{Set}}$}}}
\put(37,95){\makebox(0,0)[l]{\scriptsize{$\mathrmbfit{dom}^{\mathrm{op}}$}}}
\put(36,24){\makebox(0,0)[l]{\scriptsize{$\mathrmbfit{tup}$}}}
\put(0,58){\makebox(0,0){{$\xRightarrow{\;\,\tau\,}$}}}
\put(-15,105){\vector(-1,-1){30}}
\put(15,105){\vector(1,-1){30}}
\put(-45,45){\vector(1,-1){30}}
\put(45,45){\vector(-1,-1){30}}
%\qbezier(-12,105)(-60,60)(-12,15)\put(-12,15){\vector(1,-1){0}}
\put(0,-40){\makebox(0,0){\footnotesize{$\mathrmbf{Tbl}
= \bigl(\mathrmbfit{1}_{\mathrmbf{Set}}{\;\downarrow\,}\mathrmbfit{tup}\bigr)$}}}
\put(0,-70){\makebox(0,0){\footnotesize{$\mathrmbf{Tables}$}}}
\end{picture}
\end{tabular}}}
%
%%%%%%%%%%%%%%%%%%%%%%%%%%%%%%%%%%%%%%%%%%%%%%%%%%%%%%%%%%%%%%%%%%%%%%%%%%%%%%%%
%%%%%%%%%%%%%%%%%%%%%%%%%%%%%%%%%%%%%%%%%%%%%%%%%%%%%%%%%%%%%%%%%%%%%%%%%%%%%%%%
\\\\
\\\\
{{\begin{tabular}{c}
\setlength{\unitlength}{0.52pt}
\begin{picture}(100,180)(-40,-70)
%\put(0,118){\makebox(0,0){\footnotesize{$\mathrmbf{DOM}$}}}
\put(-24.5,118){\makebox(0,0)[l]{\footnotesize{$\mathrmbf{DOM}=\mathrmbf{Dom}^{\scriptscriptstyle{\Uparrow}}$}}}
\put(50,60){\makebox(0,0){\footnotesize{$\mathrmbf{CLS}$}}}
\put(-50,60){\makebox(0,0){\footnotesize{$\mathrmbf{LIST}$}}}
\put(0,0){\makebox(0,0){\footnotesize{$\mathrmbf{Set}^{\!\scriptscriptstyle{\Uparrow}}$}}}
\put(-35,95){\makebox(0,0)[r]{\scriptsize{$\mathring{\mathrmbfit{sign}}$}}}
\put(-35,24){\makebox(0,0)[r]{\scriptsize{$\mathring{\mathrmbfit{sort}}$}}}
\put(37,95){\makebox(0,0)[l]{\scriptsize{$\mathring{\mathrmbfit{data}}$}}}
\put(36,24){\makebox(0,0)[l]{\scriptsize{$\mathring{\mathrmbfit{sort}}$}}}
%
%\put(0,58){\makebox(0,0){{$\xRightarrow{\;\,\mathring{\sigma}\,}$}}}
\qbezier(0,50)(-8,42)(-16,34)
\qbezier(0,50)(8,42)(16,34)
%\qbezier(0,30)(10,30)(20,30)\qbezier(0,30)(0,20)(0,10)
%
\put(-15,105){\vector(-1,-1){30}}
\put(15,105){\vector(1,-1){30}}
\put(-45,45){\vector(1,-1){30}}
\put(45,45){\vector(-1,-1){30}}
%\qbezier(-12,105)(-60,60)(-12,15)\put(-12,15){\vector(1,-1){0}}
\put(-5,-40){\makebox(0,0){\footnotesize{$
\mathrmbf{DOM}
%= \bigl(\mathrmbf{Cxt}{\,\Uparrow\,}\mathrmbf{Dom}\bigr)
%\mathrmbf{DOM}
%= \mathrmbf{LIST}{\times_{\mathrmbf{Set}^{\!\scriptscriptstyle{\Uparrow}}}}\mathrmbf{CLS}
\underset{\textstyle{
= \mathrmbf{LIST}{\times_{\mathrmbf{Set}^{\!\scriptscriptstyle{\Uparrow}}}}\mathrmbf{CLS}}}
{{
= \mathrmbf{Dom}^{\scriptscriptstyle{\Uparrow}}
%\hspace{40pt}
= \bigl(\mathrmbf{Cxt}{\,\Uparrow\,}\mathrmbf{Dom}\bigr)
}}
%\mathrmbf{DOM}
%= \bigl(\mathrmbfit{1}_{\mathrmbf{SET}}{\;\downarrow\,}\mathring{\mathrmbfit{sort}}\bigr)
$}}}
%
%\put(0,-75){\makebox(0,0){\footnotesize{$\mathrmbf{Schemed}\,\;\mathrmbf{Domains}$}}}
%
\end{picture}
\end{tabular}}}
%
%%%%%%%%%%%%%%%%%%%%%%%%%%%%%%%%%%%%%%%%%%%%%%%%%%%%%%%%%%%%%%%%%%%%%%%%%%%%%%%%
&
%%%%%%%%%%%%%%%%%%%%%%%%%%%%%%%%%%%%%%%%%%%%%%%%%%%%%%%%%%%%%%%%%%%%%%%%%%%%%%%%
%
{{\begin{tabular}{c}
\setlength{\unitlength}{0.52pt}
\begin{picture}(100,180)(-40,-70)
\put(0,118){\makebox(0,0){\footnotesize{$\mathrmbf{DB}$}}}
%\put(-35,122){\vector(1,0){20}}
%\put(-20,114){\vector(-1,0){20}}
%\put(-40,126){\makebox(0,0){\footnotesize{$
%\xhookrightarrow{\;\,{inc}\,}$}}}
%\put(-40,110){\makebox(0,0){{
%\xleftarrow{\;\;\;\;\;\;\;}$}}}
%\put(-40,110){\makebox(0,0){\scriptsize{${db}$}}}
%\put(0,118){\makebox(0,0){\footnotesize{$\widehat{\mathrmbf{DB}}$}}}
%\put(-15.5,118){\makebox(0,0)[l]{\footnotesize{$\widehat{\mathrmbf{DB}}=\mathrmbf{Tbl}^{\scriptscriptstyle{\Downarrow}}$}}}
\put(60,60){\makebox(0,0){\footnotesize{$\mathrmbf{DOM}^{\mathrm{op}}$}}}
\put(-50,60){\makebox(0,0){\footnotesize{$\mathrmbf{SET}$}}}
\put(0,0){\makebox(0,0){\footnotesize{$\mathrmbf{SET}$}}}
\put(-35,95){\makebox(0,0)[r]{\scriptsize{$\mathring{\mathrmbfit{key}}$}}}
\put(-35,24){\makebox(0,0)[r]{\scriptsize{$\mathrmbfit{1}_{\mathrmbf{SET}}$}}}
\put(37,95){\makebox(0,0)[l]{\scriptsize{$\mathring{\mathrmbfit{dom}}^{\mathrm{op}}$}}}
\put(36,24){\makebox(0,0)[l]{\scriptsize{$\mathring{\mathrmbfit{tup}}$}}}
\put(0,58){\makebox(0,0){{$\xRightarrow{\;\,{\tau}\,}$}}}
%\put(0,58){\makebox(0,0){{$\xRightarrow{\;\,\mathring{\tau}\,}$}}}
%\put(-28,118){\makebox(0,0){\footnotesize{$\xhookrightarrow{\;\,\,}$}}}
%\put(-55,105){\vector(0,-1){30}}
%\put(-40,105){\vector(2,-1){70}}
\put(-15,105){\vector(-1,-1){30}}
\put(15,105){\vector(1,-1){30}}
\put(-45,45){\vector(1,-1){30}}
\put(45,45){\vector(-1,-1){30}}
%\qbezier(-12,105)(-60,60)(-12,15)\put(-12,15){\vector(1,-1){0}}
%
\put(-105,-28){\makebox(0,0)[l]{\footnotesize{$
\mathrmbf{DB} = \mathrmbf{Tbl}^{\!\scriptscriptstyle{\Downarrow}}
=     \bigl({(\mbox{-})}^{\mathrm{op}}{\,\Downarrow\,}\mathrmbf{Tbl}\bigr)
%\cong \bigl(\mathrmbf{Cxt}{\,\Downarrow\,}\mathrmbf{Tbl}\bigr)
$}}}
%
%\put(-65,-48){\makebox(0,0)[l]{\footnotesize{$
%\widehat{\mathrmbf{DB}}
%= \bigl(\mathrmbfit{1}_{\mathrmbf{SET}}{\;\downarrow\,}\mathring{\mathrmbfit{tup}}\bigr)$}}}
%
%\put(0,-75){\makebox(0,0){\footnotesize{$\mathrmbf{Databases}$}}}
%
\end{picture}
\end{tabular}}}
\\ & \\
$\mathring{\mathrmbf{Dom}} = \int_\mathrmit{data}:
\mathrmbf{Cls}\xrightarrow{\,\hat{\mathrmbfit{dom}}\;}\mathrmbf{Adj}$ 
& 
$\mathrmbf{Db} = \int_\mathrmit{data}:
\mathrmbf{Cls}^{\mathrm{op}}\!\xrightarrow{\,\hat{\mathrmbfit{db}}\;}\mathrmbf{Adj}$
\\ & \\
  {\footnotesize{$\mathrmbf{Schemed}\,\;\mathrmbf{Domains}$}}
& {\footnotesize{$\mathrmbf{Databases}$}}
\end{tabular}}}
\end{center}
\caption{Diagram/Comma/Fibered Contexts}
\label{fig:diag:comma:cxts}
\end{figure}

% not needed
%%%%%%%%%%%%%%%%%%%%%%%%%%%%%%%%%%%%%%%%%%%%%%%%%%%%%%%%%%%%%%%%%%%%%%%%%%%%%%%%
%%%%%%%%%%%%%%%%%%%%%%%%%%%%%%%%%%%%%%%%%%%%%%%%%%%%%%%%%%%%%%%%%%%%%%%%%%%%%%%%
%

%%%%%%%%%%%%%%%%%%%%%%%%%%%%%%%%%%%%%%%%%%%%%%%%%%%%%%%%%%%%%%%%%%%%%%%%%%%%%%%%%%%%%%%%%%
%%%%%%%%%%%%%%%%%%%%%%%%%%%%%%%%%%%%%%%%%%%%%%%%%%%%%%%%%%%%%%%%%%%%%%%%%%%%%%%%%%%%%%%%%%
%
%\newpage
%\subsubsection{Morphisms}
%\label{sub:sec:morphisms}
%%%%%%%%%%%%%%%%%%%%%%%%%%%%%%%%%%%%%%%%%%%%%%%%%%%%%%%%%%%%%%%%%%%%%%%%%%%%%%%%%%%%%%%%%%
%%%%%%%%%%%%%%%%%%%%%%%%%%%%%%%%%%%%%%%%%%%%%%%%%%%%%%%%%%%%%%%%%%%%%%%%%%%%%%%%%%%%%%%%%%

%
\begin{table}
\begin{center}
{\fbox{\scriptsize{\setlength{\extrarowheight}{2pt}
{\begin{tabular}
%\begin{tabular}
{|@{\hspace{3pt}}l@{\hspace{10pt}}c@{\hspace{10pt}}l@{\hspace{5pt}}|}
\hline
%one 
&
${\langle{\mathrmbf{R}_{2},\mathrmbfit{Q}_{2}}\rangle} 
\xrightarrow{{\langle{\mathrmbfit{R},\,\varsigma}\rangle}}
{\langle{\mathrmbf{R}_{1},\mathrmbfit{Q}_{1}}\rangle}$
&
$\mathrmbf{DOM}$
\\\hline\hline
%two 
&
${\langle{\mathrmbf{R}_{2},\mathrmbfit{S}_{2}}\rangle} 
\xrightarrow{\;{\langle{\mathrmbfit{R},\varphi}\rangle}\;}
{\langle{\mathrmbf{R}_{1},\mathrmbfit{S}_{1}}\rangle}$
&
$\mathring{\mathrmbf{Dom}}(\mathcal{A})$
% \cong \mathrmbf{LIST}(X)$
%\\\hline
%&
%{\footnotesize{$
%{\langle{\mathrmbf{R}_{2},\mathrmbfit{K}_{2}}\rangle}
%\xleftarrow{{\langle{\mathrmbfit{R},\kappa}\rangle}}
%{\langle{\mathrmbf{R}_{1},\mathrmbfit{K}_{1}}\rangle}$}}
%&
%$\mathrmbf{SET}$
\\\hline\hline
&
{\footnotesize{$
%\mathcal{R}_{2}=
{\langle{\mathrmbf{R}_{2},\mathrmbfit{S}_{2},\mathcal{A}_{2}}\rangle}
\xrightarrow{{\langle{\mathrmbfit{R},\hat{\varphi},f,g}\rangle}}
{\langle{\mathrmbf{R}_{1},\mathrmbfit{S}_{1},\mathcal{A}_{1}}\rangle}
%=\mathcal{R}_{1}
$}}
&
$\mathring{\mathrmbf{Dom}}$
\\\hline
\end{tabular}}}}}
\end{center}
\caption{Schemed Domain Morphisms}
\label{tbl:sch:dom:morphs}
\end{table}
\begin{table}
\begin{center}
{\fbox{\scriptsize{\setlength{\extrarowheight}{2pt}
{\begin{tabular}
%\begin{tabular}
{|@{\hspace{3pt}}l@{\hspace{10pt}}c@{\hspace{10pt}}l@{\hspace{5pt}}|}
\hline
%one 
&
${\langle{\mathrmbf{R}_{2},\mathrmbfit{T}_{2}}\rangle} 
\xleftarrow{{\langle{\mathrmbfit{R},\,\xi}\rangle}}
{\langle{\mathrmbf{R}_{1},\mathrmbfit{T}_{1}}\rangle}$
&
$\mathrmbf{DB}$
\\\hline
&
{\footnotesize{$
%\mathcal{R}_{2}=
{\langle{\mathrmbf{R}_{2},\mathrmbfit{Q}_{2},\mathrmbfit{K}_{2},\tau_{2}}\rangle} 
\xleftarrow{{\langle{\mathrmbfit{R},\varsigma,\kappa}\rangle}}
{\langle{\mathrmbf{R}_{1},\mathrmbfit{Q}_{1},\mathrmbfit{K}_{1},\tau_{1}}\rangle}
%=\mathcal{R}_{1}
$}}
&
(proj)
\\\hline\hline
%two 
&
${\langle{\mathrmbf{R}_{2},\mathrmbfit{T}_{2}}\rangle} 
\xleftarrow{\;{\langle{\mathrmbfit{R},\psi}\rangle}\;}
{\langle{\mathrmbf{R}_{1},\mathrmbfit{T}_{1}}\rangle}$
&
$\mathrmbf{Db}(\mathcal{A})$
\\\hline
&
{\footnotesize{$
%\mathcal{R}_{2}=
{\langle{\mathrmbf{R}_{2},\mathrmbfit{S}_{2},\mathrmbfit{K}_{2},\tau_{2}}\rangle}
\xleftarrow{{\langle{\mathrmbfit{R},\varphi,\kappa}\rangle}}
{\langle{\mathrmbf{R}_{1},\mathrmbfit{S}_{1},\mathrmbfit{K}_{1},\tau_{1}}\rangle}
%=\mathcal{R}_{1}
$}}
&
(proj)
\\\hline\hline
%two 
&
${\langle{\mathrmbf{R}_{2},\mathrmbfit{T}_{2},\mathcal{A}_{2}}\rangle} 
\xleftarrow{\;{\langle{\mathrmbfit{R},\hat{\psi},f,g}\rangle}\;}
{\langle{\mathrmbf{R}_{1},\mathrmbfit{T}_{1},\mathcal{A}_{1}}\rangle}$
&
$\mathrmbf{Db}$
\\\hline
&
{\footnotesize{$
%\mathcal{R}_{2}=
{\langle{\mathrmbf{R}_{2},\mathrmbfit{S}_{2},\mathcal{A}_{2},\mathrmbfit{K}_{2},\tau_{2}}\rangle}
\xleftarrow{{\langle{\mathrmbfit{R},\hat{\varphi},f,g,\hat{\kappa}}\rangle}}
{\langle{\mathrmbf{R}_{1},\mathrmbfit{S}_{1},\mathcal{A}_{1},\mathrmbfit{K}_{1},\tau_{1}}\rangle}
%=\mathcal{R}_{1}
$}}
&
(proj)
\end{tabular}}}}}
\end{center}
\caption{Relational Database Morphisms}
\label{tbl:rel:db:morphs}
\end{table}
%
%\vspace{-20pt}

.

%%%%%%%%%%%%%%%%%%%%%%%%%%%%%%%%%%%%%%%%%%%%%%%%%%%%%%%%%%%%%%%%%%%%%%%%%%%%%%%%%%%%%%%%%%
\newpage
\subsubsection{Passages}\label{sub:sub:sec:pass}
%%%%%%%%%%%%%%%%%%%%%%%%%%%%%%%%%%%%%%%%%%%%%%%%%%%%%%%%%%%%%%%%%%%%%%%%%%%%%%%%%%%%%%%%%%

%\begin{proposition}\label{prop:proj:pass:cocts}
%
%The database projection passages (LHS) Tbl.~\ref{defs:proj:pass} are cocontinuous.
%
%\end{proposition}
%\begin{proof}
%\newpage
The table projection passages (RHS) below
%Tbl.~\ref{defs:proj:pass} 
are cocontinuous \cite{kent:fole:era:tbl},
either by using fibrations ($\scriptstyle{\int}$)
or by using comma contexts ($\scriptstyle{\downarrow}$).
%Those labeled $\scriptstyle{\int}$ are cocontinuous by a result using fibrations.
%Those labeled $\scriptstyle{\downarrow}$ are cocontinuous by a result using comma contexts.
%See the paper ''The \texttt{FOLE} Table'' \cite{kent:fole:era:tbl}.
%\hfill\rule{5pt}{5pt}
%\end{proof}
%
\begin{table}
\begin{center}
{{\footnotesize{\begin{tabular}{|@{\hspace{10pt}}r@{\hspace{20pt}}l@{\hspace{10pt}}l@{\hspace{10pt}}|}
\hline
$\mathrmbf{LIST}=
\mathrmbf{List}^{\!\scriptscriptstyle{\Uparrow}} 
%= \bigl({(\mbox{-})}^{\mathrm{op}}{\,\Downarrow\,}\mathrmbf{List}\bigr)
\xrightarrow[{(\text{-})}{\,\circ\,}\mathrmbfit{sort}]{\;\mathrmbfit{sort}^{\!\scriptscriptstyle{\Uparrow}}\;}
%\bigl({(\mbox{-})}^{\mathrm{op}}{\,\Downarrow\,}\mathrmbf{\mathrmbf{Set}}\bigr) = 
\mathrmbf{\mathrmbf{Set}}^{\!\scriptscriptstyle{\Uparrow}}$
&
$\mathrmbf{List}\xrightarrow{\;\mathrmbfit{sort}\;}\mathrmbf{Set}$
&
$\overset{\int}{}$
\\\hline
$\mathrmbf{DOM}=
\mathrmbf{Dom}^{\!\scriptscriptstyle{\Uparrow}} 
%= \bigl({(\mbox{-})}^{\mathrm{op}}{\,\Downarrow\,}\mathrmbf{Dom}\bigr)
\xrightarrow[{(\text{-})}{\,\circ\,}\mathrmbfit{data}]
{\mathring{\mathrmbfit{data}}{\,\triangleq\,}\mathrmbfit{data}^{\!\scriptscriptstyle{\Uparrow}}\;}
%\bigl({(\mbox{-})}^{\mathrm{op}}{\,\Downarrow\,}\mathrmbf{Cls}\bigr) = 
\mathrmbf{Cls}^{\!\scriptscriptstyle{\Uparrow}}$
&
$\mathrmbf{Dom}\xrightarrow{\mathrmbfit{data}}\mathrmbf{Cls}$ 
&
$\overset{\int}{}$
\\
$\mathrmbf{DOM}=
\mathrmbf{Dom}^{\!\scriptscriptstyle{\Uparrow}} 
%= \bigl({(\mbox{-})}^{\mathrm{op}}{\,\Downarrow\,}\mathrmbf{Dom}\bigr)
\xrightarrow[{(\text{-})}{\,\circ\,}\mathrmbfit{sign}]
{\mathring{\mathrmbfit{sign}}{\,\triangleq\,}\mathrmbfit{sign}^{\!\scriptscriptstyle{\Uparrow}}\;}
%\bigl({(\mbox{-})}^{\mathrm{op}}{\,\Downarrow\,}\mathrmbf{Cls}\bigr) = 
\mathrmbf{List}^{\!\scriptscriptstyle{\Uparrow}}$
&
$\mathrmbf{Dom}\xrightarrow{\mathrmbfit{sign}}\mathrmbf{List}$ 
&
$\overset{\downarrow}{}$
\\\hline
$\mathrmbf{DB}^{\mathrm{op}}
%= \bigl(\mathrmbf{Tbl}^{\!\scriptscriptstyle{\Downarrow}}\bigr)^{\mathrm{op}} 
{\,\cong\,}
(\mathrmbf{Tbl}^{\mathrm{op}})^{\!\scriptscriptstyle{\Uparrow}}
\xrightarrow
[{(\text{-})^{\mathrm{op}}}{\,\circ\,}\mathrmbfit{dom}]
{\mathring{\mathrmbfit{dom}}{\,\triangleq\,}\mathrmbfit{dom}^{\!\scriptscriptstyle{\Uparrow}}}
\mathrmbf{Dom}^{\!\scriptscriptstyle{\Uparrow}}
=\mathrmbf{DOM}$
&
$\mathrmbf{Tbl}^{\mathrm{op}}\xrightarrow{\;\mathrmbfit{dom}\;}\mathrmbf{Dom}$
&
$\overset{\int}{}$
\\
$\mathrmbf{DB}^{\mathrm{op}}
%= \bigl(\mathrmbf{Tbl}^{\!\scriptscriptstyle{\Downarrow}}\bigr)^{\mathrm{op}} 
{\,\cong\,}
(\mathrmbf{Tbl}^{\mathrm{op}})^{\!\scriptscriptstyle{\Uparrow}}
\xrightarrow
[{(\text{-})^{\mathrm{op}}}{\,\circ\,}\mathrmbfit{sign}]
{\mathring{\mathrmbfit{sign}}{\,\triangleq\,}\mathrmbfit{sign}^{\!\scriptscriptstyle{\Uparrow}}}
\mathrmbf{List}^{\!\scriptscriptstyle{\Uparrow}}
=\mathrmbf{LIST}$
&
$\mathrmbf{Tbl}^{\mathrm{op}}\xrightarrow{\;\mathrmbfit{sign}\;}\mathrmbf{List}$
&
$\overset{\downarrow}{}$
\\
$\mathrmbf{DB}^{\mathrm{op}}
%=\bigl(\mathrmbf{Tbl}^{\!\scriptscriptstyle{\Downarrow}}\bigr)^{\mathrm{op}} 
\cong \bigl(\mathrmbf{Tbl}^{\mathrm{op}}\bigr)^{\!\scriptscriptstyle{\Uparrow}}
\xrightarrow[{(\text{-})^{\mathrm{op}}}{\,\circ\,}\mathrmbfit{data}]
{\mathring{\mathrmbfit{data}}{\,\triangleq\,}\mathrmbfit{data}^{\!\scriptscriptstyle{\Uparrow}}}
\mathrmbf{Cls}^{\!\scriptscriptstyle{\Uparrow}}$
%\cong \bigl((\mathrmbf{Cls}^{\mathrm{op}})^{\!\scriptscriptstyle{\Downarrow}}\bigr)^{\mathrm{op}}
&
$\mathrmbf{Tbl}^{\mathrm{op}}\xrightarrow{\!\mathrmbfit{data}\;}\mathrmbf{Cls}$
&
$\overset{\int}{}$
\\\hline
\end{tabular}}}}
\end{center}
\caption{Projection Passages}
\label{defs:proj:pass}
\end{table}
%
%Recall that the limit operation is a passage
%${\left(\mathrmbf{Cat}{\,\Downarrow\,}\mathrmbf{Tbl}\right)}^{\mathrm{op}} \xrightarrow{\mathrmbfit{lim}} %\mathrmbf{Tbl}$.
%
%\begin{definition}
%The join passage is defined to be the composition 
%{\footnotesize\[
%\mathrmbfit{join} = \mathrmbfit{dgm}^{\mathrm{op}} \circ \mathrmbfit{lim} :
%{\mathrmbf{Db}}^{\mathrm{op}} \rightarrow \mathrmbf{Tbl}.
%\]\normalsize}
%\end{definition}
%
%Hence,
%if a relational database 
%$\mathcal{D} = {\langle{\mathcal{S},\mathcal{E},\mathrmbfit{K},\tau}\rangle}$
%is regarded to be a diagram of tables
%$\mathrmbfit{T} : \mathrmbf{R}^{\mathrm{op}} \rightarrow (\mathrmbf{Set}{\downarrow}\mathrmbfit{tup}_{\mathcal{E}})$,
%the join of $\mathcal{D}$ is the limit of 
%$\mathrmbfit{T} 
%= \{ \mathrmbfit{T}(r) \mid r \in \mathrmbf{R} \}
%= \{ {\langle{X_{r},I_{r},s_{r},\mathcal{E},K_{r},t_{r}}\rangle} \mid r \in \mathrmbf{R} \}$,
%the join of all the relational tables.
%It covariantly maps a database morphism to a table morphism between the target and source joins.
%The following fact helps compute the database join.
%
Hence,
we have various commuting diagrams as in Prop.~\ref{prop:passage:lim:colim}.
%In particular, we have Cor.~\ref{cor:db:join:sch:ref}.

\newpage

%%%%%%%%%%%%%%%%%%%%%%%%%%%%%%%%%%%%%%%%%%%%%%%%%%%%%%%%%%%%%%%%%%%%%%%%%%%%%%%%%%%%%%%%%%
%\mbox{}\newline\rule{120pt}{1pt}{\fbox{\textbf{ Work Zone }}}\rule{120pt}{1pt}\newline
%%%%%%%%%%%%%%%%%%%%%%%%%%%%%%%%%%%%%%%%%%%%%%%%%%%%%%%%%%%%%%%%%%%%%%%%%%%%%%%%%%%%%%%%%%

%
\begin{corollary}\label{cor:db:join:sch:ref}\mbox{}
%\begin{itemize}
%\item 
The signed domain of the join (limit) of a database is 
the sum (colimit) of the underlying schemed domain.
%\item 
The signature of the colimit of a schemed domain is
the sum (colimit) of the underlying schema.
%\item 
Hence,
the signature of the join of a database is the
reference signature (colimit) of the underlying schema. 
%\end{itemize}
%
\begin{center}
{{\begin{tabular}{c}
\setlength{\unitlength}{0.5pt}
\begin{picture}(360,230)(-180,-115)
\put(50,100){\makebox(0,0)[r]{\footnotesize{$\mathrmbf{DB} = 
%((\mathrmbf{Tbl}^{\mathrm{op}})^{\!\scriptscriptstyle{\Uparrow}})^{\mathrm{op}} {\,\cong\,}
%\mathrmbf{Tbl}^{\!\scriptscriptstyle{\Downarrow}} \cong 
\bigl({(\mbox{-})}{\,\Downarrow\,}\mathrmbf{Tbl}\bigr)$}}}
\put(70,0){\makebox(0,0)[r]{\footnotesize{$\mathrmbf{DOM}^{\mathrm{op}} = 
%(\mathrmbf{Dom}^{\!\scriptscriptstyle{\Uparrow}})^{\mathrm{op}} 
%(\mathrmbf{Dom}^{\mathrm{op}})^{\!\scriptscriptstyle{\Downarrow}}  = 
\bigl({(\mbox{-})}{\,\Uparrow\,}\mathrmbf{Dom}\bigr)^{\mathrm{op}}$}}}
\put(70,-100){\makebox(0,0)[r]{\footnotesize{$\mathrmbf{LIST}^{\mathrm{op}} = 
%(\mathrmbf{List}^{\!\mathrm{op}})^{\scriptscriptstyle{\Downarrow}} {\,\cong\,}
%(\mathrmbf{List}^{\!\scriptscriptstyle{\Uparrow}})^{\mathrm{op}} = 
\bigl({(\mbox{-})}{\,\Uparrow\,}\mathrmbf{List}\bigr)^{\mathrm{op}}$}}}
\put(160,100){\makebox(0,0){\footnotesize{$\mathrmbf{Tbl}$}}}
\put(165,0){\makebox(0,0){\footnotesize{$\mathrmbf{Dom}^{\mathrm{op}}$}}}
\put(165,-100){\makebox(0,0){\footnotesize{$\mathrmbf{List}^{\mathrm{op}}$}}}
\put(90,110){\makebox(0,0){\scriptsize{$\mathrmbfit{lim}$}}}
\put(100,-10){\makebox(0,0){\scriptsize{$\mathrmbfit{colim}^{\mathrm{op}}$}}}
\put(100,-110){\makebox(0,0){\scriptsize{$\mathrmbfit{colim}^{\mathrm{op}}$}}}
%\put(40,50){\makebox(0,0)[r]{\scriptsize{$(\mathrmbfit{dom}^{\!\scriptscriptstyle{\Uparrow}})^{\mathrm{op}}
\put(40,50){\makebox(0,0)[r]{\scriptsize{$\mathring{\mathrmbfit{dom}}^{\mathrm{op}} = 
%(\mathrmbfit{dom}^{\mathrm{op}})^{\!\scriptscriptstyle{\Downarrow}} \cong 
\bigl({(\mbox{-})}{\,\Downarrow\,}\mathrmbfit{dom}^{\mathrm{op}}\bigr)$}}}
\put(165,50){\makebox(0,0)[l]{\scriptsize{$\mathrmbfit{dom}^{\mathrm{op}}$}}}
\put(-165,0){\makebox(0,0)[r]{\scriptsize{$\mathring{\mathrmbfit{sch}}^{\mathrm{op}} = 
%(\mathrmbfit{sch}^{\mathrm{op}})^{\!\scriptscriptstyle{\Downarrow}}{\,\cong\,}
\bigl({(\mbox{-})}{\,\Downarrow\,}\mathrmbfit{sign}\bigr)^{\mathrm{op}}$}}}
\put(40,-50){\makebox(0,0)[r]{\scriptsize{$\mathring{\mathrmbfit{sign}}^{\mathrm{op}} = 
%(\mathrmbfit{sign}^{\mathrm{op}})^{\!\scriptscriptstyle{\Downarrow}}{\,\cong\,}
\bigl({(\mbox{-})}{\,\Uparrow\,}\mathrmbfit{sign}\bigr)^{\mathrm{op}}$}}}
\put(223,0){\makebox(0,0)[l]{\scriptsize{$\mathrmbfit{sign}^{\mathrm{op}}$}}}
\put(165,-50){\makebox(0,0)[l]{\scriptsize{$\mathrmbfit{sign}^{\mathrm{op}}$}}}
\put(65,100){\vector(1,0){60}}
\put(65,0){\vector(1,0){60}}
\put(65,-100){\vector(1,0){60}}
\put(0,85){\vector(0,-1){70}}
\put(160,85){\vector(0,-1){70}}
\put(0,-15){\vector(0,-1){70}}
\put(160,-15){\vector(0,-1){70}}
\put(-160,85){\line(0,-1){170}}
\put(-145,85){\oval(30,30)[tl]}
\put(-145,-85){\oval(30,30)[bl]}
\put(-140,-100){\vector(1,0){0}}
\put(215,85){\line(0,-1){170}}
\put(200,85){\oval(30,30)[tr]}
\put(200,-85){\oval(30,30)[br]}
\put(195,-100){\vector(-1,0){0}}
\end{picture}
\end{tabular}}}
\end{center}
%
%If $\mathcal{S}=\mathring{\mathrmbfit{sch}}(\mathcal{D})$ is the schema of database $\mathcal{D}$
%and table $\mathrmbfit{T} = \mathrmbfit{lim}(\mathcal{D})$ is the join of $\mathcal{D}$, 
For any database $\mathcal{D}$
with schema $\mathcal{S}=\mathring{\mathrmbfit{sch}}(\mathcal{D})$
and join table $\mathrmbfit{T} = \mathrmbfit{lim}(\mathcal{D})$,
%then 
the signature of $\mathrmbfit{T}$ is the reference signature of $\mathcal{S}$:
$\mathrmbfit{sign}(\mathrmbfit{T})
=\mathrmbfit{colim}(\mathcal{S})$.
%$\mathrmbfit{sch}(\mathrmbfit{lim}(\mathcal{D}))
%=\mathrmbfit{colim}(\mathring{\mathrmbfit{sch}}(\mathcal{D}))$
\end{corollary}
\begin{proof}
This is an instance of the result in Prop.~\ref{prop:passage:lim:colim},
since 
%(discussion above)
$\mathrmbf{Tbl}$ is a complete context,
$\mathrmbf{Dom}$ and $\mathrmbf{List}$ are cocomplete contexts, 
and the passages
$\mathrmbf{Tbl}^{\mathrm{op}}\xrightarrow{\;\mathrmbfit{dom}}\mathrmbf{Dom}$ and
$\mathrmbf{Dom}\xrightarrow{\;\mathrmbfit{sign}\;}\mathrmbf{List}$ 
are cocontinuous.
\hfill\rule{5pt}{5pt}
\end{proof}
%
%As we have shown,
%for any entity classification $\mathcal{E} = {\langle{X,Y,\models_{\mathcal{E}}}\rangle}$,
%the context of $\mathcal{E}$-tables $\mathrmbf{Cat}(\mathcal{E})$ is complete.
%Hence,
%for any database schema $\mathcal{S}$
%the join of arbitrary $\mathcal{S}$-databases 
%can be constructed by using only 
%the join of the empty database (the terminal $\mathcal{E}$-table) and 
%the join of $\mathcal{E}$-databases with binary span $X$-schemas 
%(two $\mathcal{E}$-tables connected through a third).
%
%%%%%%%%%%%%%%%%%%%%%%%%%%%%%%%%%%%%%%%%%%%%%%%%%%%%%%%%%%%%%%%%%%%%%%%%%%%%%%%%%%%%%%%%%%
%\newpage
%\subsubsection{Continuity/Cocontinuity.}\label{sub:sub:sec:cont:cocont}
%%%%%%%%%%%%%%%%%%%%%%%%%%%%%%%%%%%%%%%%%%%%%%%%%%%%%%%%%%%%%%%%%%%%%%%%%%%%%%%%%%%%%%%%%%
%

%%%%%%%%%%%%%%%%%%%%%%%%%%%%%%%%%%%%%%%%%%%%%%%%%%%%%%%%%%%%%%%%%%%%%%%%%%%%%%%%%%%%%%%%%%
%\mbox{}\newline\rule{120pt}{1pt}{\fbox{\textbf{ Work Zone }}}\rule{120pt}{1pt}\newline
%%%%%%%%%%%%%%%%%%%%%%%%%%%%%%%%%%%%%%%%%%%%%%%%%%%%%%%%%%%%%%%%%%%%%%%%%%%%%%%%%%%%%%%%%%

%
\begin{table}
\begin{center}
{\fbox{\footnotesize{\begin{tabular}{r@{\hspace{20pt}}l}
$\mathrmbf{LIST}
= \mathrmbf{List}^{\!\scriptscriptstyle{\Uparrow}}
\xrightarrow{\;\mathrmbfit{lim},\mathrmbfit{colim}\;}\mathrmbf{List}$
&
\textit{schemas}
\vspace{.4mm}
\\\hline
%$\mathrmbf{CLS}
%= \mathrmbf{Cls}^{\!\scriptscriptstyle{\Uparrow}}
%\xrightarrow{\;\mathrmbfit{colim}\;} \mathrmbf{Cls}$
%&
%\textit{type domain diagrams}
%\\
$\mathrmbf{DOM}
= \mathrmbf{Dom}^{\!\scriptscriptstyle{\Uparrow}}
\xrightarrow{\;\mathrmbfit{lim},\mathrmbfit{colim}\;}\mathrmbf{Dom}$
&
%\textit{schemed domains}
\\
%\fbox{check}
$\mathring{\mathrmbf{Dom}} = \int_\mathrmit{data}\!\hat{\mathrmbfit{dom}}
\xrightarrow{\;\mathrmbfit{lim},\mathrmbfit{colim}\;} \mathrmbf{Dom}$
&
\textit{schemed domains}
\vspace{.8mm}
\\\hline
$\mathrmbf{DB}
= \mathrmbf{Tbl}^{\scriptscriptstyle{\Downarrow}}
\xrightarrow{\;\mathrmbfit{lim},\mathrmbfit{colim}\;}\mathrmbf{Tbl}$
&
%\textit{databases}
\\
%\fbox{check}
$
\mathrmbf{Db} = \int_\mathrmit{data}\!\hat{\mathrmbfit{db}}
\xrightarrow{\;\mathrmbfit{lim},\mathrmbfit{colim}\;} 
\mathrmbf{Tbl}$
&
\textit{databases}
\end{tabular}}}}
\end{center}
\caption{Lim (Colim) Passages}
\label{tbl:colim:lim:pass}
\end{table}
%

%%%%%%%%%%%%%%%%%%%%%%%%%%%%%%%%%%%%%%%%%%%%%%%%%%%%%%%%%%%%%%%%%%%%%%%%%%%%%%%%%%%%%%%%%%
%%%%%%%%%%%%%%%%%%%%%%%%%%%%%%%%%%%%%%%%%%%%%%%%%%%%%%%%%%%%%%%%%%%%%%%%%%%%%%%%%%%%%%%%%%
%\newpage
\subsubsection{Bridges}
\label{sub:sec:bridges}
%%%%%%%%%%%%%%%%%%%%%%%%%%%%%%%%%%%%%%%%%%%%%%%%%%%%%%%%%%%%%%%%%%%%%%%%%%%%%%%%%%%%%%%%%%
%%%%%%%%%%%%%%%%%%%%%%%%%%%%%%%%%%%%%%%%%%%%%%%%%%%%%%%%%%%%%%%%%%%%%%%%%%%%%%%%%%%%%%%%%%

Tbl.\,\ref{bridge:descr}
list the bridges defined and used in this paper
and
Tbl.\,\ref{adjoints:composites}
lists the bridge adjoints and bridge composites.
\begin{table}
\begin{center}
{\fbox{\footnotesize{\setlength{\extrarowheight}{1.6pt}
\begin{tabular}[t]
{|
l@{\hspace{16pt}}l@{\hspace{16pt}}l@{\hspace{16pt}}l@{\hspace{16pt}}l
|}
\hline
$\mathrmbf{DOM}$
&
\S\,\ref{sub:sec:sch:dom:gen}
%\textit{schemed domain}
&
Fig.\,\ref{fig:sch:dom:mor}
%$\varsigma$
&
$\varsigma :
\mathrmbfit{Q}_{2}\Rightarrow\mathrmbfit{R}{\;\circ\;}\mathrmbfit{Q}_{1}$
&
%\textit{projection}
\comment{
%%%%%%%%%%%%%%%%%%%%%%%%%%%%%%%%%%%%%%%%%%%%%%%%%%%%%%%%%%%%%%%%%%%%%%%%%%%%%%%%
\\\cline{3-5}
%%%%%%%%%%%%%%%%%%%%%%%%%%%%%%%%%%%%%%%%%%%%%%%%%%%%%%%%%%%%%%%%%%%%%%%%%%%%%%%%
$\mathrmbf{CLS}$
&
\S\,\ref{sub:sub:sec:dom:proj}
&
&
$\gamma :
\mathrmbfit{C}_{2}\Rightarrow\mathrmbfit{R} \circ \mathrmbfit{C}_{1}$
&
}
%%%%%%%%%%%%%%%%%%%%%%%%%%%%%%%%%%%%%%%%%%%%%%%%%%%%%%%%%%%%%%%%%%%%%%%%%%%%%%%%
\\\hline
%\\\cline{2-5}
%%%%%%%%%%%%%%%%%%%%%%%%%%%%%%%%%%%%%%%%%%%%%%%%%%%%%%%%%%%%%%%%%%%%%%%%%%%%%%%%
$\mathrmbf{Dom}(\mathcal{A})$
&
\S\,\ref{sub:sub:sec:dom:typ:dom:lower}
&
Fig.\,\ref{fig:sch:dom:mor:A}
&
$\varphi : \mathrmbfit{S}_{2}\Rightarrow\mathrmbfit{R} \circ \mathrmbfit{S}_{1}$
&
%%%%%%%%%%%%%%%%%%%%%%%%%%%%%%%%%%%%%%%%%%%%%%%%%%%%%%%%%%%%%%%%%%%%%%%%%%%%%%%%
\\\cline{2-5}
%%%%%%%%%%%%%%%%%%%%%%%%%%%%%%%%%%%%%%%%%%%%%%%%%%%%%%%%%%%%%%%%%%%%%%%%%%%%%%%%
$\mathrmbf{Dom}$
%&
%\textit{schemed domain}
&
\S\,\ref{sub:sec:rel:sch:typ:dom} 
&
Fig.\,\ref{fig:sch:dom:mor:typ:dom}
%& 
%$\hat{\varphi}$
&
$\acute{\phi} : 
\mathrmbfit{S}_{2}
\Leftarrow
\mathrmbfit{R}{\,\circ\,}\mathrmbfit{S}_{1}
{\,\circ\,}
\acute{\mathrmbfit{dom}}_{{\langle{f,g}\rangle}}$
&
%\textit{upper}
%%%%%%%%%%%%%%%%%%%%%%%%%%%%%%%%%%%%%%%%%%%%%%%%%%%%%%%%%%%%%%%%%%%%%%%%%%%%%%%%
\\\cline{4-5}
%%%%%%%%%%%%%%%%%%%%%%%%%%%%%%%%%%%%%%%%%%%%%%%%%%%%%%%%%%%%%%%%%%%%%%%%%%%%%%%%
%$\mathrmbf{Dom}$
&
& 
%$\hat{\varphi}$
&
$\grave{\phi} : 
\mathrmbfit{S}_{2}{\,\circ\,}\grave{\mathrmbfit{dom}}_{{\langle{f,g}\rangle}}
\Leftarrow
\mathrmbfit{R}{\,\circ\,}\mathrmbfit{S}_{1}$
&
%\textit{upper}

\comment{% not needed?
%%%%%%%%%%%%%%%%%%%%%%%%%%%%%%%%%%%%%%%%%%%%%%%%%%%%%%%%%%%%%%%%%%%%%%%%%%%%%%%%
\\\cline{3-5}
%%%%%%%%%%%%%%%%%%%%%%%%%%%%%%%%%%%%%%%%%%%%%%%%%%%%%%%%%%%%%%%%%%%%%%%%%%%%%%%%
$\mathrmbf{Set}$
&
%\textit{tuple}
\S\,\ref{sub:sec:append:tuple}
& 
%\textit{lower}
Fig.\,\ref{fig:tup:brid:adj}
%Tbl.\,\ref{tbl:tup:pass:tbl:db}
&
{\footnotesize{$
\acute{\tau}_{{\langle{f,g}\rangle}} :
{f^{\ast}}^{\mathrm{op}}{\circ\;}\mathrmbfit{tup}_{\mathcal{A}_{2}}
\Leftarrow
\mathrmbfit{tup}_{\mathcal{A}_{1}}
$}}
&
%\textbf{levo}
%%%%%%%%%%%%%%%%%%%%%%%%%%%%%%%%%%%%%%%%%%%%%%%%%%%%%%%%%%%%%%%%%%%%%%%%%%%%%%%%
\\\cline{4-5}
%%%%%%%%%%%%%%%%%%%%%%%%%%%%%%%%%%%%%%%%%%%%%%%%%%%%%%%%%%%%%%%%%%%%%%%%%%%%%%%%
&
& 
&
{\footnotesize{$
\grave{\tau}_{{\langle{f,g}\rangle}} :
\mathrmbfit{tup}_{\mathcal{A}_{2}}
\Leftarrow
{\scriptstyle\sum}_{f}^{\mathrm{op}}{\circ\;}\mathrmbfit{tup}_{\mathcal{A}_{1}}$}}
&
%\textbf{dextro}
}% not needed?

%%%%%%%%%%%%%%%%%%%%%%%%%%%%%%%%%%%%%%%%%%%%%%%%%%%%%%%%%%%%%%%%%%%%%%%%%%%%%%%%
%%%%%%%%%%%%%%%%%%%%%%%%%%%%%%%%%%%%%%%%%%%%%%%%%%%%%%%%%%%%%%%%%%%%%%%%%%%%%%%%
\\\hline\hline
%%%%%%%%%%%%%%%%%%%%%%%%%%%%%%%%%%%%%%%%%%%%%%%%%%%%%%%%%%%%%%%%%%%%%%%%%%%%%%%%
%%%%%%%%%%%%%%%%%%%%%%%%%%%%%%%%%%%%%%%%%%%%%%%%%%%%%%%%%%%%%%%%%%%%%%%%%%%%%%%%
$\mathrmbf{DB}$
&
\S\,\ref{sub:sec:rel:db:gen} 
%\textit{database}
&
Fig.\,\ref{fig:db:mor:gen}
&
$\xi :
\mathrmbfit{T}_{2}\Leftarrow\mathrmbfit{R}^{\mathrm{op}}{\circ}\mathrmbfit{T}_{1}$.
&
%\textit{original}
%%%%%%%%%%%%%%%%%%%%%%%%%%%%%%%%%%%%%%%%%%%%%%%%%%%%%%%%%%%%%%%%%%%%%%%%%%%%%%%%
\\\cline{3-5}
%%%%%%%%%%%%%%%%%%%%%%%%%%%%%%%%%%%%%%%%%%%%%%%%%%%%%%%%%%%%%%%%%%%%%%%%%%%%%%%%
$\mathrmbf{SET}$
&
%\textit{key}
&
%$\kappa$
&
$\kappa :
\mathrmbfit{K}_{2}\Leftarrow\mathrmbfit{R}^{\mathrm{op}}{\circ\;}\mathrmbfit{K}_{1}$
&
%\textit{projection}
%%%%%%%%%%%%%%%%%%%%%%%%%%%%%%%%%%%%%%%%%%%%%%%%%%%%%%%%%%%%%%%%%%%%%%%%%%%%%%%%
\\\hline
%\\\cline{2-5}
%%%%%%%%%%%%%%%%%%%%%%%%%%%%%%%%%%%%%%%%%%%%%%%%%%%%%%%%%%%%%%%%%%%%%%%%%%%%%%%%
$\mathrmbf{Db}(\mathcal{A})$
&
\S\,\ref{sub:sub:sec:rel:db:typ:dom:lower}
&
Fig.\,\ref{fig:db:A:mor:adj}
&
$\psi :
\mathrmbfit{T}_{2}\Leftarrow\mathrmbfit{R}^{\mathrm{op}}{\circ}\mathrmbfit{T}_{1}$
&
%%%%%%%%%%%%%%%%%%%%%%%%%%%%%%%%%%%%%%%%%%%%%%%%%%%%%%%%%%%%%%%%%%%%%%%%%%%%%%%%
\\\cline{2-5}
%%%%%%%%%%%%%%%%%%%%%%%%%%%%%%%%%%%%%%%%%%%%%%%%%%%%%%%%%%%%%%%%%%%%%%%%%%%%%%%%
$\mathrmbf{Db}$
&
\S\,\ref{sub:sub:sec:rel:db:typ:dom:upper}
& 
Fig.\,\ref{fig:db:mor:Db}
&
$\acute{\psi} :
\mathrmbfit{T}_{2}\Leftarrow\mathrmbfit{R}^{\mathrm{op}}\circ\mathrmbfit{T}_{1}\circ\acute{\mathrmbfit{tbl}}_{{\langle{f,g}\rangle}}$
&
%\textit{upper}
%%%%%%%%%%%%%%%%%%%%%%%%%%%%%%%%%%%%%%%%%%%%%%%%%%%%%%%%%%%%%%%%%%%%%%%%%%%%%%%%
\\\cline{4-5}
%%%%%%%%%%%%%%%%%%%%%%%%%%%%%%%%%%%%%%%%%%%%%%%%%%%%%%%%%%%%%%%%%%%%%%%%%%%%%%%%
%$\mathrmbf{Db}$
&
& 
%$\hat{\psi}$
&
$\grave{\psi} : 
\mathrmbfit{T}_{2}\circ\grave{\mathrmbfit{tbl}}_{{\langle{f,g}\rangle}}\Leftarrow\mathrmbfit{R}^{\mathrm{op}}\circ\mathrmbfit{T}_{1}$
&
%\textit{upper}
%%%%%%%%%%%%%%%%%%%%%%%%%%%%%%%%%%%%%%%%%%%%%%%%%%%%%%%%%%%%%%%%%%%%%%%%%%%%%%%%
\\\cline{3-5}
%%%%%%%%%%%%%%%%%%%%%%%%%%%%%%%%%%%%%%%%%%%%%%%%%%%%%%%%%%%%%%%%%%%%%%%%%%%%%%%%
%
%%%%%%%%%%%%%%%%%%%%%%%%%%%%%%%%%%%%%%%%%%%%%%%%%%%%%%%%%%%%%%%%%%%%%%%%%%%%%%%%
%%%%%%%%%%%%%%%%%%%%%%%%%%%%%%%%%%%%%%%%%%%%%%%%%%%%%%%%%%%%%%%%%%%%%%%%%%%%%%%%
%\\\hline\hline
%%%%%%%%%%%%%%%%%%%%%%%%%%%%%%%%%%%%%%%%%%%%%%%%%%%%%%%%%%%%%%%%%%%%%%%%%%%%%%%%
%%%%%%%%%%%%%%%%%%%%%%%%%%%%%%%%%%%%%%%%%%%%%%%%%%%%%%%%%%%%%%%%%%%%%%%%%%%%%%%%
%$\mathrmbf{Tbl}$
&
%\textit{table inclusion}
%\S\,\ref{sub:sub:sec:rel:db:typ:dom:upper}
& 
%$\hat{\chi}$,$\hat{\tau}$
%\textit{upper}
%Tbl.\,\ref{fig:db:mor:Db},
Fig.\,\ref{fig:db:mor:Db}
%{fig:incl:bridge:tbl}
&
{\footnotesize{$
\acute{\chi}_{{\langle{f,g}\rangle}} :
\acute{\mathrmbfit{tbl}}_{{\langle{f,g}\rangle}}\circ\mathrmbfit{inc}_{\mathcal{A}_{2}}
\Leftarrow
\mathrmbfit{inc}_{\mathcal{A}_{1}}
$}}
&
%\textbf{levo}
%%%%%%%%%%%%%%%%%%%%%%%%%%%%%%%%%%%%%%%%%%%%%%%%%%%%%%%%%%%%%%%%%%%%%%%%%%%%%%%%
\\\cline{4-5}
%%%%%%%%%%%%%%%%%%%%%%%%%%%%%%%%%%%%%%%%%%%%%%%%%%%%%%%%%%%%%%%%%%%%%%%%%%%%%%%%
&
& 
&
{\footnotesize{$
\grave{\chi}_{{\langle{f,g}\rangle}} : 
\mathrmbfit{inc}_{\mathcal{A}_{2}}
\Leftarrow
\grave{\mathrmbfit{tbl}}_{{\langle{f,g}\rangle}}\circ\mathrmbfit{inc}_{\mathcal{A}_{1}}
$}}
&
%\textbf{dextro}
\\\hline
\end{tabular}}}}
\end{center}
\caption{Bridges}
\label{bridge:descr}
\end{table}
%

%$\alpha, \beta, \gamma, \delta, \epsilon, \varepsilon, \zeta, \eta, 
%\theta, \vartheta,\iota, \kappa, \lambda, \mu, \nu, \xi, o, 
%\pi, \varpi, \rho, \varrho, \sigma, \varsigma, \tau, \upsilon, \phi, \varphi, \chi, \psi, \omega$

%$\hat{\tau}_{{\langle{f,g}\rangle}}$ defined
%in \S\,2.4.2 of the Table paper

%\vspace{-180pt}

%%%%%%%%%%%%%%%%%%%%%%%%%%%%%%%%%%%%%%%%%%%%%%%%%%%%%%%%%%%%%%%%%%%%%%%%%%%%%%%%%%%%%%%%%%
%%%%%%%%%%%%%%%%%%%%%%%%%%%%%%%%%%%%%%%%%%%%%%%%%%%%%%%%%%%%%%%%%%%%%%%%%%%%%%%%%%%%%%%%%%
%\newpage
%\subsection{Adjoints.}\label{sub:sec:append:adjoints}
%%%%%%%%%%%%%%%%%%%%%%%%%%%%%%%%%%%%%%%%%%%%%%%%%%%%%%%%%%%%%%%%%%%%%%%%%%%%%%%%%%%%%%%%%%
%%%%%%%%%%%%%%%%%%%%%%%%%%%%%%%%%%%%%%%%%%%%%%%%%%%%%%%%%%%%%%%%%%%%%%%%%%%%%%%%%%%%%%%%%%

%
\begin{table}
\begin{center}
{\fbox{\footnotesize\setlength{\extrarowheight}{4pt}
$\begin{array}{|@{\hspace{5pt}}l@{\hspace{15pt}}l@{\hspace{5pt}}|}
\hline
\multicolumn{1}{|l}{\text{\bfseries levo}} & \multicolumn{1}{l|}{\text{\bfseries dextro}}
\\ \hline\hline
\acute{\varphi}:\mathrmbfit{S}_{2}\Rightarrow\mathrmbfit{R}\circ\mathrmbfit{S}_{1}\circ{f^{\ast}}
& 
\grave{\varphi}:\mathrmbfit{S}_{2}\circ{{\Sigma}_{f}}\Rightarrow\mathrmbfit{R}\circ\mathrmbfit{S}_{1}
\\
\acute{\varphi}=(\mathrmbfit{S}_{2}\circ\eta_{f})\bullet(\grave{\varphi}\circ{f^{\ast}})
& 
\grave{\varphi}=(\acute{\varphi}\circ{{\Sigma}_{f}})\bullet(\mathrmbfit{R}\circ\mathrmbfit{S}_{1}\circ\varepsilon_{f})
%\\
%\acute{\varphi}=(\acute{\psi}\circ\mathrmbfit{sign}_{\mathcal{A}_{2}}^{\mathrm{op}})^{\mathrm{op}}
%&
%\grave{\varphi}=(\grave{\psi}\circ\mathrmbfit{sign}_{\mathcal{A}_{1}}^{\mathrm{op}})^{\mathrm{op}}
\\ \hline
\grave{\iota}_{{\langle{f,g}\rangle}} :
{f^{\ast}}{\circ\;}\mathrmbfit{inc}_{\mathcal{A}_{2}}
\Rightarrow
\mathrmbfit{inc}_{\mathcal{A}_{1}} 
&
\acute{\iota}_{{\langle{f,g}\rangle}} :
\mathrmbfit{inc}_{\mathcal{A}_{2}}
\Rightarrow 
{{\Sigma}_{f}}{\circ\;}{f^{\ast}} 
\\
\grave{\iota}_{{\langle{f,g}\rangle}} = 
\bigl({f^{\ast}}{\;\circ\;}\acute{\iota}_{{\langle{f,g}\rangle}}\bigr)
{\;\bullet\;}
\bigl(\varepsilon_{{\langle{f,g}\rangle}}{\;\circ\;}\mathrmbfit{inc}_{\mathcal{A}_{1}}\bigr)
&
\acute{\iota}_{{\langle{f,g}\rangle}} = 
\bigl(\eta_{{\langle{f,g}\rangle}}{\;\circ\;}\mathrmbfit{inc}_{\mathcal{A}_{2}}\bigr)
{\;\bullet\;} 
\bigl({{\Sigma}_{f}}{\;\circ\;}\grave{\iota}_{{\langle{f,g}\rangle}}\bigr)
\\ \hline
\multicolumn{2}{|c|}{
\varsigma:\mathrmbfit{Q}_{2}\Rightarrow\mathrmbfit{R}{\,\circ\,}\mathrmbfit{Q}_{1}
}
\\
\multicolumn{2}{|c|}{
(\acute{\varphi} \circ \mathrmbfit{inc}_{\mathcal{A}_{2}}) 
\bullet 
(\mathrmbfit{R} \circ \mathrmbfit{S}_{1}  \circ \grave{\iota}_{{\langle{f,g}\rangle}})
= \varsigma = 
(\mathrmbfit{S}_{2} \circ \acute{\iota}_{{\langle{f,g}\rangle}})
\bullet 
(\grave{\varphi} \circ \mathrmbfit{inc}_{\mathcal{A}_{1}}) 
}
%
%%%%%%%%%%%%%%%%%%%%%%%%%%%%%%%%%%%%%%%%%%%%%%%%%%%%%%%%%%%%%%%%%%%%%%%%%%%%%%%%
\\ \hline\hline
%%%%%%%%%%%%%%%%%%%%%%%%%%%%%%%%%%%%%%%%%%%%%%%%%%%%%%%%%%%%%%%%%%%%%%%%%%%%%%%%
%
\acute{\psi}:\mathrmbfit{T}_{2}\Leftarrow\mathrmbfit{R}^{\mathrm{op}}\circ\mathrmbfit{T}_{1}\circ\acute{\mathrmbfit{tbl}}_{{\langle{f,g}\rangle}}
&
\grave{\psi}:\mathrmbfit{T}_{2}\circ\grave{\mathrmbfit{tbl}}_{{\langle{f,g}\rangle}}\Leftarrow\mathrmbfit{R}^{\mathrm{op}}\circ\mathrmbfit{T}_{1}
\\
\acute{\psi}=(\grave{\psi}\circ\acute{\mathrmbfit{tbl}}_{{\langle{f,g}\rangle}})\bullet(\mathrmbfit{T}_{2}\circ\varepsilon_{{\langle{f,g}\rangle}})
&
\grave{\psi}=(\mathrmbfit{R}^{\mathrm{op}}\circ\mathrmbfit{T}_{1})\circ\eta_{{\langle{f,g}\rangle}}\bullet(\acute{\psi}\circ\grave{\mathrmbfit{tbl}}_{{\langle{f,g}\rangle}})
\\ \hline
\acute{\chi}_{{\langle{f,g}\rangle}}:\acute{\mathrmbfit{tbl}}_{{\langle{f,g}\rangle}}\circ\mathrmbfit{inc}_{\mathcal{A}_{2}}\Leftarrow\mathrmbfit{inc}_{\mathcal{A}_{1}}
&
\grave{\chi}_{{\langle{f,g}\rangle}}:\mathrmbfit{inc}_{\mathcal{A}_{2}}\Leftarrow\grave{\mathrmbfit{tbl}}_{{\langle{f,g}\rangle}}\circ\mathrmbfit{inc}_{\mathcal{A}_{1}}
\\
\acute{\chi}_{{\langle{f,g}\rangle}}=(\eta_{{\langle{f,g}\rangle}}\circ\mathrmbfit{inc}_{\mathcal{A}_{1}})\bullet(\acute{\mathrmbfit{tbl}}_{{\langle{f,g}\rangle}}\circ\grave{\chi}_{{\langle{f,g}\rangle}})
&
\grave{\chi}_{{\langle{f,g}\rangle}}=(\grave{\mathrmbfit{tbl}}_{{\langle{f,g}\rangle}}\circ{\chi}_{{\langle{f,g}\rangle}})\bullet(\varepsilon_{{\langle{f,g}\rangle}}\circ\mathrmbfit{inc}_{\mathcal{A}_{2}})
\\ \hline
\comment{
\acute{\tau}_{{\langle{f,g}\rangle}}:(f^{\ast})^{\mathrm{op}} \circ \mathrmbfit{tup}_{\mathcal{A}_{2}}\Leftarrow\mathrmbfit{tup}_{\mathcal{A}_{1}} 
&
\grave{\tau}_{{\langle{f,g}\rangle}}:\mathrmbfit{tup}_{\mathcal{A}_{2}}\Leftarrow{\scriptstyle\sum}_{f}^{\mathrm{op}} \circ \mathrmbfit{tup}_{\mathcal{A}_{1}} 
\\
\acute{\tau}_{{\langle{f,g}\rangle}}=(\varepsilon_{f}^{\mathrm{op}} \circ \mathrmbfit{tup}_{\mathcal{A}_{1}})\bullet((f^{\ast})^{\mathrm{op}} \circ \grave{\tau}_{{\langle{f,g}\rangle}})
&
\grave{\tau}_{{\langle{f,g}\rangle}}=({\scriptstyle\sum}_{f}^{\mathrm{op}} \circ \acute{\tau}_{{\langle{f,g}\rangle}})\bullet(\eta_{f}^{\mathrm{op}} \circ \mathrmbfit{tup}_{\mathcal{A}_{2}})
\\ \hline
}
\multicolumn{2}{|c|}{
\xi:\mathrmbfit{T}_{2}\Leftarrow\mathrmbfit{R}^{\mathrm{op}}{\!\circ}\mathrmbfit{T}_{1}
}
\\
\multicolumn{2}{|c|}{
(\acute{\psi}\circ\mathrmbfit{inc}_{\mathcal{A}_{1}})
\bullet
(\mathrmbfit{T}_{2}\circ\acute{\chi}_{{\langle{f,g}\rangle}})
=
\xi
=
(\mathrmbfit{R}^{\mathrm{op}}{\!\circ}\mathrmbfit{T}_{1})
\circ
\grave{\chi}_{{\langle{f,g}\rangle}}
\bullet
(\grave{\psi}\circ\mathrmbfit{inc}_{\mathcal{A}_{2}})
}
%
%%%%%%%%%%%%%%%%%%%%%%%%%%%%%%%%%%%%%%%%%%%%%%%%%%%%%%%%%%%%%%%%%%%%%%%%%%%%%%%%
\\ \hline\hline
%%%%%%%%%%%%%%%%%%%%%%%%%%%%%%%%%%%%%%%%%%%%%%%%%%%%%%%%%%%%%%%%%%%%%%%%%%%%%%%%
%
\acute{\tau}_{{\langle{f,g}\rangle}}:(f^{\ast})^{\mathrm{op}} \circ \mathrmbfit{tup}_{\mathcal{A}_{2}}\Leftarrow\mathrmbfit{tup}_{\mathcal{A}_{1}} 
&
\grave{\tau}_{{\langle{f,g}\rangle}}:\mathrmbfit{tup}_{\mathcal{A}_{2}}\Leftarrow{\scriptstyle\sum}_{f}^{\mathrm{op}} \circ \mathrmbfit{tup}_{\mathcal{A}_{1}} 
\\
\acute{\tau}_{{\langle{f,g}\rangle}}=(\varepsilon_{f}^{\mathrm{op}} \circ \mathrmbfit{tup}_{\mathcal{A}_{1}})\bullet((f^{\ast})^{\mathrm{op}} \circ \grave{\tau}_{{\langle{f,g}\rangle}})
&
\grave{\tau}_{{\langle{f,g}\rangle}}=({\scriptstyle\sum}_{f}^{\mathrm{op}} \circ \acute{\tau}_{{\langle{f,g}\rangle}})\bullet(\eta_{f}^{\mathrm{op}} \circ \mathrmbfit{tup}_{\mathcal{A}_{2}})
%\\ \hline
%
%\\ \hline
%
%%%%%%%%%%%%%%%%%%%%%%%%%%%%%%%%%%%%%%%%%%%%%%%%%%%%%%%%%%%%%%%%%%%%%%%%%%%%%%%%
\\ \hline
%%%%%%%%%%%%%%%%%%%%%%%%%%%%%%%%%%%%%%%%%%%%%%%%%%%%%%%%%%%%%%%%%%%%%%%%%%%%%%%%
\multicolumn{2}{|c|}{
\kappa{\;\bullet\;}\tau_{2}
= 
(\mathrmbfit{R}^{\mathrm{op}}{\circ\;}\tau_{1}){\;\bullet\;}
(\hat{\varphi}^{\mathrm{op}}{\;\circ\;}\hat{\tau}_{{\langle{f,g}\rangle}})
=
(\mathrmbfit{R}^{\mathrm{op}}{\circ\;}\tau_{1}){\;\bullet\;}
{\mathring{\mathrmbfit{tup}}_{{\langle{f,g}\rangle}}(\hat{\varphi})}
}
\\ 
\multicolumn{2}{|c|}{
(\mathrmbfit{R}^{\mathrm{op}}{\circ\;}\mathrmbfit{S}_{1}^{\mathrm{op}}{\!\circ\;}\acute{\tau}_{{\langle{f,g}\rangle}})
{\;\bullet\;}
(\acute{\varphi}^{\mathrm{op}}{\!\circ\;}\mathrmbfit{tup}_{\mathcal{A}_{2}})
=
\hat{\varphi}^{\mathrm{op}}{\;\circ\;}\hat{\tau}_{{\langle{f,g}\rangle}}
%\mathring{\mathrmbfit{tup}}_{{\langle{f,g}\rangle}}(\hat{\varphi})
= 
(\grave{\varphi}^{\mathrm{op}}{\!\circ\;}
\mathrmbfit{tup}_{\mathcal{A}_{1}})
{\;\bullet\;}
(\mathrmbfit{S}_{2}^{\mathrm{op}}{\!\circ\;}\grave{\tau}_{{\langle{f,g}\rangle}})
}
\\ \hline
\end{array}$}}
\end{center}
\caption{Bridge Adjoints and Composites}
\label{adjoints:composites}
\end{table}

\end{document}